\documentclass{ws-rv9x6}
\usepackage{subfigure}   % required only when side-by-side / subfigures are used
\usepackage{ws-rv-thm}   % comment this line when `amsthm / theorem / ntheorem` package is used
\usepackage[square]{ws-rv-van}   % numbered citation & references (default)
\makeindex
\usepackage{dsfont}
\usepackage{amsmath}
\DeclareMathOperator{\Tr}{Tr}
\DeclareMathOperator{\Str}{Str}
\usepackage{graphicx}
\usepackage{mathrsfs}
\usepackage{axodraw2}
\def\signofmetric{0}
\definecolor{Red}{cmyk}{0,1,1,0}
\definecolor{Blue}{cmyk}{1,1,0,0}

\if0\signofmetric
\def\BDpos{}
\def\BDneg{-}
\def\BDplus{+}
\def\BDminus{-}
\fi

\if1\signofmetric
\def\BDpos{-}
\def\BDneg{}
\def\BDplus{-}
\def\BDminus{+}
\fi

\if2\signofmetric
\def\BDpos{{\color{Red}\oplus}}
\def\BDneg{{\color{Red}\ominus}}
\def\BDplus{\oplus}
\def\BDminus{\ominus}
\fi

\if3\signofmetric
\def\BDpos{{\color{Red}\ominus}}
\def\BDneg{{\color{Red}\oplus}}
\def\BDplus{\ominus}
\def\BDminus{\oplus}
\fi

\renewcommand{\theequation}{\arabic{section}.\arabic{equation}}
\newlength{\captsize}           \let\captsize=\footnotesize
\newlength{\captwidth}          
\newlength{\beforetableskip}    
\newcommand{\capt}[1]{\begin{minipage}{\captwidth}
              \let\normalsize=\captsize
              \caption[0]{#1}
              \end{minipage}\\ \vspace{\beforetableskip}}
\makeatletter
      \long\def\@makecaption#1#2{\vskip 10 \p@
      \setbox\@tempboxa\hbox{\textbf{#1:} #2}%     %% boldface
      \ifdim \wd\@tempboxa >\hsize
            \textbf{#1:} #2\par                 %% boldface
      \else
         \hbox to \hsize{\box\@tempboxa\hfil}%         %%no centering
      \fi}
\makeatother
\def\beq{\begin{equation}}
\def\eeq{\end{equation}}
\newenvironment{Eqnarray}%
     {\arraycolsep 0.14em\begin{eqnarray}}{\end{eqnarray}}
\def\beqa{\begin{Eqnarray}}
\def\eeqa{\end{Eqnarray}}
\def\bea{\begin{Eqnarray*}}
\def\eea{\end{Eqnarray*}}
\def\ifmath#1{\relax\ifmmode #1\else $#1$\fi}
\def\half{\ifmath{{\textstyle{\frac{1}{2}}}}}

\def\quarter{\ifmath{{\textstyle{\frac{1}{4}}}}}

\def\eq#1{eq.~(\ref{#1})}
\def\Eq#1{Eq.~(\ref{#1})}
\def\Eqs#1#2{Eqs.~(\ref{#1}) and (\ref{#2})}
\def\eqs#1#2{eqs.~(\ref{#1}) and (\ref{#2})}
\def\eqss#1#2#3{eqs.~(\ref{#1}), (\ref{#2}) and (\ref{#3})}
\def\eqst#1#2{eqs.~(\ref{#1})--(\ref{#2})}
\def\Eqst#1#2{Eqs.~(\ref{#1})--(\ref{#2})}

\def\mathbold#1{\boldsymbol{#1}}

\def\bra#1{\left\langle #1\right|}
\def\ket#1{\left| #1\right\rangle}

\def\vev#1{\left\langle #1\right\rangle}
\def\ls#1{\ifmath{_{\lower1.5pt\hbox{$\scriptstyle #1$}}}}
\def\lsup#1{^{\lower 6pt\hbox{$\scriptstyle#1$}}}
\def\llsup#1{^{\lower 3pt\hbox{$\scriptstyle#1$}}}

\def\nicefrac#1#2{\hbox{${\frac{#1}{#2}}$}}

\def\sb{s_\beta}

\def\hl{h^0}
\def\hh{H^0}
\def\ha{A^0}

\def\phm{\phantom{-}}

\def\centeron#1#2{{\setbox0=\hbox{#1}\setbox1=\hbox{#2}\ifdim
\wd1>\wd0\kern.5\wd1\kern-.5\wd0\fi
\copy0\kern-.5\wd0\kern-.5\wd1\copy1\ifdim\wd0>\wd1
\kern.5\wd0\kern-.5\wd1\fi}}
\def\ltap{\;\centeron{\raise.35ex\hbox{$<$}}{\lower.65ex\hbox{$\sim$}}\;}
\def\gtap{\;\centeron{\raise.35ex\hbox{$>$}}{\lower.65ex\hbox{$\sim$}}\;}
\def\gsim{\mathrel{\gtap}}
\def\lsim{\mathrel{\ltap}}
\def\metric{g}

\def\reversed#1{{#1}_r}

\def\newcdot{\kern.06em{\cdot}\kern.06em}
\def\sigmabar{\overline{\sigma}}
\def\slashchar#1{\setbox0=\hbox{$#1$}           % set a box for #1
   \dimen0=\wd0                                 % and get its size
   \setbox1=\hbox{/} \dimen1=\wd1               % get size of /
   \ifdim\dimen0>\dimen1                        % #1 is bigger
      \rlap{\hbox to \dimen0{\hfil/\hfil}}      % so center / in box
      #1                                        % and print #1
   \else                                        % / is bigger
      \rlap{\hbox to \dimen1{\hfil$#1$\hfil}}   % so center #1
      /                                         % and print /
   \fi}                                        %

\def\nn{\nonumber}

\renewcommand\body{
\setcounter{footnote}{0}
\def\thefootnote{\arabic{footnote}}
\def\@makefnmark{{$^{\rm \@thefnmark}$}}
}
\newcommand{\bi}{\begin{itemize}}
\newcommand{\ei}{\end{itemize}}
\newcommand{\ben}{\begin{enumerate}}
\newcommand{\een}{\end{enumerate}}

\def\beqa{\begin{Eqnarray}}
\def\eeqa{\end{Eqnarray}}
\def\beq{\begin{equation}}
\def\eeq{\end{equation}}
\def\mw{m_W}
\def\mz{m_Z}

\def\vev#1{\langle #1 \rangle}
\def\ubar{\bar u}
\def\vbar{\bar v}
\def\Psibar{\overline{\Psi}}
\def\T{{\mathsf T}}

\def\cpvx{{\rm CP}{\setbox0=\hbox{{\rm CP}}\kern-0.7\wd0{\bf \slash}%
        \setbox1=\hbox{\bf \slash}\kern0.7\wd0\kern-\wd1}}

\def\RPVx{{\rm R_p}{\setbox0=\hbox{${\rm R_p}$}\kern-0.7\wd0{\bf \slash}%
        \setbox1=\hbox{\bf \slash}\kern0.7\wd0\kern-\wd1}}

\def\LVx{{\rm L}{\setbox0=\hbox{{\rm L}}\kern-0.7\wd0{\bf \slash}%
        \setbox1=\hbox{\bf \slash}\kern0.7\wd0\kern-\wd1}}

\def\phm{\phantom{-}}

\def\hh{H}
\def\hl{h}
\def\ha{A}

\def\mhh{m_{H}}
\def\mhl{m_{h}}
\def\mha{m_{A}}
\def\mhpm{m_{H^\pm}}

\def\sb  {s_{\beta}}

\def\cosbma{\cos(\beta-\alpha)}
\def\sinbma{\sin(\beta-\alpha)}

\def\nn{\nonumber}

\def\ifmath#1{\relax\ifmmode #1\else $#1$\fi}
\def\half{\ifmath{\tfrac12}}
\def\quarter{\ifmath{\tfrac14}}
\def\eighth{\ifmath{\tfrac18}}
\def\nicefrac#1#2{\hbox{$\frac{#1}{#2}$}}

\def\Buildrel#1\under#2{\mathrel{\mathop{#2}\limits_{#1}}}
\renewcommand{\thefootnote}{\arabic{footnote}}

\def\lsim{{~\raise.15em\hbox{$<$}\kern-.85em\lower.35em\hbox{$\sim$}~}}
\def\gsim{{~\raise.15em\hbox{$>$}\kern-.85em\lower.35em\hbox{$\sim$}~}}
\def\ls#1{\ifmath{_{\lower1.5pt\hbox{$\scriptstyle #1$}}}}
\def\lsup#1{^{\lower 2pt\hbox{$\scriptstyle#1$}}}

\def\wh{\widehat}
\def\wt{\widetilde}

\def\chini{\widetilde\chi^0_i}

\def\chipma{\widetilde\chi^\pm_1}
\def\chipma{\widetilde\chi^\pm_1}
\def\chipmb{\widetilde\chi^\pm_2}
\def\chipa{\widetilde\chi^+_1}
\def\chipb{\widetilde\chi^+_2}

\def\mchipa{M_{\widetilde\chi^+_1}}
\def\mchipb{M_{\widetilde\chi^+_2}}
\def\mchina{M_{\widetilde\chi^0_1}}
\def\mchinb{M_{\widetilde\chi^0_2}}
\def\mchinc{M_{\widetilde\chi^0_3}}
\def\mchind{M_{\widetilde\chi^0_4}}

\def\ha{A^0}
\def\mhpm{m_{H^\pm}}

\def\mhh{m_{H}}
\def\mhl{m_{h}}
\def\mha{m_{A}}
\def\mz{m_Z}
\def\mw{m_W}

\def\mstopa{M_{\widetilde t_1}}
\def\mstopb{M_{\widetilde t_2}}

\def\msusyy{M_{\rm S}^2}

\def\deltaxi{\delta_{\xi}}
\def\deltaeta{\delta_{\eta}}
\def\thetabar{\theta^\dagger}
\def\etabar{\eta^\dagger}
\def\xibar{\xi^\dagger}
\def\dalpha{\dot\alpha}
\def\dbeta{\dot\beta}
\def\dgamma{\dot\gamma}

\def\bra#1{\left\langle #1\right|}
\def\ket#1{\left| #1\right\rangle}
\def\Dbar{\overline{D}}

\def\cpvx{{\rm CP}{\setbox0=\hbox{{\rm CP}}\kern-0.7\wd0{\bf \slash}%
        \setbox1=\hbox{\bf \slash}\kern0.7\wd0\kern-\wd1}}

\def\RPVx{{\rm R_p}{\setbox0=\hbox{${\rm R_p}$}\kern-0.7\wd0{\bf \slash}%
        \setbox1=\hbox{\bf \slash}\kern0.7\wd0\kern-\wd1}}

\def\sinb{\sin\beta}
\def\cosb{\cos\beta}

\def\hh{H}
\def\hl{h}
\def\ha{A}

\def\mhh{m_{H}}
\def\mhl{m_{h}}
\def\mha{m_{A}}
\def\mhpm{m_{H^\pm}}

\newcommand{\of}[1]{\left( #1 \right)}
\newcommand{\sqof}[1]{\left[ #1 \right]}

\newcommand{\W}{\mathcal{W}}
\newif\ifarxivsubmission

\arxivsubmissiontrue
\ifarxivsubmission
  \crop[off]{}
   \usepackage{fancyhdr}
   \fancypagestyle{plain}{%
   	\fancyhf{}%
	\fancyfoot[C]{\thepage}%
	\rhead{SCIPP 17/11}%
	}
\fi

\begin{document}

\ifarxivsubmission
   \chapter*{Supersymmetric Theory and Models}%{TASI Lectures on  Supersymmetry}
   \vspace{-24pt}
\else
   \chapter[Supersymmetric Theory and Models]{Supersymmetric Theory and Models}
\fi

%
%\chapter[
%Supersymmetry
%]{Supersymmetric Theory and Models}\label{ra_ch1}

\author[]{Howard E.~Haber\textsuperscript{1} and Laurel Stephenson Haskins\textsuperscript{1,2}
}

\address{
\textsuperscript{1}Santa Cruz Institute for Particle Physics,\\
University of California, Santa Cruz, CA 95064, USA\\
\vspace{6pt}
\textsuperscript{2}Racah Institute of Physics,\\
Hebrew University, Jerusalem 91904, Israel 
}

\begin{abstract}
In these introductory lectures, we review the theoretical tools used in constructing supersymmetric field theories  and their application to physical models. 
We first introduce the technology of two-component spinors, which is convenient for describing spin-$\half$ fermions.  After motivating why a theory of nature may be supersymmetric at the TeV energy scale, we show how supersymmetry (SUSY) arises as an extension of the Poincar\'e algebra of spacetime symmetries.  We then obtain 
the representations of the SUSY algebra and discuss its simplest realization in the Wess-Zumino model. 
In order to have a systematic approach for obtaining supersymmetric Lagrangians, we 
introduce the formalism of superspace and superfields and recover the Wess-Zumino Lagrangian. These methods are then extended to encompass supersymmetric abelian and non-abelian gauge theories coupled to supermatter.
Since supersymmetry is not an exact symmetry of nature, it must ultimately be broken.   We discuss several mechanisms of SUSY-breaking (both spontaneous and explicit) and briefly survey various proposals for realizing SUSY-breaking in nature.
Finally, we construct the
the Minimal Supersymmetric extension of the Standard Model (MSSM), and 
consider the implications for the future of SUSY in particle physics.
\end{abstract}
%\markright{Customized Running Head for Odd Page} % default is Chapter Title.
\body

\tableofcontents

\section{Introduction to the TASI-2016 Supersymmetry Lectures}
\label{Intro}
These lectures were first presented at the 2016 Theoretical Advanced Study Institute (TASI-2016) in Boulder, CO.
Four ninety-minute lectures were given, with the aim of presenting the basic theoretical techniques of supersymmetry
needed for the construction of a supersymmetric extension of the Standard Model of particle physics.   
The lectures were pitched at an elementary level, assuming that
the students were well versed in quantum field theory, gauge theory and the Standard Model, but with no assumed prior knowledge 
of supersymmetry.  Nevertheless, some aspects of these lectures may also be useful to the reader with some prior 
knowledge of supersymmetry.

It is possible to introduce the technology of supersymmetry theory using four-component spinor notation that is familiar to all students of quantum field theory.  However, it is our view that employing two-component spinor notation greatly simplifies the presentation of the theoretical structure of supersymmetry in 3$+$1 spacetime dimensions.
Thus, in Section~\ref{sec:spinhalf}, we introduce the two-component spinor notation in some detail and discuss how it is related to the better known four-component spinor notation.
This material is based heavily on a comprehensive review of Dreiner, Haber and Martin that is presented in Ref.\cite{Dreiner:2008tw}.   In this review, it is shown that practical calculations in quantum field theory can be carried out entirely within the framework of the two-component spinor notation, which include the development of Feynman rules for two-component spinors.   However, at the end of Section 1, we are slightly less ambitious and revert to four-component fermion notation for the purpose of computing scattering and decay processes.  In particular, we provide a translation between two and four-component spinor notation, and develop four-component spinor Feynman rules that treat both Dirac and Majorana fermions on the same footing.

In Section~\ref{sec:motivation}, we present the motivation for TeV-scale supersymmetry.  Namely, why is it that we feel compelled to introduce a supersymmetric extension of the Standard Model, despite the great success of the Standard Model in describing collider data and the absence of significant evidence for new physics beyond the Standard Model.
With this motivation in mind, we are ready to explore the theoretical aspects of supersymmetry.

Since this is not a review article, we do not feel compelled to present a comprehensive list of references.  Nevertheless, it is instructive to assemble a list of books and lecture notes on supersymmetry, many of which we have found quite useful in preparing these lectures.   Thus, we draw your attention to the following books listed in Refs.\cite{WessBagger,Gates,Srivastava,Piguet,Freund,MullerKirsten,West,Lopuszanski,Bailin,Buchbinder,Soni,Galperin,Polonsky,Mohapatra,Drees,Baer,Aitchison,Binetruy,Terning,MullerKirsten2,Labelle,Shifman,sugra1,weinberg3,MDine,Manoukian,sugra2,Raby} and the following reviews and lecture notes listed in Ref.\cite{Taylor:1983su,Nilles:1983ge,Haber:1984rc,Sohnius:1985qm,Lahanas:1986uc,Haber:1993wf,Derendinger,Lykken,Martin:1997ns,Giudice:1998bp,bilalsusy,Petrov:2001hz,FigueroaO'Farrill:2001tr,Chung:2003fi,Luty:2005sn,RamseyMusolf:2006vr,Shirman:2009mt,GKane,susy,Bertolini:2013via}.  The reader is warned that conventions vary widely among these references.  Apart from the two possible choices for the spacetime metric (either the mostly minus metric used in these lectures or the mostly plus metric), there are many different choices in the definition of a variety of quantities, often involving different choices of signs.  Of these many conventions, we believe that the ones employed in these lecture notes are probably closest to those that appear in Ref.\cite{Sohnius:1985qm}.\footnote{We also note that although Ref.\cite{Martin:1997ns} employs the mostly plus metric, one can obtain a version of Martin's Supersymmetry Primer in the mostly minus metric by changing one line in the LaTeX  source code.  This alternative version of the Primer closely matches the conventions employed in these lectures.}

In Section~\ref{sec:SUSYalgebra}, we show how the algebra of the Poincar\'e group can be extended to obtain the supersymmetry (SUSY) algebra.  The representations of the $N=1$ SUSY
algebra are elucidated, and the Wess-Zumino model is presented as the simplest realization of a supersymmetric field theory.  In Section~\ref{sec:superspace}, we take some of the mystery out of constructing a SUSY Lagrangian by introducing the concepts of superspace and superfields.   This formalism allows one to construct supersymmetric field theories without any guesswork.  In Section~\ref{sec:gaugetheories}, the formalism of supersymmetric gauge theories is developed.  In Section~\ref{SSB}, we examine supersymmetry breaking, which is necessary for accommodating the observation that the elementary particles observed today are not each accompanied by an equal-mass superpartner.
Finally, in Section~\ref{sec:MSSM}, we construct the Minimal Supersymmetric extension of the Standard Model (MSSM).  
We end these lectures in Section~\ref{sec:future} with a brief discussion of what lies ahead for supersymmetry.

 \section{Spin-1/2 fermions in quantum field theory}
 \label{sec:spinhalf}

We begin these lectures with a treatment of  spin-$\half$ fermions in
quantum field theory.   In most introductory courses in relativistic quantum field
theory, the student first encounters fermion fields in the treatment of a
relativistic theory of electrons and photons.   The electron is
represented by a four-component Dirac fermion field, and the free field electron
Lagrangian yields the Dirac equation.   The four components represent
two degrees of freedom corresponding to the electron 
and two degrees
of freedom corresponding to the positron.  Feynman rules for quantum
electrodynamics are developed and the vector-like nature of the $e^+
e^-$ coupling to photons leads to some important simplifications.

The theory of electroweak interactions involves chiral
interactions of fermions with gauge bosons.   Left-handed and
right-handed fermions transform differently under the electroweak
gauge group, which may appear strange to students trained to think in
terms of four-component Dirac fermions.   Nevertheless, after electroweak
symmetry breaking, the mass-eigenstate fermion fields can be
identified.  All massive fermion states are charged under U(1)$_{\rm
  EM}$ and are thus represented by Dirac fermion fields.  The
neutrinos are massless, but only the left-handed neutrinos and
right-handed antineutrinos are present in the theory.  Thus, one can
still use four-component fermion fields (by applying the
appropriate chiral projection operators on the neutrino fields).  Hence, 
the four-component techniques of quantum
electrodynamics are easily accommodated and Feynman rules for the
fermion fields are obtained in a straightforward manner.   
  
However, the observation of
neutrino mixing phenomena implies that neutrinos are massive, which
requires new physics beyond the Standard Model of the electroweak
interactions.   Models of neutrino mass often include neutral
self-conjugate fermion states with two degrees of freedom, called
Majorana fermions.   Such
states can be described using four-component fermion fields that are
constrained by an appropriate conjugation condition.  However,
the resulting field theory description of systems of Majorana and
Dirac fermions is somewhat awkward.  Moreover, the 
Feynman rules for interacting Majorana fermions require some care.

Returning to first principles, one can ask how spin-$\half$ fermions arise
in quantum field theory.   In  Section~\ref{sec:spin_half_rep}, we
shall demonstrate that the fundamental building blocks employed in 
constructing spin-$\half$ quantum fields are
two-component spinors corresponding to the two-dimensional representations of the
Lorentz group.  
A neutral Majorana fermion is then represented
by a two-component fermion field.
Dirac fermions arise when one considers
theories of two mass-degenerate two-component fermions, which can
be combined to make a charged four-component Dirac fermion.  This is
completely analogous to the case of spin-0 bosons, in which a neutral
boson is represented by a real scalar field and a charged boson is
represented by a complex scalar field (whose real and imaginary parts
constitute two mass-degenerate real scalars).

The development of two-component spinor technology has a number of
benefits.  First, it provides an elegant unified description of
Majorana and Dirac fermions.  Second, it is very convenient to employ
the two-component spinor formalism in theories of chiral interactions.
Finally, it will prove especially useful in developing the formalism
of supersymmetry, which is the main focus of these lectures.

Because most students see the four-component spinor formalism first
and are therefore more familiar with it, we shall devote Section~\ref{sec:24} to the translation between the two- and four-component formalisms. 
Finally, in Section~\ref{sec:Feynman} we demonstrate how Feynman rules involving four-component
fermion fields can be extended to incorporate Majorana fermions.

This section is based on a comprehensive review of Dreiner, Haber and
Martin\cite{Dreiner:2008tw}, where many references to the original
literature can be found.

\subsection{Two-component spinor technology}
\label{sec:spin_half_rep}

\subsubsection{Orthochronous Lorentz transformations}

Quantum spin-$\half$ fields transform
under a two-dimensional irreducible representation of the Lorentz
group.   Thus, we first examine the properties that define a Lorentz transformation\cite{sexl}.
Under an active
Lorentz transformation, $\Lambda^\mu{}_\nu$, a four-vector $p^\mu$
transforms as
\beq
p^{\prime\,\mu}=\Lambda^\mu{}_\nu p^\nu\,.
\eeq 
The condition that $g_{\mu\nu}p^\mu p^\nu$ is invariant under
Lorentz transformations implies that\footnote{In our conventions, the Minkowski metric tensor is
$g_{\mu\nu}={\rm diag}(1\,,\,-1\,,\,-1\,,\,-1)$.} 
\beq \label{lambdarelation}
\Lambda^\mu{}_\nu g_{\mu\rho}\Lambda^\rho{}_\lambda=g_{\nu\lambda}.
\eeq
That is, $\Lambda\in$\,O(3,1).  \Eq{lambdarelation} implies that $\Lambda$ possesses the following two
properties: (i)~$\rm{det}~\Lambda=\pm 1$ and (ii)~$|\Lambda^0{}_0|\geq
1$.  Thus, Lorentz transformations fall into four disconnected
classes denoted by a pair of signs, $\left(\rm{sgn}[\rm{det}~\Lambda]\,,\,
\rm{sgn}[\Lambda^0{}_0]\right)$.  The proper
orthochronous Lorentz transformations correspond to $(+,+)$ and are
continuously connected to the identity.
 
 The most general proper orthochronous
Lorentz transformation, characterized by a
rotation angle $\theta$ about an axis $\mathbold{\widehat n}$
($\mathbold{\vec\theta}\equiv\theta \mathbold{\widehat n}$) and a
boost vector $\mathbold{\vec \zeta}\equiv
\mathbold{\hat v}\tanh^{-1}\beta$ (where $\mathbold{\hat{v}}\equiv
\mathbold{\vec{v}}/|\mathbold{\vec{v}}|$ is the unit velocity vector and
$\beta\equiv |\mathbold{\vec v}|/c$),\footnote{Henceforth, we shall
  work in particle physics
units where $\hbar=c=1$.}
is a $4\times 4$
matrix given by:
\begin{equation} \label{lambda44}
\Lambda=\exp\left(-\half i\theta^{\alpha\beta}s_{\alpha\beta}\right)
=\exp\left(
-i\mathbold{{\vec\theta}\cdot}\boldsymbol{\vec s}
-i\mathbold{{\vec\zeta}\cdot}\boldsymbol{\vec k}\right)\,,
\end{equation}
where $\theta^{\alpha\beta}$ is antisymmetric, with
$\theta^i \equiv \half\epsilon^{ijk} \theta_{jk}$,
$\zeta^i\equiv\theta^{i0}=-\theta^{0i}$, and
\begin{equation} \label{explicitsmunu}
(s_{\alpha\beta})^\mu{}_\nu=i(g_\alpha{}^\mu\,g_{\beta\nu}-g_\beta{}^\mu
\,g_{\alpha\nu})\,,
\end{equation}
with $s^i\equiv\half\epsilon^{ijk}s_{jk}$ and $k^i\equiv s^{0i}=-s^{i0}$.
We have employed a notation where the lower case Latin indices $i,j,k=1,2,3$ and $\epsilon^{123}=+1$.

Note that the $s^{\mu\nu}$ are antisymmetric $4\times 4$ matrices, {\it i.e.},
$s^{\mu\nu}=-s^{\nu\mu}$, and satisfy the
commutation relations,
\beq \label{eq:comm-rels}
[s^{\alpha\beta},s^{\rho\sigma}] = i(g^{\beta\rho}\,s^{\alpha\sigma} -
g^{\alpha\rho}\,s^{\beta\sigma} - g^{\beta\sigma}\,s^{\alpha\rho} +
g^{\alpha\sigma}\,s^{\beta\rho} ).
\eeq
It follows from \eqs{lambda44}{explicitsmunu} that an
infinitesimal orthochronous Lorentz transformation is
given by
\beq \label{inflambda4}
\Lambda^\mu{}_\nu\simeq\delta^\mu{}_\nu+\theta^\mu{}_\nu
\simeq (\mathds{1}_{4\times 4}-i\mathbold{{\vec\theta}\cdot}\boldsymbol{\vec s}
-i\mathbold{{\vec\zeta}\cdot}\boldsymbol{\vec k})^\mu{}_\nu\,,
\eeq
where $\mathds{1}_{4\times 4}$ is the $4\times 4$ identity matrix, 
and we have used $\theta^\mu{}_\nu=-\theta_\nu{}^\mu$.

\subsubsection{Finite-dimensional Representations of the Lorentz Group}

A generic spin-$s$ field $\Phi$ transforms as
\beq
\Phi(x) \rightarrow \Phi'(x^{\prime}) = M_R(\Lambda)\Phi(x)\,,
\eeq
where $M_R\equiv\exp\bigl(-\half i \theta_{\mu\nu}S^{\mu\nu}\bigr)$ and
the $S_{\mu\nu}$ constitute finite-dimensional irreducible matrix
representations of the Lie algebra of the Lorentz group.  The $S^{\mu\nu}$ satisfy 
the same commutation relations as the $s^{\mu\nu}$ given in \eq{eq:comm-rels}.
It is convenient to denote the six independent generators defined by the
$S^{\mu\nu}$ as
\beq \label{jkdef}
S^i \equiv\half \epsilon^{ijk} S_{jk}\,,\qquad\qquad K^i \equiv S^{0i}\,,
\eeq
where $i,j,k=1,2,3$. The $S^i$ generate
three-dimensional rotations in space and the $K^i$ generate the
Lorentz boosts.  It then follows that
\beq
M_R\equiv\exp\left(-i\mathbold{{\vec\theta}\newcdot}\boldsymbol{\vec
S} -i\mathbold{{\vec\zeta}\newcdot}\boldsymbol{\vec K}\right)\,.
\eeq
The $S^i$ and $K^i$ satisfy the
commutation relations,
\begin{align}
[S^i\, , \,S^j] &= \epsilon^{ijk} S^k\,,\\
 [S^i\, , \,K^j] &= \epsilon^{ijk} K^k\,,\\
 [K^i\, , \,K^j] &= - \epsilon^{ijk} S^k\,.
\end{align}
We define the following linear combinations of the generators,
\beq
\mathbold{\vec S_+} \equiv\half (\mathbold{\vec S}+
i\mathbold{\vec K})\,,\qquad\quad
\mathbold{\vec S_-} \equiv\half
(\mathbold{\vec S}- i\mathbold{\vec K}),
\eeq
which satisfy the commutation relations, 
\begin{align}
[S_+^i\,,\,S_+^j] &= i\epsilon^{ijk}S_+^k\,, \\
[S_-^i \, , \, S_-^j ] & = i \epsilon^{ijk}S_-^k\,, \\
[S_{\pm}^i\,,\,S_{\mp}^j  ] &= 0\,,
\end{align}
corresponding to two independent (complexified) SU(2) Lie algebras.
Thus, the representations of the Lorentz algebra are characterized by $(s_1,s_2)$, where
the $s_i$ are half-integers.
For example, $(0,0)$ corresponds to a scalar field and
$(\half,\half)$ corresponds to a four-vector field.

\subsubsection{Two-component spinors}

Spin-1/2 fermion fields transform under the spinor representations, $(\half,0)$ corresponding to $\boldsymbol{\vec{S}}_+=\half\boldsymbol{\vec\sigma}$ and
$\boldsymbol{\vec{S}}_-=0$,
 and $(0,\half)$ corresponding to $\boldsymbol{\vec{S}}_+=0$ and
$\boldsymbol{\vec{S}}_-=\half\boldsymbol{\vec\sigma}$.  That is, the 
Lorentz transformation matrices acting on spinor fields may be written in terms of the Pauli spin matrices $\sigma^1$, $\sigma^2$, and $\sigma^3$ as follows,
\beq \label{halfzero}
\!\!\!\!\!\!\!\!\!\!\!
(\half, 0): \hspace{1.2cm}  M=\exp\left(-\nicefrac{i}{2} \mathbold{\vec\theta\newcdot\vec\sigma}-\half\mathbold{\vec\zeta
\newcdot\vec\sigma}\right),
\eeq
which via a similarity transformation is equivalent to the matrix representation, $(M^{-1})^{\T} =i\sigma^2 M (i\sigma^2)^{-1}$, and
\beq \label{zerohalf}
(0,\half): \hspace{1cm}   [M^{-1}]^\dagger=\exp\left(-\nicefrac{i}{2}
\mathbold{\vec\theta\newcdot\vec\sigma}
+\half\mathbold{\vec\zeta\newcdot\vec\sigma}\right), 
\eeq
which via a similarity transformation is equivalent to the matrix representation, $M^*=i\sigma^2 [M^{-1}]^\dagger (i\sigma^2)^{-1}$.

Thus, the Lorentz transformation law for two-component $(\half,0)$ fields can be written in two equivalent ways, 
\beq
 \xi'_\alpha=M_\alpha{}^\beta\,\xi_\beta\,,\qquad\quad
\xi^{\prime\,\alpha}=[(M^{-1})^{\T}]^\alpha{}_\beta\,\xi^\beta\,,
\eeq
where $\alpha,\beta=1,2$.
Likewise,  the Lorentz transformation law for two-component $(0,\half)$ fields can be written in two equivalent ways, 
\beq
\xi^{\prime\,\dagger\,\dot\alpha}=
[(M^{-1})^\dagger]^{\dot\alpha}{}_{\dot\beta}\,\xi^{\dagger\,\dot\beta}
\,,\qquad\quad 
\xi^{\prime\,\dagger}_{\dot\alpha}=
[M^*]_{\dot{\alpha}}{}^{\dot\beta}\xi^\dagger_{\dot\beta}\,.
\eeq
The $(0,\half)$ fields are related to the $(\half,0)$ fields by hermitian conjugation,
\beq
\xi^\dagger_{\dot\alpha}\equiv(\xi_\alpha)^\dagger\,,\qquad\quad \xi^{\dagger\,\dot\alpha}\equiv(\xi^\alpha)^\dagger\,.
\eeq
It is conventional to employ undotted indices for the spinor
components of $(\half,0)$ fields and dotted indices for the spinor components of $(0,\half)$ fields.

As noted below \eqs{halfzero}{zerohalf}, respectively,
each of the two equivalent representation matrices, $M$ and
$(M^{-1})^{\T}$  in the case of $(\half,0)$, and $(M^{-1})^\dagger$
and $M^*$ in the case of $(0,\half)$,
are related by a similarity transformation involving the antisymmetric matrices,
\beq
i\sigma^2=\left(\begin{matrix} \phm 0&\quad 1\\
-1&\quad
0\end{matrix}\right)=\epsilon^{\alpha\beta}=\epsilon^{\dot\alpha\dot\beta}\,,
\eeq
and
\beq
(i\sigma^2)^{-1}=-i\sigma^2=\epsilon_{\alpha\beta}=\epsilon_{\dot\alpha\dot\beta}\,,
\eeq
which define the epsilon symbols with undotted and dotted indices.  Note that the epsilon symbols with raised and lowered
indices differ by an overall sign. 
Moreover, they can be used to
raise and lower the 
%two-component 
spinor indices,

\beq \label{raiseindex}
\xi^\alpha =\epsilon^{\alpha\beta}\,\xi_\beta\,,\qquad
\xi_\alpha =\epsilon_{\alpha\beta}\,\xi^\beta,\qquad
\xi^{\dagger\,\dot\alpha}
=\epsilon^{\dot\alpha\dot\beta}\,\xi^\dagger_{\dot\beta}\,,\qquad
\xi^\dagger_{\dot\alpha}
=\epsilon_{\dot\alpha\dot\beta}\,\xi^{\dagger\,\dot\beta}.
\eeq

The products of two epsilon symbols with undotted and with dotted indices,
respectively, satisfy,
\begin{align}
&\epsilon_{\alpha\beta} \epsilon^{\gamma\delta} =
-\delta_\alpha^\gamma \delta_\beta^\delta
+\delta_\alpha^\delta \delta_\beta^\gamma
,\\
&\epsilon_{\dot{\alpha}\dot{\beta}} \epsilon^{\dot{\gamma}\dot{\delta}} =
-\delta_{\dot{\alpha}}^{\dot{\gamma}}\delta_{\dot{\beta}}^{\dot{\delta}}
+\delta_{\dot{\alpha}}^{\dot{\delta}}\delta_{\dot{\beta}}^{\dot{\gamma}}\,,
\end{align}
where $\delta_{\dot\alpha}^{\dot\beta}=\delta_\alpha^\beta$
and the two-index symmetric Kronecker delta symbol 
with undotted indices is defined by
$\delta^1_1=\delta^2_2=1$ and $\delta_1^2=\delta_2^1=0$.
In particular,
\beq
\epsilon_{\alpha\gamma}\,\epsilon^{\gamma\beta}=\delta_\alpha^\beta\,,\qquad\quad
\epsilon_{\dot\alpha\dot\gamma}\,\epsilon^{\dot\gamma\dot\beta}=\delta_{\dot\alpha}^{\dot\beta}\,.
\eeq

Finally, we introduce the $\sigma$-matrices:
\begin{align}
\sigma^\mu_{\alpha\dot\beta}=(\mathds{1}_{2\times 2}\,;\,
\mathbold{\vec\sigma}) \,,\qquad
\sigmabar^{\mu\,\dot\alpha\beta}=(\mathds{1}_{2\times 2}\,;\, 
-\mathbold{\vec\sigma})\,,
\end{align}
where $\mathds{1}_{2\times 2}$ is the $2\times 2$ identity matrix.  The spinor index
structure derives from the relations,
\beq \label{MM}
(M^\dagger)^{\dot\alpha}{}_{\dot\beta}\sigmabar^{\mu\dot\beta\gamma}
M_\gamma{}^\delta=\Lambda^\mu{}_\nu\sigmabar^{\nu\,\dot\alpha\delta}\,,\qquad
(M^{-1})_\alpha{}^\beta\sigma^\mu_{\beta\dot\gamma}[(M^{-1})^\dagger]^{\dot\gamma}{}_{\dot\delta}=\Lambda^\mu{}_\nu\sigma^\nu_{\alpha\dot\delta}\,.
\eeq
Note that the matrix $M$ and its inverse have the same spinor index
structure (and likewise for the matrix $M^\dagger$ and its inverse).

We will sometimes find it useful to relate the $\sigma^\mu$ and $\overline{\sigma}^\mu$ matrices using the identities
\beq
\sigma^\mu_{\alpha{\dot{\alpha}}} = \epsilon_{\alpha\beta}
\epsilon_{\dot{\alpha}\dot{\beta}} \sigmabar^{\mu\,\dot{\beta}\beta}
\,, \qquad\quad \sigmabar^{\mu\,\dot{\alpha}\alpha} =
\epsilon^{\alpha\beta} \epsilon^{\dot{\alpha}\dot{\beta}}
\sigma^{\mu}_{\beta\dot{\beta}}\,.
\eeq
The significance of $\sigma^\mu$ is that Lorentz 4-vectors can be built
from spinor bilinears. For example, $\chi^\alpha \of{x} \sigma^\mu_{\alpha \dot{\beta}} \xi^{\dot{\beta}}\of{x} $ transforms as a Lorentz 4-vector,
\beqa
\chi^{\,\prime\,\alpha}(x')\sigma^\mu_{\alpha\dot\beta}
\xi^{\prime\,\dagger\,\dot\beta}(x')
&=&\chi^\alpha(x)
[M^{-1}\sigma^\mu(M^{-1})^\dagger]_{\alpha\dot\beta}\xi^{\dagger\,\dot\beta}(x)
 \\
&
=&\Lambda^\mu{}_\nu\,\chi(x)^\alpha\sigma^\nu_{\alpha\dot\beta}
\xi^{\dagger\,\dot\beta}(x)
\,,
\eeqa
after making use of \eq{MM}.  Spinor
indices can be suppressed by adopting a summation
convention where we contract indices as follows:
\beq \label{contract}
{}^\alpha{}_\alpha\qquad {\rm and} \qquad
{}_{\dot{\alpha}}{}^{\dot{\alpha}}\,.
\eeq
For example,
\beqa \xi\eta
&\equiv & \xi^\alpha\eta_\alpha ,
\\
\xi^\dagger \eta^\dagger &\equiv & \xi^\dagger_{\dot\alpha} \eta^{\dagger\,\dot
\alpha} ,
\\
\xi^\dagger\sigmabar^\mu\eta &\equiv &  \xi^\dagger_{\dot{\alpha}}
\sigmabar^{\mu\dot{\alpha}\beta}\eta_\beta , 
\\
\xi\sigma^\mu \eta^\dagger &\equiv & \xi^{{\alpha}} \sigma^{\mu}_{\alpha
\dot \beta} \eta^{\dagger\,\dot \beta} .
 \eeqa 
In particular, for
anticommuting spinors,
\beqa
\eta\xi\equiv\eta^\alpha\xi_\alpha&=&-\xi_\alpha\eta^\alpha=+\xi^\alpha\eta_\alpha=\xi\eta\,.\\
\etabar\xibar\equiv \etabar_{\dalpha}{\xibar}^{\dalpha}&=&-{\xibar}^{\dalpha} \eta^\dagger_{\dalpha}=\xi^\dagger_{\dalpha} {\etabar}^{\dalpha}=\xibar\etabar\,.
\eeqa

The behavior of spinor products
under hermitian conjugation is noteworthy,
\beq
(\xi \Sigma \eta)^\dagger = \eta^\dagger \reversed{\Sigma} \xi^\dagger\,,
\quad (\xi \Sigma \eta^\dagger)^\dagger = \eta \reversed{\Sigma}
\xi^\dagger\,,
\quad (\xi^\dagger  \Sigma \eta)^\dagger = \eta^\dagger \reversed{\Sigma}
\xi\,,
\eeq
where in each case $\Sigma$ stands for any sequence of alternating
$\sigma$ and $\sigmabar$ matrices, and $\reversed{\Sigma}$ is
obtained by reversing the order of the $\sigma$
and $\sigmabar$ matrices that appear in $\Sigma$.

From the sigma matrices, one can construct the
antisymmetrized products,
\begin{align}
(\sigma^{\mu\nu})_\alpha{}^\beta
&\equiv \tfrac14 i\left(\sigma^\mu{}_{\!\!\!\!\alpha\dot{\gamma}}
\sigmabar^{\nu\dot{\gamma}\beta}-\sigma^\nu{}_{\!\!\!\!\alpha\dot{\gamma}}
\sigmabar^{\mu\dot{\gamma}\beta}\right)\,,
\\
(\sigmabar^{\mu\nu})^{\dot{\alpha}}{}_{\dot{\beta}} &\equiv
\tfrac14 i\left(\sigmabar^\mu{}^{\dot{\alpha}\gamma}
\sigma^\nu{}_{\!\!\!\!\gamma\dot{\beta}}-\sigmabar^\nu{}^{\dot{\alpha}\gamma}
\sigma^\mu{}_{\!\!\!\!\gamma\dot{\beta}}\right)\,. 
\end{align}

With this notation, we may write the $(\half,0)$ and $(0,\half)$ transformation
matrices, respectively, as
\beqa
M&=&\exp\left(-\half i\theta^{\mu\nu}\sigma_{\mu\nu}\right)\,,
 \\
(M^{-1})^\dagger &=&
\exp\left(-\half i\theta^{\mu\nu}\sigmabar_{\mu\nu}\right)\,,
\eeqa 
where the $\theta^{\mu\nu}$ are defined below \eq{lambda44}.

Consider a  pure boost of an on-shell two-component spinor from its rest frame to
the frame where $p^\mu=(E_{\boldsymbol{p}}\,,\,\boldsymbol{\vec
p})$, with $E_{\boldsymbol{p}}=(|{\mathbold{\vec p}}|^2+m^2)^{1/2}$.
In this case, setting $\theta^{ij}=0$ (corresponding to no rotation),
we obtain,
 \beqa
M&=&\exp\left(-\half\mathbold{\vec\zeta\newcdot\vec\sigma}\right)
=\sqrt{\frac{p\newcdot\sigma}{m}}=\frac{(E_{\boldsymbol{p}}+m)\mathds{1}_{2\times 2}
-\mathbold{\vec\sigma\newcdot\vec
p}}{\sqrt{2m(E_{\boldsymbol{p}}+m)}}
\,,\label{mboost}
\\
(M^{-1})^\dagger&=&\exp\left(+\half\mathbold{\vec\zeta\newcdot\vec\sigma}\right)
=\sqrt{\frac{p\newcdot\sigmabar}{m}}=\frac{(E_{\boldsymbol{p}}+m) \mathds{1}_{2\times 2}
+\mathbold{\vec\sigma\newcdot\vec
p}}{\sqrt{2m(E_{\boldsymbol{p}}+m)}} \,.\label{mstarboost}
 \eeqa 
The matrix square roots, $\sqrt{p\newcdot\sigma}$ and
$\sqrt{p\newcdot\sigmabar}$, appearing in \eqs{mboost}{mstarboost} are defined to be the unique non-negative
definite hermitian matrices
whose squares are equal
to the non-negative definite hermitian matrices
$\BDpos p\newcdot\sigma$ and $\BDpos p\newcdot\sigmabar$,
respectively.\footnote{Note that $\BDpos p\newcdot\sigma$
and $\BDpos p\newcdot\sigmabar$ are non-negative
matrices due to the implicit mass-shell condition
satisfied by $p^\mu$.}

\subsubsection{Useful identities}

The following identities can be used to systematically simplify
expressions involving products of $\sigma$ and $\sigmabar$
matrices,
\begin{align}
& \sigma^\mu_{\alpha\dot{\alpha}}
\sigmabar_\mu^{\dot{\beta}\beta} = \BDpos 2 \delta_{\alpha}^{\beta}
\delta^{\dot{\beta}}_{\dot{\alpha}}, \label{sigid1}
\\
& \sigma^\mu_{\alpha\dot{\alpha}} \sigma_{\mu\beta\dot{\beta}} =
\BDpos 2 \epsilon_{\alpha\beta}
\epsilon_{\dot{\alpha}\dot{\beta}}\,, \label{sigid2} \\
&
\sigmabar^{\mu\dot{\alpha}\alpha}
\sigmabar_\mu^{\dot{\beta}\beta} = \BDpos 2 \epsilon^{\alpha\beta}
\epsilon^{\dot{\alpha}\dot{\beta}}\,, {}\label{sigid3}
\\
& {[\sigma^\mu\sigmabar^\nu + \sigma^\nu \sigmabar^\mu
]_\alpha}^\beta = \BDpos 2\metric^{\mu\nu}
\delta_{\alpha}^{\beta}\,, {}\label{sigid4}
\\
&[\sigmabar^\mu\sigma^\nu + \sigmabar^\nu \sigma^\mu
]^{\dot{\alpha}}{}_{\dot{\beta}} = \BDpos 2\metric^{\mu\nu}
\delta^{\dot{\alpha}}_{\dot{\beta}}\,, {}\label{sigid5}
\\
& \sigma^\mu \sigmabar^\nu \sigma^\rho = \BDpos \metric^{\mu\nu}
\sigma^\rho \BDminus \metric^{\mu\rho} \sigma^\nu \BDplus
\metric^{\nu\rho} \sigma^\mu \BDplus i \epsilon^{\mu\nu\rho\kappa}
\sigma_\kappa\,, {}\label{sigsigsig1}
\\
& \sigmabar^\mu \sigma^\nu \sigmabar^\rho = \BDpos \metric^{\mu\nu}
\sigmabar^\rho \BDminus \metric^{\mu\rho} \sigmabar^\nu \BDplus
\metric^{\nu\rho} \sigmabar^\mu \BDminus i
\epsilon^{\mu\nu\rho\kappa} \sigmabar_\kappa\,,{} \label{sigsigsig2}
\end{align}
where $\epsilon^{0123}=-\epsilon_{0123}=+1$ in our conventions.
%Computations of cross sections and decay rates often require
The traces of alternating products of $\sigma$ and $\sigmabar$
matrices are given by,
\begin{align}
&{\rm Tr}[\sigma^\mu \sigmabar^\nu ] = {\rm Tr}[\sigmabar^\mu
\sigma^\nu ] = \BDpos 2 \metric^{\mu\nu} \,, {}
\\
&{\rm Tr}[\sigma^\mu \sigmabar^\nu \sigma^\rho \sigmabar^\kappa ] =
2 \left ( \metric^{\mu\nu} \metric^{\rho\kappa} - \metric^{\mu\rho}
\metric^{\nu\kappa} + \metric^{\mu\kappa} \metric^{\nu\rho} + i
\epsilon^{\mu\nu\rho\kappa} \right )\,, \qquad\phantom{xx} {}
\\
&{\rm Tr}[\sigmabar^\mu \sigma^\nu \sigmabar^\rho \sigma^\kappa ] =
2 \left ( \metric^{\mu\nu} \metric^{\rho\kappa} - \metric^{\mu\rho}
\metric^{\nu\kappa} + \metric^{\mu\kappa} \metric^{\nu\rho} - i
\epsilon^{\mu\nu\rho\kappa} \right )\,. {}
\end{align}
Traces involving
an odd number of $\sigma$ and $\sigmabar$ matrices cannot arise,
since there is no way to connect the spinor indices consistently.
Additional identities involving $\sigma^{\mu\nu}$ and
$\sigmabar^{\mu\nu}$ can be found in Ref.~\cite{Dreiner:2008tw}. 
%%%%%%%%%%%%%%%%%%%%%%%%%%%%%%%%%%%%%%%%%%%%%%%%%%%%

Finally, we examine some useful identities involving bilinear spinor
quantities.  Although the two-component spinor fields appearing in
these lectures are anticommuting, one also may encounter commuting
two-component spinor wave functions.  Thus, it is convenient to denote
an arbitrary two-component spinor by $z_i$, and a sign factor,
$(-1)^A=+1 [-1]$, for commuting [anticommuting] spinors, respectively.
Then, the following
identities hold:
\begin{align}
&z_1 z_2 = -(-1)^A z_2 z_1 {}
\\
&z_1^\dagger z_2^\dagger = -(-1)^A z_2^\dagger  z_1^\dagger {}
\\
&z_1 \sigma^\mu z_2^\dagger = (-1)^A z_2^\dagger \sigmabar^\mu z_1 {} \label{eq:sigmucom} \\
&z_1 \sigma^\mu \sigmabar^\nu z_2 = -(-1)^A z_2 \sigma^\nu
\sigmabar^\mu z_1
 {}\\
&z_1^\dagger \sigmabar^\mu \sigma^\nu z_2^\dagger = -(-1)^A z_2^\dagger
\sigmabar^\nu \sigma^\mu z_1^\dagger
{}\\
&z_1^\dagger \sigmabar^\mu \sigma^\rho \sigmabar^\nu z_2=(-1)^A z_2
\sigma^\nu \sigmabar^\rho \sigma^\mu z_1^\dagger\,.{}
\end{align}

In many cases, it is convenient to rewrite a product of two bilinear
spinor quantities in terms of products in which 
the individual spinors appear in a different order.   Below, we provide five different Fierz
identities, which are valid for both commuting and anticommuting spinors,
\begin{align}
(z_1 z_2)(z_3 z_4) &= -(z_1 z_3) (z_4 z_2) - (z_1 z_4)(z_2 z_3)\,, {}
\\
(z_1^\dagger z_2^\dagger)(z_3^\dagger z_4^\dagger) &= 
- (z_1^\dagger z_3^\dagger)
(z_4^\dagger z_2^\dagger) - (z^\dagger_1 z^\dagger_4) 
(z_2^\dagger  z^\dagger_3)\,, {}
\\ (z_1 \sigma^\mu z_2^\dagger)(z_3^\dagger \sigmabar_\mu z_4) &= \BDneg
2 (z_1 z_4) (z_2^\dagger z^\dagger_3)\,, {}
\\
 (z_1^\dagger \sigmabar^\mu z_2)(z^\dagger_3 \sigmabar_\mu z_4) &= \BDpos
\phm 2 (z_1^\dagger z^\dagger_3) (z_4 z_2)\,, {}
\\
 (z_1 \sigma^\mu z^\dagger_2)(z_3 \sigma_\mu z^\dagger_4) &= \BDpos \phm 2
(z_1 z_3) (z^\dagger_4 z^\dagger_2)\,.{}
\end{align}
An exhaustive 
%comprehensive 
list of Fierz identities 
can be found in Appendix B of Ref.\cite{Dreiner:2008tw}. 

\subsubsection{Free field theories of two-component fermions}
The $(\half,0)$ spinor field $\xi_\alpha(x)$ describes a neutral
{{Majorana fermion}}.  The free-field Lagrangian is:
\beq
\mathscr{L}= \BDpos i\xi^\dagger \sigmabar^\mu\partial_\mu\xi - \half m
(\xi \xi + \xi^\dagger \xi^\dagger )\,,
\eeq
which is hermitian up to a total divergence since we can rewrite the above Lagrangian as
\beq
\mathscr{L}= \half i \xi^\dagger \sigmabar^\mu\!\!\stackrel{\leftrightarrow}{\partial}_{\!\mu}\!\xi - \half m
(\xi \xi + \xi^\dagger \xi^\dagger )+\text{total divergence}\,,
\eeq
where $ \xi^\dagger \sigmabar^\mu\!\!\stackrel{\leftrightarrow}{\partial}_{\!\mu}\!\xi \equiv
\xi^\dagger \sigmabar^\mu(\partial_\mu\xi) - (\partial_\mu\xi)^\dagger \sigmabar^\mu\,\xi$. 

Generalizing to a multiplet of two-component fermion fields,
$\hat{\xi}_{\alpha i}(x)$, labeled by flavor index $i$, the free Lagrangian is
\beq
\mathscr{L}= \BDpos i{\hat\xi}^{\dagger\,i}\sigmabar^\mu\partial_\mu\hat\xi_i
- \half M^{ij}\hat\xi_i\hat\xi_j
- \half M_{ij}{\hat\xi}^{\dagger\,i}{\hat\xi}^{\dagger\,j}\,,
\eeq
where hermiticity implies that $M_{ij}\equiv (M^{ij})^*$ is a complex symmetric
matrix.
To identify the physical fermion fields, we express the so-called \textit{interaction eigenstate fields}, 
$\hat\xi_{\alpha i}(x)$,  in terms of \textit{mass-eigenstate fields}
\beq
\xi(x)=\Omega^{-1}\hat\xi(x),
\eeq
where $\Omega$ is unitary and chosen such that
\beq
\Omega^{\T} M\, \Omega = \boldsymbol{m} = {\rm diag}(m_1,m_2,\ldots),
\eeq
where the $m_i$ are non-negative real numbers.
In linear algebra, this is called the
{\textit{Takagi diagonalization}} of a complex symmetric matrix
$M$\cite{takagi,horn}.\footnote{Subsequently, it was recognized in
Refs.\cite{horn2,horn3} that the Takagi diagonalization was first
established for nonsingular complex symmetric matrices by Autonne
\cite{autonne}.}
To compute the values of the diagonal elements of $\boldsymbol{m}$, we note
that
\beq
\Omega^{\T} MM^\dagger \Omega^\ast= \boldsymbol{m}^2 .
\eeq
Since $MM^\dagger$ is hermitian, it can be diagonalized by a unitary
matrix.   Thus, the $m_i$ of the Takagi diagonalization are
the non-negative square-roots of the eigenvalues of $MM^\dagger$.
In terms of the mass eigenstate fields,
\beq
\mathscr{L}= \BDpos
i\xi^{\dagger\,i}\sigmabar^\mu\partial_\mu\xi_i- \half m_{i}(\xi_i\xi_i+
\xi^{\dagger\,i}\xi^{\dagger\,i})\,.
\eeq
\clearpage

\begin{example}[The Seesaw Mechanism\cite{seesaw1,seesaw2,seesaw3,seesaw4,seesaw5}]

The~seesaw~Lagrangian for the two-component fermions $\psi_1$ and $\psi_2$ is
\beq
\mathscr{L}=i\left(\psi^{\dagger\,1}\,\overline\sigma^\mu\partial_\mu\psi_1+
\psi^{\dagger\,2}\,\overline\sigma^\mu\partial_\mu\psi_2\right)-M^{ij}\psi_i\psi_j
-M_{ij}\psi^{\dagger\,i}\,\psi^{\dagger\,j}  \,,
\eeq
where
\beq
M^{ij}={\left(\begin{array}{cc} 0 &\,\,\, m_D\\ m_D &\,\,\,
M\end{array}\right)}\,,
\eeq
and (without loss of generality) $m_D$ and $M$ are real and
positive. The Takagi diagonalization of this matrix is
\beq
\Omega^T M
\Omega=M_D,\label{takagidef}
\eeq
 where
\begin{align}
\Omega=\left(\begin{array}{cc} \phm i\cos\theta &\quad \sin\theta \\
-i\sin\theta &\quad \cos\theta\end{array}\right)\,,\qquad\quad
M_D= \left(\begin{array}{cc} m_- &\quad 0\\ 0 &\quad  m_+ \end{array}\right)  \,,
\end{align}
with
\beq
m_\pm=\half\left[\sqrt{M^2+4m_D^2}\pm M\right]
\eeq
and
\beq
\sin 2\theta=\frac{2m_D}{\sqrt{M^2+4m_D^2}}\,.
\eeq
If $M\gg m_D$, then the corresponding fermion masses are
$m_-\simeq m_D^2/M$ and $m_+\simeq M$, with $\sin\theta\simeq
m_D/M$.  The mass eigenstates, $\chi_i$ are given by
$\psi_i=\Omega_i{}^j\chi_j$; to leading order in $m_d/M$,
\begin{align}
i\chi\ls{1} \simeq \psi_1-\frac{m_D}{M}\psi_2\,,\qquad\quad
\chi\ls{2} \simeq  \psi_2+\frac{m_D}{M}\psi_1\,.
\end{align}
Indeed, one can check that: 
\beq
\begin{split}
 \half
m_D(\psi_1\psi_2+\psi_2\psi_1  )+\tfrac12 M\psi_2 & \psi_2 +{\rm
h.c.} \\
& \simeq\frac12\left[\frac{m_D^2}{M}\chi\ls{1}\chi\ls{1}+
M\chi\ls{2}\chi\ls{2}+{\rm h.c.}\right]\,,
\end{split}
\eeq
 which corresponds to a theory
of two Majorana fermions---one very light and one very heavy
({\textit{the seesaw}}).
\end{example}

In any theory containing a multiplet of fields, one can check for the existence of global symmetries.
The simplest case is a theory of a pair of two-component $(\half,0)$  fermion fields $\chi$ and $\eta$, with
the free-field Lagrangian,
\beq \label{DiracLag2}
\mathscr{L}= \BDpos i\chi^\dagger\sigmabar^\mu\partial_\mu\chi \BDplus
  i\eta^\dagger\sigmabar^\mu\partial_\mu\eta-m(\chi\eta+
\chi^\dagger\eta^\dagger)\,.
\eeq
The Lagrangian given in \eq{DiracLag2} possesses a U(1) global symmetry, $\chi\to e^{i\theta}\chi$ and $\eta\to e^{-i\theta}\eta$.
That is, $\chi$ and $\eta$ are oppositely charged. 
 The corresponding  mass matrix is
\beq
M = \left(\begin{matrix} 0&\quad m\\ m&\quad 0
\end{matrix}\right).
\eeq
  Performing the Takagi diagonalization yields two degenerate two-component fermions of mass $m$.  However, the corresponding mass-eigenstates are not eigenstates of charge.\footnote{This is the analog of a free field theory of a complex scalar boson
$\Phi$ with a mass term, $\mathscr{L}_{\rm mass}=-m^2|\Phi|^2$.   Writing
$\Phi=(\phi_1+i\phi_2)/\sqrt{2}$, we can write Lagrangian in terms of $\phi_1$ and $\phi_2$ with a diagonal mass
term.  But, $\phi_1$ and $\phi_2$ do not correspond to states of definite charge.}
Together, $\chi$ and $\eta^\dagger$ constitute a single (four-component) {\textit{Dirac fermion}}.

More generally, consider a collection of
charged Dirac fermions  represented by
pairs of two-component interaction eigenstate fields
$\hat\chi_{\alpha i}(x)$, $\hat\eta_{\alpha }^i(x)$, with
\beq
\mathscr{L}=
  i{\hat\chi}^{\dagger i}\sigmabar^\mu\partial_\mu\hat\chi_i
+
  i{\hat\eta}^\dagger_{i}\sigmabar^\mu\partial_\mu\hat\eta^i
-M^i{}_j \hat\chi_i\hat\eta^j
-M_i{}^j  {\hat\chi}^{\dagger i}\hat\eta^\dagger_{ j}\,,
\eeq
where $M$ is a complex matrix with matrix elements denoted by
$M^i{}_j$ (note the placement of the flavor indices $i$ and $j$), and $M_i{}^j\equiv (M^i{}_j)^*$.

We denote the mass eigenstate fields by $\chi_i$ and
$\eta^i$ and the unitary matrices $L$ and~$R$, such that
$\hat\chi_i=L_i{}^k\chi_k$ and 
$\hat\eta^i=R^i{}_k\eta^k$,
and 
\beq \label{LTMR}
L^{\T} M R= {\boldsymbol{m}}={\rm diag}(m_1,m_2,\ldots),
\eeq
where the $m_i$ are non-negative real numbers.
This is the singular value
decomposition of a complex matrix (see, e.g., Refs.\cite{horn2,horn3}). Noting that
\beq
R^\dagger(M^\dagger M) R \,=\, {\boldsymbol{m}}^2\,,\label{svd}
\eeq
the diagonal elements of $\boldsymbol{m}$ are
the non-negative square roots of the
corresponding eigenvalues of $M^\dagger M$.
In terms of the
mass eigenstate fields,
\beq
\label{lagDiracdiag}
\mathscr{L}= i{\chi}^{\dagger i}\sigmabar^\mu\partial_\mu\chi_i+
  i{\eta}^\dagger_i \sigmabar^\mu \partial_\mu\eta^i
- m_i(\chi_i\eta^i + \chi^{\dagger i} \eta^\dagger_i)\,.
\eeq

\subsubsection{Fermion--scalar interactions}

The most general set of interactions
with the scalars of the theory $\hat\phi_I$
are then given by:
\beq
\mathscr{L}_{\rm int} = -\half \hat Y^{Ijk} \hat\phi_I\hat\psi_j\hat\psi_k
-\half \hat Y_{Ijk}\hat\phi^{I} {\hat\psi}^{\dagger\,j} {\hat\psi}^{\dagger\,k}
\,,
\eeq
where $\hat Y_{Ijk}\equiv  (\hat Y^{Ijk})^*$ and $\hat\phi^I\equiv (\hat\phi_I)^*$.
 The flavor index $I$ runs over a collection of
real scalar fields $\hat\varphi_i$ and pairs of complex scalar fields
$\hat\Phi_j$ and $\hat\Phi^j\equiv(\hat\Phi_j)^*$
(where a complex field and its
conjugate are counted separately).
The Yukawa
couplings $\hat Y^{Ijk}$ are symmetric under interchange of $j$ and
$k$.

The mass-eigenstate
basis $\psi$ is related to the interaction-eigenstate basis $\hat \psi$ by
a unitary transformation,
\begin{align}
\hat \psi \equiv \begin{pmatrix}\hat\xi \\ \hat\chi \\ \hat\eta
\end{pmatrix}= U \psi
\equiv \begin{pmatrix}\Omega &\quad 0& \quad0 \\
                0 & \quad L &\quad 0 \\
                0 & \quad 0 &\quad  R\end{pmatrix}
\begin{pmatrix}\xi \\ \chi \\ \eta\end{pmatrix} \,,
\end{align}
where $\Omega$, $L$, and $R$ are constructed as described previously.
Likewise a unitary transformation yields the scalar mass-eigenstates via $\hat\phi=V\phi$.
Thus, in terms of mass-eigenstate fields:
\beq
\mathscr{L}_{\rm int} = -\half Y^{Ijk} \phi_I\psi_j\psi_k
-\half Y_{Ijk} \phi^{I} {\psi}^{\dagger\,j} {\psi}^{\dagger\,k}
\,,
\eeq
where
$Y^{Ijk}=V_J{}^I U_m{}^j U_n{}^k \hat Y^{Jmn}$.

\subsubsection{Fermion--gauge boson interactions}

In the gauge-interaction basis for the
two-component fermions the corresponding interaction
Lagrangian is given by
\beq \label{eq:lintG}
\mathscr{L}_{\rm int} =
- g_a A_a^{\mu} {\hat\psi}^{\dagger\,i}\,
\sigmabar_\mu ({\boldsymbol T}^a)_i{}^j \hat\psi_j \,,
\eeq
where the index $a$ labels the (real or complex) vector bosons
$A_a^\mu$ and is summed over.
If the gauge symmetry is unbroken, then the index $a$ runs over the
adjoint representation of the gauge group, and the $({\boldsymbol T}^a)_i{}^j$
are hermitian representation matrices\footnote{For a $U(1)$ gauge
group, the $\mathbold{T}^a$ are replaced by real numbers
corresponding to the U(1) charges of the $(\half,0)$
fermions.}
of the gauge group acting on the
fermions.  There is a separate coupling
$g_a$ for each simple group or U(1) factor of the
gauge group G.

In the case of spontaneously broken gauge theories, one must
diagonalize the vector boson squared-mass matrix.  The form of
\eq{eq:lintG} still applies where $A_\mu^a$ are gauge boson fields of
definite mass, although in this case for a fixed value of $a$, the
product $g_a{\boldsymbol T}^a$ is
some linear combination of the original $g_a {\boldsymbol T}^a$ of the
unbroken theory.
That is, the hermitian matrix gauge field $(A_\mu)_i{}^j\equiv
A_\mu^a (\boldsymbol{T^a})_i{}^j$ appearing in \eq{eq:lintG} can always
be re-expressed in terms of the
\textit{physical} mass eigenstate gauge boson fields.
If an unbroken U(1)
symmetry exists, then the physical gauge bosons will also
be eigenstates of the
conserved U(1)-charge.\footnote{\label{vectormass}
In terms of the physical gauge boson fields, $A_\mu^a\boldsymbol{T^a}$
consists of a sum over real neutral gauge fields
multiplied by hermitian generators, and
complex charged gauge fields
multiplied by non-hermitian generators.
For example, in the electroweak Standard
Model, ${\rm G}={\rm SU}(2)\times$U(1) with gauge bosons and
generators $W_\mu^a$ and ${\boldsymbol T}^a=\half\tau^a$ for SU(2),
and $B_\mu$ and $\boldsymbol{Y}$ for U(1),
where the $\tau^a$ are the usual Pauli matrices.
After diagonalizing the gauge boson squared-mass matrix,
$$
gW_\mu^a \boldsymbol{T^a}+ g' B_\mu \boldsymbol{Y}=
\frac{g}{\sqrt{2}}(W_\mu^+\boldsymbol{T^+}
+W_\mu^-\boldsymbol{T^-})+
\frac{g}{\cos\theta_W}\left(\boldsymbol{T^3}-\boldsymbol{Q}
\sin^2\theta_W\right)Z_\mu+e\boldsymbol{Q}A_\mu\,,
$$
where
$\boldsymbol{Q}=\boldsymbol{T^3}+\boldsymbol{Y}$ is the generator
of the unbroken U(1)$_{\rm EM}$,
$\boldsymbol{T^\pm}\equiv \boldsymbol{T^1}\pm i\boldsymbol{T^2}$,
and $e=g\sin\theta_W=g'\cos\theta_W$.
The massive gauge boson charge-eigenstate fields
of the broken theory
consist of a charged massive gauge boson pair,
$W^\pm\equiv (W^1\mp iW^2)/\sqrt{2}$, a neutral massive gauge boson,
$Z\equiv W^3\cos\theta_W-B\sin\theta_W$, and the massless photon,
$A\equiv W^3\sin\theta_W+B\cos\theta_W$.}

In terms of mass-eigenstate fermion fields,
\beq
\mathscr{L}_{\rm int} =
- A_a^{\mu} \psi^{\dagger\,i}\,
\sigmabar_\mu (G^a)_i{}^j \psi_j \,,
\eeq
where $G^a= g_a U^\dagger {\boldsymbol{T}}^a U$
(no sum over~$a$).

The case of gauge interactions
of charged Dirac fermions can be treated as follows. Consider pairs of $(\half,0)$
interaction-eigenstate fermions
$\hat\chi_i$ and $\hat\eta^i$ that  transform as conjugate representations
of the gauge group (hence the difference in the flavor index heights).
The Lagrangian for the gauge interactions
of Dirac fermions can be written in the form:
\beq
\mathscr{L}_{\rm int} =
- g_a A_a^{\mu} \hat\chi^{\dagger\,i}\, \sigmabar_\mu
({\boldsymbol T}^a)_i{}^j \hat\chi_j
+ g_a A_a^{\mu} \hat\eta^\dagger_{\,i}\, \sigmabar_\mu
({\boldsymbol T}^a)_j{}^i \hat\eta^j \,,
\eeq
where the $A_\mu^a$ are gauge boson mass-eigenstate fields.
Here we have used the fact that if $({\boldsymbol T}^a)_i{}^j$ are the
representation matrices for the $\hat\chi_i$, then the $\hat\eta^i$
transform in the complex conjugate representation with generator
matrices $-({\boldsymbol T}^a)^*
= -({\boldsymbol T}^a)^T$.
In terms of mass-eigenstate fermion fields,
\beq
\mathscr{L}_{\rm int} =
- A_a^{\mu}\left[ {\chi}^{\dagger\,i}\, \sigmabar_\mu
(G_L^a)_i{}^j \chi_j
- {\eta}^\dagger_{\,i}\, \sigmabar_\mu
(G_R^a)_j{}^i \eta^j\right] \,,
\eeq
where
$G_L^a=g_a L^\dagger {\boldsymbol{T}}^a L$ and
$G_R^a=g_a R^\dagger {\boldsymbol{T}}^a R$ (no sum over~$a$).

\subsection[Correspondence between the
two- and four-component spinor notations]{Correspondence between the
two-component and four-component spinor notations}
\label{sec:24}

Most pedagogical treatments of calculations in particle physics
employ four-component Dirac spinor notation, which combines distinct irreducible
representations of the Lorentz symmetry algebra.  Parity-conserving theories such as QED and
QCD and their Feynman rules are especially well-suited to four-component
spinor notation.  In light of the widespread familiarity with four-component spinor
techniques, we provide in this section a translation between 
two-component and four-component spinor notation.

\subsubsection{From two-component to four-component spinor notation}
The correspondence between the two-component and four-component
spinor language is most easily exhibited
in the basis in which $\gamma_5$ is diagonal (this is called the {\it
chiral} representation).
Employing 2$\times$2 matrix blocks,
the gamma matrices are given by:
\begin{align}
\gamma^\mu = \begin{pmatrix} 0 & \quad \sigma^\mu_{\alpha{\dot{\beta}}}\\
\sigmabar^{\mu{\dot{\alpha}}\beta} &\quad 0\end{pmatrix}
\,,\quad
\gamma_5 \equiv i\gamma^0\gamma^1\gamma^2\gamma^3=\begin{pmatrix}
-\delta_\alpha{}^\beta & \quad 0\\ 0 &\quad \delta^{\dot{\alpha}}{}_{\dot{\beta}}
\end{pmatrix}
\,.
\end{align}
The chiral projections operators are
\begin{align}
P_L\equiv \half(1-\gamma_5)\,, \label{eq:projL} \\
P_R\equiv \half(1+\gamma_5)\,. \label{eq:projR}
\end{align}
In addition, we identify the generators of the Lorentz group
in the reducible $(\half,0)\oplus (0,\half)$ representation\footnote{In most textbooks,
$\Sigma^{\mu\nu}$ is called $\sigma^{\mu\nu}$.  Here, we use the
former symbol so that there is no confusion with the two-component
definition of $\sigma^{\mu\nu}$.}
\beq
\half\Sigma^{\mu\nu}\equiv\frac{i}{4}[\gamma^\mu,\gamma^\nu]=
\begin{pmatrix} \sigma^{\mu\nu}{}_\alpha{}^\beta & \quad 0\\ 
0 & \quad\sigmabar^{\mu\nu}{}^{\dot{\alpha}}{}_{\dot{\beta}}\end{pmatrix}\,,
\eeq
where $\Sigma^{\mu\nu}$ satisfies the duality relation,
$\gamma\ls{5}\Sigma^{\mu\nu}=\half i \epsilon^{\mu\nu\rho\tau}\Sigma_{\rho\tau}$.

A four-component Dirac spinor field, $\Psi(x)$, is made up of two
mass-degenerate two-component spinor fields, $\chi_\alpha(x)$ and
$\eta_\alpha(x)$ as follows:
\beq \label{diracspinor}
\Psi(x)\equiv\begin{pmatrix} \chi_\alpha(x)
\\[4pt] \eta^{\dagger\,\dot{\alpha}}(x)\end{pmatrix}\,.
\eeq
Note that $P_L$ and $P_R$ project out the upper and lower components,
respectively.
The Dirac conjugate field 
$\overline{\Psi}$ and the charge conjugate field $\Psi^c$ are
defined by
\beqa
\overline{\Psi}(x)&\equiv&\Psi^\dagger A =
\bigl(\eta^\alpha(x),\chi^\dagger_{\dot{\alpha}}(x)\bigr)\,,{} \label{psibar} \\[6pt]
\Psi^c(x)&\equiv&C\overline{\Psi}^{\T}(x)=
\begin{pmatrix} \eta_\alpha(x) \\[4pt] \chi^{\dagger\,\dot{\alpha}}(x)
\end{pmatrix}\,,  {}
\eeqa
where the Dirac conjugation matrix $A$ and the charge conjugation
matrix $C$ satisfy
\beq
A\gamma^\mu A^{-1}={\gamma^\mu}^\dagger\,,\qquad\qquad\qquad C^{-1}
\gamma^\mu C=-{\gamma^\mu}^{\T}\,.
\eeq
It is conventional to impose two additional conditions:
\beq \label{extraconditions}
\Psi=A^{-1}\Psibar^\dagger\,, \qquad\qquad (\Psi^c)^c=\Psi\,.
\eeq
The first of these conditions together with \eq{psibar} is equivalent
to the statement that $\Psibar\Psi$ is hermitian.
The second condition corresponds to the statement
that the (discrete)
charge conjugation transformation applied twice is equal to the
identity operator.  It then follows that
\beq \label{AC}
A^\dagger=A\,,\qquad\quad C^{\T}=-C\,,\qquad\quad (AC)^{-1}=(AC)^*\,.
\eeq
In the chiral representation, $A$ and $C$ are explicitly given by
\beq
A=\begin{pmatrix} 0 &\quad \delta^{\dot{\alpha}}{}_{\dot{\beta}} \\
\delta_\alpha{}^\beta &\quad   0\end{pmatrix}\,,\qquad
C =\begin{pmatrix} \epsilon_{\alpha\beta}& \quad   0\\
                 0 &\quad  \epsilon^{\dot{\alpha}\dot{\beta}}\end{pmatrix}
\,.
\eeq
Note the \textit{numerical} equalities, $A=\gamma^0$ and
$C=i\gamma^0\gamma^2$, although these identifications do not respect
the structure of the undotted and dotted indices specified above.

Finally, we note the following results, which are easily derived:
\beqa
&& \hspace{-0.4in} A\Gamma A^{-1} = \eta\ls{\Gamma}^A\Gamma^\dagger\,,\qquad
\eta\ls{\Gamma}^A=\begin{cases} +1\,, &
\text{\quad for $\Gamma=\mathds{1}\,,\,\gamma^\mu\,,\,
\gamma^\mu\gamma\ls{5}\,,\,\Sigma^{\mu\nu}$,}\\  -1\,, &
\text{\quad for $\Gamma=\gamma\ls{5}
\,,\,\Sigma^{\mu\nu}\gamma\ls{5}$\,,}\end{cases} \label{aagamma}
\\[6pt]
&&  \hspace{-0.4in}  C^{-1}\Gamma C= \eta\ls{\Gamma}^C \Gamma^{\T}\,,\qquad
\eta\ls{\Gamma}^C=\begin{cases} +1\,, &
\text{\quad for $\Gamma=\mathds{1}\,,\,\gamma\ls{5}\,,\,
\gamma^\mu\gamma\ls{5}$\,,} \\  -1\,, &
\text{\quad for $\Gamma=\gamma^\mu\,,\,\Sigma^{\mu\nu}
\,,\,\Sigma^{\mu\nu}\gamma\ls{5}$\,.}\end{cases} \label{ccgamma} 
\eeqa

\subsubsection{Four-component spinor bilinear covariants}

The Dirac bilinear covariants are quantities that are quadratic in the
Dirac spinor fields and transform irreducibly as Lorentz tensors. These
may be constructed from the corresponding quantities that are
quadratic in the two-component spinors.  To construct a translation
table between the two-component spinor and four-component spinor forms
of the bilinear covariants, we first define two Dirac spinor fields,
\beq
\Psi_1(x)\equiv\left(\begin{array}{c}{\chi_1} (x) \\[4pt]
{\eta^\dagger_1}(x)\end{array}\right)\,, \qquad\quad
\Psi_2(x)\equiv\left(\begin{array}{c}{\chi_2} (x) \\[4pt]
{\eta^\dagger_2}(x)\end{array}\right)\, ,
\eeq
where spinor indices have been suppressed.   It follows that,
\beqa
 &&\overline\Psi_1 \Psi_2 = \eta_1\chi_2 +
\chi^\dagger_1\eta^\dagger_2\,, \label{bilinear1}\\
&& \overline\Psi_1\gamma_5\Psi_2 = -\eta_1\chi_2 +
\chi^\dagger_1\eta^\dagger_2\,,\\
&& \overline\Psi_1\gamma^\mu\Psi_2 = \chi_1^\dagger\sigmabar^\mu\chi_2
       +\eta_1\sigma^\mu \eta^\dagger_2\,,\\
&& \overline\Psi_1\gamma^\mu\gamma_5\Psi_2 =
-\chi^\dagger_1\sigmabar^\mu\chi_2
       +\eta_1\sigma^\mu \eta^\dagger_2\,, \\
&& \overline\Psi_1\Sigma^{\mu\nu}\Psi_2 = 2(\eta_1 \sigma^{\mu\nu}
       \chi_2 + \chi^\dagger_1 \sigmabar^{\mu\nu} \eta^\dagger_2)\,,\\
&& \overline\Psi_1\Sigma^{\mu\nu}\gamma_5 \Psi_2 = -2(\eta_1
\sigma^{\mu\nu}
       \chi_2 - \chi^\dagger_1 \sigmabar^{\mu\nu} \eta^\dagger_2)
       \,.\label{bilinear6}
\eeqa
The above results can be used to to obtain the translations given
in Table~\ref{tab:24}.
\clearpage

\begin{table}[t!]
\begin{center}
\renewcommand{\arraystretch}{1.5}
\setlength{\tabcolsep}{1.5pc}
\caption{\small Relating the Dirac bilinear covariants written in
  terms of four-component Dirac spinor fields to the corresponding quantities
  expressed in terms of two-component spinor fields using the notation
  of \eq{diracspinor}.  These results apply to both commuting and
  anticommuting spinors.  In the latter case,
one may alternatively write $ \overline\Psi_1\gamma^\mu P_R\Psi_2 =
-\eta^\dagger _2\sigmabar^\mu\eta_1$,
etc. [cf.~\eq{eq:sigmucom}].
}
\label{tab:24}
\vskip 0.05in
\begin{tabular}{|l|l|} \hline
$\overline\Psi_1 P_L \Psi_2 = \eta_1\chi_2$
    &$\overline\Psi\lsup c_1 P_L \Psi_2^c = \chi_1\eta_2$ \\
$\overline\Psi_1 P_R \Psi_2 = \chi^\dagger_1\eta^\dagger_2$
    &$\overline\Psi_1\lsup c P_R \Psi_2^c = \eta^\dagger_1\chi^\dagger_2$ \\
$\overline\Psi\lsup c_1 P_L \Psi_2 = \chi_1\chi_2$
    &$ \overline\Psi_1 P_L \Psi_2^c = \eta_1\eta_2$ \\
$\overline\Psi_1 P_R \Psi^c_2 = \chi^\dagger_1\chi^\dagger_2$
  &$\overline\Psi\lsup c_1 P_R \Psi_2 = \eta^\dagger_1\eta^\dagger_2$ \\
$\overline\Psi_1 \gamma^\mu P_L\Psi_2 = \chi^\dagger_1\sigmabar^\mu\chi_2$
  &$\overline\Psi\lsup c_1 \gamma^\mu P_L\Psi_2^c =
                          \eta^\dagger_1\sigmabar^\mu\eta_2$\\
$\overline\Psi\lsup c_1\gamma^\mu P_R\Psi^c_2 = \chi_1\sigma^\mu
                                       \chi^\dagger_2$
  &$\overline\Psi_1\gamma^\mu P_R\Psi_2 = \eta_1\sigma^\mu
                                       \eta^\dagger_2 $ \\
$\overline\Psi_1 \Sigma^{\mu\nu}P_L \Psi_2
                = 2\,\eta_1\sigma^{\mu\nu}\chi_2$
  &$\overline\Psi_1\lsup c \Sigma^{\mu\nu}P_L \Psi_2^c
                          = 2\,\chi_1\sigma^{\mu\nu}\eta_2$ \\
$\overline\Psi_1 \Sigma^{\mu\nu}P_R \Psi_2
                = 2\,\chi^\dagger_1\sigmabar^{\mu\nu}\eta^\dagger_2$
  &$\overline\Psi_1\lsup c \Sigma^{\mu\nu}P_R \Psi_2^c
                          =
    2\,\eta^\dagger_1\sigmabar^{\mu\nu}\chi^\dagger_2$\\
\hline
\end{tabular}
\end{center}
\vskip -0.15in
\end{table}

When $\Psi_2 = \Psi_1$, the bilinear covariants listed in
\eqst{bilinear1}{bilinear6} are either hermitian or
anti-hermitian.   Using \eq{aagamma}, it follows that
$\overline\Psi \Gamma \Psi$ is
hermitian for $\Gamma = \mathds{1}_{4\times 4},\ i\gamma_5,\ \gamma^\mu,\ \gamma^\mu \gamma_5, \ \Sigma^{\mu\nu}$, and $i \Sigma^{\mu\nu}\gamma_5$.

One can also define Majorana bilinear covariants.  A four-component
Majorana fermion field is defined by the condition,
\beq \label{majcond}
\Psi_M(x)=\Psi^c_M(x)=C\overline\Psi_M^T(x)=\begin{pmatrix} \xi_\alpha(x) \\[3pt] \xi^{\dot\alpha\,\dagger}(x)\end{pmatrix}\,.
\eeq
\Eqst{bilinear1}{bilinear6} and the results of Table~\ref{tab:24}
may also be applied to four-component Majorana
spinors, $\Psi_{M1}$ and $\Psi_{M2}$, by setting $\xi_1\equiv\chi_1=\eta_1$,
and $\xi_2\equiv\chi_2=\eta_2$, respectively.
This implements the Majorana
condition given in \eq{majcond}
and imposes additional
restrictions on the Majorana bilinear covariants.  In particular, the
\textit{anticommuting} Majorana four-component fermion fields
satisfy the following additional identities,
\beqa
\overline\Psi_{M1}\Psi_{M2}&=&\overline\Psi_{M2}\Psi_{M1}
\,,{}\label{M1}\\
\overline\Psi_{M1}\gamma\ls{5}\Psi_{M2}&=&\overline\Psi_{M2}
\gamma\ls{5}\Psi_{M1}\,,{}\label{M2}\\
\overline\Psi_{M1}\gamma^\mu \Psi_{M2}&=&
-\overline\Psi_{M2}\gamma^\mu\Psi_{M1}\,,{}\label{M3}\\
\overline\Psi_{M1}\gamma^\mu\gamma\ls{5} \Psi_{M2}&=&
\overline\Psi_{M2}\gamma^\mu\gamma\ls{5}\Psi_{M1}\,,{}\label{M4}\\
\overline\Psi_{M1}\Sigma^{\mu\nu} \Psi_{M2}&=&
-\overline\Psi_{M2}\Sigma^{\mu\nu}\Psi_{M1}\,,{}\label{M5}\\
\overline\Psi_{M1}\Sigma^{\mu\nu}\gamma\ls{5} \Psi_{M2}&=&
-\overline\Psi_{M2}\Sigma^{\mu\nu}\gamma\ls{5}\Psi_{M1}\,. {}\label{M6}
\eeqa

If
$\Psi_{M1}=\Psi_{M2}\equiv\Psi_M$, then \eqst{M1}{M6} yield
\beq
\overline\Psi_{M}\gamma^\mu
\Psi_{M}=\overline\Psi_{M}\Sigma^{\mu\nu}
\Psi_{M}=\overline\Psi_{M}\Sigma^{\mu\nu}\gamma\ls{5} \Psi_{M}=0\,.\\
\eeq
One additional useful result for Majorana fermion fields is:
\beq
\overline\Psi_{M1}\gamma^\mu P_L\Psi_{M2}=
-\overline\Psi_{M2}\gamma^\mu P_R\Psi_{M1}\,.
\eeq

%
%%%%%%%%%%%%%%%%%%%%%%%%%%%%%%%%%%%%%%%%%%%%%%%%%%%%%%%%%%%%%
\subsection{Feynman Rules for Dirac and Majorana fermions}
\label{sec:Feynman}

The application of four-component fermion
techniques in parity-violating theories is straightforward for
processes involving Dirac fermions.   However, the inclusion of
Majorana fermions involves some subtleties that require elucidation.
In light of the widespread familiarity with four-component spinor
techniques, we shall develop four-component fermion Feynman rules
 that treat Dirac and Majorana fermions on equal 
footing\cite{Dreiner:2008tw,Gates:1987ay,Denner:1992me,Kleiss:2009hu}.\footnote{
For a comprehensive set of 
two-component fermion Feynman rules, see Ref.~\cite{Dreiner:2008tw}.}

Consider first the Feynman rule for the four-component fermion
propagator.
Virtual Dirac fermion lines can either correspond to $\Psi$ or
$\Psi^c$.  Here, there is no ambiguity in the propagator Feynman rule,
since for free Dirac fermion fields,
\beq
\left\langle 0\right|T[\Psi(x)\Psibar(y)]
\left|0\right\rangle=
\left\langle 0\right |T[\Psi^c(x)\overline{\Psi^c}(y)]
\left|0\right\rangle\,,
\eeq
so that the Feynman rules for the propagator of a $\Psi$ and $\Psi^c$
line, exhibited below, are identical.
The same rule also applies to a four-component Majorana fermion $\Psi_M$.
\begin{center}
\begin{picture}(200,50)(-135,-16)
\thicklines
\LongArrow(-110,25)(-70,25)
\ArrowLine(-130,15)(-50,15)
\put(-90,30){$p$}
\put(20,10){$\displaystyle
 {\frac{i(\slashchar{p}\BDplus m)}
 {p^2 \BDminus m^2 \BDplus i\epsilon}}$}
\end{picture}
\end{center}

\vspace{-0.2in}
Consider next a set of neutral Majorana fermions $\Psi_{Mi}$ and
charged Dirac fermions $\Psi_i$,
\beq
\Psi_{Mi} =
\begin{pmatrix}
\xi_i
\\[4pt]
\xi^\dagger_i
\end{pmatrix},
\qquad
\Psi_i =
\begin{pmatrix}
\chi_i \\[4pt]
\eta^\dagger_i
\end{pmatrix},
\eeq
 interacting with a neutral scalar $\phi$ or
vector boson $A_\mu$.  The interaction Lagrangian in terms of two-component
fermions is
\beqa \!\!\!\!\!\!\!\! \!\!\!\!\!\!\!\!
\mathscr{L}_{\rm int} &=& -\half(\lambda^{ij}\xi_i\xi_j+\lambda_{ij}
\xi^{\dagger\,i}\xi^{\dagger\,j})\phi-(\kappa^i{}_j\chi_i\eta^j+\kappa_i{}^j
\chi^{\dagger\,i}\eta^\dagger_j)\phi\ {} \nonumber \\
&&\, -G_i{}^j\,\xi^{\dagger\,i}\sigmabar^\mu\xi_j A_\mu
-[(G_L)_i{}^j\chi^{\dagger\,i}\sigmabar^\mu\chi_j
+(G_R)_i{}^j\eta^{\dagger\,i}\sigmabar^\mu\eta_j]A_\mu\,,{} \label{lint1}
\eeqa
where $\lambda$ is a complex symmetric matrix with 
$\lambda^{ij}\equiv\lambda^*_{ij}$,
$\kappa$ is an arbitrary complex matrix with $\kappa_i{}^j\equiv (\kappa^i{}_j)^*$,
and $G$, $G_L$ and $G_R$
are hermitian matrices.
Converting to four-component spinor notation (see Problem 1), the resulting Feynman rules 
are shown below.
\clearpage

%%%%%%%%%%%%%%%%%%%%%%%%%%%%%
\begin{figure}[t!]
\begin{center}
\begin{picture}(200,68)(40,0)
\DashLine(10,40)(60,40)5
\ArrowLine(60,40)(100,70)
\ArrowLine(100,10)(60,40)
\Text(30,30)[]{$\scriptstyle\phi$}
\Text(70,20)[]{$\scriptstyle\Psi_{Mj}$}
\Text(70,67)[]{$\scriptstyle\Psi_{Mi}$}
\Text(140,40)[l]{$-i(\lambda^{ij}P_L+\lambda_{ij} P_R)$}
\end{picture}
\end{center}
\vspace{0.2in}

%%%%%%%%%%%%%%%%%%%%%
\begin{center}
\begin{picture}(200,68)(40,0)
\Photon(60,40)(10,40){3}{5}
\ArrowLine(60,40)(100,70)
\ArrowLine(100,10)(60,40)
\Text(30,25)[]{$\scriptstyle A_\mu$}
\Text(70,20)[]{$\scriptstyle\Psi_{Mj}$}
\Text(70,67)[]{$\scriptstyle\Psi_{Mi}$}
\Text(140,40)[l]{$-i\gamma_\mu[G_i{}^j P_L-G_j{}^i P_R]$}
\end{picture}
\end{center}
\vspace{0.2in}
%%%%%%%%%%%%%%%%%%%%%%%%%%%%
\begin{center}
\begin{picture}(200,68)(40,0)
\DashLine(10,40)(-40,40)5
\ArrowLine(10,40)(50,70)
\ArrowLine(50,10)(10,40)
\Text(-20,30)[]{$\scriptstyle\phi$}
\Text(20,20)[]{$\scriptstyle\Psi_j$}
\Text(20,67)[]{$\scriptstyle\Psi_i$}
\DashLine(160,40)(110,40)5
\ArrowLine(160,40)(200,70)
\ArrowLine(200,10)(160,40)
\Text(75,40)[]{or}
\Text(130,30)[]{$\scriptstyle\phi$}
\Text(170,15)[]{$\scriptstyle\Psi^{cj}$}
\Text(170,65)[]{$\scriptstyle\Psi^{ci}$}
\Text(240,40)[l]{$-i(\kappa^i{}_j P_L+\kappa_j{}^i P_R)$}
\end{picture}
\end{center}
%%%%%%%%%%%%%%%%%%%%%%
\vspace{0.2in}
%%%%%%%%%%%%%%%%%%%%%
\begin{center}
\begin{picture}(200,68)(40,0)
\Photon(60,40)(10,40){3}{5}
\ArrowLine(60,40)(100,70)
\ArrowLine(100,10)(60,40)
\Text(30,25)[]{$\scriptstyle A_\mu$}
\Text(70,20)[]{$\scriptstyle\Psi_j$}
\Text(70,67)[]{$\scriptstyle\Psi_i$}
\Text(140,40)[l]{$-i\gamma_\mu[(G_L)_i{}^j P_L+(G_R)_i{}^j P_R]$}
\end{picture}
\end{center}
%%%%%%%%%%%%%%%%%%%%%
\vspace{0.2in}
%%%%%%%%%%%%%%%%%%%%%
\begin{center}
\begin{picture}(200,68)(40,0)
\Photon(60,40)(10,40){3}{5}
\ArrowLine(60,40)(100,70)
\ArrowLine(100,10)(60,40)
\Text(50,90)[]{or}
\Text(200,90)[]{or}
\Text(30,25)[]{$\scriptstyle A_\mu$}
\Text(70,15)[]{$\scriptstyle\Psi^{cj}$}
\Text(70,65)[]{$\scriptstyle\Psi^{ci}$}
\Text(140,40)[l]{$\phm i\gamma_\mu[(G_L)_i{}^j P_L+(G_R)_i{}^j P_R]$}
\end{picture}
\end{center}
%%%%%%%%%%%%%%%%%%%%%
%\caption{\small Feynman rules}
\end{figure}

The arrows on the Dirac fermion lines depict the flow of the
conserved charge.  A Majorana fermion is self-conjugate, so
its arrow simply reflects the structure of $\mathscr{L}_{\rm int}$;
{\it i.e.}, $\overline\Psi_M$ [$\Psi_M$] is represented by
an arrow pointing out of [into] the vertex.  The arrow directions
determine the placement of the $u$ and $v$ spinors in an
invariant amplitude.

For vertices involving Dirac fermions, one has a choice of either
using the Dirac field or its charge conjugated field.  The Feynman
rules corresponding to these two choices are related, due to the
following identity, 
\beq \label{CC}
\overline\Psi^c_i\Gamma\Psi^c_j=-\Psi_i^T C^{-1}\Gamma C\overline\Psi_j^T=
\overline\Psi_j C\Gamma^T
C^{-1}\Psi_i=\eta^C\ls{\Gamma}\overline\Psi_j\Gamma \Psi_i\,,
\eeq
where we have used \eq{ccgamma}.   Note that the extra minus sign that
arises in the penultimate step above is due to the anticommutativity
of the fermion fields.

Next, consider the interaction of fermions with charged bosons $\Phi$ and $W$ (assumed
to have charge equal to that of $\chi$ and $\eta^\dagger$).  The corresponding interaction Lagrangian is given by:
\beqa
\mathscr{L}_{\rm int} &=&
-\Phi[(\kappa_1)^i{}_j\xi_i\eta^j
+(\kappa_2)_{ij}\xi^{\dagger i} \chi^{\dagger j}] 
 -\Phi^\dagger[(\kappa_2)^{ij}\xi_i\chi_j
+(\kappa_1)_i{}^j\xi^{\dagger i}_i \eta^{\dagger}_j] \nonumber
\\
&&
\BDminus W_\mu[(G_1)_j{}^i\chi^{\dagger j}\sigmabar^\mu\xi_i
-(G_2)_{ij}\xi^{\dagger i}\sigmabar^\mu \eta^j]  \nonumber \\
&& \BDminus W_\mu^\dagger[(G_1)^j{}_i\xi^{\dagger i}\sigmabar^\mu\chi_j
-(G_2)^{ij}\eta^{\dagger}_j\sigmabar^\mu\xi_i]\,,\label{lint2}
\eeqa
where $\kappa_1$, $\kappa_2$, $G_1$ and $G_2$ 
are complex matrices.    Converting to four-component spinor notation,
the corresponding Feynman rules are:
\begin{center}
\begin{picture}(200,78)(20,0)
\DashArrowLine(-60,40)(-10,40)5
\ArrowLine(-10,40)(30,70)
\ArrowLine(30,10)(-10,40)
\Text(-50,30)[]{$\scriptstyle\Phi$}
\Text(-10,18)[]{$\scriptstyle\Psi_{Mi}$}
\Text(-10,67)[]{$\scriptstyle\Psi_{j}$}
\Text(45,40)[l]{or}
\DashArrowLine(80,40)(130,40)5
\ArrowLine(170,70)(130,40)
\ArrowLine(130,40)(170,10)
\Text(100,30)[]{$\scriptstyle\Phi$}
\Text(140,18)[]{$\scriptstyle\Psi_{Mi}$}
\Text(140,67)[]{$\scriptstyle\Psi^{cj}$}
\Text(200,40)[l]{$-i(\kappa_1{}^i{}_j P_L+\kappa_{2ij} P_R)$}
\end{picture}
\end{center}
\vspace{0.04in}
%%%%%%%%%%%%%%%%%%%%%%%%%%%%
\begin{center}
\begin{picture}(200,68)(20,0)
\DashArrowLine(-10,40)(-60,40)5
\ArrowLine(30,70)(-10,40)
\ArrowLine(-10,40)(30,10)
\Text(-40,30)[]{$\scriptstyle\Phi$}
\Text(0,18)[]{$\scriptstyle\Psi_{Mi}$}
\Text(0,67)[]{$\scriptstyle\Psi_{j}$}
\Text(45,40)[l]{or}
\DashArrowLine(130,40)(80,40)5
\ArrowLine(130,40)(170,70)
\ArrowLine(170,10)(130,40)
\Text(100,30)[]{$\scriptstyle\Phi$}
\Text(140,18)[]{$\scriptstyle\Psi_{Mi}$}
\Text(140,67)[]{$\scriptstyle\Psi^{cj}$}
\Text(200,40)[l]{$-i(\kappa_2{}^{ij}P_L+\kappa_{1i}{}^j P_R)$}
\end{picture}
\end{center}
\vspace{0.04in}
%%%%%%%%%%%%%%%%%%%%%
\begin{center}
\begin{picture}(200,78)(20,0)
\Photon(-20,40)(30,40){3}{5}
\ArrowLine(30,40)(70,70)
\ArrowLine(70,10)(30,40)
\ArrowLine(5.05,40)(4.95,40)
\Text(0,30)[]{$\scriptstyle W$}
\Text(40,18)[]{$\scriptstyle\Psi_{Mj}$}
\Text(40,67)[]{$\scriptstyle\Psi_{i}$}
\Text(160,40)[l]{$-i\gamma^\mu(G_{1i}{}^jP_L-G_{2ji} P_R)$}
\end{picture}
\end{center}
\vspace{0.3in}
%%%%%%%%%%%%%%%%%%%%%%%%%%%%
\begin{center}
\begin{picture}(200,68)(20,0)
\Photon(-20,40)(30,40){3}{5}
\ArrowLine(70,70)(30,40)
\ArrowLine(30,40)(70,10)
\ArrowLine(5.05,40)(4.95,40)
\Text(0,30)[]{$\scriptstyle W$}
\Text(40,18)[]{$\scriptstyle\Psi_{Mj}$}
\Text(40,67)[]{$\scriptstyle\Psi^{ci}$}
%\Text(20,40)[l]
%
\Text(170,40)[l]{$i\gamma^\mu(G_{2ji} P_R -G_{1i}{}^j P_L)$}
\Text(20,95)[]{or}
\Text(210,95)[]{or}
\end{picture}
\end{center}
\vspace{0.3in}
%%%%%%%%%%%%%%%%%%%%%
\begin{center}
\begin{picture}(200,68)(20,0)
\Photon(30,40)(-20,40){3}{5}
\ArrowLine(70,70)(30,40)
\ArrowLine(30,40)(70,10)
\ArrowLine(4.95,40)(5.05,40)
\Text(0,30)[]{$\scriptstyle W$}
\Text(40,18)[]{$\scriptstyle\Psi_{Mj}$}
\Text(40,67)[]{$\scriptstyle\Psi_{i}$}
%\Text(20,40)[l]
%
\Text(160,40)[l]{$-i\gamma^\mu(G_1{}^i{}_j P_L-G_{2}{}^{ji} P_R)$}
\end{picture}
\end{center}
\vspace{0.3in}
%%%%%%%%%%%%%%%%%%%%%
\begin{center}
\begin{picture}(200,68)(20,0)
\Photon(30,40)(-20,40){3}{5}
\ArrowLine(30,40)(70,70)
\ArrowLine(70,10)(30,40)
\ArrowLine(4.95,40)(5.05,40)
\Text(20,95)[]{or}
\Text(210,95)[]{or}
\Text(0,30)[]{$\scriptstyle W$}
\Text(40,18)[]{$\scriptstyle\Psi_{Mj}$}
\Text(40,67)[]{$\scriptstyle\Psi^{ci}$}
%\Text(20,40)[l]
%
\Text(170,40)[l]{$i\gamma^\mu(G_{2}{}^{ji} P_R -G_{1}{}^i{}_j P_L)$}
\end{picture}
\end{center}
%%%%%%%%%%%%%%%%%%%%%

When the interaction Lagrangians given in \eqs{lint1}{lint2} are
converted to four-component spinor notation (see Problems 1 and 2 at
the end of this section), there is an equivalent form in which
$\mathscr{L}_{\rm int}$ is written in terms of charge-conjugated Dirac
four-component fields [after using \eq{CC}].  Thus, the  Feynman rules involving
Dirac fermions can take two possible forms, as
shown above.
As previously noted, the direction of an
arrow on a Dirac fermion line indicates the
direction of the fermion charge flow (whereas the arrow on
the Majorana fermion line is unconnected to charge flow).
However, we are free to choose either a
$\Psi$ or $\Psi^c$ line to represent a Dirac fermion at any place in a
given Feynman graph.\footnote{Since the charge of $\Psi^c$ is opposite
in sign to the charge
of $\Psi$, the corresponding arrow directions of the $\Psi$ and $\Psi^c$
lines must point in opposite directions.}
For any decay or scattering process,
a suitable choice of either the $\Psi$-rule or the $\Psi^c$-rule
at each vertex (the choice can be different at different vertices)
will guarantee that
the arrow directions on fermion lines flow continuously through
the Feynman diagram.  Then, to evaluate an invariant amplitude,
one should traverse \textit{any} continuous fermion
line (either $\Psi$ or $\Psi^c$)
by moving antiparallel to the direction of the fermion arrows.

For a given process, there may be a number of distinct
choices for the arrow directions on the Majorana fermion lines,
which may depend on whether one represents a given Dirac fermion by
$\Psi$ or $\Psi^c$.
However, different choices do {\it not} lead to independent Feynman
diagrams.
When computing an invariant amplitude, one
first writes down the relevant
Feynman diagrams with no arrows on any Majorana
fermion line.  The number of distinct graphs contributing to the
process is then determined.  Finally, one makes some choice for
how to distribute the arrows on the Majorana fermion lines
and how to label Dirac fermion lines (either as the field $\Psi$ or its
charge conjugate $\Psi^c$) in a manner consistent
with the Feynman rules for the vertices previously given.
The end result for the invariant
amplitude (apart from an overall unobservable phase)
does not depend on the choices
made for the direction of the fermion arrows.

Using the above procedure, the Feynman rules for the
external fermion wave functions are the same for Dirac and Majorana fermions:
\begin{itemlist}
\item
$u(\boldsymbol{\vec p},s)$: incoming $\Psi$ [or $\Psi^c$]
with momentum $\boldsymbol{\vec p}$ parallel to the arrow direction,
\item
$\bar u(\boldsymbol{\vec p},s)$: outgoing $\Psi$ [or $\Psi^c$] with
momentum $\boldsymbol{\vec p}$ parallel to the arrow direction,
\item
$v(\boldsymbol{\vec p},s)$: outgoing $\Psi$ [or $\Psi^c$] with
momentum $\boldsymbol{\vec p}$ anti-parallel to the arrow direction,
\item
$\bar v(\boldsymbol{\vec p},s)$: incoming $\Psi$ [or $\Psi^c$] with
momentum $\boldsymbol{\vec p}$ anti-parallel to the arrow direction.
\end{itemlist}

We now consider the application of the Feynman rules presented above
to some $2\to 2$ scattering processes involving a Majorana fermion
either as an external state or as an internal line.

\begin{example}[$\boldsymbol{\Psi(p_1)\Psi(p_2)\to\Phi(k_1)\Phi(k_2)}$ via $\boldsymbol{\Psi_M}$-exchange]

Here, $\Phi$ is a charged scalar.
The contributing Feynman graphs are:

\vspace{12pt}

\begin{picture}(450,85)(125,-25)
\thicklines
\ArrowLine(185,-15)(125,-15)
\DashArrowLine(185,-15)(245,-15){5}
\ArrowLine(125,45)(185,45)
\DashArrowLine(185,45)(245,45){5}
\ArrowLine(185,45)(185,-15)
\put(165,12){$\Psi_M$}
\put(130,50){$\Psi$}
\put(130,-25){$\Psi^c$}
\ArrowLine(360,-15)(300,-15)
\DashLine(360,-15)(390,15){5}
\DashArrowLine(390,15)(420,45){5}
\ArrowLine(300,45)(360,45)
\DashLine(360,45)(390,15){5}
\DashArrowLine(390,15)(420,-15){5}
\ArrowLine(360,45)(360,-15)
\put(340,12){$\Psi_M$}
\put(305,50){$\Psi$}
\put(305,-25){$\Psi^c$}
\end{picture}

\vspace{12pt}

\noindent
Following the arrows on the fermion lines in reverse,
the invariant amplitude is given by,
\beqa \label{mex2}
i\mathcal{M}&=&
(-i)^2\bar v(\boldsymbol{\vec p}_2,s_2)(\kappa_1 P_L+\kappa_2^* P_R)
\left[\frac{i(\slashchar{p_1}-\slashchar{k_1}+m)}{t-m^2} 
+\frac{i(\slashchar{k_1}-\slashchar{p_2}+m)}{u-m^2}\right]\nonumber \\
&&\qquad \times (\kappa_1 P_L+\kappa_2^* P_R) u(\boldsymbol{\vec p}_1,s_1)\,,{}
\eeqa
where $t\equiv (p_1-k_1)^2$, $u\equiv (p_2-k_1)^2$ and
$m$ is the Majorana fermion mass.  The sign of each diagram is
determined by the relative permutation of spinor wave functions
appearing in the amplitude (the overall sign of the amplitude is
unphysical).
In the present example, in
both terms appearing in \eq{mex2}, the spinor wave functions appear in 
the same order (first $\boldsymbol{\vec p}_2$ and then
$\boldsymbol{\vec p}_1$),
implying a relative plus sign between the two terms.

One can check that  \textrm{$i\mathcal{M}$} is antisymmetric under
interchange of the two initial electrons. 
This is most easily verified by
taking the transpose of the invariant amplitude (the latter is a
complex number whose value is not changed by transposition).
It is convenient to adopt the convention in which the (commuting) $u$
and $v$ spinor wave functions are related via,
\beqa
v(\boldsymbol{\vec p},s) &=& C\ubar(\boldsymbol{\vec p},s)^{\T}
\,,\qquad\qquad\quad
u(\boldsymbol{\vec p},s) = C\vbar(\boldsymbol{\vec p},s)^{\T}\,,
\label{uvspinrelation1} \\
\vbar(\boldsymbol{\vec p},s) &=& -u(\boldsymbol{\vec p},s)^{\T}C^{-1}
\,,\qquad\quad\,
\ubar(\boldsymbol{\vec p},s) = -v(\boldsymbol{\vec p},s)^{\T}C^{-1}\,.
\label{uvspinrelation2}
\eeqa
where $C$ is the charge conjugation matrix.
Using \eqs{uvspinrelation1} {uvspinrelation2}, 
the transposed amplitude can be simplified by employing the relation,
\begin{align}
\bar v(\boldsymbol{\vec p}_2,s_2)\Gamma u(\boldsymbol{\vec p}_1,s_1)=
-\eta^C\ls{\Gamma}
\bar v(\boldsymbol{\vec p}_1,s_1)\Gamma u(\boldsymbol{\vec
  p}_2,s_2)\,, \label{vgamu}
\end{align}
which is a consequence of \eq{ccgamma}.

\end{example}

\begin{example}[$\boldsymbol{\Psi(p_1)\Psi^c(p_2)\!\to\!\Psi_M(p_3)\Psi_M(p_4)}$~\!via\! charged\! $\boldsymbol{\Phi}$-exchange]

In addition to a possible $s$-channel annihilation graph,
the contributing Feynman graphs can be represented by 
either diagram set (i) or diagram set (ii) shown below, where each set
contains a $t$-channel and $u$-channel graph, respectively.
\clearpage

\noindent Diagram set (i):

\begin{picture}(450,85)(130,-10)
\thicklines
\ArrowLine(125,-15)(185,-15)
\ArrowLine(185,-15)(245,-15)
\ArrowLine(125,45)(185,45)
\ArrowLine(185,45)(245,45)
\DashArrowLine(185,45)(185,-15){5}
\put(220,50){$\Psi_M$}
\put(220,-25){$\Psi_M$}
\put(130,50){$\Psi$}
\put(130,-25){$\Psi^c$}
\ArrowLine(290,-15)(350,-15)
\Line(380,15)(350,-15)
\ArrowLine(380,15)(410,45)
\ArrowLine(290,45)(350,45)
\Line(350,45)(380,15)
\ArrowLine(380,15)(410,-15)
\DashArrowLine(350,45)(350,-15){5}
\put(410,50){$\Psi_M$}
\put(410,-25){$\Psi_M$}
\put(295,50){$\Psi$}
\put(295,-25){$\Psi^c$}
\end{picture}

\vskip 0.5in
\noindent
Diagram set (ii):

\begin{picture}(450,85)(130,-10)
\thicklines
\ArrowLine(185,-15)(125,-15)
\ArrowLine(245,-15)(185,-15)
\ArrowLine(125,45)(185,45)
\ArrowLine(185,45)(245,45)
\DashArrowLine(185,45)(185,-15){5}
\put(220,50){$\Psi_M$}
\put(220,-25){$\Psi_M$}
\put(130,50){$\Psi$}
\put(130,-25){$\Psi$}
\ArrowLine(350,-15)(290,-15)
\Line(380,15)(350,-15)
\ArrowLine(410,45)(380,15)
\ArrowLine(290,45)(350,45)
\Line(350,45)(380,15)
\ArrowLine(380,15)(410,-15)
\DashArrowLine(350,45)(350,-15){5}
\put(410,50){$\Psi_M$}
\put(410,-25){$\Psi_M$}
\put(295,50){$\Psi$}
\put(295,-25){$\Psi$}
\end{picture}
\vskip 0.5in

The amplitude is evaluated by following the arrows on the fermion
lines in reverse.   Either diagram set (i) or set (ii) may be chosen to evaluate the invariant amplitude.
We again employ \eq{ccgamma} to derive the relation,
\beq \label{vgamv}
\bar v(\boldsymbol{\vec p}_2,s_2)\Gamma v(\boldsymbol{\vec p}_4,s_4)=
-\eta^C\ls{\Gamma}
\bar u(\boldsymbol{\vec p}_4,s_4)\Gamma u(\boldsymbol{\vec p}_2,s_2)\,,
\eeq
which can be used in comparing the invariant amplitude obtained by
using diagram sets (i) and (ii).  One can check that 
the invariant amplitudes resulting from diagram sets
(i) and (ii) differ by an overall minus sign, which is unphysical.
The overall minus sign arises due to the fact that the corresponding
order of the spinor wave functions
differs by an odd permutation [e.g.,
for the $t$-channel graphs, compare 3142 and
3124 for (i) and (ii) respectively].  For the same
reason, there is a relative minus sign between the $t$-channel and
$u$-channel graphs for either diagram set [e.g., compare 3142
and 4132 in diagram set(i)].

If $s$-channel annihilation contributes, its contribution to the
invariant amplitude is easily obtained.
Relative to the $t$-channel graph of diagram set
(ii) above, the $s$-channel graph shown below
comes with an extra minus sign (since 2134 is odd with respect to 3124).
\end{example}

\begin{picture}(450,95)(28,-40)
\SetScale{0.8}
\thicklines
\ArrowLine(185,15)(125,-25)
\ArrowLine(305,-25)(245,15)
\ArrowLine(125,55)(185,15)
\ArrowLine(245,15)(305,55)
\DashLine(185,15)(245,15){5}
\put(220,45){$\Psi_M$}
\put(220,-25){$\Psi_M$}
\put(110,45){$\Psi$}
\put(110,-25){$\Psi$}
\end{picture}

\vskip -0.05in

In the computation of the unpolarized cross-section, non-standard spin
projection operators can arise in the evaluation of the interference
terms (see Appendix D of Reference \cite{Haber:1984rc}), such as
\beqa
 \sum_s u({\boldsymbol{\vec p}},s) v^T({\boldsymbol{\vec p}},s)
 = (\slashchar{p} + m)C^T\,, \qquad
\sum_s \bar{u}^T({\boldsymbol{\vec p}},s)
\bar{v}({\boldsymbol{\vec p}},s)
=  C^{-1}(\slashchar{p} - m)\,,\nonumber 
\eeqa
which requires additional manipulation of the charge conjugation
matrix~$C$.  However, these non-standard spin projection operators can be
avoided by judicious use of spinor wave function product relations of the kind
obtained in \eqs{vgamu}{vgamv}.

\subsection{Problems}
\begin{problem}
Convert the interaction Lagrangian given by \eq{lint1}  to
four-component spinor notation.
Show that the end result is 
\beqa
\mathscr{L}_{\rm int}&=& -\half(\lambda^{ij}\Psibar_{Mi} P_L\Psi_{Mj}
+\lambda_{ij}\Psibar_{M}\llsup{i}P_R\Psi_{M}\llsup{j})\phi
-\Psibar\llsup{\,j}(\kappa^i{}_j P_L +\kappa_i{}^j P_R)\Psi_{i}\phi  \nonumber \\
&& \BDminus\half\Psibar_{Mi}\gamma^\mu\left[(G^a)_i{}^j P_L-(G^a)_j{}^i P_R\right]
\Psi_{Mj}  \nonumber \\
&& 
\BDminus \left[(G_L^a)_i{}^j\Psibar\llsup{\,i}\gamma^\mu P_L\Psi_{j}
+(G_R^a)_i{}^j\Psibar\llsup{\,i}\gamma^\mu
P_R\Psi_{j}\right]A^a_\mu\,, \label{lint14}
\eeqa
where the $\Psi_{Mj}$ are a set of (neutral) Majorana
four-component fermions and the $\Psi_{j}$ are a set of Dirac four-component fermions.
\end{problem}

\begin{problem}
Convert the interaction Lagrangian given by \eq{lint2}  to
four-component spinor notation.
Show that the end result is 
\beqa 
\!\!\!\!\!\!\!\!\!\!
\mathscr{L}_{\rm int} &=&
-\left[(\kappa_1)^i{}_j\Psibar\llsup{\,j}P_L\Psi_{Mi}
+(\kappa_2)_{ij}\Psibar\llsup{j}P_R\Psi_{M}^i\right]\Phi
\nonumber \\
&&
\BDminus \left[(G_1)_j{}^i\Psibar\llsup{\,j}\gamma^\mu P_L\Psi_{Mi}
+(G_2)_{ij}\Psibar\llsup{\,j}\gamma^\mu P_R\Psi_{M}^i\right]W_\mu
+{\rm h.c.} \label{lintc4}
\eeqa
\end{problem}
\begin{problem}
Derive \eq{vgamu}.  Then,
verify that the invariant amplitude given by \eq{mex2} is
antisymmetric under the interchange of the two initial electrons.
\end{problem}

\begin{problem}
Derive \eq{vgamv}.  Then, verify that the invariant amplitude 
for the scattering process considered in Example 3 obtained
from diagram sets (i) and (ii), respectively, differ by an overall
minus sign.
\end{problem}

%
%%%%%%%%%%%%%%%%%%%%%%%%%%%%%%%%%%%%%%%%%%%%%%%%%%%
\section{Motivation for TeV-scale supersymmetry}
\label{sec:motivation}
\renewcommand{\theequation}{\arabic{section}.\arabic{equation}}
\setcounter{equation}{0}

The Standard Model (SM) of particle physics has been remarkably
successful for describing the observed behavior of the fundamental
particles and their interactions\cite{Langacker}. 
Indeed, there are no definitive departures from the Standard Model observed in experiments conducted at high energy collider facilities.  
Nevertheless, some fundamental microscopic phenomena must necessarily lie outside of the purview of the SM.
These include: neutrinos with non-zero mass\cite{numass}; dark matter\cite{darkmatter}; the suppression of CP-violation in the strong interactions (the so-called strong CP problem\cite{Kim:2008hd}); gauge coupling unification\cite{guts}; the baryon asymmetry of the universe\cite{White:2016nbo}; inflation in the early universe\cite{inflation}; dark energy\cite{darkenergy}; and the gravitational interaction.  None of these phenomena can be explained within the framework of the SM alone.

Consequently, the SM should be regarded at best as a low-energy effective field theory~\cite{eft}, which is valid below some high energy scale.  
That is, new high energy scales must exist where more fundamental physics resides.
In this section, we explain why one might expect to find this new
physics at the TeV scale. We discuss the \textit{principle of
  naturalness}, and how supersymmetry provides a natural mechanism for
avoiding the quadratic sensitivity of the  squared-masses of
elementary scalar particles to ultraviolet physics.
%avoid the problem of quadratic divergences in the squared-masses of elementary scalar particles.

\subsection{Why the \textrm{TeV} scale?}
The classical gravitational interaction lies outside the SM.   Using
the fundamental constants, $\hbar$, $c$ and Newton's gravitational
constant $G_N$, one can construct a quantity with the units of energy
called the Planck scale,
\beq
M_{\rm PL}c^2\equiv \left(\frac{\hbar c^5}{G_N}\right)^{1/2}\simeq
1.2\times 10^{19}~{\rm GeV}\,.
\eeq
The significance of the Planck scale can be seen as follows.
At the Planck energy scale, the quantum mechanical
aspects of gravity can no longer be neglected.
The gravitational energy of a particle of mass $m$,
evaluated at its Compton wavelength, $r_c=\hbar/(mc)$,
\beq
\Phi\sim\frac{G_N m^2}{r_c}=\frac{G_N m^3 c}{\hbar}\lsim 2mc^2\,,
\eeq
must be below $2mc^2$ to avoid particle-antiparticle pair
creation by the gravitational field. Hence, up to 
$\mathcal{O}(1)$ constants, we conclude that $m\lsim M_{\rm
PL}$.\footnote{Note that for $m=M_{\rm PL}$, the Schwarzschild radius
$r_s\equiv 2G_N m/c^2\simeq r_c$, which provides additional evidence
that the quantum mechanical nature of gravity cannot be neglected at
energy scales above the Planck scale.}
Since particle-antiparticle pair creation is an inherently quantum
mechanical phenomenon, 
%we conclude that above the Planck energy scale,
quantum gravitational effects can no longer be ignored at the Planck scale. 
%no longer be ignored.  
Thus, the SM cannot be a
fundamental theory of particles and interactions at energy scales of
order the Planck scale and above.

There must be an energy scale $\Lambda$ at which the Standard Model
breaks down.  Based on the arguments given above, it follows that the
upper bound on $\Lambda$ is the Planck scale.   But, it is possible
that $\Lambda$ lies significantly below the Planck scale.
For example, a credible theory of neutrino masses (e.g., the type-I seesaw model~\cite{numass}) posits the existence of a right-handed electroweak singlet Majorana neutrino of mass of order $10^{14}~{\rm GeV}$.   
Henceforth, we shall define $\Lambda$ to be the lowest energy scale at
which the SM breaks down.  

The predictions made by the SM depend on a number of parameters that
must be taken as input to the theory.   These parameters cannot be
predicted, since their values are 
sensitive to unknown ultraviolet (UV) physics.
%, whose values cannot be predicted since the physics at very
%high energies is not known, one cannot predict their values.
In the 1930s, it was already appreciated
that a critical difference exists between the behavior of boson and
fermion masses~\cite{Weisskopf:1939zz}.  Fermion masses are
logarithmically sensitive to UV physics~\cite{Weisskopf:1934} due to
the chiral symmetry of massless fermions, which implies that the radiative
correction to the tree-level fermion mass is of the form,
\beq
\delta m_F\sim m_F\ln(\Lambda^2/m_F^2)\,,
\eeq
which vanishes in the limit of $m_F\to 0$.
In contrast, no such symmetry exists for bosons (in the absence of supersymmetry), and consequently we expect quadratic sensitivity of
the boson squared-mass to UV physics,
$\delta m^2_B\sim \Lambda^2\,.$

These observations have important consequences for the fundamental physics
that describes the Higgs boson.  
In the SM,  the Higgs boson squared-mass is given by $m_h^2=\lambda v^2$ and the W
boson squared-mass is $m_W^2=\tfrac{1}{4}g^2 v^2$, where
$\vev{\Phi^0}=v/\sqrt{2}=174$~GeV is the vacuum
expectation value of the neutral Higgs field, $\lambda$~is
the Higgs self-coupling [cf.~\eq{vofphi}], and $g$ is the SU(2) gauge coupling.
Together, these imply that
\beq
\frac{m_h^2}{m_W^2}=\frac{4\lambda}{g^2}\,,
\eeq
which one would expect to be roughly of $\mathcal{O}(1)$.  The Higgs
boson with mass 125~GeV
satisfies this expectation.  

However, the existence of the Higgs boson is a consequence of a spontaneously broken scalar potential, 
\beq
V(\Phi)=-\mu^2(\Phi^\dagger\Phi)+\half\lambda(\Phi^\dagger\Phi)^2\,,\label{vofphi}
\eeq
where $\mu^2=\half\lambda v^2$ at
the minimum of the scalar potential.
% in terms of the vacuum expectation value $v$ of the Higgs field.  
The parameter $\mu^2$ is quadratically sensitive to $\Lambda$.  Hence, to obtain $v=246$~GeV in a theory
where $v\ll \Lambda$ requires a significant fine-tuning of the ultraviolet parameters of the fundamental theory.
Indeed, the one-loop contributions to the squared mass parameter $\mu^2$ are expected to be of
order $(g^2/16\pi^2)\Lambda^2$.  Setting this quantity to be of order of $v^2$ (to avoid an \textit{unnatural} cancellation
between the tree-level parameter and the loop corrections) yields
\beq
\Lambda\simeq 4\pi v/g\sim {O}(1~{\rm TeV})\,.
\eeq
Thus, a \textit{natural} theory of electroweak symmetry breaking
(EWSB) appears to require new TeV scale physics beyond the SM associated with the EWSB dynamics.

%%%%%%%%%%%%%%%%%%%%%%%%%%%%%%%%%%%%%%%%%%%%%%%%%%

\subsection{The modern principle of naturalness}
This principle of naturalness was 
first introduced by Weisskopf in a paper published in 1939\cite{Weisskopf:1939zz}.
In the abstract of this 1939 paper, Weisskopf wrote,
``the self-energy of charged particles obeying Bose statistics is found to be quadratically divergent...,'' and concluded that in theories of elementary bosons, new phenomena must enter at an energy scale of $m/e$ (where $e$ is the relevant coupling). 
In modern particle physics, naturalness is often associated with the question, ``how do we understand the magnitude of the EWSB scale?''
In the absence of
new physics beyond the SM, its natural value would be the
Planck scale (or perhaps the grand unification scale or the seesaw scale that
controls neutrino masses).

There have been a number of theoretical proposals to explain the origin of the EWSB energy scale:
(1) naturalness is restored by  a symmetry principle--supersymmetry (SUSY)--which ties the bosons to
the more well-behaved fermions\cite{Witten,Susskind}; (2) the Higgs boson is an approximate Goldstone boson, the only other
known mechanism for keeping an elementary scalar light\cite{dewsb}; (3) the Higgs boson is a composite scalar, with an inverse length of
order the TeV-scale\cite{dewsb}; (4) extra spatial dimensions beyond three provide new mechanisms
for naturally large hierarchies of scales\cite{RS,extradim};
(5)~classical scale invariance and its minimal violation via quantum
anomalies\cite{Bardeen:1995kv,Meissner:2006zh,Iso:2009ss,Tavares:2013dga,Gorsky:2014una,Helmboldt:2016mpi} can generate a Higgs mass via dimensional transmutation\cite{Coleman:1973jx}; and
(6)~the EWSB scale arises due to some vacuum selection mechanism
(either anthropic\cite{Agrawal:1998xa} or cosmological\cite{Graham:2015cka,Arkani-Hamed:2016rle}).
Finally, maybe none of these explanations are relevant, and the EWSB
energy scale 
%(which appears to us to be highly fine-tuned) 
is simply the result of some initial condition whose origin will never be discernible.

Of course, these are lectures on supersymmetry.  Thus, we shall
motivate SUSY at the TeV scale as a potential solution of the
so-called hierarchy problem:
why is the scale of EWSB so much smaller than the Planck
scale?

 \subsection{\mbox{Avoiding quadratic UV-sensitivity 
with elementary scalars}}
\label{quadratic}
First, consider a lesson from history.
The electron self-energy in classical electromagnetism goes
like $e^2/a$, where $a$ is the classical radius of the electron. For a
point-like electron, $a\rightarrow 0$; hence the electron self-energy diverges linearly. In the quantum
theory, fluctuations of the electromagnetic fields (in the
``single electron theory'') generate a quadratic divergence.
 If
these divergences are not canceled, one would expect 
QED to break down at an energy of order $m_e/e$,
 far below the Planck scale.

The linear and quadratic divergences will cancel exactly if
one makes a bold hypothesis: the existence of the positron
(with a mass equal to that of the electron but of opposite
charge).
Weisskopf was the first to demonstrate this cancellation in
1934\cite{Weisskopf:1934}.\footnote{Actually the cancellation was not present
in the initial publication, but thanks to
a letter from Wendell Furry, the correct result was published in an erratum.}
This is an historical example in which 
 a symmetry implies the existence of a partner particle that cancels
 the dangerously large UV contribution to the particle mass. 

The motivation for SUSY may be viewed analogously\cite{Hitoshi,Hitoshi2}, with the electron playing the role of SM particles and the
positron playing the role of superpartners. SUSY
associates a fermionic superpartner with every SM particle and vice versa, thus doubling the  SM spectrum. SUSY relates the self-energy of the
elementary scalar boson to the self-energy of its fermionic partner.  Since the latter is only logarithmically sensitive to~$\Lambda$, we conclude
that the quadratic sensitivity of the scalar squared-mass to
UV physics must exactly cancel.
Naturalness is restored!

However,
since no superpartners degenerate in mass with the
corresponding SM particles exist in nature,  SUSY must be a broken symmetry.
Although the fundamental origin of SUSY-breaking is yet to be understood,
the effective scale of SUSY-breaking cannot be much larger than of
order a few TeV, 
if SUSY is responsible for the origin of the EWSB scale.

\enlargethispage{\baselineskip}
The absence of any evidence for SUSY at the LHC\cite{nosusy} is a cause for some
concern\cite{susy}.  This has led to some discussion of the so-called little
hierarchy problem\cite{little,little2,little3} which reflects the observation that the
effective SUSY-breaking mass scale is somewhat separated from the scale of EWSB.
Nevertheless, if evidence for supersymmetric
phenomena in the TeV or multi-TeV regime were to be eventually established at 
the LHC or at a future collider facility
(with an energy reach beyond the LHC\cite{vlhc}), it would be viewed as a spectacularly
successful explanation of the large hierarchy between the (multi-)TeV scale and
Planck scale. In this case, the remaining little hierarchy would
perhaps be regarded as a less pressing issue.

%%%%%%%%%%%%%%%%%%%%%%%%%%%%%%%%%%%%%%%%%%%%%%%%%%

%

%%%%%%%%%%%%%%%%%%%%%%%%%%%%%%%%%%%%%%%%%%%%%%%%%%

%%%%%%%%%%%%%%%%%%%%%%%%%%%%%%%%%%%%%%%%%%%%%%%%%%

\section{Supersymmetry: first steps}
\renewcommand{\theequation}{\arabic{section}.\arabic{equation}}
\setcounter{equation}{0}
\label{sec:SUSYalgebra}

The supersymmetry algebra is a generalization of the Lie algebra of the
Poincar\'e group of spacetime symmetries.   
In this section we begin by reviewing the representations of the
Poincar\'e group.  We then present the supersymmetry algebra and
examine its representations.   The consequences of super-Poincar\'e
invariance, in terms of the vacuum energy and the bosonic and
fermionic degrees of freedom, are discussed.   Finally, we exhibit how
these properties are manifested in the simplest supersymmetric field
theory of spin-0 and spin-$\half$ particles (the so-called Wess-Zumino
model\cite{Wess:1974tw}), and demonstrate how the SUSY algebra is realized.

\subsection{Review of the Poincar\'e algebra}
\label{sec:Poincare}

The Poincar\'e group consists of Lorentz transformations and spacetime
translations\cite{sexl}.  That is, under a Poincar\'e transformation, the
spacetime coordinates transform as $x^{\prime\,\mu}=\Lambda^\mu{}_\nu
x^\nu +a^\mu$, where $\Lambda$ is given by \eq{lambda44} and $a^\mu$ is a
constant four-vector. Under a Lorentz transformation $\Lambda$ and a
spacetime translation $a$, the field $\psi_\alpha$ of spin $s$
transforms as,
\beq \label{poin1}
\psi^\prime_\alpha(x) = {\exp\bigl( -\half i \theta_{\mu\nu}
S^{\mu\nu} \bigr)_\alpha}^\beta\ \psi_\beta\left(\Lambda^{-1}(x-a)\right)\,,
\eeq
where we have used $x=\Lambda^{-1}(x'-a)$ and redefined the dummy variable
$x'$ by removing the prime.  The Poincar\'e algebra is obtained by
considering an infinitesimal Poincar\'e transformation.  Expanding in
a Taylor series about $\Lambda=\mathds{1}_{4\times 4}$ and 
$a=0$, we may rewrite \eq{poin1}
as\footnote{The operators $\mathds{1}$, $P^\mu$ and $L^{\mu\nu}$ include an
implicit factor of $\delta_\alpha{}^\beta$, whereas the spin operator
$S^{\mu\nu}$ depends non-trivially on $\alpha$ and $\beta$ (except for
the case of spin zero, when $S=0$).}
\beq \label{poin2}
\psi^\prime_\alpha(x)\simeq\bigl[\mathds{1}+ia_\mu P^\mu
-\nicefrac{i}{2} \theta_{\mu\nu}
(L^{\mu\nu}+ S^{\mu\nu}) \bigr]_\alpha{}^\beta\ \psi_\beta (x)\,,
\eeq
where 
$\mathds{1}$ is the unit operator, $P^\mu\equiv i\partial^\mu$ and $L^{\mu\nu}\equiv i(x^\mu\partial^\nu
-x^\nu\partial^\mu)$ are the linear and angular momentum operators, respectively, and $S^{\mu\nu}$ depends on the representation;
for spin-\half \ two-component fermions,
\begin{align}
S^{\mu\nu} = 
\begin{cases}
\sigma^{\mu\nu} & \mathrm{for\ (\half,0)\ fields}; \\
\overline{\sigma}^{\mu\nu} & \mathrm{for\ (0,\half)\ fields}.
\end{cases}
\end{align}  
The Poincar\'e algebra consists of ten generators $P^\mu$ and
$J^{\mu\nu}\equiv L^{\mu\nu}+S^{\mu\nu}$ 
(where $J^{\mu\nu}=-J^{\nu\mu}$), which obey the following commutation
relations:
\begin{align}
\left[P^\mu,P^\nu\right] &= 0\,,\label{spoincarealg1}
 \\
\left[J^{\mu\nu},P^{\lambda}\right] &=
i(g^{\nu\lambda}P^\mu-g^{\mu\lambda}P^\nu)\,,
\label{spoincarealg2}
 \\
\left[J^{\alpha\beta},J^{\rho\sigma}\right]
&= i(g^{\beta\rho}\,J^{\alpha\sigma} -
g^{\alpha\rho}\,J^{\beta\sigma} - g^{\beta\sigma}\,J^{\alpha\rho} +
g^{\alpha\sigma}\,J^{\beta\rho})\,.\label{spoincarealg3}
\end{align}

The Poincar\'e algebra possesses two independent Casimir operators (these are
polynomial functions of the generators that commute with the
generators $P^\mu$ and $J^{\mu\nu}$), which are given by
\begin{align}
P^2\equiv
P_\mu P^\mu
\qquad \mathrm{and} \qquad
w^2\equiv
w_\mu w^\mu,
\end{align}
where $w^\mu$ is the  Pauli-Lubanski vector,
\begin{align}
w^\mu\equiv-\half\epsilon^{\mu\nu\rho\lambda}J_{\nu\rho}P_\lambda\,,
\end{align}
in a convention where $\epsilon^{0123}=1$.   Explicitly,
\begin{align} \label{paulilubanski}
w^\mu=(\mathbold{\vec{J}\cdot\vec{P}}\,;\,P^0\mathbold{\vec{J}}+
\mathbold{\vec{K}\times\vec{P}})\,,
\end{align}
where $J^i\equiv\half\epsilon^{ijk}J_{jk}$ and $K^i\equiv J^{0i}$.
Note that
\begin{align}\label{wP}
w_\mu P^\mu=0 \qquad \mathrm{and} \qquad [w_\mu\,,\,P_\nu]=0.
\end{align} 

The unitary representations of the Poincar\'e algebra can be
labeled by the eigenvalues of $P^2$ and $w^2$ when acting on the
physical states with non-negative energy $P^0$.
The eigenvalue of $P^2$ is $m^2$, where $m$ is the
mass of the physical state.  To see the physical interpretation of $w^2$, we first consider the case of $m\neq 0$.  In
this case, it is convenient to evaluate $w^2$ in the particle rest
frame.
In this frame, $w^\mu=(0\,;\,m\mathbold{\vec
S})$, where $S^i$ is defined in \eq{jkdef}.  Hence, $w^2=-m^2\mathbold{\vec
S}\llsup{\,2}$, with eigenvalues $-m^2 s(s+1)$, $s=0,\half,1,\ldots$.   We
conclude that massive (positive energy) states can be labeled by
$(m,s)$, where $m$ is the mass and $s$ is the spin of the state.

If $m=0$, the previous analysis is not valid, since we cannot evaluate
$w^2$ in the rest frame.  Nevertheless, if we take the $m\to 0$ limit,
it follows from the results above that either $w^2=0$, or the
corresponding states have infinite spin.  We reject the second
possibility (which does not appear to be realized in nature), in which
case $w^2=\lim_{m\to 0} (-m^2\boldsymbol{\vec S}\llsup{\,2})=0$.
Thus, we must solve the equations, $w^2=P^2=w_\mu
P^\mu=0$.  It is simplest to
choose a frame in which $P=P^0(1;0,0,1)$ where $P^0>0$.  In this frame,
it is easy to show that $w=w^0(1;0,0,1)$.  That is, in any Lorentz frame,
\beq \label{helicitydef}
w^\mu=h P^\mu\,,
\eeq
where $h$ is called the helicity operator.  In particular,
\beq
[h\,,\,P^\mu]=[h\,,\,J^{\mu\nu}]=0\,,
\eeq
which means that the eigenvalues of $h$ can be used to label states of
the irreducible massless representations of the Poincar\'e algebra.
From \eq{helicitydef}, we
derive\footnote{We define the differential operator
$L^i\equiv\half\epsilon^{ijk}L_{jk}$.  Then, noting that
$\mathbold{\vec L}=\mathbold{\vec x\times\vec P}$, it follows that $\mathbold{\vec{L}\cdot\vec{P}}=0$.
Hence, $\mathbold{\vec{J}\cdot\vec{P}}=(\mathbold{\vec{L}}+\mathbold{\vec{S}})\cdot\mathbold{\vec{P}}=\mathbold{\vec{S}\cdot\vec{P}}$.}
\beq \label{hdefinition}
h=\frac{w^0}{P^0}=\frac{\mathbold{\vec{J}\cdot\vec{P}}}{P^0}
=\frac{\mathbold{\vec{S}\cdot\vec{P}}}{|\boldsymbol{\vec{P}}|}=\boldsymbol{\vec{S}\newcdot\hat{P}}\,,
\eeq
after noting that $P^0=|\boldsymbol{\vec{P}}|$ for massless states.
Eigenvalues of $h$ are called the helicity (and are denoted by
$\lambda$);
its spectrum consists of non-negative half-integers,
$\lambda=0,\pm\half,\pm 1,\ldots$.
Under a CPT transformation, $\lambda\to -\lambda$.  
Thus, in any quantum field theory realization of massless particles, 
both~~$\pm|\lambda|$ helicity states must appear in the theory.
It is common to refer to a massless (positive energy) state
of helicity $\lambda$ as having spin $|\lambda|$.

\subsection{The supersymmetry (SUSY) algebra}
\label{SUSYalg}
In the 1960s, Coleman and Mandula proved
a very powerful no-go theorem
that showed that in quantum field theories in $3+1$ dimensional
spacetime with a mass gap, the only possible symmetry incorporating Poincar\'e
transformations and a global internal symmetry group of transformations 
must be a trivial tensor product of the two groups\cite{Coleman:1967ad}.  
Subsequently, Haag, {\L}opusza{\'{n}}ski and Sohnius proved that  the only
possible extension of the Poincar\'e algebra involves the addition
of new fermionic generators that transform either as a $(\half,0)$ or
$(0,\half)$ under the Lorentz algebra, denoted by $Q^i_{\alpha}$ and
its hermitian conjugate
$Q^{\dagger}_{\dot\alpha i}\equiv (Q^i_\alpha)^\dagger$, respectively,
where $i=1,2,\ldots N$\cite{Lopuszanski,Haag:1974qh}.
In these lectures, we shall focus exclusively on the case of $N=1$, in which case the subscript $i$
can be dropped.

We therefore begin by examining the structure of
the $N=1$ SUSY algebra, which is obtained by adding one
$(\half,0)$ and one $(0,\half)$ generator to the Poincar\'e algebra,
denoted by $Q_\alpha$ and $Q^\dagger_{\dot\alpha}$, respectively.
These two-component spinor generators have no explicit dependence
on the spacetime coordinate  and are thus
invariant under spacetime translations.  That is,
\beqa
\exp\left(-ia_\mu P^\mu\right)Q_\alpha \exp\left(ia_\mu P^\mu\right)&=&Q_\alpha\,,\\[6pt]
\exp\left(-ia_\mu P^\mu\right)Q^\dagger_{\dot\alpha} \exp\left(ia_\mu P^\mu\right)&=&Q^\dagger_{\dot\alpha}\,,
\eeqa
where the $a_\mu$ are real parameters.  Working to first order in $a_\mu$, it follows that
the spinor generators
must commute with the translation generator~$P^\mu$,
\beq \label{QP}
[Q_\alpha\,,\,P^\mu]=[Q^\dagger_{\dot\alpha}\,,\,P^\mu]=0\,.
\eeq

The commutation relations given in \eq{QP} can also be deduced by
employing the following algebraic argument.  Using the known
transformation properties of $Q_\alpha$, $Q^\dagger_{\dot\alpha}$ and
$P^\mu$ under the
Poincar\'e algebra, it follows that $[Q_\alpha\,,\,P^\mu]$ must consist
of generators whose transformation properties are consistent with the
tensor product,
\beq
(\half,0)\otimes(\half,\half)=(1,\half)\oplus(0,\half)\,,
\eeq
under the Poincar\'e algebra.  But according to the
Haag-{\L}opuszanski-Sohnius theorem, there are no $(1,\half)$ generators.
This argument still leaves open the possibility that
$[Q_\alpha\,,\,P^\mu]\propto \sigma^\mu_{\alpha\dot\beta}Q^{\dagger\,\dot\beta}$.
However, it can be shown using the Jacobi identity
 that
the proportionality constant must be zero.

The transformation properties of $Q_\alpha$ and
$Q^\dagger_{\dot\alpha}$
under the  Poincar\'e algebra yield their
commutation relations with the $J^{\mu\nu}$,
\beq
[Q_\alpha\,,\,J^{\mu\nu}]=(\sigma^{\mu\nu})_\alpha{}^\beta Q_\beta\,,
\qquad\qquad
[Q^\dagger_{\dot\alpha}\,,\,J^{\mu\nu}]
=-Q^\dagger_{\dot\beta}(\sigmabar^{\mu\nu})^{\dot\beta}{}_{\dot\alpha}\,.
\eeq
The Coleman-Mandula theorem implies that one cannot obtain a consistent algebraic
structure by postulating commutation relations for the $Q_\alpha$ and $Q^\dagger_{\dot\alpha}$.
However, by declaring $Q_\alpha$ and $Q^\dagger_{\dot\alpha}$ to be \textit{fermionic}
generators, one can postulate \textit{anticommutation} relations for $Q_\alpha$ and $Q^\dagger_{\dot\alpha}$
such that the generators $\{P^\mu\,,\,J^{\mu\nu}\,,\,Q_\alpha\,,\,Q^\dagger_{\dot\alpha}\}$ form
a closed algebraic system.  We therefore consider the three possible anticommutation relations,
along with their transformation properties with respect to the Poincar\'e algebra,
\beqa
\{Q_\alpha\,,\,Q_\beta\}\qquad & \qquad (\half,0)\otimes(\half,0)=(1,0)\oplus(0,0)\,,\label{qqcomm1}\\
\{Q^\dagger_{\dot\alpha}\,,\,Q^\dagger_{\dot\beta}\}\qquad & \qquad (0,\half)\otimes (0,\half)=(0,1)\oplus(0,0)\,,\label{qqcomm2}\\
\{Q_\alpha\,,\,Q^\dagger_{\dot\beta}\}\qquad & \qquad (\half,0)\otimes (0,\half)=(\half,\half)\,.\label{qqcomm3}
\eeqa
\Eqs{qqcomm1}{qqcomm2} imply that
\beqa
\{Q_\alpha\,,\,Q^\beta\}&=&s(\sigma^{\mu\nu})_\alpha{}^\beta J_{\mu\nu}+k\delta_\alpha{}^\beta\mathds{1}\,,\label{QQsk1}\\[6pt]
\{Q^{\dagger\dot\alpha}\,,\,Q^\dagger_{\dot\beta}\}\ &=& s^* (\sigmabar^{\mu\nu})^{\dot\alpha}{}_{\dot{\beta}}J_{\mu\nu}
+k^*\delta^{\dot\alpha}{}_{\dot\beta}\mathds{1}\,,\label{QQsk2}
\eeqa
where $s$ and $k$ are complex numbers
and \eq{QQsk2} is the hermitian conjugate of \eq{QQsk1}.  Note that
we have raised and/or lowered some of the spinor indices for convenience.
Since $[Q_\alpha,P^\lambda]=[Q^\dagger_{\dot\alpha},P^\lambda]=0$ and $[J_{\mu\nu},P^\lambda]\neq 0$,
it follows that $s=0$.  If we now lower all spinor indices, \eqs{QQsk1}{QQsk2} with $s=0$ yield
\beq
\{Q_\alpha\,,\,Q_\beta\}=k\epsilon_{\beta\alpha}\mathds{1}\,,\qquad\quad
\{Q^\dagger_{\dot\alpha}\,,\,Q^\dagger_{\dot\beta}\}=k^*\epsilon^{\dot\beta\dot\alpha}\mathds{1}\,,
\eeq
and we conclude that $k=0$, since the left-hand sides of the above equations are symmetric under
the interchange of spinor indices, whereas the right hand sides are antisymmetric.
Hence,
\beq
\{Q_\alpha\,,\,Q_\beta\}=\{Q^\dagger_{\dot\alpha}\,,\,Q^\dagger_{\dot\beta}\}=0\,.
\eeq

\Eq{qqcomm3} implies that the remaining anticommutation relation must be of the form
\beq \label{QQt}
\{Q_\alpha\,,\,Q^\dagger_{\dot\beta}\}=t\sigma^\mu_{\alpha\dot\beta}P_\mu\,,
\eeq
where $t$ is a complex number. Multiplying \eq{QQt} by $\sigmabar^{\nu\dot\beta\alpha}$ and using
$\Tr(\sigma^\mu\sigmabar^\nu)=2g^{\mu\nu}$, it follows that
\beq \label{sigbarqq}
\sigmabar_\mu^{\dot\beta\alpha}\{Q_\alpha\,,\,Q^\dagger_{\dot\beta}\}=2tP_\mu\,.
\eeq
In particular, for $\mu=0$, \eq{sigbarqq} relates the energy $P^0$ to the SUSY generators:
\beq \label{pzero}
2tP^0=Q_1 Q_1^\dagger+Q_1^\dagger Q_1+Q_2 Q_2^\dagger+Q_2^\dagger Q_2\,.
\eeq
Since $P^0\geq m$ for physical states of mass $m$ and the right-hand side of \eq{pzero} is positive semi-definite,
it follows that $t$ must be real and positive.\footnote{We reject the possibility of $t=0$, in
which case $Q=Q^\dagger=0$ and the SUSY algebra reduces to the Poincar\'e algebra.}
One can rescale the definition of the fermionic generators $Q$ and $Q^\dagger$ such that
$t=2$.  In this convention,
\beq
\{Q_\alpha\,,\,Q^\dagger_{\dot\beta}\}=2\sigma^\mu_{\alpha\dot\beta}P_\mu\,.
\eeq

To summarize, the $N=1$ SUSY algebra
is spanned by the generators $\{P^\mu\,,\,J^{\mu\nu}\,,\,Q_\alpha\,,\,Q^\dagger_{\dot\alpha}\}$, which
satisfy \eqst{spoincarealg1}{spoincarealg3} and
\beqa
[Q_\alpha\,,\,P^\mu]&=&[Q^\dagger_{\dot\alpha}\,,\,P^\mu]=0\,,\label{susyalg1}\\
\left[Q_\alpha\,,\,J^{\mu\nu}\right]&=&(\sigma^{\mu\nu})_\alpha{}^\beta Q_\beta\,,\label{susyalg2}\\
\left[Q^\dagger_{\dot\alpha}\,,\,J^{\mu\nu}\right]&=&-Q^\dagger_{\dot\beta}
(\sigmabar^{\mu\nu})^{\dot\beta}{}_{\dot\alpha}\,,\label{susyalg3}\\
\{Q_\alpha\,,\,Q_\beta\}&=&\{Q^\dagger_{\dot\alpha}\,,\,Q^\dagger_{\dot\beta}\}=0\,,\label{susyalg4}\\
\{Q_\alpha\,,\,Q^\dagger_{\dot\beta}\}&=&2\sigma^\mu_{\alpha\dot\beta}P_\mu\,.\label{susyalg5}
\eeqa

Note that \eqst{susyalg1}{susyalg5} are unchanged under the U(1) phase transformation,
\beq
Q_\alpha\to e^{-i\chi}Q_\alpha\,,\qquad\quad Q^{\dagger}_{\dot\alpha}\to e^{i\chi}Q^{\dagger}_{\dot\alpha}\,,
\eeq
whereas the generators $P^\mu$ and $J^{\mu\nu}$ are not transformed.
One can therefore extend the $N=1$ SUSY algebra by adding a bosonic generator $R$ such that
\beqa
e^{i\chi R}Q_\alpha e^{-i\chi R}&=&e^{-i\chi}Q_\alpha\,,\label{R1}\\
e^{i\chi R}Q^\dagger_{\dot\alpha} e^{-i\chi R}&=&e^{i\chi}Q^\dagger_{\dot\alpha}\,.\label{R2}
\eeqa
Expanding out to first order in $\chi$, one easily derives the commutation relations,
\beqa
\left[R\,,\,Q_\alpha\right]&=&-Q_\alpha\,,\label{susyalg6} \\
\left[R\,,\,Q^\dagger_{\dot\alpha}\right]&=&Q^\dagger_{\dot\alpha}\,.\label{susyalg7}
\eeqa
We therefore say that the generator $Q_\alpha$ has an $R$-charge of $-1$.  Since $P^\mu$ and $J^{\mu\nu}$
are uncharged under the U(1)$_R$ transformation, it follows that
\beq \label{susyalg8}
[R\,,\,P^\mu]=[R\,,\,J^{\mu\nu}]=0\,.
\eeq
Thus, \eqst{spoincarealg1}{spoincarealg3},
(\ref{susyalg1})--(\ref{susyalg5}) and (\ref{susyalg6})--(\ref{susyalg8})
define the maximally extended
$N=1$ SUSY algebra, which includes an additional
continuous U(1)$_R$ symmetry.

\subsection{Representations of the $N=1$ SUSY algebra}
In Section~\ref{sec:Poincare}, we  identified the two Casimir operators of the Poincar\'e
algebra, $P^2$ and $w^2$, and noted that the
representations of the Poincar\'e algebra can be labeled by the eigenvalues of
the Casimir operators acting on the physical states.
We saw that 
the massive
representations can be labeled by their mass and spin, $(m,s)$.  For a fixed value of $m$, the
corresponding spin-$s$ representations are $(2s+1)$-dimensional.
For massless states, we defined
 the helicity operator
$h=\boldsymbol{\vec{S}\newcdot\hat{P}}$ [cf.~\eq{hdefinition}], with
 eigenvalues $\lambda=0,\pm\half,\pm 1\,\ldots$.
We also noted that $\lambda$ changes sign under a CPT
transformation.  Hence, the massless positive energy
representations of the Poincar\'e algebra are specified by $|\lambda|$.
For the case of $\lambda=0$, the corresponding representation is one-dimensional.
For any non-zero 
choice for $\lambda$,
the corresponding representation is two-dimensional and reducible,
as both $\pm|\lambda|$ helicity states must appear.

The unitary representations of the $N=1$ SUSY algebra can be determined
by using similar techniques\cite{Salam:1974za,Sokatchev:1975gg}.  First, we identify the Casimir operators, which
commute with all the SUSY algebra generators, $\{P^\mu\,,\,J^{\mu\nu}\,,\,Q_\alpha\,,\,Q^{\dagger\dot\alpha}\}$.
It is clear that $P^2$ is a Casimir operator, since $Q_\alpha$ and $Q^{\dagger\dot\alpha}$ commute
with $P^\mu$.  However, $w^2$ is \textit{not} a Casimir operator of the SUSY
algebra.  To establish this result, it is straightforward to use the (anti-)commutation
relations of the SUSY algebra to prove that:
\beq \label{wQQ}
\left[w^\mu\,,\,Q_\alpha\right]=i(\sigma^{\mu\nu})_{\alpha}{}^\beta Q_\beta P_\nu\,,\qquad\quad
\left[w^\mu\,,\,Q^\dagger_{\dot\alpha}\right]= i(\sigmabar^{\mu\nu})^{\dot\beta}{}_{\dot\alpha}Q^\dagger_{\dot\beta}P_\nu\,.
\eeq
Using these results, it is straightforward to derive:
\beqa
[w^2\,,\,Q_\alpha]&=&2i\sigma^{\mu\nu}{}_\alpha{}^\beta Q_\beta w_\mu P_\nu-\tfrac{3}{4}P^2 Q_\alpha\,,\label{wtwo} \\
\left[w^2\,,\,Q^\dagger_{\dot\alpha}\right]&=&2i\sigmabar^{\mu\nu\dot\beta}{}_{\dot\alpha} Q^\dagger_{\dot\beta} w_\mu P_\nu
-\tfrac{3}{4}P^2 Q^\dagger_{\dot\alpha}\,.\label{wtwodag}
\eeqa
Thus, $w^2$ does not commute with the fermionic generators of the SUSY algebra.  One consequence
of this result is that the representations of the SUSY
algebra consist of supermultiplets that contain particles of equal
mass but with different spins.

In order to deduce the possible spins that make up an irreducible supermultiplet, we shall identify a second Casimir
operator of the $N=1$ SUSY algebra.  We begin by defining the operator
\beq
B^\mu\equiv  w^\mu+\tfrac{1}{4}Q^\dagger\sigmabar^\mu Q\,.
\eeq
Using \eqss{susyalg4}{susyalg5}{wQQ}, one can derive
\beq \label{BQQ}
[B^\mu\,,\,Q_\alpha]=-\half P^\mu Q_\alpha\,,\qquad\qquad [B^\mu\,,\,Q^\dagger_{\dot\alpha}]=\half P^\mu Q^\dagger_{\dot\alpha}\,.
\eeq
The four-vector operator $B^\mu$ possesses some of the properties of the Pauli-Lubanski vector $w^\mu$.
In particular,
\begin{align} 
[B^\mu\,,\,B^\nu] &=i\epsilon^{\mu\nu\rho\lambda}B_\rho P_\lambda ; \label{BmuBnu} \\
[B^\mu, P^\nu] & = 0; \label{BP}\\
[B^\mu, J^{\nu\lambda} ] & = i \of{ g^{\mu\nu} B^\lambda - g^{\mu\lambda} B^\nu } . \label{Bvector}
\end{align}
One may be tempted to conjecture that $B^2\equiv B_\mu B^\mu$
is a Casimir operator of the SUSY algebra.  However, $[B^2\,,\,Q_\alpha]\neq 0$, so we must look further.
The structure of \eq{BQQ} suggests that we define
\beq \label{Cmunu}
C^{\mu\nu}\equiv B^\mu P^\nu-B^\nu P^\mu\,.
\eeq
It then follows that
\beq
[C^{\mu\nu}\,,\,Q_\alpha]=[C^{\mu\nu}\,,\,Q^\dagger_{\dot\alpha}]=[C^{\mu\nu}\,,\,P^\lambda]=0\,,
\eeq
where the first two commutators vanish as a consequence of \eq{BQQ} and the last commutator
vanishes as a consequence of \eq{BP}.  Moreover, \eqs{spoincarealg2}{Bvector} imply that $P^\mu$ and $B^\mu$
are Lorentz four-vectors, in which case $C^{\mu\nu}$ is a second-rank Lorentz tensor.  Hence
\beq \label{C2def}
C^2\equiv C_{\mu\nu} C^{\mu\nu}=2[B^2 P^2-(B\newcdot P)^2]\,,
\eeq
satisfies
\beq
[C^2\,,\,P^\mu]=[C^2\,,\,J^{\mu\nu}]=[C^2\,,\,Q_\alpha]=[C^2\,,\,Q^\dagger_{\dot\alpha}]=0\,.
\eeq

We conclude that $P^2$ and $C^2$ are the two
Casimir operators of the $N=1$ SUSY algebra.
Representations of the $N=1$ SUSY algebra can therefore be
labeled by the eigenvalues of $P^2$ and $C^2$ when acting on the
physical states.\footnote{As in the case of the Poincar\'e algebra,
we restrict our considerations to
states of non-negative energy $P^0$.} The eigenvalue of $P^2$ is
$m^2$, where $m$ is the mass. To understand the physical meaning of $C^2$, we will consider massive and massless supermultiplets separately.

\subsubsection{Massive $N=1$ supermultiplets}
To see the physical interpretation of $C^2$, we first consider the case of $m\neq 0$, so that we are free to evaluate the Lorentz scalar $C^2$ in the particle rest frame.
In this frame,
\beq \label{bmudef}
B^\mu=(\tfrac{1}{4}Q^\dagger\sigmabar^0 Q\,;\,mS^i+\tfrac{1}{4}Q^\dagger\sigmabar^i Q),
\eeq
where $S^i$ is defined in \eq{jkdef}.
%\beq
%S^i = \half \epsilon^{ijk} S_{jk}.
%\eeq
We then compute,
\begin{align}
C^2  & = 2\left[B^2 P^2-(B\newcdot P)^2\right] 
 =2m^2\left[B^2-B_0^2\right] 
 = -2m^2 B^i B^i,
\end{align}
where $B^iB^i\equiv |\boldsymbol{\vec B}|^2$.  
Moreover, if we define the rest-frame operator,
\beq \label{caljdef}
\mathcal{J}^i\equiv \frac{1}{m}B^i=S^i+\frac{1}{4m}Q^\dagger\sigmabar^i Q\,,
\eeq
then it follows from \eq{BmuBnu} that
\beq
[\mathcal{J}^i\,,\,\mathcal{J}^j]=i\epsilon^{ijk} \mathcal{J}^k\,.
\eeq

The eigenvalues of $\mathcal{J}^i \mathcal{J}^i$ are $j(j+1)$ for $j=0,\half,1,\tfrac{3}{2}\,\ldots$.  Hence, the
eigenvalues of
\beq
C^2=-2m^4 \mathcal{J}^i \mathcal{J}^i
\eeq
are $-2m^4 j(j+1)$.  We conclude that for positive energy, timelike $P^\mu$,
the unitary irreducible representations of the $N=1$ SUSY algebra are labeled by $(m,j)$, where $j$ is called
the \textit{superspin} of the supermultiplet.
The states of an irreducible $N=1$ massive supermultiplet of superspin $j$ are exhibited
in Table~\ref{massivesuperplet}.   
The explicit construction of these
states and a discussion of their properties is presented in Section~\ref{App}.

\begin{table}[t!]
    \caption{\small States of an $N=1$ massive supermultiplet of superspin $j$.  An interpretation
is provided for $j=s$ and $j=s+\half$ where $s$ is a non-negative integer.
The bosonic and fermionic degrees of freedom (D.o.f.) of the supermultiplet coincide
and is equal to $2(2j+1)$.\label{massivesuperplet}}
\vskip 0.1in
{
\addtolength\tabcolsep{2pt}
\begin{tabular}{cccc}
\hline
Spin & D.o.f. & Interpretation ($j=s$) & Interpretation ($j=s+\half$) \\ \hline
$j$ & $2(2j+1)$ & complex spin-$s$ boson & ``complex'' spin-($s+\half$) fermion \\
$j+\half$ & $2j+2$ & spin-($s+\half)$ fermion & real spin-$(s+1)$ boson \\
$j-\half$ & $2j$ &  spin-($s-\half)$ fermion & real spin-$s$ boson \\ \hline
\end{tabular}}
  \end{table}

\begin{example}[The massive chiral supermultiplet, $\boldsymbol{j=0}$]

For $j=0$, only $j_3=0$ is possible, in which case the massive
supermultiplet is made up of two states of spin 0 and two states of
spin $\half$.  The two spin-0 states can be combined into a single complex
scalar state, and the two spin-$\half$ states can be identified as the two
components of a two-component Majorana fermion.    In this case the
$j-\half$ row of Table~\ref{massivesuperplet} is not relevant. 
\end{example}

It can be shown (see Problem \ref{pr:jhalf}) that the massive
supermultiplet of superspin $\half$ consists of a (real) spin-1 boson,
a (real) spin-0 boson and two mass-degenerate Majorana fermions, which
can be combined into a single Dirac fermion (called a \textit{complex}
fermion in Table~\ref{massivesuperplet}).
As expected,  in both the $j=0$ and $j=\half$ cases exhibited above,
the number of bosonic degrees of freedom of the
supermultiplet equals the number of fermionic degrees of freedom.

\subsubsection{Massless $N=1$ supermultiplets}

We now examine the case of zero-mass positive energy states, where $P^2=0$ and $P^0>0$.  If one multiplies
\eq{susyalg5} by $P^\rho P^\lambda\sigmabar_\rho^{\dot\gamma\alpha}\sigmabar_\gamma^{\dot\beta\tau}$,
one can easily derive the anticommutation relation,
\beq
\{P^\rho\sigmabar_\rho^{\dot\gamma\alpha}Q_\alpha\,,\,P^\lambda Q^\dagger_{\dot\beta}\sigmabar^{\dot\beta\tau}\}=2P^2 P^\mu\sigmabar_\mu^{\dot\gamma\tau}\,.
\eeq
Thus, for $P^2=0$ we have,
\beq \label{PoperatorP}
\bra{\Psi}\{P^\rho\sigmabar_\rho^{\dot\gamma\alpha}Q_\alpha\,,\,P^\lambda Q^\dagger_{\dot\beta}\sigmabar^{\dot\beta\tau}\}\ket{\Psi}=0\,,
\eeq
for any state $\ket{\Psi}$.
In the space of one-particle states, only positively-normed states exist.  Noting that
$(P^\mu\sigmabar_\mu^{\dot\alpha\beta}Q_\beta)^\dagger=P^\mu Q^\dagger_{\dot\beta}\sigmabar_\mu^{\dot\beta\alpha}$,
\eq{PoperatorP} implies that as operators on the space of one-particle states,
\beq \label{zeroops}
P^\rho\sigmabar_\rho^{\dot\gamma\alpha}Q_\alpha=
P^\lambda Q^\dagger_{\dot\beta}\sigmabar_\lambda^{\dot\beta\tau}=0\,,\qquad
\text{for}~~P^2=0\,.
\eeq
Using this result, one can evaluate the Casimir operator $C^2$, defined in \eq{C2def}, in the case of $P^2=0$.
In particular, using $w_\mu P^\mu = 0$ and \eq{zeroops},
\beq
C^2=-2(B\newcdot P)^2=-\tfrac{1}{8}(Q^\dagger_{\dot\alpha}\sigmabar^{\dot\alpha\beta}_\mu Q_\beta P^\mu)^2=0\,.
\eeq

The same conclusion can be obtained by choosing
the standard reference frame,  $P^\mu=P^0(1\,;\,0\,,\,0\,,\,1)$,
for lightlike four-vectors.   In this reference frame, the anticommutators given in \eqs{susyalg4}{susyalg5}
simplify to
\beqa
\{Q_1\,,\,Q^\dagger_1\}&=& 0\,,\qquad\quad\,\,
\{Q_2\,,\,Q^\dagger_2\}=4P_0\,,\label{qqmassless1}\\
\{Q_1\,,\,Q_1\}&=&\{Q_2\,,\,Q_2\}=\{Q_1\,,\,Q_2\}=0\,,\label{qqmassless2}\\
 \{Q^\dagger_1\,,\,Q^\dagger_1\}&=&\{Q^\dagger_2\,,\,Q^\dagger_2\}=\{Q^\dagger_1\,,\,Q^\dagger_2\}=0\,.\label{qqmassless3}
\eeqa
Hence,
\beq
C^2=-2(B\newcdot P)^2=-\half P_0^2(Q_1^\dagger Q_1)^2=\half P_0^2 Q_1^\dagger Q_1^\dagger Q_1 Q_1=0\,.
\eeq

\Eq{zeroops} implies a number of other operator identities when acting on the space of one-particle states.
Using \eq{susyalg5}, one easily derives
\beq
[Q^\alpha Q_\alpha\,,\,Q^\dagger_{\dot\beta}]=4P_\mu\sigma^\mu_{\alpha\dot\beta}Q^\alpha\,,\qquad\quad
[Q^\dagger_{\dot\alpha}Q^{\dagger\,\dot\alpha}\,,\,Q_\beta]=-4P_\mu\sigma^\mu_{\alpha\dot\beta}Q^{\dagger\,\dot\beta}\,.
\eeq
Applying \eq{zeroops} then yields
\beq \label{QQQQ}
[Q^\alpha Q_\alpha\,,\,Q^\dagger_{\dot\beta}]=[Q^\dagger_{\dot\alpha}Q^{\dagger\,\dot\alpha}\,,\,Q_\beta]=0\,,\qquad
\text{for}~~P^2=0\,.
\eeq
Then, for any one-particle state $\ket{\Psi}$, \eqss{susyalg4}{susyalg5}{QQQQ} yield
\beqa
P_\mu\sigma^\mu_{\alpha\dot\alpha}Q^\beta Q_\beta\ket{\Psi}&=&\half\{Q_\alpha\,,\,Q^\dagger_{\dot\alpha}\}Q^\beta Q_\beta\ket{\Psi}
=\half Q_\alpha Q^\dagger_{\dot\alpha}Q^\beta Q_\beta\ket{\Psi}\nonumber \\
&=&\half Q_\alpha [Q^\dagger_{\dot\alpha}\,,\,Q^\beta Q_\beta]\ket{\Psi}=0\,.
\eeqa
A similar computation of $P_\mu\sigma^\mu_{\alpha\dot\alpha}Q^\dagger_{\dot\beta}Q^{\dagger\,\dot\beta}$
allows us to conclude that
\beq
P_\mu Q^\beta Q_\beta\ket{\Psi}=P_\mu Q^\dagger_{\dot\beta}Q^{\dagger\,\dot\beta}\ket{\Psi}=0\,,\qquad \text{for}~~P^2=0\,,
\eeq
after multiplying through by $\sigmabar_\nu^{\dot\alpha\alpha}$ and evaluating the resulting trace.
As we are only interested in positive energy states, we conclude that as operators on the space of one-particle states,
\beq \label{QQQQ0}
Q^\beta Q_\beta=Q^\dagger_{\dot\beta}Q^{\dagger\,\dot\beta}=0\,,\qquad \text{for}~~P^2=0~~\text{and}~~P^0>0\,.
\eeq

In order to identify the massless supermultiplets of one-particle states, it is convenient to define
\beq \label{Lmudef}
L^\mu\equiv  \half(w^\mu+B^\mu)=w^\mu+\tfrac{1}{8}Q^\dagger\sigmabar^\mu Q\,.
\eeq
Note 
$[Q_\alpha ,P^\mu ]=[Q^\dagger_{\dot{\alpha}}, P^\mu ]=0$ and
$[w_\mu, P_\nu ]=0$
imply that
\beq \label{PL}
[P^\mu\,,\,L^\nu]=0\,.
\eeq
Using \eqss{susyalg4}{susyalg5}{wQQ}, one can easily derive
\beq \label{LQQ}
[L^\mu\,,\,Q_\alpha]=-\tfrac{1}{4}(\sigma^\mu\sigmabar^\nu)_\alpha{}^\beta Q_\beta P_\nu\,,\qquad
[L^\mu\,,\,Q^\dagger_{\dot\alpha}]=\tfrac{1}{4}(\sigmabar^\nu\sigma^\mu)^{\dot\beta}{}_{\dot\alpha} Q^\dagger_{\dot\beta}P_\nu\,.
\eeq
A straightforward computation then gives:
\beq \label{LLcomm}
[L^\mu\,,\,L^\nu]=i\epsilon^{\mu\nu\rho\lambda}(L_\rho+\tfrac{1}{16}Q^\dagger\sigmabar_\rho Q)P_\lambda\,.
\eeq
When $P^2=0$, we impose the results of \eq{zeroops} to obtain
\beq \label{Lprops0}
P^\mu L_\mu=[L^\mu\,,\,Q_\alpha]=[L^\mu\,,\,Q^\dagger_{\dot\alpha}]=0\,,\qquad \text{for}~~P^2=0\,.
\eeq
Moreover, if we employ the identity
\beq
\epsilon^{\mu\nu\rho\lambda}\sigmabar_\rho=\half i(\sigmabar^\nu\sigma^\mu\sigmabar^\lambda-\sigmabar^\lambda\sigma^\mu\sigmabar^\nu)\,,
\eeq
[which is a consequence of \eq{sigsigsig1}],
it then follows from \eq{zeroops} that
\beq \label{epsQP}
\epsilon^{\mu\nu\rho\lambda}Q^\dagger\sigmabar_\rho Q P_\lambda=0\,,\qquad \text{for}~~P^2=0\,.
\eeq

Hence, in the massless case, \eq{LLcomm} simplifies to
\beq \label{LLcomm0}
[L^\mu\,,\,L^\nu]=i\epsilon^{\mu\nu\rho\lambda}L_\rho P_\lambda\,,\qquad \text{for}~~P^2=0\,.
\eeq
Finally, we evaluate $L^\mu L_\mu$ for the positive energy massless one-particle states.
As in the analysis of the Poincar\'e algebra, we shall assume that $w^\mu w_\mu=\lim_{m\to
  0} (-m^2\boldsymbol{\vec S}\llsup{\,2})=0$.
Using \eq{epsQP}, it follows that
\beq
w^\mu Q^\dagger\sigmabar_\mu Q= -\half\epsilon^{\mu\nu\rho\lambda}J_{\nu\rho}P_\lambda Q^\dagger\sigmabar_\mu Q=0\,.
\eeq
In light of \eq{sigid3},
we obtain
\beqa
(Q^\dagger\sigmabar^\mu Q)(Q^\dagger\sigmabar_\mu Q)&=&2\epsilon^{\dot\alpha\dot\gamma}\epsilon^{\beta\tau}Q^\dagger_{\dot\alpha}
Q_\beta Q^\dagger_{\dot\gamma}Q_\tau=2\epsilon^{\dot\alpha\dot\gamma}\epsilon^{\beta\tau}
Q^\dagger_{\dot\alpha}[2P_\mu\sigma^\mu_{\beta\dot\gamma}-Q^\dagger_{\dot\gamma}Q_\beta]Q_\tau\nonumber \\
&=& 2(Q^\dagger_{\dot\alpha}Q^{\dagger\,\dot\alpha})(Q^\beta Q_\beta)-4P^\mu Q^\dagger\sigmabar_\mu Q=0\,,
\eeqa
after applying the operator identities given in \eqs{zeroops}{QQQQ0}.  Hence,
\beq \label{lmulmu}
L^\mu L_\mu=0\,,\qquad \text{for}~~P^2=0~~\text{and}~~P^0>0\,.
\eeq

When $P^2=0$ and $P^0>0$, the properties of $L^\mu$ [cf.~eqs.~(\ref{PL}), (\ref{Lprops0}), (\ref{LLcomm0})
and (\ref{lmulmu})] match precisely the
properties of the Pauli-Lubanski vector.
Thus, we must solve the equations $L^2=P^2=L_\mu
P^\mu=0$.  In
a reference frame in which $P^\mu=P^0(1\,;\,0\,,\,0\,,\,1)$ and $P^0>0$, 
it follows that $L^\mu=L^0(1\,;\,0\,,\,0\,,\,1)$.  Consequently, in any Lorentz frame,
\beq \label{shelicitydef}
L^\mu=\mathcal{K} P^\mu\,,
\eeq
where $\mathcal{K}\equiv L^0/P^0$ is called the superhelicity operator.
More explicitly, in a frame where $P^\mu=P^0(1\,;\,0\,,\,0\,,\,1)$,
\beq \label{mathcalkdef}
\mathcal{K}=h+\frac{1}{8P^0}\left(Q_1^\dagger Q_1+Q_2^\dagger Q_2\right)\,,
\eeq
where $h\equiv w^0/P^0=\boldsymbol{\vec S\newcdot\hat P}$ is the usual helicity operator acting
on massless one-particle states.  By virtue of \eqs{susyalg1}{Lprops0}, it follows that
\beq \label{KQQ}
[\mathcal{K}\,,\,P^\mu]=[\mathcal{K}\,,\,Q_\alpha]=[\mathcal{K}\,,\,Q^\dagger_{\dot\alpha}]=0\,.
\eeq

Hence, the states of the massless supermultiplet are eigenstates of
$\mathcal{K}$, with possible
eigenvalues $\kappa=0,\pm\half,\pm 1,\pm\tfrac{3}{2},\ldots$.  In contrast, $h$ does not commute with
$Q_\alpha$ and $Q^\dagger_{\dot\alpha}$.  Thus, the different states of the massless
supermultiplet will have different helicities.
We conclude that for positive energy, timelike $P^\mu$,
the irreducible representations of the $N=1$ SUSY algebra are labeled by
the eigenvalue $\kappa$ of the superhelicity operator, which is called
the \textit{superhelicity} of the massless supermultiplet.  Moreover, an $N=1$ massless supermultiplet with superhelicity $\kappa$ consists of two massless
states with helicity $\kappa$ and $\kappa-\half$, respectively.\footnote{In the literature, it is more
common to define $L^\mu=(\mathcal{K}+\half)P^\mu$, in which case the helicities of the massless $N=1$ supermultiplet
are $\kappa+\half$ and $\kappa$ (e.g., see refs.~\cite{Srivastava,Buchbinder}).
In our opinion, the definition of the superhelicity operator given in \eq{shelicitydef} is cleaner.}

Any quantum field theory realization of supersymmetry respects CPT symmetry.  Since the helicity
changes sign under a CPT transformation, it follows that any
irreducible massless supermultiplet with superhelicity~$\kappa$ must
be accompanied by the corresponding CPT-conjugate states that make up an
irreducible massless supermultiplet with superhelicity
$-\kappa+\half$.  
Hence, without loss of generality, we can restrict the possible values
of the superhelicity to $\kappa=\half,1,\tfrac32,\ldots$.  
These results are summarized in
Table~\ref{masslesssuperplet}.  The explicit construction of the
states of an irreducible massless supermultiplet and a discussion of their properties is presented in Section~\ref{App}.

\begin{table}[t!]
\caption{\small States of an $N=1$ massless supermultiplet of superhelicity $\kappa$
and the corresponding
CPT conjugates which comprise an $N=1$ massless super\-multiplet of superhelicity $-\kappa+\half$.
An interpretation
is provided for $\kappa=s$ and $\kappa=s-\half$, where $s$ is a positive integer.
In the special case of $\kappa=\half$, the scalar boson of the supermultiplet is complex, whereas for $\kappa=1,\tfrac{3}{2},2,\ldots$,
the bosonic member of the supermultiplet is real
with nonzero spin.  In all cases, the number of bosonic and fermionic degrees of
freedom (D.o.f.) coincide and are equal to
2.\label{masslesssuperplet}}
\vskip 0.1in
{
\begin{tabular}{cccc} \hline
Helicities & D.o.f. & Interpretation ($\kappa=s$) & Interpretation ($\kappa=s-\half$) \\ \hline
$\kappa$\,,\,$-\kappa$ & $2$ & spin-$s$ boson & spin-$(s-\half)$ fermion \\
$\kappa-\half$\,,\,$-\kappa+\half$  & $2$ & spin-$(s-\half)$ fermion & spin-$(s-1)$ boson \\ \hline
\end{tabular}}
\end{table}

\begin{example}[A massless chiral supermultiplet, with $\boldsymbol{\kappa=\half}$]
Including the CPT-conjugates, this supermultiplet contains two states of helicity 0, and two states of helicity $\pm\half$,
respectively, which yields a massless complex scalar and a
massless Majorana fermion.  We recognize this as the massless limit of
a massive $j=0$ chiral supermultiplet.  
\end{example}

\begin{example}[a massless gauge supermultiplet, with $\boldsymbol{\kappa=1}$]
Including the CPT-conjugates, this supermultiplet contains two states of helicity $\pm\half$
and two states of helicity $\pm 1$, which yields a massless
Majorana fermion and a massless spin-1 particle.  This is a gauge supermultiplet
 (e.g the photino and the photon of
supersymmetric QED).  
\end{example}

In Problem \ref{pr:spin2}, you will show that
a massless supermultiplet with $\kappa=2$ and its CPT-conjugates
contains
a massless spin-$\tfrac{3}{2}$ and a massless spin 2 particle, which
is realized in supergravity by the
gravitino and the graviton, respectively.

\subsection{Consequences of super-Poincar\'e invariance}

A Poincar\'e invariant quantum field theory respects the Poincar\'e algebra generated
by $\{P^\mu\,,\,J^{\mu\nu}\}$, which satisfy commutation relations given by
\eqst{spoincarealg1}{spoincarealg3}.  One of the basic postulates of
Poincar\'e-invariant quantum field theory  states that a translationally-invariant,
Lorentz-invariant vacuum $\ket{0}$ exists such that\cite{Roman},
\beq \label{pvacuum}
P^\mu\ket{0}=0\,,\qquad\quad J^{\mu\nu}\ket{0}=0\,.
\eeq
In particular, $\bra{0}P^\mu\ket{0}=0$.  Indeed if $\bra{0}P^\mu\ket{0}\neq 0$, then
the vacuum would not be invariant under Lorentz transformations.  This is easily proven
by taking the vacuum expectation value of 
\beq
\exp\left(\tfrac{1}{2}i\theta_{\rho\tau}J^{\rho\tau}\right)  P^\mu
  \exp\left(-\tfrac{1}{2}i\theta_{\rho\tau}J^{\rho\tau}\right)=\Lambda^{\mu}{}_{\nu}P^\nu\,,
\eeq
where the $\theta_{\rho\tau}=-\theta_{\rho\tau}$ parameterize the $4\times 4$ Lorentz transformation
matrix $\Lambda^\mu{}_\nu$
[cf.~\eqs{lambda44}{explicitsmunu}].  
 Using $J^{\mu\nu}\ket{0}=0$,
it follows that
\beq
\bra{0}P^\mu\ket{0}=\Lambda^{\mu}{}_{\nu}\bra{0}P^\nu\ket{0}\,,
\eeq
which holds for all Lorentz transformations $\Lambda$.  Thus, it follows that $\bra{0}P^\mu\ket{0}=0$.

A super-Poincar\'e invariant quantum field theory respects the SUSY algebra
generated by
$\{P^\mu\,,\,J^{\mu\nu}\,,\,Q_\alpha\,,\,Q^{\dagger\dot\alpha}\}$.
The SUSY algebra generators satisfy
the commutation relations of the Poincar\'e algebra and the
(anti)commutation relations given by
\eqst{susyalg1}{susyalg5}.  Two important consequences can be established:
\vskip 0.1in

1. \textit{The vanishing of the vacuum energy is a necessary and sufficient condition
for the existence of a global supersymmetric vacuum.}
\vskip 0.1in

2. \textit{In a theory governed by a supersymmetric action, for a fixed non-zero $P_\mu$ the number of bosonic
and fermionic degrees of freedom coincide.}
\vskip 0.1in

\noindent
We address these two results in the next two subsections.

\subsubsection{The vacuum energy of a globally supersymmetric theory}

In order to prove that the vanishing of the vacuum energy is a necessary and sufficient condition
for the existence of a global supersymmetric vacuum, we consider
the anticommutation relations of the fermionic generators of the SUSY
algebra,
\beq \label{QQanti}
\{Q_\alpha\,,\,Q^\dagger_{\dot\beta}\}=2\sigma^\mu_{\alpha\dot\beta}P_\mu\,.
\eeq
Following the derivation of \eq{pzero},
\beq \label{QQpzero}
P^0=\tfrac{1}{4}\left[Q_1 Q_1^\dagger+Q_1^\dagger Q_1+Q_2 Q_2^\dagger+Q_2^\dagger Q_2\right]\,.
\eeq
Since the right-hand side of \eq{pzero} is
positive semi-definite (and neither $Q$ nor $Q^\dagger$ is the zero operator), it
follows that
\beq
\vev{0\,|P^0\,|\,0}=0\quad\Longleftrightarrow\quad Q_\alpha\ket{0}=0\,.
\eeq
In particular, $Q_\alpha\ket{0}=0$ implies that the vacuum is supersymmetric, in the same way
that $P^\mu\ket{0}=J^{\mu\nu}\ket{0}=0$ imply that the vacuum is translationally-invariant
and Lorentz-invariant.\footnote{Equivalently,
$\bra{0}\{Q_\alpha\,,\,Q^\dagger_{\dot\beta}\}\ket{0}=0$,
by covariance with respect to the SUSY algebra,
since there are no spinor quantities with one undotted and one dotted index
that can appear on the right hand side of this equation.
Hence, $Q_\alpha\ket{0}=0$, which then yields $\bra{0}P^0\ket{0}=0$.}

However, this proof is troubling for two separate reasons.  First, suppose that the action
of the theory is invariant under supersymmetric transformations, but the vacuum is not
preserved by supersymmetry.  In this case, $Q_\alpha\ket{0}\neq 0$, and we say that
supersymmetry is spontaneously broken.  Then, \eq{QQpzero} implies that
$\bra{0}P^0\ket{0}>0$, which contradicts \eq{pvacuum}.  Thus, it appears that the
spontaneous breaking of supersymmetry is not possible without breaking Lorentz invariance.
Perhaps a more fundamental objection is that the concept of the vacuum energy is
usually considered to be unphysical
in non-gravitational theories, as it is commonly asserted that only energy differences are physical.
Thus, it seems to be a matter of convention to choose the vacuum energy such that $\bra{0}P^0\ket{0}=0$.

To overcome the objections raised above, we re-examine the concept of the vacuum energy
in relativistic (non-gravitational) quantum field theory.  Using the
Noether procedure, the conserved
canonical energy-momentum tensor, $T^{(c)}_{\mu\nu}$ can be obtained,
which satisfies $\partial^\mu T^{(c)}_{\mu\nu}=0$.\footnote{The arguments given here do not depend on
whether one employs the canonical energy momentum tensor or the
improved symmetrized energy-momentum tensor.}
One can then formally compute the vacuum energy
density by summing over the vacuum Feynman diagrams of the theory.  By Lorentz covariance~\cite{Witten},
\beq
\bra{0}T^{(c)}_{\mu\nu}\ket{0}=\mathcal{E}g_{\mu\nu}\,,
\eeq
where $\mathcal{E}$ is typically UV divergent.  Since the Hamiltonian density is identified
as $\mathscr{H}=T_{00}$, it follows that $\mathcal{E}$ is the vacuum energy density.
However, one is always free to define a new subtracted energy-momentum tensor,
\beq
T_{\mu\nu}\equiv T^{(c)}_{\mu\nu}-\mathcal{E} g_{\mu\nu}\,,
\eeq
which is a Lorentz-covariant expression.\footnote{For example, in the quantum theory
of free fields, the vacuum energy is set to zero by defining the Hamiltonian density to be
normal ordered.}
By construction, $\partial^\mu T_{\mu\nu}=0$ and
\beq
\bra{0}T_{\mu\nu}\ket{0}=0\,.
\eeq
The energy-momentum tensor $T_{\mu\nu}$ plays
a distinguished role in relativistic quantum field theory, since it can be used
to construct the generators of spacetime translations,
\beq \label{vacp}
P_\mu=\int d^3 x\ T_\mu{}^0\,,
\eeq
that satisfy $\bra{0}P_\mu\ket{0}=0$.  Indeed, $P_\mu$ defined by \eq{vacp} is a four-vector with respect
to Lorentz transformations.
Likewise, one can construct a distinguished angular momentum tensor $M_{\mu\nu\lambda}$ that
can be used to construct the generators of Lorentz transformations
\beq
J_{\mu\nu}=\int d^3 x M_{\mu\nu}{}^0\,,
\eeq
which satisfy $\bra{0}J_{\mu\nu}\ket{0}=0$.

However, in a supersymmetric theory, another choice of the energy-momentum tensor is
natural.  The fermionic generators $Q_\alpha$ and $Q^{\dagger\dot\alpha}$ of the SUSY algebra
are time-independent (conserved) quantities that are obtained by integrating the zeroth component
of the supercurrents,
\beq \label{Qint}
Q_\alpha=\int d^3 x J_\alpha^0\,,\qquad\qquad Q^{\dagger\dot\alpha}=\int d^3 x J^{\dagger\dot\alpha\,0}\,.
\eeq
In a theory governed by a supersymmetric Lagrangian, the supercurrents $J_\alpha^\mu$ and $J^{\dagger\dot\alpha\,\mu}$
are related by supersymmetry to
an energy-momentum tensor, denoted by $T^{(\rm SUSY)}_{\mu\nu}$.
Then, the proper interpretation of \eq{QQanti} is~\cite{deWit}
\beq
\{Q_\alpha\,,\,Q^\dagger_{\dot\beta}\}=2\sigma^\mu_{\alpha\dot\beta}\int d^3 x\, T^{(\rm SUSY)}{}_\mu{}^0\,.
\eeq
One can then rewrite the above anticommutation relation as:
\beq \label{revisedQQ}
\{Q_\alpha\,,\,Q^\dagger_{\dot\beta}\}=2\sigma^\mu_{\alpha\dot\beta}P_\mu+2E_0\sigma^0_{\alpha\dot\beta}\,,
\eeq
where $P_\mu$ is defined by \eq{vacp} and
\beq \label{Ezero}
E_0\equiv\int d^3 x\,\bra{0}T^{(\rm SUSY)}{}_0{}^0\ket{0}\,.
\eeq
If $E_0=0$ (which corresponds to $T^{(\rm SUSY)}_{\mu\nu}=T_{\mu\nu}$),
then we recover the standard SUSY algebra, and the
vacuum is supersymmetric.  If $E_0\neq 0$, then \eq{revisedQQ} is consistent with $\bra{0}P^\mu\ket{0}=0$
(which is required by the Lorentz-invariant vacuum) and with $Q_\alpha\ket{0}\neq 0$.  In particular,
$E_0$ serves as an order parameter for broken supersymmetry.

Note that $E_0\geq 0$ since \eq{revisedQQ} implies that:
\beq
E_0=\tfrac{1}{4}\bra{0}Q_1 Q_1^\dagger+Q_1^\dagger Q_1+Q_2 Q_2^\dagger+Q_2^\dagger Q_2\ket{0}\geq 0\,.
\eeq
In supersymmetric theories, it is common to call $E_0$ the vacuum energy.  Thus, if supersymmetry is
spontaneously broken, then this definition of the vacuum energy is not compatible with usual
conventions of quantum field theory in which the vacuum energy is defined to be zero.

Although the conclusions obtained above are correct, the derivation of \eq{revisedQQ} is still somewhat formal.
Indeed
if the vacuum breaks supersymmetry, then the integrals in \eq{Qint} do not converge when integrated
over an infinite volume (this is an infrared divergence), so strictly
speaking the fermionic generators $Q_\alpha$ and $Q^{\dagger\dot\alpha}$ are undefined.\footnote{Moreover,
given a non-zero value for $\bra{0}T^{(\rm SUSY)}{}_0{}^0\ket{0}$, which is a constant by
translational invariance, one sees that $E_0$ defined in \eq{Ezero} also diverges in the infinite volume limit.}
Nevertheless,
the supercurrents are conserved, as expected in a supersymmetric theory with no \textit{explicit}
supersymmetry breaking.  In section~\ref{goldstino}, we will demonstrate that given a supersymmetric
Lagrangian, if the vacuum breaks supersymmetry
then a massless Goldstone fermion exists in the spectrum.  The long range forces mediated by
this massless particle are responsible for the non-convergence of the integrals in \eq{Qint}.
Equivalently, in a spontaneously-broken globally supersymmetric theory, applying $Q_\alpha$ to the vacuum
creates a zero-momentum massless fermionic state, which is a state of infinite norm~\cite{weinberg3}.

\subsubsection{Equality of bosonic and fermionic degrees of freedom in
  super\-symmetric theories}

In a theory governed by a supersymmetric action, for a fixed non-zero $P_\mu$ the number of bosonic
and fermionic degrees of freedom coincide.  To prove this result, we first observe that the application
of $Q_\alpha$ or $Q^\dagger_{\dot\alpha}$ to a physical state changes that state by adding half a unit of spin.
An explicit example of this behavior can be seen in \eqs{massiveplet2}{massiveplet3}.  We can summarize
this behavior in the following schematic equations,
\beq
Q_\alpha\ket{B}=\ket{F}\,,\qquad\quad Q_\alpha\ket{F}=\ket{B}\,,
\eeq
and similarly for the application of $Q^\dagger_{\dot\alpha}$, where $\ket{B}$ is a bosonic state
and $\ket{F}$ is a fermionic state.  It is convenient to introduce an operator, denoted by $(-1)^F$,
with the following properties:
\beq
(-1)^F\ket{B}=\ket{B}\,,\qquad\quad (-1)^F\ket{F}=-\ket{F}\,.
\eeq
Note that
\beqa
Q_\alpha(-1)^F\ket{F}&=&-Q_\alpha\ket{F}=-\ket{B}\,, \\
(-1)^F Q_\alpha\ket{F}&=&(-1)^F\ket{B}=\ket{B}\,,
\eeqa
and similarly for the application of $Q^\dagger_{\dot\alpha}$.  It follows that $Q_\alpha$
[and $Q^\dagger_{\dot\alpha}$] anticommute with $(-1)^F$,
\beq \label{minusF}
\{Q_\alpha\,,\,(-1)^F\}=\{Q^\dagger_{\dot\alpha}\,,\,(-1)^F\}=0\,.
\eeq

Using \eq{minusF}, we can evaluate the following trace over physical states,
\beqa
\Tr\left[(-1)^F\{Q_\alpha\,,\,Q^\dagger_{\dot\beta}\}\right]&=&
\Tr\left[(-1)^F(Q_\alpha Q^\dagger_{\dot\beta}+Q^\dagger_{\dot\beta}Q_\alpha)\right] \nonumber \\
&=& \Tr\left[-Q_\alpha (-1)^F Q^\dagger_{\dot\beta}+(-1)^F Q^\dagger_{\dot\beta}Q_\alpha\right]\nonumber \\
&=& \Tr\left[-Q^\dagger_{\dot\beta}Q_\alpha (-1)^F+Q^\dagger_{\dot\beta}Q_\alpha (-1)^F\right]\nonumber \\
&=&0\,,
\eeqa
after a cyclic permutation within the trace at the penultimate step.
Employing \eq{susyalg5},
we conclude that
\beq \label{traceF}
\Tr (-1)^F=0\,,\qquad \text{for any fixed non-zero $P^\mu$}\,.
\eeq
For a fixed non-zero eigenvalue $p^\mu$ obtained by applying the
momentum operator $P^\mu$ to a physical
state, 
\beq
\Tr (-1)^F=\sum_{\{r\}} \bra{p^\mu,\{r\}}(-1)^F\ket{p^\mu,\{r\}}=N_B(p^\mu)-N_F(p^\mu)=0\,,
\eeq
where $\{r\}$ indicates all other quantum numbers of the physical state.  Thus, the number of
bosonic ($N_B$) and fermionic ($N_F$) degrees of freedom coincide.

We have already observed that \eq{traceF} is satisfied by all
positive energy representations of the SUSY algebra.  The
proof above demonstrates that the equality of bosonic and fermionic
degrees of freedom in supersymmetric theories is far more general.
Indeed, the only case where this equality can break down is when~$P^\mu=0$,
corresponding to the vacuum state of the supersymmetric theory.\footnote{For
example, Witten showed that in an SU($N$) supersymmetric Yang-Mills
theory, $\Tr (-1)^F=N$ for the supersymmetric ground
state~\cite{Witten2}.}

%%%%%%%%%%%%%%%%%%%%%%%%%%%%%%%%%%%%%%%%%%%%%%%%%%

\subsection{Supersymmetric theories of spin-0 and spin-\half\ particles}

The simplest supermultiplet contains a complex scalar and a
two-component (Majorana) fermion, of common mass $m$.  The case of
$m\neq 0$ corresponds to superspin $j\!=\!0$ and the case of $m\!=\!0$
corresponds to superhelicity $\half$ and its CPT-conjugate.

\subsubsection{The Wess-Zumino Lagrangian}

A Lagrangian that respects the SUSY algebra is given by
\begin{align}
\mathscr{L}=(\partial_\mu A)^\dagger(\partial^\mu A)+ i \psi^\dagger \sigmabar^\mu \partial_\mu \psi-\left|\frac{dW}{dA}\right|^2-\frac12\left[\frac{d^2 W}{dA^2}\,\psi\psi+\left(\frac{d^2 W}{dA^2}\right)^{\!\!\dagger}\!\!\psi^\dagger\psi^\dagger\right]\,,
\label{eq:LWZoriginal}
\end{align}
where $A$ is a complex scalar,\footnote{Employing $A$ for a complex
  scalar field rather than $\phi$ follows the notation first introduced
  in Ref.\cite{WessBagger}.  It should not be confused with the
  notation for a vector field, which will henceforth be denoted
  by $V$.} 
$\psi$ and $\psi^\dagger$ are two-component spinors, and $W=W(A)$
[called the \textit{superpotential}]
is a holomorphic function of $A$ (\textit{i.e.}, a function of $A$ and \textit{not}~$A^\dagger$).
If $W(A)$ is (at most) a cubic polynomial in $A$, then the above
Lagrangian yields a renormalizable 
quantum field theory called the \textit{Wess-Zumino model}.
For example, a simple quadratic superpotential,
$W=\half mA^2$, describes
a free theory of a complex scalar and a Majorana fermion of common mass $|m|$.
An interacting theory is
obtained by including a cubic term in the superpotential,
\begin{align}
W=\half mA^2+\tfrac{1}{3}g A^3\,.\label{wcubic}
\end{align}
Without loss of generality, we can assume that $m$ and $g$ are
non-negative (by appropriate rephasing of $A$ and $\psi$).  Then,
inserting \eq{wcubic} into \eq{eq:LWZoriginal} yields the Wess-Zumino
Lagrangian,
\begin{align}
\begin{split}
\mathscr{L}&=
(\partial_\mu A)^\dagger(\partial^\mu A)+ i \psi^\dagger \sigmabar^\mu \partial_\mu \psi-\half m(\psi\psi+\psi^\dagger\psi^\dagger)-m^2(A^\dagger A)
 \\
&\quad -g(A\psi\psi+A^\dagger\psi^\dagger\psi^\dagger)-mg(A^\dagger A)(A+A^\dagger)-g^2(A^\dagger A)^2 \,.
\end{split} \label{wzlag}
\end{align}
As expected, the boson and fermion are mass-degenerate.  Moreover, SUSY imposes relations among the couplings.  In this model, we see that the quartic scalar coupling is the square of the Yukawa (scalar-fermion-fermion) coupling.

In order to employ four-component Feynman rules, it is convenient to
convert the Wess-Zumino Lagrangian into four-component fermion form.
Writing $A=(S+iP)/\sqrt{2}$, where $S$ and $P$ are hermitian fields, we obtain
\begin{align}
\begin{split}
\mathscr{L}&=\half (\partial_\mu S)^2+\half (\partial_\mu P)^2-\half m^2(S^2+P^2)+\half \overline\Psi_M(i\gamma^\mu\partial_\mu-m)\Psi_M
 \\
&\quad -\frac{g}{\sqrt{2}}\left[S\overline\Psi_M\psi_M-iP\overline\Psi_M\gamma\ls{5}\Psi_M\right]-\frac{mg}{\sqrt{2}}S(S^2+P^2)
-\tfrac{1}{4}g^2(S^2+P^2)^2\,.
\end{split}
\end{align}
Note that this Lagrangian separately conserves C, P and T.  We identify $S$ as a scalar and $P$ as a pseudoscalar.

\subsubsection{Invariance of the Wess-Zumino Lagrangian with respect
  to SUSY transformations}

The Wess-Zumino Lagrangian given by \eq{wzlag} is invariant with respect to global supersymmetry transformations.  Explicitly, these transformations depend on an
infinitesimal Grassmann (anticommuting) two-component spinor parameter $\xi$ that is independent of the spacetime position $x$,
\begin{align}
\delta_\xi A &= \sqrt{2}\,\xi \psi\,,\label{susytr1}\\
\delta_\xi\psi_\alpha
&=
- i\sqrt{2} (\sigma^\mu \xi^\dagger)_\alpha\> \partial_\mu A-\sqrt{2}\,\xi_\alpha\left(\frac{dW}{dA}\right)^{\!\!\dagger}\,.\label{susytr2}
\end{align}
By hermitian conjugation, one also obtains
\begin{align}
\delta_\xi A^\dagger &= \sqrt{2}\,\xi^\dagger \psi^\dagger\,,\label{susytr3}
 \\
\delta_\xi\psi^\dagger_{\dot{\alpha}}
&=
 i \sqrt{2}(\xi\sigma^\mu)_{\dot{\alpha}}\>   \partial_\mu A^\dagger-\sqrt{2}\,\xi^\dagger_{\dot\alpha}\left(\frac{dW}{dA}\right)\,.\label{susytr4}
 \end{align}
 Applying these transformation laws to \eq{wzlag}, one obtains a result of the form
 \begin{align}
 \deltaxi\mathscr{L}=\partial_\mu K^\mu\,.
 \label{eq:Kmu}
 \end{align}
 That is, the action of the Wess-Zumino Model, $S=\int d^4 x\,\mathcal{L}$, is invariant under global SUSY transformations; \textit{i.e.}, $\deltaxi S=0$.

 But, how do we know that the transformation laws just introduced correspond to SUSY transformations?  Recall that for ordinary spacetime translations,
 \beq
 e^{ia\newcdot P}\Phi(x)e^{-ia\newcdot P}=\Phi(x+a)\,,
 \eeq
 which in infinitesimal form is given by
 \beq
 i\bigl[P^\mu\,,\,\Phi(x)\bigr]=\partial^\mu\Phi(x)\,,
 \eeq
 where $\Phi=A$ or $\psi$. Equivalently, for an infinitesimal translation, 
 \beq
 \delta_a\Phi(x)\equiv\Phi(x+a)-\Phi(x)\simeq a^\mu\partial_\mu\Phi(x)=ia^\mu\bigl[P^\mu\,,\,\Phi(x)\bigr]\,.
 \eeq
 Likewise, since $Q$ and $Q^\dagger$ are the generators of SUSY-translations, we expect
 \beq \label{susytranslate}
 \deltaxi\Phi(x)=i\bigl[\xi Q+\xi^\dagger Q^\dagger\,,\,\Phi(x)\bigr]\,.
 \eeq
Consider the commutator of two SUSY-translations:
 \beqa
 (\deltaeta\deltaxi-\deltaxi\deltaeta)\Phi(x)&=&\biggl[i(\eta Q+\eta^\dagger Q^\dagger)\,,\,
 \bigl[i(\xi Q+\xi^\dagger Q^\dagger)\,,\,\Phi(x)\bigr]\biggr] 
  -(\xi\longleftrightarrow\eta)\nn
  \\
 &=&\biggl[\bigl[i(\eta Q+\eta^\dagger Q^\dagger)\,,\,i(\xi Q+\xi^\dagger Q^\dagger)\bigr]\,,\,\Phi(x)\biggr]\,,
 \eeqa
 after employing the Jacobi identity for the double commutators.  Using the SUSY algebra,
 $$
 \bigl[\eta Q\,,\,\xi^\dagger Q^\dagger\bigr]=2(\eta\sigma^\mu\xi^\dagger) P_\mu\,.
 $$
 Note that the anticommutator has been converted into a commutator due to the fact that $\eta$ and $\xi$ are anticommuting two-component spinors.  Likewise, 
 $$
 \bigl[\eta Q\,,\,\xi Q\bigr]=\bigl[\eta^\dagger Q^\dagger\,,\,\xi^\dagger Q^\dagger\bigr]=0\,.
 $$
 Hence, we end up with 
 \beqa
 \bigl[\deltaeta\,,\,\deltaxi\bigr]\Phi(x)&=&2(\xi\sigma^\mu\eta^\dagger-\eta^\dagger\sigma^\mu\xi^\dagger)\bigl[P_\mu\,,\,\Phi(x)\bigr] \nn
  \\
&=& -2i(\xi\sigma^\mu\eta^\dagger-\eta^\dagger\sigma^\mu\xi^\dagger)\partial_\mu\Phi(x)\,.
 \eeqa
Likewise, a similar computation yields,
 \beqa
  \bigl[\deltaeta\,,\,\deltaxi\bigr]A(x)&=&-2i(\xi\sigma^\mu\eta^\dagger-\eta^\dagger\sigma^\mu\xi^\dagger)\partial_\mu A(x)\,, 
  \\[6pt]
   \bigl[\deltaeta\,,\,\deltaxi\bigr]\psi_\alpha(x)&=&-2i(\xi\sigma^\mu\eta^\dagger-\eta^\dagger\sigma^\mu\xi^\dagger)\partial_\mu \psi_\alpha +R\,, \label{eq:Remainder}
 \eeqa
 where the remainder $R$ vanishes after imposing the classical field
 equations for $\psi_\alpha(x)$, as you will verify in Problem \ref{pr:R}.
We conclude that the SUSY algebra is realized \textit{on-shell}, \textit{i.e.}, after employing the classical field equations.
 
 It is instructive to employ
 Noether's theorem, which states that an invariance of the action under
 a continuous symmetry implies the existence of a conserved current.
 Since we have explicitly identified the SUSY
 transformations, we can 
use Noether's theorem to determine the corresponding conserved supercurrent.  Using $\deltaxi\mathscr{L}=\partial_\mu K^\mu$, the resulting conserved Noether supercurrents are
 \begin{align}
 \xi^\alpha J_\alpha^\mu+\xi^\dagger_{\dot\alpha} J^{\dagger\,\mu\dot\alpha}=\sum_\Phi \deltaxi\Phi\,\frac{\delta\mathscr{L}}{\delta(\partial_\mu \Phi)}-K^\mu\,,
 \end{align}
 where the sum is taken over $\Phi=A$, $\psi$.  Note that the supercurrent has both a Lorentz index and a spinor index.
 Noether's theorem states that the supercurrent is conserved \textit{after imposing the classical field equations}.  That is, 
\beq
 \partial_\mu J^\mu_\alpha=\partial_\mu J^{\dagger\,\mu\dot\alpha}=0\,.
\eeq
 
The supercharges are defined in the usual way (as previously noted):
 \begin{align} 
Q_\alpha=\int d^3 x J_\alpha^0\,,\qquad\qquad Q^{\dagger\dot\alpha}=\int d^3 x J^{\dagger\dot\alpha\,0}\,.\label{QJ}
\end{align}
These are expressions that depend on the fields $A$ and $\psi$.
One can now employ the canonical commutation relations of the boson field $A$ and the canonical anticommutation relations of the fermion field $\psi$ to verify that 
\beq \label{QandCCR}
\{Q_\alpha\,,\,Q_\beta\}=\{Q^\dagger_{\dot\alpha}\,,\,Q^\dagger_{\dot\beta}\}=0\,,\qquad\quad
\{Q_\alpha\,,\,Q^\dagger_{\dot\beta}\}=2\sigma^\mu_{\alpha\dot\beta}P_\mu\,,
\eeq
where $P_\mu$ is the Noether charge of spacetime translations given in \eq{vacp}.
%\beq
%P_\mu=\int d^3 x\,\widetilde{T}_\mu{}^{0}\,.
%\eeq

\subsection{The SUSY algebra realized off-shell}
\label{offshell}

The SUSY transformation laws of the Wess-Zumino Lagrangian exhibited in \eqs{susytr1}{susytr2} are not in an
optimal form for two reasons.  First, in the case of a cubic
superpotential $W(A)$, the transformation law for $\psi_\alpha$ is non-linear in the fields.
Second, the SUSY algebra is only realized on-shell.
We can address both these issues by introducing an auxiliary complex scalar field $F(x)$.
Consider the alternative Lagrangian,
\beqa
\mathscr{L}&=&(\partial_\mu A)^\dagger(\partial^\mu A)+ i \psi^\dagger \sigmabar^\mu \partial_\mu \psi
+F^\dagger F 
+F\,\frac{dW}{dA}+F^\dagger\left(\frac{dW}{dA}\right)^{\!\!\dagger} \nn \\
&& -\frac12\left[\frac{d^2 W}{dA^2}\,\psi\psi+\left(\frac{d^2 W}{dA^2}\right)^{\!\!\dagger}\!\!\psi^\dagger\psi^\dagger\right]\,.
\label{eq:LWZF}
\eeqa
The field $F(x)$ is auxiliary since $\mathscr{L}$ does not depend on $\partial_\mu F$ and
$\partial_\mu F^\dagger$.  That is, $F$ and $F^\dagger$ are non-dynamical fields.

We can trivially solve for $F$ and $F^\dagger$ using the classical field equations,
\begin{align}
\frac{\partial\mathscr{L}}{\partial F}&=0\qquad\Longrightarrow \qquad F^\dagger=-\frac{dW}{dA}\,,\label{fs}
\\
\frac{\partial\mathscr{L}}{\partial
  F^\dagger}&=0\qquad\Longrightarrow\qquad F=
-\left(\frac{dW}{dA}\right)^{\!\!\dagger}\,.\label{f}
\end{align}
Hence, \eqs{fs}{f} yield,
\beq
F^\dagger F+F\,\frac{dW}{dA}+F^\dagger\left(\frac{dW}{dA}\right)^{\!\!\dagger}=-\left|\frac{dW}{dA}\right|^2\,.
\eeq
Plugging this result back into \eq{eq:LWZF},
we recover the general form of the Wess-Zumino Lagrangian given by \eq{eq:LWZoriginal}.

The Lagrangian including the auxiliary fields given by \eq{eq:LWZF} is
also invariant under SUSY translations.  The appropriately modified
SUSY transformation laws are now given by
\begin{align}
\delta_\xi A &= \sqrt{2}\,\xi \psi\,,\label{offshell1}
\\
\delta_\xi\psi_\alpha
&=
- i\sqrt{2} (\sigma^\mu \xi^\dagger)_\alpha\> \partial_\mu A+\sqrt{2}\,\xi_\alpha F\,,\label{offshell2}
\\
\deltaxi F&=-i\sqrt{2}\,\xi^\dagger\sigmabar^\mu\partial_\mu\psi\,.\label{offshell3}
\end{align}
By hermitian conjugation, one also obtains
\begin{align}
\delta_\xi A^\dagger &= \sqrt{2}\,\xi^\dagger \psi^\dagger\,,
 \\
\delta_\xi\psi^\dagger_{\dot{\alpha}}
&=
 i \sqrt{2}(\xi\sigma^\mu)_{\dot{\alpha}}\>   \partial_\mu A^\dagger+\sqrt{2}\,\xi^\dagger_{\dot\alpha}F^\dagger\,,
  \\
\deltaxi F^\dagger&=i\sqrt{2}(\partial_\mu\psi^\dagger)\sigmabar^\mu\xi\,.
 \end{align}
 Applying these transformation laws to \eq{eq:LWZF}, one obtains a result of the form
 \begin{align}
 \deltaxi\mathscr{L}=\partial_\mu K^{\prime\,\mu}\,,
 \label{eq:Kpmu}
 \end{align}
where the explicit form for $K^{\prime\,\mu}$ is  to be determined in
Problem~\ref{pr:kprime}.  Moreover, as you will verify in Problem \ref{pr:xieta},
 \begin{align}
  \bigl[\deltaeta\,,\,\deltaxi\bigr]\Phi(x)=-2i(\xi\sigma^\mu\eta^\dagger-\eta^\dagger\sigma^\mu\xi^\dagger)\partial_\mu \Phi(x)\,,
  \end{align}
 for $\Phi=A$, $\psi$ and $F$ \textit{without} the need to impose the classical field equations.
Thus, the Wess-Zumino Lagrangian with auxiliary fields included as in
\eq{eq:LWZF}  is invariant under SUSY translations, and the SUSY algebra is realized \textit{off-shell}, \textit{i.e.}, without requiring that the fields satisfy their classical field equations.

The following two observations will be particularly useful as we move
forward.  First, note that the mass dimensions of the fields are given
by $[A]=1$, $[\psi]=\tfrac{3}{2}$ and $[F]=2$, which is consistent with
the requirement that $[\mathscr{L}]=4$ (since the action is
dimensionless in units of $\hbar=1$).  Then,
\eqst{offshell1}{offshell3} are dimensionally consistent if $[\xi]=\half$.
 Second, note that $\deltaxi F$ given in \eq{offshell3}  is a total
 derivative.  Indeed, $\deltaxi F$ is a total derivative as a consequence of dimensional analysis and the linearity of the SUSY transformation laws.  This implies that $\deltaxi F$ must involve $\partial_\mu$, since $[\partial_\mu]=1$.
An important consequence of this observation is that
$\int \!d^4x\, F$ is invariant under SUSY transformations.

\subsection{Counting bosonic and fermionic degrees of freedom}
It is instructive to count both the on-shell and off-shell bosonic and
fermionic degrees of freedom in the Wess-Zumino model, which
is a theory of a complex scalar and a
two-component fermion.   

A complex scalar possesses two real
degrees of freedom.  Note that applying the classical field equations
(in this case the inhomogeneous Klein-Gordon equation) does not affect
the number of scalar
degrees of freedom, but only the spacetime dependence of the scalar
field.  The two-component fermion $\psi_\alpha$ possesses two complex degrees of
freedom, which yields four real degrees of
freedom.\footnote{Equivalently, we can count $\psi$ and $\psi^\dagger$ as four independent degrees of freedom.} 
Applying the classical field equations,
\beq \label{diraceq}
i\sigmabar^\mu\partial_\mu\psi=\left(\frac{d^2 W}{dA^2}\right)^{\!\!\dagger}\psi^\dagger\,,
\eeq
which relate $\psi$ and $\psi^\dagger$, thereby eliminating two of the
four degrees of freedom.\footnote{If $d^2 W/dA^2=0$, then
  $i\sigmabar^\mu\partial_\mu\psi=0$ yields a relation between 
  $\psi_1$ and $\psi_2$.}
By taking the derivative of \eq{diraceq}, one can eliminate
$\psi^\dagger$ using the hermitian conjugate of \eq{diraceq}.
The resulting equation for $\psi$ is the inhomogeneous Klein-Gordon
equation, which does not further affect the number of
fermionic degrees of freedom.
Thus, the Wess-Zumino model possesses two on-shell bosonic and two fermionic
degrees of freedom.  

The counting of the off-shell degrees of freedom can be performed by examining the
Lagrangian [\eq{eq:LWZF}] expressed in terms of the propagating and
auxiliary fields.  In this case, we count two real degrees of freedom for
the complex scalar, four real degrees of freedom for the two-component
fermion and two real degrees of freedom for the complex auxiliary
field~$F$.  That is, the Wess-Zumino model possesses four bosonic and four fermionic
off-shell degrees of freedom.  

Thus, the number of bosonic and fermionic degrees of freedom match in both on-shell and off-shell counting.

\subsection{Lessons from the Wess-Zumino Model}

In our study of the Wess-Zumino model, we provided a Lagrangian that
incorporated the fields of a known supermultiplet.  However, it was
rather mysterious how this Lagrangian was obtained.  It was
even more mysterious how we came up with the correct SUSY
transformation laws for the various fields. 
Moreover, it was quite laborious to verify that the proposed SUSY
transformation laws satisfy the SUSY algebra and the action is 
invariant under super-Poincar\'e transformations.

We also learned that in order for the SUSY transformation laws to
respect the SUSY algebra off-shell, one must introduce additional
auxiliary fields.  One additional benefit of doing so is that the 
corresponding SUSY transformation laws are now linear in all the fields. 
For this reason, we introduced the auxiliary field $F$, which can be
used to write down the SUSY translation-invariant quantity $\int\! d^4
x \, F(x)$.  This observation actually provides an important clue for how to
construct a SUSY Lagrangian.

As we shall demonstrate in  Section~\ref{sec:superspace}, it is possible to develop a formalism in which, starting with
a known supermultiplet, one can trivially construct a Lagrangian that
is invariant under super-Poincar\'e transformations.   Moreover, this
formalism will provide explicit forms for the SUSY transformation
laws that automatically respect the SUSY algebra.

\subsection{\mbox{Appendix: Constructing the states of a supermultiplet}}
\label{App}

In this subsection, we provide further details on the construction of
the states of the massive and massless supermultiplets, which yields
the results presented in Tables~\ref{massivesuperplet} and
\ref{masslesssuperplet}.

\subsubsection{States of a massive supermultiplet of superspin $j$}

To construct the states of the massive supermultiplet,
we note that in the rest frame, the anticommutators given in \eqs{susyalg4}{susyalg5}
simplify to
\beqa
\{Q_1\,,\,Q^\dagger_1\}&=&\{Q_2\,,\,Q^\dagger_2\}=2m\,,\label{restframeQQ}\\
\{Q_1\,,\,Q_1\}&=&\{Q_2\,,\,Q_2\}=\{Q_1\,,\,Q_2\}=0\,,\label{restframeanti} \\
 \{Q^\dagger_1\,,\,Q^\dagger_1\}&=&\{Q^\dagger_2\,,\,Q^\dagger_2\}=\{Q^\dagger_1\,,\,Q^\dagger_2\}=0\,.
 \eeqa
 All states in a supermultiplet with superspin $j$ are simultaneous eigenstates of $P^2$,
 $\mathcal{J}^i \mathcal{J}^i$ and $\mathcal{J}^3$ with eigenvalues $m^2$, $j(j+1)$ and
 $j_3$, respectively, where the possible values of $j_3$
 are $-j,-j+1,\ldots,j-1,j$.  

For a fixed value of the superspin $j$,
 there exists a distinguished state of the supermultiplet that is a 
 simultaneous eigenstate of $P^2$, $\mathcal{J}^i\mathcal{J}^i$ and $\mathcal{J}^3$,
 denoted by $\ket{\Omega}$, which satisfies\footnote{Recall that if $\ket{s,m_s}$ are eigenstates
 of $\boldsymbol{\vec S}\llsup{\,2}$ and $S^3$ with corresponding eigenvalues $s(s+1)$ and $m_s$
 respectively, then
 $$
 S_{\pm}\ket{s,m_s}=\sqrt{(s\mp m_s)(s\pm m_s+1)}\ket{s,m_s\pm 1}\,.
 $$
 }
 \beq \label{Omegastate}
 Q_\beta\ket{\Omega}=0\,,\qquad\quad S_+\ket{\Omega}=0\,,
 \eeq
 where $S_{\pm}\equiv S^1\pm iS^2$.
 To verify that a state $\ket{\Omega}$ exists that is annihilated by $Q_\beta$,
let us assume the contrary.  Suppose that
 a simultaneous eigenstate of $P^2$, $\mathcal{J}^i\mathcal{J}^i$ and $\mathcal{J}^3$,
 denoted by $\ket{\Psi}$, is not annihilated by $Q_\beta$.  In the rest frame, \eq{BQQ}
 yields
 \beq \label{JQQ}
 [\mathcal{J}^i\,,\,Q_\beta]=[\mathcal{J}^i\,,\,Q^\dagger_{\dot\beta}]=0\,,
 \eeq
 so it follows that $Q_\beta\ket{\Psi}$ is also a simultaneous eigenstate of $P^2$, $\mathcal{J}^i\mathcal{J}^i$
 and $\mathcal{J}^3$.  By assumption, $Q_\beta\ket{\Psi}$ is not annihilated by $Q_\alpha$, so we
 conclude that $Q_\alpha Q_\beta\ket{\Psi}$ is also a simultaneous eigenstate of $P^2$, $\mathcal{J}^i\mathcal{J}^i$
 and $\mathcal{J}^3$.   But we now arrive at a contradiction, since \eq{restframeanti} yields
 \beq
 Q_\gamma\left(Q_\alpha Q_\beta\ket{\Psi}\right)=0\,.
 \eeq
 Consequently, there must be at least one state of the supermultiplet that satisfies
$Q_\beta\ket{\Omega}=0$.  Using \eqs{caljdef}{Omegastate}, it follows that
\beq
 \mathcal{J}^i\ket{\Omega}=S^i\ket{\Omega}\,.
 \eeq
 If $S_+\ket{\Omega}=0$, then it follows that $\ket{\Omega}$ is also
 a simultaneous eigenstate of $\boldsymbol{\vec{S}}\llsup{\,2}$ and $S^3$ with corresponding eigenvalues
 $j(j+1)$ and $j$.  Moreover, this state must be unique under the assumption that the
 superspin $j$ supermultiplet is an \textit{irreducible} representation of the $N=1$ supersymmetry
 algebra.

 Note that \eq{wQQ} when evaluated in the rest frame yields:
 \beq \label{siQcomm}
 [S^i\,,\,Q_\alpha]=i\sigma^{i0}{}_\alpha{}^\beta Q_\beta\,,\qquad\quad
 [S^i\,,\,Q^\dagger_{\dot\alpha}]=i\sigmabar^{i0\dot\beta}{}_{\dot\alpha}Q^\dagger_{\dot\beta}\,.
 \eeq
Hence, one can define additional states of the supermultiplet,
\beq
\ket{\Omega(j_3)}\equiv (S_-)^{j-j_3}\ket{\Omega}\,,\qquad\quad\text{for}~~j_3=-j,-j+1,\ldots,j-1,j\,,
\eeq
all of which satisfy
\beq \label{Omegastatej3}
Q_\alpha\ket{\Omega(j_3)}=0\,,
\eeq
as a result of \eq{siQcomm}.  As before, $\mathcal{J}^i\ket{\Omega(j_3)}=S^i\ket{\Omega(j_3)}$ as
a consequence of \eqs{caljdef}{Omegastatej3}.  It follows that
$\ket{\Omega(j_3)}$ is also a simultaneous eigenstate of $\boldsymbol{\vec{S}}\llsup{\,2}$ and $S^3$ with corresponding eigenvalues
 $j(j+1)$ and~$j_3$. That is,
 \beq \label{Omegajj3}
 \ket{\Omega(j_3)}=\ket{j,j_3}\,,
 \eeq
 where the rest-frame spin and its projection along the $z$-axis are explicitly indicated.

Starting from $\ket{\Omega(j_3)}=\ket{j,j_3}$, one can now construct the remaining states of the massive supermultiplet
by considering the series of states for each possible value of $j_3$,
 $
 \ket{\Omega(j_3)}\,,\, Q^\dagger_{\dot\alpha}\ket{\Omega(j_3)}\,,\, Q^\dagger_{\dot\alpha}Q^\dagger_{\dot\beta}\ket{\Omega(j_3)}\,,\,\ldots\,.
 $
This series of states terminates due to \eq{restframeanti} and only four independent states survive (for a given fixed value of $j_3$),
\beq \label{mstatesj}
\ket{\Omega(j_3)}\,,\quad Q^\dagger_1\ket{\Omega(j_3)}\,,\quad  Q^\dagger_2\ket{\Omega(j_3)}\,,\quad
 Q^\dagger_1Q^\dagger_2\ket{\Omega(j_3)}\,.
 \eeq

All the states of \eq{mstatesj} are mass-degenerate (with mass $m\neq 0$).  The spins of
these states can be determined by applying the operators $\boldsymbol{\vec S}\llsup{\,2}$ and $S^3$.
By virtue of \eq{Omegajj3}, we already know that $\ket{\Omega(j_3)}$ is a spin-$j$ state
with $S^3$-eigenvalue $j_3$.  Next, one can use \eq{siQcomm} to derive:
\beqa
[S^i\,,\,Q^\dagger_{\dot\alpha}Q^\dagger_{\dot\beta}]&=&iQ^\dagger_{\dot\gamma}\left[\sigmabar^{i0\dot\gamma}{}_{\dot\alpha}
Q^\dagger_{\dot\beta}-\sigmabar^{i0\dot\gamma}{}_{\dot\beta}Q^\dagger_{\dot\alpha}\right]\,,\\
\left[\boldsymbol{\vec S}\llsup{\,2}\,,\,Q^\dagger_{\dot\alpha}Q^\dagger_{\dot\beta}\right]&=&2iQ^\dagger_{\dot\gamma}\left[
\sigmabar^{i0\dot\gamma}{}_{\dot\alpha} Q^\dagger_{\dot\beta}
-\sigmabar^{i0\dot\gamma}{}_{\dot\beta} Q^\dagger_{\dot\alpha}\right]S^i\,.
\eeqa
It immediately follows that:
\beqa
[S^i\,,\,Q_1^\dagger Q_2^\dagger]&=&iQ_1^\dagger Q_2^\dagger\Tr \sigma^{i0}=0\,,\label{SiQ1Q2}\\
\left[\boldsymbol{\vec S}\llsup{\,2}\,,\,Q_1^\dagger Q_2^\dagger\right]&=&2iQ_1^\dagger Q_2^\dagger S^i\Tr \sigma^{i0}=0\,.
\label{S2Q1Q2}
\eeqa
Applying \eqs{SiQ1Q2}{S2Q1Q2} to the state $\ket{\Omega(j_3)}$, it follows that $Q_1^\dagger Q_2^\dagger
\ket{\Omega(j_3)}$ is also a spin-$j$ state with $S^3$-eigenvalue $j_3$.
This result is easily understood.  Noting that we can write
\beq
Q_1^\dagger Q_2^\dagger=\half\epsilon^{\dot\alpha\dot\beta}Q^\dagger_{\dot\alpha}Q^\dagger_{\dot\beta}\,,
\eeq
it follows that $Q_1^\dagger Q_2^\dagger$ is a \textit{scalar} operator.  This is consistent with the
fact that the antisymmetric part of the tensor product of two SU(2) spinor representations is an SU(2) singlet.
Thus, $Q_1^\dagger Q_2^\dagger\ket{\Omega(j_3)}$ and $\ket{\Omega(j_3)}$ possess the same eigenvalues
with respect to $\boldsymbol{\vec S}\llsup{\,2}$ and $S^3$.

To determine the properties of $Q_1^\dagger\ket{\Omega(j_3)}$  and $Q_2^\dagger\ket{\Omega(j_3)}$,
we first note that $Q_\alpha$ is a spinor operator 
that imparts spin-$\half$ to any state it acts on. 
Moreover, \eq{siQcomm} yields:
\beq \label{Qdaghalf}
\hspace{-0.2in}
S^3 Q_1^\dagger\ket{\Omega(j_3)}=(j_3+\half)Q_1^\dagger\ket{\Omega(j_3)}\,,
\quad
S^3 Q_2^\dagger\ket{\Omega(j_3)}=(j_3-\half)Q_2^\dagger\ket{\Omega(j_3)}.
\eeq
Hence, one can employ the standard results from the theory of angular momentum addition in quantum mechanics,
which relates the tensor product basis to the total angular momentum basis.  In particular,
\beq
\ket{j\,,\,m}=\sum_{m_1,m_2}\ket{j_1\,,\,m_1}\otimes\ket{j_2\,,\,m_2}\vev{j_1\,\, j_2\,;\, m_1\,\, m_2\,|\,j\,\, m}\,,
\eeq
where $\vev{j_1\,\, j_2\,;\, m_1\,\, m_2\,|\,j\,\, m}$ are the
Clebsch-Gordon (C-G) coefficients.  We employ the Condon-Shortly
phase conventions in which the C-G coefficients are real and symmetric.  In the present application,
we require the following two C-G coefficients (taking the upper and lower
signs, respectively),
\begin{align}
\ket{\half\,,\, \pm\half}\otimes\ket{j\,,\,
  m\mp\half}=&\left(\frac{j+\half\pm
               m}{2j+1}\right)^{1/2}\!\ket{j+\half\,,\, m}\nonumber \\
&\mp\left(\frac{j+\half\mp m}{2j+1}\right)^{1/2}\!\ket{j-\half\,,\, m},  
\end{align}
Eqs.~(\ref{Omegajj3}), (\ref{Qdaghalf}), (\ref{SiQ1Q2}) and
(\ref{S2Q1Q2}) imply that
\begin{align}
& \ket{\Omega(j_3)}=\ket{j\,,\,j_3}\,, \label{massiveplet1}\\
& Q_1^\dagger \ket{\Omega(j_3)}=\left(\frac{j+j_3+1}{2j+1}\right)^{1/2}\ket{j+\half\,,\,j_3+\half}
-\left(\frac{j-j_3}{2j+1}\right)^{1/2}\ket{j-\half\,,\,j_3+\half}\,, \label{massiveplet2} \\
& Q_2^\dagger \ket{\Omega(j_3)}=\left(\frac{j-j_3+1}{2j+1}\right)^{1/2}\ket{j+\half\,,\,j_3-\half}
+\left(\frac{j+j_3}{2j+1}\right)^{1/2}\ket{j-\half\,,\,j_3-\half}\, ,  \label{massiveplet3}\\
& Q_1^\dagger Q_2^\dagger \ket{\Omega(j_3)}=\ket{j\,,\,j_3}\,.\label{massiveplet4}
\end{align}
In particular, if $j_3\neq j$ then \eqs{massiveplet2}{massiveplet3} imply that $Q_1^\dagger \ket{\Omega(j_3)}$
and $Q_2^\dagger \ket{\Omega(j_3)}$ are orthogonal linear combinations of spin-($j\pm\half$) states
(although these states are eigenstates of $S^3$ as shown in \eq{Qdaghalf}).  If
$j_3=\pm j$ then $Q_1^\dagger \ket{\Omega(j)}$ and  $Q_2^\dagger \ket{\Omega(-j)}$
are states of spin-($j+\half$), since both these states
are eigenstates of $\boldsymbol{\vec S}\llsup{\,2}$ and $S^3$ with eigenvalues $(j+\half)(j+\tfrac{3}{2})$
and $\pm(j+\half)$, respectively.

Note that since $[P^2,Q_\alpha]=[P^2,Q^\dagger_{\dot\alpha}]=0$, it follows that 
all the states of the supermultiplet,
$\ket{\Omega(j_3)}\,,\,Q^{\dagger\,1}\ket{\Omega(j_3)}\,,\,Q^{\dagger\,2}\ket{\Omega(j_3)}\,,\,Q^{\dagger\,1}
Q^{\dagger\,2}\ket{\Omega(j_3)}$, are mass-degenerate, with common
mass $m$.
The states of an $N=1$ massive supermultiplet of superspin $j$ are exhibited
in Table~\ref{massivesuperplet}.

In summary, there are $4(2j+1)$ mass-degenerate states in a massive
supermultiplet of superspin $j$, which are explicitly given by
\eqst{massiveplet1}{massiveplet4}, for $j_3=-j,-j+1,\ldots,j-1,j$.  In
general, a massive supermultiplet of superspin $j$ is made up of
$2(2j+1)$ states of spin $j$, $2j+2$ states of spin $(j+\half)$ and
$2j$ states of spin ($j-\half$).  The extra two states for the case of
spin-$(2j+1)$ arise when $j_3=\pm j$, in which cases $Q_1^\dagger
\ket{\Omega(j)}$ and $Q_2^\dagger \ket{\Omega(-j)}$ are pure states of
spin $(j+\half)$ as previously noted.  Note that the number of
fermionic and bosonic degrees of freedom of the massive supermultiplet
coincide and is equal to $2(2j+1)$.  These results are summarized in Table~\ref{massivesuperplet}.

\subsubsection{States of a massless supermultiplet of superhelicity $\kappa$}
To construct the states of an irreducible massless supermultiplet, we
choose the standard reference frame, $P^\mu=P^0(1\,;\,0\,,\,0\,,\,1)$, 
for lightlike four-vectors.
In this reference frame, the anticommutators given in \eqs{susyalg4}{susyalg5}
simplify to those exhibited in \eqst{qqmassless1}{qqmassless3}.  All the
states in the massless supermultiplet are
simultaneous eigenstates of $P^2$ and the superhelicity operator $\mathcal{K}$, with eigenvalues
$m^2$ and $\kappa$, respectively, where the possible values of
$\kappa=0,\pm\half,\pm 1,\pm\tfrac32,\ldots$.

For a fixed value of the superhelicity $\kappa$, there exists a distinct state of the supermultiplet,
 denoted by $\ket{\Omega}$, that satisfies:
 \beq \label{Omegastate0}
 Q_\beta\ket{\Omega}=0\,,\qquad\quad \mathcal{K}\ket{\Omega}=\kappa\ket{\Omega}\,.
 \eeq
To verify that a state $\ket{\Omega}$ exists that is annihilated by $Q_\beta$,
let us assume the contrary. 
% To verify that such a state must exist, let us assume the contrary. 
Suppose that a state of
 the massless supermultiplet, denoted by $\ket{\Psi}$ exists that is not annihilated by
$Q_\beta$.  Due to \eq{KQQ}, it follow that $Q_\beta\ket{\Psi}$ must also be a state of
the massless supermultiplet.  Arguing as we did below \eq{Omegastate}, we again arrive at a contradiction.
Consequently, there must be at least one state of the supermultiplet that satisfies
$Q_\beta\ket{\Omega}=0$.   Moreover, a state that satisfies
\eq{Omegastate0} must be unique under the assumption that the
 massless supermultiplet with superhelicity $\kappa$ is an \textit{irreducible} representation of the $N=1$ SUSY
 algebra.

The states of the massless supermultiplet are obtained by considering the series,
\beq
\ket{\Omega}\,,\, Q^\dagger_{\dot\alpha}\ket{\Omega}\,,\,Q^\dagger_{\dot\beta} Q^\dagger_{\dot\alpha}\ket{\Omega}\,.
\eeq
However, $Q^\dagger_{\dot\beta} Q^\dagger_{\dot\alpha}\ket{\Omega}=0$ as a result of
\eq{QQQQ0}, and $P^\lambda Q^\dagger_{\dot\beta}\sigmabar_\lambda^{\dot\beta\tau}\ket{\Omega}=0$
as a consequence of \eq{zeroops}.  Thus, in contrast to the massive
supermultiplet, the massless supermultiplet contains only two states.
These two states are eigenvalues of the helicity operator $h$.
To determine the corresponding helicities, we shall employ
the standard reference frame where $P^\mu=P^0(1\,;\,0\,,\,0\,,\,1)$.
Since \eq{zeroops} yields $Q_1=Q^\dagger_1=0$, it follows that the massless $N=1$ supermultiplet consists of the
two states, $\ket{\Omega}$ and $Q^\dagger_2\ket{\Omega}$.  Using \eqs{mathcalkdef}{Omegastate0},
the helicities of these two states can be determined,
\beqa
h\ket{\Omega}&=&\left[\mathcal{K}-\frac{1}{8P^0}\left(Q_1^\dagger Q_1+Q_2^\dagger Q_2\right)\right]\ket{\Omega}
=\kappa\ket{\Omega}\,,\label{helicityOmega}\\[8pt]
hQ^\dagger_2\ket{\Omega}&=&\left[\mathcal{K}-\frac{1}{8P^0}\left(Q_1^\dagger Q_1+Q_2^\dagger Q_2\right)\right]Q^\dagger_2\ket{\Omega} \nonumber \\
&=& \left[\kappa Q_2^\dagger-\frac{1}{8P^0} Q_2^\dagger\left(2P_\mu\sigma^\mu_{22}-Q_2^\dagger Q_2\right) \right. \nonumber \\
&&\qquad\quad
-\left.\!\!\frac{1}{8P^0} Q_1^\dagger\left(2P_\mu\sigma^\mu_{12}-Q_2^\dagger Q_1\right)\right]\ket{\Omega} \nonumber \\
&=& \left[\kappa-\tfrac{1}{4}(\sigma^0_{22}-\sigma^3_{22})\right]Q^\dagger_2\ket{\Omega}
=(\kappa-\half) Q^\dagger_2\ket{\Omega}\,.\label{helicityQomega}
\eeqa
Indeed, the superhelicity $\kappa$ is the maximal helicity of the massless $N=1$ supermultiplet.
Thus, an irreducible $N=1$ massless supermultiplet with superhelicity $\kappa$ consists of two massless
states with helicity $\kappa$ and $\kappa-\half$, respectively.
These results are summarized in Table~\ref{masslesssuperplet}.

%%%%%%%%%%%%%%%%%%%%%%%%%%%%%%%%%%%%%%%%%%%%%%%%%%%%%
\subsection{Problems}

\begin{problem}
\label{pr:jhalf}
Show that the massive $j=\half$ supermultiplet corresponds to a real vector field, a real scalar field and a Dirac fermion field.
\end{problem}
\begin{problem}
Derive the following three commutation relations:
\beq
[B^\mu\,,\,Q_\alpha]=-\half P^\mu Q_\alpha\,,\qquad\qquad
[B^\mu\,,\,Q^\dagger_{\dot\alpha}]=\half P^\mu
Q^\dagger_{\dot\alpha}\,,
\eeq
\beq
[B^\mu\,,\,B^\nu]=i\epsilon^{\mu\nu\rho\lambda}B_\rho P_\lambda\,,
\eeq
where $B^\mu$ is defined in \eq{bmudef}.
\end{problem}
\begin{problem}
\label{pr:Lcomms}
Derive the following two commutation relations,
\begin{align} 
[L^\mu\,,\,Q_\alpha]=-\tfrac{1}{4}(\sigma^\mu\sigmabar^\nu)_\alpha{}^\beta Q_\beta P_\nu\,,\qquad\quad
[L^\mu\,,\,Q^\dagger_{\dot\alpha}]=\tfrac{1}{4}(\sigmabar^\nu\sigma^\mu)^{\dot\beta}{}_{\dot\alpha}
  Q^\dagger_{\dot\beta}P_\nu\,,
\end{align}
where $L^\mu$ is defined in \eq{Lmudef}.
\end{problem}

\begin{problem}
\label{pr:spin2}
Show that a massless supermultiplet with $\kappa=2$ and 
its CPT-conjugates corresponds to
a massless spin-$\tfrac{3}{2}$ and a massless spin 2 particle, which
is realized in supergravity by the
gravitino and the graviton. 
\end{problem}

 \begin{problem}
Obtain the explicit form for $K^\mu$ in \eq{eq:Kmu}.
\end{problem}
\clearpage

 \begin{problem}
 \label{pr:R}
Obtain an explicit expression for $R(x)$ in \eq{eq:Remainder}, and
show that it vanishes after imposing the classical field equations for
$\psi_\alpha(x)$.  Note that this computation is non-trivial and
requires a judicious application of 
Fierz identities for two-component fermions (which can be found, e.g.,
in Appendix B of Ref.~\cite{Dreiner:2008tw}).
\end{problem}

\begin{problem}
Obtain an explicit expression for $J^\mu_\alpha$ in terms of the fields $A$ and $\psi$ in the Wess-Zumino model.
\end{problem}

\begin{problem}
%In the Wess-Zumino model, verify that \eq{QandCCR} is satisfied by employing
%the canonical commutation relations of the boson field $A$ and the
%canonical anticommutation relations of the fermion field $\psi$.
Verify, for the Wess-Zumino model, 
that the Noether supercharges defined by \eq{QJ} satisfy the SUSY algebra [cf.~\eq{QandCCR}].
\end{problem}

\begin{problem}
\label{pr:kprime}
Obtain the explicit form for $K^{\prime\,\mu}$ in \eq{eq:Kpmu}.
\end{problem}

\begin{problem}
\label{pr:xieta}
Starting from \eqst{offshell1}{offshell3}, verify that
 \begin{align*}
 \bigl[\deltaeta\,,\,\deltaxi\bigr]\Phi(x)=-2i(\xi\sigma^\mu\eta^\dagger-\eta^\dagger\sigma^\mu\xi^\dagger)\partial_\mu \Phi(x)\,,
 \end{align*}
 for $\Phi=A$, $\psi$ and $F$ without the need to impose the classical field equations.
 \end{problem}

\section{Superspace and Superfields}
\renewcommand{\theequation}{\arabic{section}.\arabic{equation}}
\setcounter{equation}{0}
\label{sec:superspace}

In the section we introduce superspace coordinates $\theta$ and $\thetabar$.
The concept of a supersymmetry transformation is then realized as a translation in superspace.
We construct superfields\cite{Ferrara:1974ac,Salam:1974jj,Salam:1976ib}, which can be expanded in powers of $\theta$ and
$\thetabar$; the corresponding expansion coefficients are the fields
of a super\-multiplet.   By introducing the spinor covariant derivative, one is
able to define the derivative of a superfield that is covariant with
respect to SUSY transformations.  This allows us to define an
irreducible chiral
superfield by imposing a derivative constraint.

Employing this formalism, we demonstrate how to construct a
SUSY Lagrangian for chiral superfields, and 
and show that the
supersymmetric action can be expressed as an integral over superspace. 
%Finally, we discuss the $R$-transformation of a superfield and revisit 
%the concept of $R$-invariance.
Finally, we discuss the improved ultraviolet behavior of SUSY and introduce the
celebrated non-renormalization theorem of $N=1$ supersymmetry\cite{GRS,SeibergNR}.
 
\subsection{Superspace coordinates and translations}
\label{sec:supercoords}

In Section~\ref{sec:SUSYalgebra} we indicated that we expect a SUSY
translation to be similar to a space-time translation, where the SUSY generators $Q$, $Q^\dagger$ replace the $P^\mu$ of ordinary space-time translations:
\begin{align}
 \deltaxi\Phi(x)=i\bigl[\xi Q+\xi^\dagger Q^\dagger\,,\,\Phi(x)\bigr]\,,
 \end{align}
 for $\Phi=A$, $\psi$ or $F$. 
 But what exactly is being translated? 
\clearpage

In this subsection,
 we extend spacetime by introducing Grassmann coordinates, $\theta^\alpha$ and $\thetabar_{\dalpha}$.  The result is an 8-dimensional \textit{superspace} with coordinates
$(x^\mu\,,\,\theta^\alpha\,,\,\thetabar_{\dalpha})$.  The
Grassmann coordinates are anticommuting coordinates; i.e., they satisfy anticommutation relations,
 \begin{align}
 \{\theta^\alpha\,,\,\theta^\beta\}=\{\thetabar_{\dalpha}\,,\,\thetabar_{\dbeta}\}=\{\theta^\alpha\,,\,\thetabar_{\dbeta}\}=0\,.
 \end{align}

One can also define derivatives with respect to $\theta$ and
$\theta^\dagger$.  It is convenient to introduce the following notation,
\beq \label{dth1}
\partial_{\alpha}\equiv \frac{\partial}{\partial\theta^\alpha}\,,\qquad\qquad
\partial^\dagger_{\dalpha}\equiv \frac{\partial}{\partial{\thetabar}^{\dalpha}}\,.
\eeq
The derivatives with respect to $\theta$ and $\theta^\dagger$ are
defined in the obvious way,
\beq
 \partial_\alpha\theta^\beta=\delta_\alpha^\beta\,,\qquad\qquad \partial^\dagger_{\dalpha}{\thetabar}^{\dbeta}=\delta_{\dalpha}^{\dbeta}\,.
\eeq
It then follows that
\begin{align}
\partial_\alpha\theta_\beta=\partial_\alpha(\epsilon_{\beta\gamma}\theta^\gamma)=-\epsilon_{\alpha\beta}\,,\qquad \partial^\dagger_{\dalpha}\thetabar_{\dbeta}=\partial^\dagger_{\dalpha}(\epsilon_{\dbeta\dot\gamma}\theta^{\dagger\dot\gamma})=-\epsilon_{\dalpha\dbeta}\,.\label{dteps}
\end{align}

Derivatives with respect to $\theta$ and $\theta^\dagger$ satisfy a
modified Leibniz rule, 
\beqa
\partial_\alpha(fg)&=&(\partial_\alpha f)g+(-1)^{\varepsilon(f)}f(\partial_\alpha g)\,,\\
\partial^\dagger_{\dalpha}(fg)&=&(\partial^\dagger_{\dalpha} f)g+(-1)^{\varepsilon(f)}f(\partial^\dagger_{\dalpha} g)\,,
\eeqa
where 
\beq
\varepsilon(f)=\begin{cases} 0\,,&\quad \text{if $f$ is Grassmann even}\,,\\  1\,,&\quad \text{if $f$
    is Grassmann odd}\,, \end{cases}
\eeq
and $f$ is Grassmann even [odd] if it is a product of an even
[odd] number of anticommuting quantities.
For example,
\beqa
\partial_\alpha(\theta\theta)&=&\partial_\alpha\bigl(\epsilon_{\gamma\beta}\theta^\gamma\theta^\beta\bigr)=\epsilon_{\gamma\beta}(\delta^\gamma_\alpha\theta^\beta-\delta_\alpha^\beta\theta^\gamma)=2\theta_\alpha\,,\label{partialtt}\\
\partial^\dagger_{\dalpha}(\thetabar\thetabar)&=&
\partial^\dagger_{\dalpha}\bigl(\epsilon_{\dbeta\dgamma}
\theta^{\dagger\dgamma}\theta^{\dagger\dbeta}\bigr)=
\epsilon_{\dbeta\dgamma}(\delta^{\dgamma}_{\dalpha}\theta^{\dagger\dbeta}-\delta_{\dalpha}^{\dbeta}\theta^{\dagger\dgamma})
=-2\theta_{\dalpha}^\dagger\,.\label{partialtdtd}
\eeqa

Likewise, one conventionally defines,
\begin{align}
\partial^{\alpha}\equiv \frac{\partial}{\partial\theta_\alpha}\,,\qquad
\partial^{\dagger\dalpha}\equiv \frac{\partial}{\partial\thetabar_{\dalpha}}\,.\label{dth2}
\end{align}
However, one needs to be careful since this notation leads to an unexpected minus sign when relating
the derivatives of \eqs{dth1}{dth2},
\beq
\partial^\alpha =-\epsilon^{\alpha\beta}\partial_\beta\,,\qquad
\partial^{\dagger\dalpha}=-\epsilon^{\dalpha\dbeta}\partial^\dagger_{\dbeta}\,. \label{eq:partialsign}
\eeq
This is the one case where the rule for raising a spinor index given
in \eq{raiseindex} does \textit{not} apply.

In order to define translations in superspace, we shall
generalize the translation operator $\exp(ix\newcdot P)$ to the super-translation operator,
 \begin{align}
 G(x,\theta,\thetabar)=\exp(ix\newcdot P+\theta Q+\thetabar Q^\dagger)\,.
 \end{align}
 We can now extend the field operator, $\Phi(x)=\exp(ix\newcdot
 P)\Phi(0)\exp(-ix\newcdot P)$ to a \textit{superfield} operator,
 \begin{align}
 \Phi(x,\theta,\thetabar)=G(x,\theta,\thetabar)\Phi(0,0,0)G^{-1}(x,\theta,\thetabar)\,.
 \end{align}
 In this way, we can realize a supersymmetry transformation as a translation in superspace.

 Using the Baker-Campbell-Hausdorff formula\cite{BrianHall},
 \beq \label{BCH}
\exp(A)\exp(B)=\exp\bigl(A+B+\half[A\,,\,B]+\cdots\bigr)\,, 
\eeq
one can prove (see Problem \ref{pr:two_super_translations}),
 \begin{align}
 G(y,\xi,\xi^\dagger)G(x,\theta,\thetabar)=G\bigl(x+y+i(\xi\sigma\thetabar-\theta\sigma\xi^\dagger),\xi+\theta,\xi^\dagger+\thetabar\bigr)\,. \label{GG}
 \end{align}
 Note the appearance in \eq{GG} of an extra non-trivial spacetime translation, $ i (\xi \sigma \thetabar - \theta \sigma \xi^\dagger)$.
Hence, it follows that
 \begin{align}
 \begin{split}
 &G(y,\xi,\xi^\dagger)\Phi(x,\theta,\thetabar) G^{-1}(y,\xi,\xi^\dagger) \\
 & \qquad = \Phi\bigl(x+y+i(\xi\sigma\thetabar-\theta\sigma\xi^\dagger),\xi+\theta,\xi^\dagger+\thetabar\bigr)\,.
\end{split}
\label{eq:Phixy}
\end{align}
For infinitesimal $y$, $\xi$ and $\xi^\dagger$, we can approximate
\beq
G(y,\xi,\xi^\dagger)\simeq \mathds{1}+i(y\newcdot P+\xi Q+\xi^\dagger Q^\dagger)\,,
\eeq
which allows us to rewrite the left-hand side of \eq{eq:Phixy} as
\begin{align}
\begin{split}
& G\of{y,\xi,\xi^\dagger} \Phi\of{ x,\theta,\thetabar} G^{-1}\of{y,\xi,\xi^\dagger} \\
&\quad \simeq \of{ \mathds{1}+i\of{y\newcdot P+\xi Q+\xi^\dagger Q^\dagger} }  \Phi\of{x,\theta,\thetabar} \of{ \mathds{1} - i \of{ y\newcdot P + \xi Q + \xi^\dagger Q^\dagger }}
\end{split} \nn
\\
& \quad \simeq \Phi\of{x,\theta,\thetabar} 
+ i y_\mu \sqof{ P^\mu, \Phi }  + i \sqof{ \xi Q, \Phi }  + i \sqof{  \xi^\dagger Q^\dagger, \Phi }.
\label{eq:GPhiG}
\end{align}
One can also Taylor expand the right-hand side of \eq{eq:Phixy}, which
to first order yields
\begin{align}
\begin{split}
&  \Phi\bigl(x+y+i(\xi\sigma\thetabar-\theta\sigma\xi^\dagger),\xi+\theta,\xi^\dagger+\thetabar\bigr)
\\
 & \qquad \qquad =    \Phi(x,\theta,\thetabar)  + \bigl[y^\mu+i(\xi\sigma^\mu\thetabar-\theta\sigma^\mu\xi^\dagger)\bigr]\partial_\mu \Phi(x,\theta,\thetabar) \\
& \qquad \qquad \quad \qquad \qquad\,\,\,\, +\bigl(\xi^\alpha\partial_\alpha+\xi^\dagger\partial^{\dagger\dot\alpha}\bigr)\Phi(x,\theta,\thetabar)\,,
 \end{split}
 \label{eq:Phixy2}
 \end{align}
where we have employed the derivatives defined in \eq{dth1}.
Comparing the first-order terms of eqns.~(\ref{eq:GPhiG}) and (\ref{eq:Phixy2}),
 we end up with expressions for the following commutators,
 \begin{align}
 \bigl[\Phi\,,\,P_\mu\bigr]&= i\,\partial_\mu\Phi\,, \label{eq:PPhi} \\
 \big[\Phi\,,\,\xi Q\bigr]&= i\,\xi^\alpha\left(\partial_{\alpha}+i(\sigma^\mu\thetabar)_\alpha\partial_\mu\right)\Phi\,, \label{eq:QPhi} \\
  \big[\Phi\,,\,\xi^\dagger Q^\dagger\bigr]&= -i\left(\partial^\dagger_{\dalpha}+i(\theta\sigma^\mu)_{\dalpha}\partial_\mu\right)\xi^{\dagger\,\dalpha}\Phi\,.\label{eq:QbarPhi}
\end{align}

The above results motivate the introduction of the following differential operators,
\begin{align}
\widehat{P}_\mu&=i\partial_\mu\,,\label{Phat}\\
\widehat{Q}_\alpha&=i\partial_\alpha-(\sigma^\mu\thetabar)_\alpha\partial_\mu\,, \label{Qhat}\\
\widehat{Q}^\dagger_{\dalpha}&=-i\partial^\dagger_{\dalpha}+(\theta\sigma^\mu)_{\dalpha}\partial_\mu\,,\label{QDhat}
\end{align}
which allow us to succinctly rewrite \eqst{eq:PPhi}{eq:QbarPhi} as follows:
 \begin{align}
 \bigl[\Phi\,,\,P_\mu\bigr]&=\widehat{P}_\mu\Phi\,, \\
 \big[\Phi\,,\,\xi Q\bigr]&=(\xi\widehat{Q})\Phi\,, \label{stranslate1}\\
  \big[\Phi\,,\,\xi^\dagger Q^\dagger\bigr]&=(\xi^\dagger\widehat{Q}^\dagger)\Phi\,.\label{stranslate2}
\end{align}

In \eq{susytranslate}, we noted that the action of an infinitesimal SUSY transformation on any field $\Phi\of{x}$ was given by
$\deltaxi\Phi(x)=i\bigl[\xi Q+\xi^\dagger Q^\dagger\,,\,\Phi(x)\bigr]$.
In light of \eqs{stranslate1}{stranslate2}, we conclude that the action of an
infinitesimal SUSY transformation on a
superfield $\Phi\of{x,\theta,\thetabar}$ is given by
 \begin{align}
 \deltaxi\Phi(x,\theta,\thetabar)=-i(\xi \widehat{Q}+\xi^\dagger \widehat{Q}^\dagger)\Phi(x,\theta,\thetabar)\,.\label{supertrans}
\end{align}

\subsection{Expansion of the superfield in powers of $\theta$ and $\thetabar$}

Consider the Taylor expansion of a superfield,
$\Phi(x,\theta,\theta^\dagger)$, in powers of $\theta$ and~$\thetabar$.  The coefficients of this expansion will be functions of
$x$, which can be interpreted as ordinary fields.  Since $\theta$ and
$\thetabar$ are anticommuting coordinates, this Taylor series
terminates after a finite number of terms.  In particular, since $\theta$ and
$\theta^\dagger$ are 
anticommuting two-component spinor
quantities,  it follows that
$(\theta_1)^2=(\theta_2)^2=(\thetabar_{\dot 1})^2=(\thetabar_{\dot
  2})^2=0$, whereas products such as $\theta_1\theta_2$ and
$\theta_1^\dagger \theta_2^\dagger$ do not vanish.  Indeed, it is easy
to check that
\begin{align}
\theta^\alpha\theta^\beta&=-\half\epsilon^{\alpha\beta}\theta\theta\,,\qquad\qquad
{\thetabar}^{\dalpha}{\thetabar}^{\dbeta}=\half\epsilon^{\dalpha\dbeta}\thetabar\thetabar\,,\nn \\
\theta_\alpha\theta_\beta&=\half\epsilon_{\alpha\beta}\theta\theta\,,\qquad\qquad\phm
{\theta}^\dagger_{\dalpha}{\theta}^\dagger_{\dbeta}=-\half\epsilon_{\dalpha\dbeta}\thetabar\thetabar\,,\nn
\end{align}
where $\theta\theta\equiv \theta^\alpha\theta_\alpha$ and $\thetabar\thetabar\equiv\thetabar_{\dalpha}
{\thetabar}^{\dalpha}$ following the convention of \eq{contract}. 
Products such as $\theta_\alpha\theta_\beta\theta_\gamma=0$, since the
spinor indices can assume at most two different values.  Finally, the
following three results are noteworthy (see Problem 17),
\begin{align}
(\theta\sigma^\mu\thetabar)\theta_\beta&=-\half \theta\theta(\sigma^\mu\thetabar)_\beta \label{eq:r1} \\
(\theta\sigma^\mu\thetabar)\thetabar_{\dbeta}&=-\half \thetabar\thetabar(\theta\sigma^\mu)_{\dbeta} \label{eq:r2} \\
(\theta\sigma^\mu\thetabar)(\theta\sigma^\nu\thetabar)&=\half g^{\mu\nu}(\theta\theta)(\thetabar\thetabar). \label{eq:r3}
\end{align}
Sometimes, we shall write
$\theta\theta\theta^\dagger\theta^\dagger\equiv
(\theta\theta)(\thetabar\thetabar)$.  In such products, there should
be no ambiguity in omitting the parentheses.

The Taylor series expansion of a complex superfield
$\Phi(x,\theta,\thetabar)$ is therefore given by,
\beqa
\Phi(x,\theta,\thetabar)&=& f(x) +\theta\zeta(x)+\thetabar\chi^\dagger(x)+\theta\theta m(x)+\thetabar\thetabar n(x) +\theta\sigma^\mu\thetabar V_\mu(x) \nn\\
&& 
+(\theta\theta)\thetabar\lambda^\dagger(x)+(\thetabar\thetabar)\theta\lambda(x)+\theta\theta\thetabar\thetabar d(x)\,,\label{phitaylor}
\eeqa
where $f$, $m$, $n$, $V_\mu$, and $d$ are complex commuting bosonic
fields and $\zeta$, $\chi$, $\lambda$ and $\psi$ are anticommuting
two-component fermionic fields.  
The SUSY transformation laws of the component fields can now be easily
obtained (see Problem~\ref{pr:fmnV}) by comparing both sides of \eq{supertrans}.

Hence, there are 16 bosonic and 16 fermionic real degrees of freedom.
If we impose the constraint, $\Phi^\dagger=\Phi$, then $f$, $d$ and $V_\mu$
are real bosonic fields, $n^\dagger=m$, $\zeta=\chi$ and $\lambda=\psi$.  In this
case, there are 8 bosonic and 8 fermionic real degrees of freedom.  In both cases, there are too many degrees of freedom to describe the supermultiplet of the Wess-Zumino model.
This is because an unconstrained complex superfield, $\Phi(x,\theta,\thetabar)$, describes a
reducible representation of the SUSY algebra.   One must impose
supersymmetric constraints to project out an irreducible
supermultiplet.\footnote{A real superfield $\Phi$ yields an off-shell
  irreducible representation with superspin $j=\half$.  More on this
  in Section~\ref{sec:gaugetheories}.}

The superfield defined in \eq{phitaylor} is an example of a
\textit{bosonic} superfield, where the Taylor series coefficients of terms even in
the number of Grassmann coordinates are commuting bosonic fields and the coefficients of terms odd in
the number of Grassmann coordinates are anticommuting fermionic fields.
Similarly, one can define a \textit{fermionic} superfield, where the
Taylor series coefficients of terms even in
the number of Grassmann coordinates are anticommuting fermionic fields and the coefficients of terms odd in
the number of Grassmann coordinates are commuting bosonic fields.
 
\subsection{Spinor covariant derivatives}

For a superfield $\Phi$, it is easy to check that neither
$\partial_\alpha\Phi$ nor $\partial_{\dalpha}\Phi$ is a superfield, since
\begin{align}
\partial_\alpha(\deltaxi\Phi)\neq  \deltaxi(\partial_\alpha\Phi)\,,\qquad\quad
\partial^\dagger_{\dalpha}(\deltaxi\Phi)\neq  \deltaxi(\partial^\dagger_{\dalpha}\Phi)\,.
\end{align}
Note that if $\Phi$ is a bosonic superfield, then the hermitian conjugate of $\partial_\alpha\Phi$ is given
by,
\beq \label{daggers}
(\partial_\alpha\Phi)^\dagger=-\partial_{\dalpha}^\dagger
\Phi^\dagger\,,
\eeq
where the minus sign above is related to the minus sign in
\eq{dteps}.
\clearpage

We therefore introduce spinor covariant derivatives $D_\alpha$ and
$\Dbar_{\dalpha}$ such that $D_\alpha\Phi$ and
$\overline{D}_{\dalpha}\Phi$ are superfields,\footnote{Note that if
  $\Phi$ is a bosonic superfield, then $D_\alpha\Phi$ and
  $\Dbar_{\dalpha}\Phi$ are fermionic superfields.}
 which implies the
following conditions must be satisfied,
\begin{align}
D_\alpha(\deltaxi\Phi)=\deltaxi(D_\alpha\Phi)\,,\qquad\quad
\Dbar_{\dalpha}(\deltaxi\Phi)=\deltaxi(\Dbar_{\dalpha}\Phi)\,.\label{dsuper}
\end{align}
Using \eq{supertrans} to express $\deltaxi\Phi$ in terms of the operators $\widehat{Q}$ and
$\widehat{Q}^\dagger$ defined in \eqs{Qhat}{QDhat}, respectively, one easily derives
\begin{align}
\{D_\alpha\,,\,\widehat{Q}_\beta\}=\{D_\alpha\,,\,\widehat{Q}^\dagger_{\dbeta}\}=\{\Dbar_{\dalpha}\,,\,\widehat{Q}_{\beta}\}=\{\Dbar_{\dalpha}\,,\,\widehat{Q}^\dagger_{\dbeta}\}=0\,.\label{antis}
\end{align}

To fix the explicit forms for the spinor covariant derivatives, we choose
the normalization of $D_\alpha$ so that it has the form
$D_\alpha=\partial_\alpha+\ldots$, where the ellipsis refers to
correction terms needed to satisfy \eqs{dsuper}{antis}.  In the case
of $\overline{D}_\alpha$, it is customary to impose the condition,
\beq \label{Dcond}
(D_\alpha\Phi)^\dagger=\overline{D}_{\dalpha}\Phi^\dagger\,,
\eeq
where $\Phi$ is a bosonic superfield, in which case
$\overline{D}_{\dalpha}=-\partial_{\dalpha}^\dagger+\ldots$
 [cf.~\eq{daggers}].   

The explicit forms for the spinor covariant derivatives that satisfy
the above conditions are given by, 
\begin{align}
D_\alpha&=\partial_\alpha-i(\sigma^\mu\thetabar)_\alpha\,\partial_\mu\,, \label{eq:D} \\
\Dbar_{\dalpha}&=-\partial^\dagger_{\dalpha}+i(\theta\sigma^\mu)_{\dalpha}\,\partial_\mu\,. \label{eq:Db}
\end{align}
In particular, $D$ and $\overline{D}$
satisfy the same anticommutation relations as $\widehat{Q}$ and $\widehat{Q}^\dagger$ (see Problem~\ref{pr:D}),
\begin{align}
\{D_\alpha\,,\,D_\beta\}=\{\Dbar_{\dot\alpha}\,,\,\Dbar_{\dot\beta}\}=0  \ \ \mathrm{and\ \  }
\{D_\alpha\,,\,\Dbar_{\dot\beta}\}=2i\sigma^\mu_{\alpha\dot\beta}\partial_\mu.
\label{eq:Dcomms}
\end{align}

One can also define spinor covariant derivatives with a raised spinor index.
In this case, it is conventional to define,
\begin{align}
D^\alpha\equiv \epsilon^{\alpha\beta}D_\beta&=-\partial^\alpha+i(\thetabar\sigmabar^\mu)^\alpha\,\partial_\mu\,, \\
\Dbar\lsup{\dalpha}\equiv \epsilon^{\dalpha\dbeta}D_\beta
&=\partial^{\dagger\dalpha}-i(\sigmabar^\mu\theta)^{\dalpha}\,\partial_\mu\,,
\end{align}
where we have employed \eq{eq:partialsign}. That is,
the spinor indices of $D_\alpha$ and $\overline{D}_{\dalpha}$ are raised
in the
conventional way according to \eq{raiseindex}.\footnote{This is in contrast to the
rule for raising the spinor indices of $\partial_\alpha$ and
$\partial^\dagger_{\dalpha}$ specified in \eq{eq:partialsign}, where
an extra minus sign appears.}   
%One can now compute
The following differential operators
%which 
will be useful later in these lectures,
\beqa
D^2&=&D^\alpha D_\alpha=-\partial^\alpha\partial_\alpha+2i(\partial^\alpha\sigma^\mu_{\alpha\dbeta}{\thetabar}\lsup{\dbeta})\partial_\mu+\thetabar\thetabar\,\square\,,\label{DD}\\
\Dbar\lsup{\,2}&=&\Dbar_{\dalpha}\Dbar\lsup{\,\dalpha}=-\partial^\dagger_{\dalpha}\partial^{\dagger\dalpha}+2i(\theta^\alpha\sigma^\mu_{\alpha\dbeta}\partial^{\dagger\dbeta})\partial_\mu+\theta\theta\square\,,\label{DbDb}
\eeqa
where $\square\equiv\partial_\mu\partial^\mu$.   One can then derive
the following identity (see Problem~\ref{pr:DDid}),
\beq \label{DDid}
[D^2,\Dbar\lsup{2}]= 4i\sigma^\mu_{\alpha
\dot{\beta}}\partial_\mu[D^\alpha,\overline D\lsup{\,\dot{\beta}}]\,.
\eeq
%\clearpage

We have employed different notation for the conjugation of the
various differential operators that appear in this subsection.
The relation of $\widehat{Q}^\dagger$ to $\widehat{Q}$ is
\textit{hermitian conjugation} in the same sense that $\hat{P}_\mu
=i\partial_\mu$ [defined in \eq{Phat}] is an hermitian operator in
quantum field theory with respect to the inner product defined by the
integration of complex fields over spacetime.
That is, the dagger on the differential operator $\widehat{Q}^\dagger$
denotes Hermitian conjugation with respect to the inner product
defined by the integration of complex superfields over
superspace.\footnote{For further details, see
  Refs.~\cite{Sohnius:1985qm,Martin:1997ns}.  Integration over
  superspace will be treated in Section~\ref{integration}.}

In contrast, the relation of $\Dbar$ to $D$ is \textit{complex
  conjugation} in the same sense that $\partial_\mu^*$ is the complex
conjugate of $\partial_\mu$.  In the latter case, the differential
operator $\partial_\mu$ is a real operator.
That is, if we define $\partial^*_\mu$ 
to be the derivative operator that acts on the field $\phi$ such that
\begin{align}
\of{\partial_\mu \phi}^\dagger = \partial_\mu^* \phi^\dagger,
\end{align}
then since $\of{\partial_\mu \phi}^\dagger = \partial_\mu \phi^\dagger$, it
follows that
$\partial^*_\mu=\partial_\mu$.  In light of \eq{Dcond},
we can therefore regard $\overline{D}$ as the complex conjugate of $D$.

\subsection{Chiral superfields}

A chiral superfield is obtained by imposing
the constraint $\Dbar_{\dalpha}\Phi=0$ on a general superfield $\Phi$.  Such a constraint is covariant with respect to
SUSY transformations, and the end result is an irreducible superfield that
corresponds to the superspin $j=0$ irreducible
representation of the SUSY algebra.  
Using \eq{eq:Db}, the constraint yields a differential equation,
\begin{align}
\Dbar_{\dalpha}\Phi=\bigl[-\partial^\dagger_{\dalpha}+i(\theta\sigma^\mu)_{\dalpha}\,\partial_\mu\bigr]\Phi(x,\theta,\thetabar)=0\,,
\end{align}
whose solution is of the form
\begin{align}
\Phi(x,\theta,\thetabar)=\exp(-i\theta\sigma^\mu\thetabar\,\partial_\mu)\Phi(x,\theta)\,.\label{eq:Phixththb}
\end{align}

We can expand $\Phi(x,\theta)$ in a
 (truncated) 
Taylor series in $\theta$,
\begin{align}
\Phi(x,\theta)=A(x)+\sqrt{2}\,\theta\psi(x)+\theta\theta F(x)\,,
\end{align}
where the factor of $\sqrt{2}$ is conventional. Plugging this into
\eq{eq:Phixththb} and using the identity  (see Problem~\ref{pr:expexp}),
\begin{align}
\exp(-i\theta\sigma^\mu\thetabar\,\partial_\mu)=1-i\theta\sigma^\mu\thetabar\,\partial_\mu-\tfrac{1}{4}(\theta\theta)(\thetabar\thetabar)\square,
\end{align}
we find after some algebraic manipulation a chiral superfield with the form,
\begin{align}
\begin{split}
\Phi(x,\theta,\thetabar) &=
A(x) + \sqrt{2}\,\theta \psi(x) + \theta\theta F(x)-i \theta\sigma^\mu\thetabar \partial_\mu A(x)\\
&\quad - \frac{i}{\sqrt{2}} (\theta\theta) 
\thetabar \sigmabar^\mu\, \partial_\mu \psi(x)-\tfrac{1}{4}(\theta\theta)(\thetabar\thetabar) \square A(x).
\end{split}
\label{eq:chiralSF}
\end{align}
Note that the chiral superfield $\Phi$ has dimension $[\Phi]=1$, in which case it
follows that the dimensions of the component fields are $[A]=1$ and
$[\psi]=\tfrac32$, as expected, whereas $[F]=2$
after making use of the dimensions of the Grassmann coordinates,
$[\theta]=[\thetabar]=-\half$.

Given a chiral superfield $\Phi$, its hermitian conjugate,
$\Phi^\dagger$, is an antichiral superfield, which is
defined by the SUSY-covariant constraint,
$D_\alpha\Phi^\dagger=0$.
Using \eq{eq:D}, the latter constraint yields a differential equation,
\begin{align}
D_{\alpha}\Phi^\dagger=\bigl[\partial_{\alpha}-i(\sigma^\mu\thetabar)_{\alpha}\,\partial_\mu\bigr]\Phi^\dagger(x,\theta,\thetabar)=0\,,
\end{align}
whose solution is of the form
\begin{align}
\Phi^\dagger(x,\theta,\thetabar)=\exp(i\theta\sigma^\mu\thetabar\,\partial_\mu)\Phi^\dagger(x,\thetabar)\,.\label{eq:Phixththb2}
\end{align}

We can expand $\Phi^\dagger(x,\thetabar)$ in a
 (truncated) 
Taylor series in $\thetabar$,
\begin{align}
\Phi^\dagger(x,\thetabar)=A^\dagger(x)+\sqrt{2}\,\thetabar\psi^\dagger(x)+\thetabar\thetabar F^\dagger(x)\,.
\end{align}
Plugging this result into \eq{eq:Phixththb2} and
following the same procedure as before, we end up with,
\begin{align}
\begin{split}
\Phi^\dagger(x,\theta,\thetabar) &=
A^\dagger(x) + \sqrt{2}\,\thetabar \psi^\dagger(x) + \thetabar\thetabar F^\dagger(x)+i \theta\sigma^\mu\thetabar \partial_\mu A^\dagger(x) \\
&\quad - \frac{i}{\sqrt{2}} (\thetabar\thetabar) 
\theta\sigma^\mu\, \partial_\mu \psi^\dagger(x)-\tfrac{1}{4}(\theta\theta)(\thetabar\thetabar) \square A^\dagger(x)\,.
\end{split}
\end{align}
Since $\Phi^\dagger$ is the hermitian conjugate of $\Phi$, we can
identify $A^\dagger$, $\psi^\dagger$ and $F^\dagger$ as the hermitian
conjugates of $A$, $\psi$ and $F$.

In calculations, it is often simpler to employ the so-called \textit{chiral representation}, in which all superfield operators $\mathcal{O}$ are modified according to
\begin{align}
\mathcal{O}_{\rm chiral}=\exp(i\theta\sigma^\mu\thetabar\,\partial_\mu)\mathcal{O}\exp(-i\theta\sigma^\mu\thetabar\,\partial_\mu)\,.
\end{align}
In the chiral representation,
\beqa
&& \widehat{Q}_\alpha=i\partial_\alpha\,,\qquad\qquad \ \
\widehat{Q}^\dagger_{\dalpha}=-i\partial^\dagger_{\dalpha}+2(\theta\sigma^\mu)_{\dalpha}\,\partial_\mu\,,\label{Qchiral}
\\
&&\Dbar_{\dalpha}=-\partial^\dagger_{\dalpha}\,,\qquad\qquad 
D_{\alpha}=\partial_{\alpha}-2i(\sigma^\mu\thetabar)_{\alpha}\,\partial_\mu\,.\label{Dchiral}
\eeqa
Thus, in the chiral representation, the requirement $\Dbar_{\dalpha}\Phi=-\partial^\dagger_{\dalpha}\Phi=0$ is simply the requirement that $\Phi$ is independent of $\thetabar$.
In the chiral representation, the chiral superfield will be denoted by
\begin{align}
\Phi_1(x,\theta)=A(x)+\sqrt{2}\,\theta\psi(x)+\theta\theta F(x)\,.\label{phione}
\end{align}
It then follows that the general expression for a chiral superfield is
\begin{align}
\Phi(x,\theta,\thetabar) =\exp(-i\theta\sigma^\mu\thetabar\,\partial_\mu)\Phi_1(x,\theta)=\Phi_1(x-i\theta\sigma^\mu\thetabar\,,\,\theta)\,.
\end{align}
It is convenient to define the shifted spacetime coordinate,
\begin{align}
y \equiv x - i \theta \sigma^\mu \thetabar,
\end{align}
so that the chiral superfield is given by,
\begin{align}
\Phi\of{x,\theta,\thetabar} = \Phi_1\of{y,\theta}.
\end{align}

The SUSY transformation laws for the fields that appear in the chiral
superfield can now be determined simply by inserting the expression
for $\Phi$ in the chiral representation given by \eq{phione} into
\eq{supertrans}.  In performing the computation, one employs the
chiral representation expressions for $\widehat{Q}$ and
$\widehat{Q}^\dagger$ given in \eq{Qchiral}.   You may verify (see Problem~\ref{pr:chiraltrans})
that the result of this calculation coincides with the SUSY
transformation laws
given previously in \eqst{offshell1}{offshell3}.

Likewise, one can define an antichiral representation in which 
\begin{align}
\mathcal{O}_{\rm antichiral}=\exp(-i\theta\sigma^\mu\thetabar\,\partial_\mu)\mathcal{O}\exp(i\theta\sigma^\mu\thetabar\,\partial_\mu)\,.
\end{align}
In the antichiral representation,
\begin{align}
\begin{split}
& \widehat{Q}^\dagger_{\dalpha}=-i\partial^\dagger_{\dalpha}\,,\qquad\qquad \ \ 
\widehat{Q}_\alpha=i\partial_{\alpha}-2(\sigma^\mu\thetabar)_{\alpha}\,\partial_\mu\,,
 \\
&D_{\alpha}=\partial_{\alpha}\,,\qquad\qquad \quad
\Dbar_{\dalpha}=-\partial^\dagger_{\dalpha}+2i(\theta\sigma^\mu)_{\dalpha}\,\partial_\mu\,.
\end{split}
\end{align}
Thus, in the antichiral representation, the requirement $D_{\alpha}\Phi^\dagger
=\partial_{\alpha}\Phi^\dagger=0$ is simply the requirement that $\Phi^\dagger$ is independent of $\theta$.
In the antichiral representation, the antichiral superfield will be denoted by
\begin{align}
\Phi_2(x,\thetabar)=A^\dagger(x)+\sqrt{2}\,\thetabar\psi^\dagger(x)+\thetabar\thetabar F^\dagger(x)\,.
\end{align}
It then follows that the general expression for an antichiral superfield is
\begin{align}
\Phi^\dagger(x,\theta,\thetabar) =\exp(i\theta\sigma^\mu\thetabar\,\partial_\mu)\Phi_2(x,\thetabar)=\Phi_2(x+i\theta\sigma^\mu\thetabar\,,\,\thetabar)\,.
\end{align}
It is convenient to define the shifted spacetime coordinate,
\begin{align}
y^\dagger \equiv x + i \theta \sigma^\mu \thetabar,
\end{align}
so that the antichiral superfield is given by,
\begin{align}
\Phi^\dagger\of{x,\theta,\thetabar} = \Phi_2\of{y^\dagger,\thetabar}.
\end{align}

%%%%
\subsection{Constructing the SUSY Lagrangian}

\subsubsection{$F$-terms}
Ultimately, our goal is to construct an action that is invariant under SUSY.  It is therefore sufficient to construct a Lagrangian that transforms under SUSY as a total derivative.
In the literature, it is common to use the nomenclature
\textit{$F$-term} to denote the coefficient of the
$\theta\theta$ term of a superfield.  This is sometimes explicitly
indicated as follows,
\begin{align}
[\Phi]_{\theta\theta}=[\Phi]_F=F. \label{Fterm}
\end{align}
Recall that in \eq{offshell3}, we demonstrated that the auxiliary
field $F(x)$ transforms as a total derivative under the SUSY
transformation laws.  But, this field is simply the coefficient of the
$\theta \theta$ term of a chiral superfield!  Indeed, the 
$F$-term of any chiral superfield transforms under a SUSY
transformation as a total derivative.  This means that such terms (and
their hermitian conjugates) are candidates for terms in a Lagrangian,
which then yields an action that is invariant under SUSY.

To discover the relevant $F$-terms for constructing a SUSY Lagrangian,
we first prove an important theorem.

\begin{theorem}
For any positive integers $n$ and $m$, 
if $\Phi$ is a chiral superfield, then so is $\Phi^n$, whereas $\Phi^n (\Phi^\dagger)^m$ is not a chiral superfield.
\end{theorem}
\begin{proof}
We first note that
\beq
\Dbar_{\dalpha}\Phi^n=n\Phi^{n-1}\Dbar_{\dalpha}\Phi=0,
\eeq
which shows $\Phi^n$ satisfies the defining constraint of a chiral superfield.
A similar computation shows that  $\Phi^n (\Phi^\dagger)^m$ does not
satisfy the required constraint.
\end{proof}

An important consequence of the above theorem is that
\beq
\sum_{n\geq 1} [a_n\Phi^n]_F+{\rm h.c.}
\eeq
is a Lorentz scalar that transforms as a total divergence, and thus is a candidate for terms in a Lagrangian whose action is invariant under SUSY.

\subsubsection{Kinetic terms}
\label{kineticterms}
To construct the kinetic terms of the SUSY Lagrangian,
we define the operator $T$,
\begin{align}
T\Phi=-\tfrac{1}{4}\Dbar\lsup2\Phi^\dagger\,,\label{Tdef}
\end{align}
where $\Dbar\lsup2\equiv \Dbar_{\dalpha}{\Dbar}\lsup{\dalpha}$.  Note
that $\Dbar_{\dalpha}(T\Phi)=0$ (due to the anticommutation relations
satisfied by $\Dbar$), so that $T\Phi$ is a chiral superfield.
In the chiral representation, with $\Phi=A+\sqrt{2}\,\theta\psi+\theta\theta F$,
\begin{align}
T\Phi=F^\dagger-i\sqrt{2}\,\theta\sigma^\mu\partial_\mu\psi^\dagger-\theta\theta\,\square A^\dagger\,.
\end{align}
Hence, the $F$-component of $\Phi T \Phi$ is given by,
\begin{align}
[\Phi T\Phi]_F & =-A\square A^\dagger+F^\dagger F+i\psi\sigma^\mu\partial_\mu\psi^\dagger\nn\\
&= (\partial_\mu A)(\partial^\mu A^\dagger)+F^\dagger F+i\psi^\dagger\sigmabar^\mu\partial_\mu\psi+\text{total derivative}\,,\label{KE}
\end{align}
which we recognize as the kinetic energy term of the
Wess-Zumino Lagrangian [cf.~\eq{eq:LWZF}].
\subsubsection{Mass terms}

To construct the mass terms of the SUSY Lagrangian, the following
theorem is useful.
\begin{theorem}
For any chiral superfield $\Phi$,
\begin{align}
[\Phi]_F=-\tfrac{1}{4} D^2\Phi\biggl|_{\theta=\thetabar=0}=\tfrac{1}{4}\partial^\alpha\partial_\alpha\Phi\biggl|_{\theta=\thetabar=0}\,.\label{phiF}
\end{align}
\end{theorem}
\begin{proof}
\Eq{phiF} follows immediately from \eq{DD}.   
\end{proof}
We can compute the $F$ term of any holomorphic function of a chiral
superfield, $W(\Phi)$,  as follows.  After making judicious use of the
chain rule,
\begin{align}
[W(\Phi)]_F&=\tfrac{1}{4}\partial^\alpha\partial_\alpha W\biggl|_{\theta=\thetabar=0} 
	=\tfrac{1}{4}\partial^\alpha\frac{dW}{d\Phi}\partial_\alpha\Phi\biggl|_{\theta=\thetabar=0}  \nn \\[6pt]
&=\frac{1}{4}\biggl\{\left(\frac{d^2 W}{d\Phi^2}\partial^\alpha\Phi\partial_\alpha\Phi\right)
+\frac{dW}{d\Phi}\partial^\alpha\partial_\alpha\Phi\biggr\}\biggl|_{\theta=\thetabar=0} \,.\label{WPhi}
\end{align}
Noting that $(\partial^\alpha\Phi\partial_\alpha\Phi)_{\theta=\thetabar=0}=-2\psi\psi$,
\eq{WPhi} yields,
\begin{align}
[W(\Phi)]_F &=-\frac12\left(\frac{d^2 W}{d\Phi^2}\right)_{\Phi=A}\psi\psi+\left(\frac{dW}{d\Phi}\right)_{\Phi=A}F\,.
\end{align}
Introducing the notation, $dW/dA\equiv (dW/d\Phi)_{\Phi=A}$,
it follows that
\begin{align}
[W(\Phi)]_F=-\frac12 \frac{d^2 W}{dA^2}\psi\psi+\frac{dW}{dA}F\,.\label{rest}
\end{align}
In the jargon of SUSY, $W(\Phi)$ is called the \textit{superpotential}.  For renormalizable theories, $\!W(\Phi)\!$ is at most cubic in~$\Phi$. 

\subsubsection{The Wess-Zumino SUSY Lagrangian using $F$-terms}
Collecting the results of \eqs{KE}{rest}, we end up with,
\beqa
\mathscr{L}&=&[\Phi T\Phi]_F+\bigl\{[W(\Phi)]_F+{\rm h.c.}\bigr\}
\nn \\[2pt]
&=&(\partial_\mu A)^\dagger(\partial^\mu A)+ i \psi^\dagger \sigmabar^\mu \partial_\mu \psi +F\frac{dW}{dA}+F^\dagger\left(\frac{dW}{dA}\right)^{\!\!\dagger}
+F^\dagger F \nn  \\
&&\quad 
-\frac12\left[\frac{d^2 W}{dA^2}\,\psi\psi+\left(\frac{d^2 W}{dA^2}\right)^{\!\!\dagger}\!\!\psi^\dagger\psi^\dagger\right]\,,\label{WZlagF}
\eeqa
after dropping total derivative terms.  We have thus recovered the
Wess-Zumino Lagrangian that was previously written down in \eq{eq:LWZoriginal}.

The proof that the Wess-Zumino action is supersymmetric, or
equivalently, $\deltaxi\mathcal{L}=\partial_\mu K^{\prime\,\mu}$, is
now trivial since 
$\mathscr{L}$ was constructed from $F$-terms, which
transform as total derivatives under SUSY transformations.
\subsubsection{An alternate form for the kinetic terms: $D$-terms and the K\"ahler potential}
\label{Kahler}
The approach of subsection~\ref{kineticterms} is not the only supersymmetric way to construct  the kinetic energy terms.
Consider an unconstrained superfield $V(x,\theta,\thetabar)$.
 Expanding $V$ as a Taylor series in $\theta$ and $\thetabar$, the
 highest order nonvanishing term is proportional to $(\theta\theta)(\thetabar\thetabar)$.  If we write 
\begin{align}
V(x,\theta,\thetabar)=\cdots+(\theta\theta)(\thetabar\thetabar)D(x)\,,
\end{align}
then one can show that $\deltaxi D(x)$ is a total derivative
using dimensional analysis as we did for $\deltaxi F(x)$ at the end of
Section~\ref{offshell}.  Hence, $D$-terms can also provide suitable terms for a SUSY Lagrangian.

We shall denote the $D$-term by,
\beq
[V]_{\theta\theta\thetabar\thetabar}=[V]_D=D\,,
\eeq
using a notation analogous to that of \eq{Fterm}.  The relevant
theorem analogous to \eq{phiF} is given below.
\begin{theorem}
For any superfield $V$,
\beq
[V]_D =\tfrac{1}{16}\Dbar\lsup{2}D^2
  V\biggl|_{\theta=\thetabar=0}=\tfrac{1}{16}(\partial^\dagger_{\dalpha}\partial^{\dagger\dalpha})(\partial^\alpha\partial_\alpha)V\biggl|_{\theta=\thetabar=0}\,.\label{VD}
\eeq
\end{theorem}
\begin{proof}
\Eq{VD} follows immediately from \eqs{DD}{DbDb}.  
\end{proof}

\noindent
For example, if $\Phi$ is a chiral superfield, one can show that (see Problem~\ref{pr:phistphiD}),
\begin{align}
[\Phi^\dagger\Phi]_D=(\partial_\mu A)(\partial^\mu A^\dagger)+F^\dagger F+i\psi^\dagger\sigmabar^\mu\partial_\mu\psi+\text{total derivative}\,,\label{phiphiD}
\end{align}
which again reproduces the kinetic energy terms of the Wess-Zumino Lagrangian.

Indeed, one can obtain candidate terms for a SUSY Lagrangian by
considering the $\theta\theta\thetabar\thetabar$ component of an
arbitrary function of a chiral superfield and its complex conjugate.
This function, denoted by $K(\Phi,\Phi^\dagger)$, is called the K\"ahler
potential.  Applying the chain rule as in our computation of
$[W(\Phi)]_F$ [cf.~\eqst{WPhi}{rest}], one can calculate (see Problem~\ref{pr:K}),
\begin{align}
\begin{split}
[K(\Phi,\Phi^\dagger)]_D=&\frac{\partial^2 K}{\partial A\partial A^\dagger}\biggl[(\partial_\mu A)(\partial^\mu A^\dagger)+F^\dagger F+\half i\psi^\dagger\sigmabar^\mu\!\!\stackrel{\leftrightarrow}{\partial}_{\!\mu}\!\psi\biggr] \\
&-\frac12\,\frac{\partial^3 K}{\partial A\partial A^{\dagger\,2}}\biggl[F\psi^\dagger\psi^\dagger+i\psi^\dagger\sigmabar^\mu\psi\partial_\mu A^\dagger\biggl] \\
&-\frac12\,\frac{\partial^3 K}{\partial A^2\partial A^{\dagger}}\biggl[F^\dagger\psi\psi-i\psi^\dagger\sigmabar^\mu\psi\partial_\mu A\biggl] \\
&+\frac14\,\frac{\partial^4 K}{\partial A^2\partial A^{\dagger\,2}}(\psi\psi)(\psi^\dagger\psi^\dagger)+\text{total derivative}\,.\label{kahler}
\end{split}
\end{align}

We conclude that the most general SUSY Lagrangian involving a chiral
superfield $\Phi$ is given by
\begin{align}
\mathscr{L}=[K(\Phi,\Phi^\dagger)]_D+\bigl\{[W(\Phi)]_F+{\rm h.c.}\bigr\}\,.\label{eq:Lgeneral}
\end{align}
The auxiliary field $F$ can be determined via its classical
field equation, which yields
\begin{align}
F=\left(\frac{\partial^2 K}{\partial A\partial A^\dagger}\right)^{-1}\left[\frac12\,\frac{\partial^3 K}{\partial A^2\partial A^{\dagger}}\psi\psi-\left(\frac{dW}{dA}\right)^\dagger\right]\,.\label{aux}
\end{align}

The case of $K(\Phi,\Phi^\dagger)=\Phi^\dagger\Phi$ reduces to the result of
\eq{phiphiD} and corresponds to the kinetic energy term of the
Wess-Zumino model as noted above.  In this case, \eq{aux} yields,
\beq
F=-\left(\frac{dW}{dA}\right)^\dagger\,,
\eeq
which reproduces the result previously obtained in \eq{f}.

More complicated K\"ahler potentials yield non-renormalizable Lagrangians.  These arise in
low-energy effective field theories (that include operators of dimension greater than four),
in supersymmetric $\sigma$-models, and in supergravity.  Such
applications lie beyond the scope of these lectures.

\subsection{$R$-invariance}
\label{Rinvariance}
Recall that the SUSY algebra can be extended by added adding a bosonic
U(1)$_R$ generator $R$ such that [cf.~\eqst{R1}{susyalg7}],
\begin{align}
\left[R\,,\,Q_\alpha\right]=-Q_\alpha\,,\qquad\quad
\left[R\,,\,Q^\dagger_{\dot\alpha}\right]=Q^\dagger_{\dot\alpha}\,.\label{Rcommute}
\end{align}
The action of ${\rm U}(1)_R$ on a superfield $\Phi$ can be represented by a differential operator $\widehat{R}$ acting on superspace,
\begin{align}
[\Phi\,,\,R]=\widehat{R}\Phi\,,
\end{align}
where
\begin{align}
\widehat{R}\equiv \theta^\alpha\partial_\alpha-\thetabar_{\dalpha}\partial^{\dagger\dalpha}-n\,,\qquad \text{with $n\in\mathbb{R}$}\,.
\end{align}
We call $n$ the \textit{weight} (or $R$-charge) of the superfield $\Phi$.  (For a \textit{real} superfield, only $n=0$ is possible.)
Under a ${\rm U}(1)_R$ transformation,
\begin{align}
\delta_a\Phi=ia[R\,,\,\Phi]=-ia\widehat{R}\Phi\,.
\end{align}
Acting on a superfield $\Phi(x,\theta,\thetabar)$,
\beq
\widehat{R}\,\Phi(x,\theta,\thetabar)=e^{ina}\,\Phi(x,e^{-ia}\theta,e^{ia}\thetabar)\,,\label{Rtrans}
\eeq
The differential operator $\widehat{R}$ satisfies the identities,
\beqa
D_\alpha \widehat{R}&=&(\widehat{R}+1)D_\alpha\,,\\
\Dbar_{\dalpha}\widehat{R}&=&(\widehat{R}-1)\Dbar_{\dalpha}\,.
\eeqa
Hence, it follows that if $\Phi$ is a chiral [antichiral] superfield, then
$\widehat{R}\Phi$ is a chiral [antichiral] superfield.

Given a chiral superfield, $\Phi=A+\sqrt{2}\,\theta\psi+\theta\theta
F$, in the chiral representation, the ${\rm U}(1)_R$ transformations of the component fields are:
\begin{align}
A&\to e^{ina}A\,,\\
\psi &\to  e^{i(n-1)a}\psi\,,\\
F&\to  e^{i(n-2)a}F\,,\label{RF}
\end{align}
after employing \eq{Rtrans}.

\begin{theorem}
\label{Rtheorem}
The kinetic energy term $[\Phi^\dagger\Phi]_D$ is automatically
$R$-invariant, whereas $[W(\Phi)]_F$ is $R$-invariant if and only if $W$ has $R$-charge equal to 2.
\end{theorem}
\begin{proof}
If $n=2$, then $F$ is invariant under a ${\rm U}(1)_R$ transformation,
 in light of \eq{RF}.  This result applies to any $F$-term.
\end{proof}

\begin{example}
[Wess-Zumino model with $\boldsymbol{W(\Phi)=\half m\Phi^2+\tfrac13 g\Phi^3}$]
If $m=0$, then the Wess-Zumino model is $R$-invariant with $n=\tfrac13$.
If $g=0$, then the Wess-Zumino model is $R$-invariant with $n=\tfrac12$.
If both $m\neq 0$ and $g\neq 0$, then the Wess-Zumino model is not $R$-invariant. 
\end{example}

\subsection{Grassmann integration and the SUSY action}
\label{integration}
A supersymmetric action can be written as an integral over superspace.
First, we introduce integration over anticommuting Grassmann
variables.   The rules of integration are\cite{Berezin},
\begin{align} \label{grules}
\int d\theta=\int d\theta^\dagger=0\,,\qquad \int \theta\,d\theta=\int
  \theta^\dagger\,d\theta^\dagger=1\,.
\end{align}
That is, integration over Grassmann variables is in some sense equivalent to differentiation.

It is conventional to define
\begin{align}
d^2\theta&\equiv -\tfrac14 d\theta^\alpha d\theta_\alpha\,,\\
d^2\thetabar&\equiv-\tfrac14 d\thetabar_{\dalpha} d\theta^{\dagger\dalpha}\,,\\
d^4\theta &\equiv d^2\theta d^2\thetabar\,,
\end{align}
which yields the following non-zero integrals,
\begin{align}
\int d^2 \theta\, (\theta\theta)=\int d^2\thetabar\,(\thetabar\thetabar)=\int d^4\theta\,(\theta\theta)(\thetabar\thetabar)=1\,.
\end{align}

It follows that for a chiral superfield,
\begin{align}
\int d^2\theta\, \Phi(x,\theta,\thetabar)=\int d^2\theta\, \Phi_1(x,\theta)=[\Phi]_F=-\tfrac14 D^2\Phi\biggl|_{\theta=\thetabar=0}\,.
\label{d2theta}
\end{align}
Likewise, for an arbitrary superfield $V(x,\theta,\thetabar)$,
\begin{align}
\int d^4\theta\, V(x,\theta,\thetabar)=[V]_D=\tfrac{1}{16}\Dbar\lsup{2} D^2 V\biggl|_{\theta=\thetabar=0}\,.\label{d4theta}
\end{align}
Thus, the most general SUSY action involving a chiral superfield $\Phi$ is
\begin{align}
S=\int d^4 x\,d^4\, \theta K(\Phi,\Phi^\dagger)+\int d^4 x\,d^2\theta\, W(\Phi)
+\int d^4 x\,d^2\thetabar\, W(\Phi^\dagger)\,.\label{S}
\end{align}

Generalizations to theories with multiple chiral superfields are
straightforward.  In the more general case, $W$ is a holomorphic
multivariable function of the chiral superfields, and $K$ is a
multivariable function of the chiral superfields and their hermitian conjugates.  For a renormalizable theory, $W$ is at most a cubic multinomial, 
\begin{align}
W(\Phi_i)=\sum_i a_i\Phi_i+\sum_{i,j} b_{ij}\Phi_i\Phi_j+\sum_{i,j,k}c_{ijk}\Phi_i\Phi_j\Phi_k\,,
\end{align}
and
\begin{align}
K(\Phi_i,\Phi_i^\dagger)=\sum_i \Phi_i^\dagger\Phi_i\,.\label{Ksimple}
\end{align}

In special cases, one can convert an integral over ``half'' of superspace (e.g. integrals over
$d^4 x\, d^2\theta$) into an integral over the full superspace.  The key
observation is that for an arbitrary superfield $V$,
\begin{align}
\int d^4x\,d^2\theta\,V(x,\theta,\thetabar)=\int d^4 x\left(-\tfrac14 D^2 V\right)\,.\label{intV}
\end{align}
On the left-hand side of \eq{intV}, the integration over $d^2 \theta$ projects out all terms proportional to $\theta\theta$.  On the right-hand side, $D^2\!=-\!\partial^\alpha\partial_\alpha$
up to total derivative terms that can be dropped because we are integrating over $d^4 x$.   Hence, $\tfrac14 \partial^\alpha\partial_\alpha$ has the effect of projecting out all terms proportional to $\theta\theta$.
Likewise,
\begin{align}
\int d^4x\,d^2\thetabar\,V(x,\theta,\thetabar)=\int d^4 x\left(-\tfrac14 \Dbar\lsup{2} V\right)\,.
\end{align}
Hence, it follows that
\begin{align}
\int d^4 x\,d^2\theta\left(-\tfrac14 \Dbar\lsup{2} V\right)=\int d^4x\,d^4\theta\,V(x,\theta,\thetabar)\,.\label{2to4}
\end{align}

\Eqs{d2theta}{d4theta} identify integrals over
half of superspace as $F$-terms and integrals over the full superspace
as $D$-terms.  However,
\eq{2to4} appears to blur the distinction between $D$-terms and $F$-terms.
For example, in the Wess-Zumino Lagrangian, the kinetic energy
term may be written as an $F$-term, $[\Phi T\Phi]_F$ [cf. \eq{WZlagF}], or
as a $D$-term, $[\Phi^\dagger\Phi]_D$, as in
 \eqs{phiphiD}{eq:Lgeneral}.  However,
consider the case of a half superspace integral of the superpotential given
in \eq{S}.  If we attempt to convert this into a full superspace
integral using \eq{2to4}, the end result is
\beq
\int d^4 x \,d^2\theta\,W(\Phi) =-4\int d^4 x\, d^4\theta\,\label{nonlocal}
\Dbar\lsup{-2} W(\Phi)\,.
\eeq
Due to the inverse differential operator, the integrand on the right-hand side of \eq{nonlocal} is a non-local
functional of chiral superfields.  This provides the distinction
between $F$-terms and $D$-terms.  In particular, any half superspace integral
that can be converted into a full superspace integral over a \textit{local}
functional of superfields will be called a $D$-term.

Having written the action in \eq{S} as an integral over superspace (for
$D$-terms) and half of superspace (for $F$-terms), one can obtain
expressions for the Green functions of quantum chiral
(and antichiral) superfields.  The corresponding two-point functions
provide expressions for the superspace propagators.   One can then
formulate a set of superspace Feynman rules and develop a
diagrammatic representation of the perturbative expansion of the
Green functions.  This was first carried out by Grisaru, Ro\u{c}ek, and
Siegel\cite{GRS}, and was applied to the perturbative computation of
the effective action.   Indeed, such techniques are quite useful since a
single supergraph (in which individual lines correspond to
superfields) is equivalent to a large number of Feynman diagrams involving
the corresponding component fields.
A comprehensive treatment of these methods are beyond the scope of these lectures.
For a pedagogical development of supergraphs and superspace Feynman rules, see e.g.~Refs.\cite{Gates,Srivastava,Buchbinder,Pokorski}.

\subsection{Improved ultraviolet behavior of supersymmetry}
\label{sec:non-renorm}

An attractive feature of supersymmetric quantum field theories is that
their ultraviolet divergences are better behaved, as compared to
ordinary quantum field theories.
Ref.\cite{GRS} demonstrated that the loop corrections to the effective
action of a supersymmetric theory of chiral superfields
can be expressed as an integral over the full superspace,
\beq \label{effact}
\sum_n \int d^4 x_1\cdots d^4 x_n\int d^4\theta\, g_n(x_1,\ldots,x_n)
F_1(x_1,\theta,\thetabar)\cdots F_n(x_n,\theta,\thetabar)\,,
\eeq
where the $F_i(x_i,\theta,\thetabar)$ are local functionals of chiral
and antichiral superfields and their covariant derivatives, and the
$g_n$ are translationally invariant functions on Minkowski space.

\Eq{effact} implies that $D$-terms are renormalized but $F$-terms
are not renormalized.  Moreover, if $F$-terms are absent at
tree-level, then they are not generated at the loop level.  Hence, the
tree-level K\"ahler potential is
renormalized by radiative corrections, whereas there are no loop
corrections to the tree-level superpotential.  This is the famous
non-renormalization theorem of $N=1$ supersymmetry.\footnote{The 
proof of the non-renormalization theorem implicitly assumes that the function $g_n$
in \eq{effact} is local.  However, the non-renormalization theorem can fail if 
the super\-symmetric theory contains massless fields as shown in
Refs.\cite{West:1990rm,Jack:1990pd,Dunbar:1991fc},  due to infrared
divergences.  For example, the inverse Laplacian operator
$\square^{-1}$ (from a massless propagator) can appear, resulting in
a non-local function $g_n$ in \eq{effact}.  One can show that
the non-renormalization theorem holds for the
Wilsonian effective action\cite{Shifman:1986zi,Shifman:1991dz}, where
the infrared effects are cut off\cite{SeibergNR,Poppitz:1996na}. \label{fnW}}
The proof of the non-renormalization theorem in Ref.\cite{GRS} relies
on the analysis of supergraphs in perturbation theory, and is beyond
the scope of these lectures.  Heuristically, this
theorem is a consequence of an exact cancellation between fermion and boson
loop contributions to the effective action due to supersymmetry.

Note that the non-renormalization of the tree-level superpotential
is simply  a consequence of the fact that the integral of a
product of chiral superfields over \textit{all} of superspace in \eq{effact} is zero due to
\eq{grules} [see Problem~\ref{pr:half}].  Moreover, the assumption that
the $F_i$ in \eq{effact} are \textit{local} functionals of chiral and
antichiral superfields is essential.  Otherwise, one could employ
\eq{nonlocal} and erroneously claim the existence of loop corrections
to the tree-level superpotential.

We now briefly explore the consequence of the non-renormalization of
the superpotential.  Consider the action of the Wess-Zumino model,
\begin{align}
S_{\mathrm{WZ}} = \int d^4 x \int d^4 \theta\, \Phi^\dagger \Phi  + \sqof{ \int d^4 x \int d^2\theta \of{\half m \Phi^2 + \tfrac{1}{3} \lambda\Phi^3}  
+ {\rm h.c.}}.
\end{align}
The non-renormalization theorem implies that renormalized fields and
parameters are related to bare fields and parameters as follows\cite{Cui},
\beq \label{bare}
\Phi_R=Z^{-1/2}\Phi\,,\qquad\quad m_R=Zm\,,\qquad\quad
\lambda_R=Z^{3/2}\lambda\,.
\eeq
where the subscript $R$ indicates renormalized quantities and the
bare quantities have no subscript.  \Eq{bare} is equivalent to the
statement that the superpotential is unrenormalized,
$W_R(\Phi_R)=W(\Phi)$. 
That is,
\beq
\half m_R\Phi_R^2 +\tfrac13 \lambda_R\Phi_R^3=\half m\Phi^2 +\tfrac13
\lambda\Phi^3\,.
\eeq
Wave function renormalization is a consequence of the
renormalization of the  K\"ahler potential ($\Phi^\dagger\Phi$ in
the case of the Wess-Zumino model). 

The non-renormalization theorem does \textit{not} assert that
the parameters of the
superpotential are not renormalized.   Indeed, \eq{bare} states that
the renormalization of the parameters $m$ and $\lambda$ are governed
by the wave function renormalization constant $Z$.   Moreover, the
wave function renormalization constants of the component fields of the chiral
superfield are~equal (i.e., $A_R\!\!=Z^{-1/2}A$ and
$\psi_R\!=\!Z^{-1/2}\psi$), as a consequence of supersymmetry.

%%%%%

In Ref.\cite{SeibergNR}, Seiberg offered a more intuitive understanding of the
non-renormalization theorem, which also forbids nonperturbative
corrections to the Wilsonian effective action [cf.~footnote~\ref{fnW}].
Seiberg's argument draws on the symmetry and
holomorphy\footnote{The fact that the superpotential is a
  holomorphic function of chiral superfields plays a critical role in
Seiberg's argument.
In contrast, the renormalization of the K\"ahler potential is
possible because the latter is a function of chiral and antichiral
superfields and hence is not holomorphic.}
\!of the
superpotential. Consider again the example of the Wess-Zumino
superpotential, $W(\Phi)=\half m\Phi^2+\tfrac13\lambda\Phi^3$.
Following Ref.\cite{SeibergNR}, one can think of $m$ and $\lambda$ as the vacuum
expectation values of chiral superfields, so that $W$ must be
holomorphic in $m$ and $\lambda$ as well as in $\Phi$. 
In light of Theorem~\ref{Rtheorem} in Section~\ref{Rinvariance}, the theory is
invariant under an enhanced ${\rm U}(1) \times {\rm U}(1)_R$
symmetry, with the charge assignments shown in
Table~\ref{tab:charges}.
\begin{table}[h!]
\begin{center}
\caption{\small Charge assignments under the ${\rm U}(1) \times {\rm U}(1)_R$ symmetry.}
\label{tab:charges}
\vskip 0.1in
\begin{tabular}{| l | r r r r |}
\hline
	& $\Phi$	& $\Phi^\dagger$	& $m$	& $\lambda$	\\
	\hline
${\rm U}(1)$ & 1		& $-1$				& $-2$		& $-$3 			\\
${\rm U}(1)_R$	& 1		&	1			&	0	&	$-$1	\\
\hline
\end{tabular}
\end{center}
\end{table}

To maintain the ${\rm U}(1) \times{\rm  }U(1)_R$ symmetry and holomorphy,
corrections to the Wilsonian effective superpotential must therefore be of the form
\begin{align}
m \Phi^2 f\of{  \frac{ \lambda \Phi }{ m} },\label{wilson}
\end{align}
where $f$ is an arbitrary holomorphic function. 
\Eq{wilson} is valid for arbitrary $\lambda$.  Thus, we can take
$|\lambda|\ll 1$, in which case perturbation theory should be valid.
Expanding in powers of the coupling constant $\lambda$, the perturbative expansion should have the
following form,
\beq
W_{\rm eff}=\sum_{n=0}^\infty
a_n\frac{\lambda^n}{m^{n-1}}\Phi^{n+2}\,.\label{Weffective}
\eeq
The terms in $W_{\rm eff}$ are represented diagrammatically by one
particle irreducible (1PI) supergraphs constructed from propagators and
three-point vertices proportional to $\lambda$.  However, one cannot construct a
one-loop (or higher) supergraph that behaves like
$\lambda^n\Phi^{n+2}$.  It is easy to show that tree-level diagrams
with $n+2$ external legs, $n$ vertices and $n-1$ propagators would
behave like $\lambda^n\Phi^{n+2}$.  But, the only 1PI tree-level
graphs are those with either two or three external legs!  Hence, we
conclude that $a_0=\half$, $a_1=\tfrac13$ and $a_n=0$ for
$n\geq 2$.\footnote{One can also conclude that $a_n=0$ for $n\geq 2$
  by noting that the Wilsonian effective action $W_{\rm eff}$ must have a smooth
  limit as $m\to 0$.} That is
$W_{\rm eff}(\Phi)=W_{\rm tree}(\Phi)$, which is the statement that the
superpotential is not renormalized.

\subsection{Problems}

\begin{problem}
\label{pr:two_super_translations}
Prove that
\begin{align}
 G(y,\xi,\xi^\dagger)G(x,\theta,\thetabar)=G\bigl(x+y+i(\xi\sigma\thetabar-\theta\sigma\xi^\dagger),\xi+\theta,\xi^\dagger+\thetabar\bigr)\,. \nn
 \end{align}
  {\sl HINT}: use the Baker-Campbell-Hausdorff formula given in \eq{BCH}. 
%$\exp(A)\exp(B)=\exp(A+B+\half[A\,,\,B]+\cdots).$
 \end{problem}

\begin{problem}
Verify that when acting on a superfield $\Phi(x,\theta,\thetabar)$,
$$
\{\widehat{Q}_\alpha\,,\,\widehat{Q}_\beta\}=\{\widehat{Q}^\dagger_{\dot\alpha}\,,\,\widehat{Q}^\dagger_{\dot\beta}\}=0\,,\qquad\quad
\{\widehat{Q}_\alpha\,,\,\widehat{Q}^\dagger_{\dot\beta}\}=2\sigma^\mu_{\alpha\dot\beta}\widehat{P}_\mu\,.
$$
\end{problem}

\begin{problem}
\label{pr:3results}
Prove \eqst{eq:r1}{eq:r3}.  The last result is an example of a
Fierz identity (see, e.g., Appendix B of Ref.\cite{Dreiner:2008tw} or
Appendix A of Ref.\cite{Bailin}).
\end{problem}

\begin{problem}
\label{pr:fmnV}
Using \eq{supertrans}, obtain the SUSY transformation laws for the
bosonic component fields,
$f$, $m$, $n$, $V_\mu$, and $d$, and the fermionic component fields,
$\zeta$, $\chi$, $\lambda$ and $\psi$, which appear in the complex
superfield defined in \eq{phitaylor}. 
\end{problem}

\begin{problem}
Suppose that $\Phi$ is a bosonic superfield.  Verify that
\eq{daggers} holds.  Then, show that \eqs{eq:D}{eq:Db} satisfy
\eq{Dcond}.
\end{problem}

\begin{problem}
Suppose that $\Phi$ is a fermionic superfield.  Show that
\eqs{daggers}{Dcond} are modified as follows:
$(\partial_\alpha\Phi)^\dagger=\partial_{\dalpha}^\dagger\Phi^\dagger$
and $(D_\alpha\Phi)^\dagger=-\Dbar_{\dalpha}\Phi^\dagger$. 
 \end{problem}

\begin{problem} 
\label{pr:D}
Show that the spinor covariant derivatives, as defined in \eq{eq:D}
and \eq{eq:Db}, satisfy the  following anticommutation relations,
$\{D_\alpha\,,\,D_\beta\}=\{\Dbar_{\dot\alpha}\,,\,\Dbar_{\dot\beta}\}=0$ and
$\{D_\alpha\,,\,\Dbar_{\dot\beta}\}=2i\sigma^\mu_{\alpha\dot\beta}\partial_\mu$.
\end{problem}

\begin{problem}
\label{pr:DDid}
Derive \eq{DDid}.
\end{problem}

\begin{problem}
\label{pr:expexp}
Prove that
\begin{align*}
\exp(-i\theta\sigma^\mu\thetabar\,\partial_\mu)=1-i\theta\sigma^\mu\thetabar\,\partial_\mu-\tfrac{1}{4}(\theta\theta)(\thetabar\thetabar)\square,
\end{align*}
where $\square\equiv\partial_\mu\partial^\mu.$
\end{problem}

\begin{problem}
\label{pr:chiraltrans}
Using  \eq{supertrans}, one can
obtain the SUSY transformation laws for the component fields $A$,
$\psi$ and $F$ in \eq{eq:chiralSF}.  Perform the calculation by
working in the chiral representation and show
that the SUSY transformation laws for $A$, $\psi$ and $F$  coincide with
the results obtained previously in \eqst{offshell1}{offshell3} for the fields of a superspin $j=0$
supermultiplet. 
\end{problem}

\begin{problem}
\label{pr:phistphiD}
If $\Phi$ is a chiral superfield, show that
\begin{align*}
[\Phi^\dagger\Phi]_D=(\partial_\mu A)(\partial^\mu A^\dagger)+F^\dagger F+i\psi^\dagger\sigmabar^\mu\partial_\mu\psi+\text{total derivative}\,.
\end{align*}
\end{problem}

\begin{problem}
\label{pr:K}
Derive \eq{kahler}.
\end{problem}

\begin{problem}
A linear superfield\cite{Ferrara:1974ac,Salam:1974jj}, $L(x,\theta,\bar\theta)$,
is defined as a constrained real scalar superfield that satisfies, $D^2 L(x,\theta,\bar\theta)=\overline D\lsup{2}
L(x,\theta,\bar\theta)=0$.   Identify the component fields
that make up the linear superfield.  Show that $\partial_\mu V^\mu=0$,
where $V^\mu$ is the component vector field of $L$. 
Check that the number of fermion and boson degrees of freedom of the linear
superfield are equal.
[HINT: the identity given by \eq{DDid} should be helpful.]
\end{problem}

\begin{problem}
Employing the operator $T$ defined in \eq{Tdef}, show that
\beq
\int d^4 x\,d^2\theta\, \Phi T\Phi = \int d^4 x\, d^4\theta
\,\Phi^\dagger\Phi\,,
\eeq
by converting the integral over half of superspace into an integral
over the full superspace.  Use the above result to conclude that
$[\Phi T \Phi]_F=[\Phi^\dagger\Phi]_D$.
\end{problem}

\begin{problem}
\label{pr:half}
If $\Phi$ is a chiral superfield and $\Phi^\dagger$ is an antichiral
superfield, show that
\beq
\int d^4 x\,d^4\theta\, \Phi(x,\theta,\thetabar)=\int d^4 x\,d^4\theta\,
\Phi^\dagger(x,\theta,\thetabar)=0\,.
\eeq
\end{problem}

%%%%%%%%%%%%%%%%%%%%%%%%%%%%%%%%%%%%%%%%%%%%%%%%%%

%

%%%%%%%%%%%%%%%%%%%%%%%%%%%%%%%%%%%%%%%%%%%%%%%%%%%
\section{Supersymmetric gauge theories}
\label{sec:gaugetheories}
\renewcommand{\theequation}{\arabic{section}.\arabic{equation}}
\setcounter{equation}{0}

In this section, we discuss the supersymmetric extension of 
gauge theories. We begin with the  vector superfield $V$, which contains
the gauge fields as well as their supersymmetric partners, the
gauginos. We discuss the behavior of $V$ under a  gauge
transformation, and the gauge-invariant interaction terms that couple
the vector superfield with one or more 
chiral superfields.   Both abelian and non-abelian gauge groups are
treated.  Finally, we construct the SUSY Lagrangians corresponding to QED and a
non-Abelian SUSY Yang-Mill theory coupled to supersymmetric matter.

\subsection{Vector superfields}

Imposing a reality condition on a complex superfield (which is a
covariant constraint with respect to SUSY transformations), we obtain the so-called real vector superfield,
\begin{align}
V(x,\theta,\thetabar)=V^\dagger(x,\theta,\thetabar)\,,
\end{align}
which will be employed in constructing supersymmetric gauge theories.
Expanding in $\theta$ and $\thetabar$,
\beqa
V&=& \ C + i\theta\chi-i\thetabar\chi^\dagger +\half i
\theta\theta(M+iN)-\half i \thetabar \thetabar (M-iN)
+\theta\sigma^\mu\thetabar V_\mu  \nn \\
&&+
i(\theta\theta)\thetabar\bigl(\lambda^\dagger-\half
i\,\sigmabar^\mu\partial_\mu\chi\bigr)-i(\thetabar
\thetabar)\theta(\lambda-\half
i\sigma^\mu\partial_\mu\chi^\dagger\bigr) \nn \\
&& +\half(\theta\theta)(\thetabar 
\thetabar)\bigl(D-\half\Box C\bigr)\,,
\label{VectorSF}
\eeqa
where $C$, $M$, $N$, $D$ and $V_\mu$ are real bosonic fields,
and $\chi$ and $\lambda$ are two-component fermion fields.  The
various factors of $i$ and $\half$ are conventional, and the
particular linear combination of fields chosen as coefficients of
$(\theta\theta)\thetabar$, $(\thetabar\thetabar)\theta$ and
$(\theta\theta)(\thetabar\thetabar)$ are convenient for later
purposes [cf.~footnote~\ref{fn38}].
%\footnote{In particular, with these choices of coefficients, $\lambda$ and
  %$D$ are invariant under super gauge transformations [cf.~\eqs{lambdaGT}{DGT}.}
Note that the superfield $V$ is dimensionless, in which case it
follows that the dimensions of the component fields are $[V_\mu]=1$ and
$[\lambda]=\tfrac32$, as expected, whereas $[C]=[D]=0$, and $[\chi]=\half$
after making use of the dimensions of the Grassmann coordinates,
$[\theta]=[\thetabar]=-\half$.

The real vector field $V_\mu$ is a candidate for a gauge
boson of an abelian U(1) gauge theory.  The corresponding field strength tensor is given by
\begin{align}
F_{\mu\nu}=\partial_\mu V_\nu-\partial_\nu V_\mu\,.
\end{align}
Indeed, this can be shown to be one of the components of the \textit{field strength superfield}, which is defined by
\begin{align}
\W_\alpha=-\tfrac{1}{4} {\Dbar}^{2} D_\alpha V\,.\label{Walpha}
\end{align}
Note that $\Dbar_{\dbeta}\W_\alpha=0$, so that $\W_\alpha$ is a spinor chiral superfield.  Evaluating the above expression, and expressing it in the chiral representation,
\begin{align}
\W_\alpha(y,\theta,\thetabar)=
-i\lambda_\alpha + \theta_\alpha D 
-\half i (\sigma^\mu\sigmabar^\nu\theta)_\alpha F_{\mu\nu}
-\theta\theta (\sigma^\mu\partial_\mu \lambda^\dagger)_\alpha ,\label{Wdef}
\end{align}
where $y\equiv x-i\theta \sigma^\mu\thetabar$.   The fermionic partner
of the gauge boson, called the gaugino, is represented by the
two-component spinor field $\lambda$.  Remarkably, the fields $C$,
$M$, $N$ and $\chi$ that are coefficients in the Taylor expansion of the vector
superfield $V$ do not appear in \eq{Wdef}.  The reason for this will
become apparent in Section~\ref{sec:gauge}.

One can work out the SUSY transformation laws of the fields, $\lambda$,
$F_{\mu\nu}$ and~$D$,
by matching component fields on both sides of the following equation,
 \begin{align}
 \deltaxi \W_\alpha=-i(\xi \widehat{Q}+\xi^\dagger \widehat{Q}^\dagger)\W_\alpha.
 \end{align}
 The end result is
\begin{align}
 \deltaxi\lambda_\alpha&= i\xi_\alpha D+\half(\sigma^\mu\sigmabar^\nu)_\alpha{}^\beta\xi_\beta F_{\mu\nu}\,, \\
 \deltaxi F_{\mu\nu}&=i\partial_\mu(\xi\sigma_\nu\lambda^\dagger-\lambda\sigma_\nu\xi^\dagger)
 -i\partial_\nu(\xi\sigma_\mu\lambda^\dagger-\lambda\sigma_\mu\xi^\dagger)\,, \\
 \deltaxi D&=\partial_\mu(\xi\sigma^\mu\lambda^\dagger+\lambda\sigma^\mu\xi^\dagger)\,.\label{delD}
\end{align}
Note that the mass dimension of the $D$-term is given by $[D]=2$.  Hence, 
dimensional analysis implies that $\delta_\xi D$ must be a total
derivative, which is confirmed in \eq{delD}.
From the above transformation laws, we conclude that $(\lambda\,,\,\lambda^\dagger\,,\,F_{\mu\nu}\,,\,D)$ forms an irreducible supermultiplet (corresponding to superhelicity $1$).

To obtain the Lagrangian for the SUSY U(1) gauge theory, note that
\begin{align}
\tfrac14[\W^\alpha \W_\alpha]_F+{\rm h.c.}&=\half i(\lambda\sigma^\mu\partial_\mu\lambda^\dagger+\lambda^\dagger\sigmabar^\mu\partial_\mu\lambda)+\half D^2-\tfrac14 F_{\mu\nu}F^{\mu\nu} \nn \\[5pt]
&=i\lambda^\dagger\sigmabar^\mu\partial_\mu\lambda+\half D^2-\tfrac14
  F_{\mu\nu}F^{\mu\nu}+\text{total derivative}. \label{WF}
\end{align}
This is the kinetic energy term for a U(1) gauge field $V_\mu$ and its
gaugino superpartner $\lambda$.  Both the gauge boson and gaugino are
massless.  The real scalar field $D$ is not dynamical; it is an
auxiliary field.

The action corresponding to the Lagrangian of
\eq{WF} can be written as an integral over half of superspace.   In
particular,  \eq{d2theta} yields,
\begin{align}
\mathscr{L}=\tfrac14\int d^2\theta\, \W^\alpha \W_\alpha+{\rm h.c.}
\end{align}
One can show that $[\W^\alpha \W_\alpha]_F$ and its hermitian
conjugate term differ only by a total derivative.  Hence, both terms
contribute equally to the action, which is given by
\begin{align}
S=\tfrac12\int d^4 x\,d^2\theta\, \W^\alpha \W_\alpha\,.
\end{align}
It is sometimes convenient to turn this integral into an integration
over the full superspace.  Using a trick analogous to the one employed
in \eq{2to4}, we end up with,
\begin{align}
S=\tfrac12\int d^4 x\,d^2\theta\, \left(-\tfrac14 \Dbar\lsup{2}\right)(D^\alpha V) \W_\alpha
=\tfrac12\int d^4 x\,d^2\theta\,d^2\thetabar(D^\alpha V) \W_\alpha
\,,
\end{align}
after using \eq{Walpha} to rewrite one factor of $\W^\alpha$ in
terms of $V$.

It is instructive to count the degrees of freedom in the irreducible
supermultiplet, $(\lambda\,,\,\lambda^\dagger\,,\,F_{\mu\nu}\,,\,D)$. 
On-shell, there are two real fermionic degrees
of freedom 
associated with the massless gaugino, after imposing the
Lagrange field equations,\footnote{Starting with two complex (or
  equivalently four real) degrees
  of freedom for the two-component gaugino field $\lambda$,
  \eq{gauginoDirac} relates the spinor components $\lambda_1$ and $\lambda_2$, thereby
  reducing the number of real degrees of freedom from four to two.}
\beq
i\sigmabar\lsup{\mu\alpha\dbeta}\partial_\mu\lambda_\beta=0\,.\label{gauginoDirac}
\eeq
This matches the two real bosonic degrees of freedom corresponding
to the two transverse polarizations of the massless gauge boson.

To count
 the off-shell bosonic degrees of freedom,
one must take into account the Bianchi identity,\footnote{Although it
  appears that the Bianchi identity yields four constraints, since the
  spacetime index $\mu$ is a free index, in fact only
  three constraints are independent.  This is because one of the four
  constraints is redundant due to the identity,
$\epsilon^{\mu\nu\rho\sigma}\partial_\mu\partial_\nu F_{\rho\sigma}=0$,
which is automatically satisfied as a result of the antisymmetry of the Levi-Civita tensor.
Physically, the Bianchi identity implies that  the three components of the
  electric field vector determine the three components of the magnetic field vector.}
\beq
\epsilon^{\mu\nu\rho\sigma}\partial_\nu F_{\rho\sigma}=0\,,\label{bianchi}
\eeq
which is satisfied independently of the field equations.  This
identity reduces the number of real degrees of freedom in the real
antisymmetric tensor $F_{\mu\nu}$ from
six to three.   Adding in the one real degree of freedom associated
with $D$, we end up with a total of four real bosonic degrees of freedom,
which matches the four real off-shell fermionic degrees of freedom corresponding to $\lambda$ and~$\lambda^\dagger$. 

\subsection{Gauge invariance}
\label{sec:gauge}
The vector superfield $V$ contains the familiar gauge field
$V_\mu$. But it also includes other component fields $C$, $\chi$, $M$
and $N$, whose meaning is less obvious. As we will see, these latter fields
turn out to be gauge artifacts.  Thus, we must examine how gauge
transformations  of the gauge field theory get promoted to gauge transformations of the
vector superfield $V$.

Let $\Lambda(x,\theta,\thetabar)$ be a chiral superfield (\textit{i.e.}, $\Dbar_{\dalpha}\Lambda=0$) and let 
$\Lambda^\dagger(x,\theta,\thetabar)$ be the corresponding antichiral superfield.  Consider the transformation,
\begin{align}
V\to V+i(\Lambda-\Lambda^\dagger)\,.
\label{eq:Vgaugetransform}
\end{align}
We assert that \eq{eq:Vgaugetransform} is a supersymmetric
generalization of the gauge transformation of an abelian gauge theory,
henceforth called a super gauge transformation.

With the help of \eq{eq:Dcomms}, it is straightforward to show that
the field strength superfield, $\W_\alpha$, is invariant under a super
gauge transformation.
Moreover, if the Taylor series of $\Lambda(x,\theta,\thetabar)$ is
written as\footnote{In contrast to the chiral superfield $\Phi$ in
  \eq{eq:chiralSF} whose mass dimension is 1, the chiral superfield $\Lambda$ is dimensionless,
  as required for consistency in light of \eq{eq:Vgaugetransform}.}
\begin{align}
\begin{split}
\Lambda(x,\theta,\thetabar)=&
\widetilde{A}(x) + \sqrt{2}\,\theta \widetilde{\psi}(x) + \theta\theta \widetilde{F}(x)-i \theta\sigma^\mu\thetabar \partial_\mu \widetilde{A}(x)  \\
& - \frac{i}{\sqrt{2}} (\theta\theta) 
\theta^\dagger \sigmabar^\mu\, \partial_\mu \widetilde{\psi}(x)-\tfrac{1}{4}(\theta\theta)(\thetabar\thetabar) \square \widetilde{A}(x)\,, 
\end{split}
\end{align}
then the impact of the super gauge transformation given by
\eq{eq:Vgaugetransform} on the component fields of $V$
is,\footnote{\label{fn38} The invariance
of $\lambda$ and $D$ under super gauge transformations is a
consequence of the particular choices made for the coefficients of
$(\theta\theta)\thetabar$, $(\thetabar\thetabar)\theta$ and
$(\theta\theta)(\thetabar\thetabar)$ in \eq{VectorSF}.}
\begin{align}
  C 
  & 
  \to C+i(\widetilde{A}-\widetilde{A}^\dagger)\,,\\
%\chi
%&
%\to \chi+\frac{1}{\sqrt{2}}\widetilde{\psi} \,, \\
\chi
&
\to \chi+ \sqrt{2}\, \widetilde{\psi} \, , \\
M+iN 
&
\to M+iN+2\widetilde{F}\,,\\
V_\mu
&
 \to V_\mu+\partial_\mu(\widetilde{A}+\widetilde{A}^\dagger)\,, \\
 \lambda
 &
  \to \lambda\,,\label{lambdaGT} \\
D 
&
\to D\,.\label{DGT}
\end{align}
Indeed, under a super gauge transformation, the gauge field $V_\mu$
transforms by an ordinary gauge transformation.  Moreover the 
field strength tensor $F_{\mu\nu}=\partial_\mu V_\nu-\partial_\nu
V_\mu$, the gaugino
field $\lambda$, and the auxiliary field $D$ are gauge invariant as
one would anticipate (consistent with the fact that the field strength
superfield $\W$ is gauge invariant).

One particularly useful gauge choice is to choose $\widetilde A$,
$\widetilde\psi$ and $\widetilde F$ such that 
\beq
C=\chi=M=N=0\,.\label{WZ}
\eeq
This is called the {{Wess-Zumino (WZ) gauge}}\cite{Wess:1974jb}.  
The existence of such a gauge implies that the fields $C$,
$\chi$, $M$, and $N$ are gauge artifacts, as previously stated.
The main drawback of the WZ gauge is that it is not a
supersymmetric gauge choice.  That is, starting from the WZ gauge and
performing a SUSY transformation on the component fields of the vector
superfield~$V$ will yield new component fields that do not satisfy the WZ gauge
condition.  

The main benefit of the WZ gauge is that it provides enormous
simplification in many practical computations.  In particular, applying
the WZ gauge condition [\eq{WZ}] to the vector superfield given in \eq{VectorSF},
\begin{align}
V_{\rm WZ}=\theta\sigma^\mu\thetabar V_\mu+
i(\theta\theta)(\thetabar\bar{\lambda})-i(\thetabar\thetabar)(\theta\lambda)+\half(\theta\theta)(\thetabar\thetabar)D\,.
\end{align}
Computing the square of $V_{\rm WZ}$ with the help of \eq {eq:r3} yields,
\begin{align}
V^2_{\rm WZ}(x,\theta,\thetabar)=\half (\theta\theta)(\thetabar\thetabar)V_\mu V^\mu\,.
\end{align}
and $V^n_{\rm WZ}(x,\theta,\thetabar)=0$ for $n=3,4,5,\dots$.  This
implies that the Taylor series for the exponential of $V_{\rm WZ}$ is
a finite series and
contains only three terms,
\begin{align}
\exp(2gV_{\rm WZ})=1+2gV_{\rm WZ}+2g^2V^2_{\rm WZ}\,. \label{e2gV}
\end{align}
This result will be especially important when we consider
gauge-invariant interactions in Section~\ref{GI}.

Finally, we consider the implications of $R$-invariance.  
Since $V$ is a real superfield, it follows from \eq{Rtrans} that,
\begin{align}
\widehat{R}V(x,\theta,\thetabar)=V(x,e^{-ia}\theta,e^{ia}\thetabar)\,.
\end{align}
In the Wess-Zumino gauge, the $R$ transformations of the component fields are given by
\begin{align}
V_\mu&\to V_\mu\,, \\
\lambda&\to e^{ia}\lambda\,, \\
D&\to D\,.
\end{align}
The Lagrangian of \eq{WF} for the SUSY gauge theory is invariant under
$R$ transformations.
In the present context, the presence of $R$-invariance is associated with the chiral symmetry of the massless gaugino.

\subsection{Gauge-invariant interactions}
\label{GI}
Suppose that $\Phi$ is a chiral superfield that is charged under the U(1) gauge group.
Then the gauge transformations of the chiral superfield and the corresponding antichiral superfield are given by,
\begin{align}
\Phi\to e^{-2ig\Lambda}\Phi\,,\qquad\quad \Phi^\dagger\to e^{2ig\Lambda^\dagger}\Phi^\dagger\,,\label{eq:Phigauge}
\end{align}
where $\Lambda$ is the chiral superfield gauge
transformation parameter introduced in
\eq{eq:Vgaugetransform}.
In the presence of gauge interactions, the kinetic energy term for the
chiral superfield given by \eq{phiphiD},
\begin{align}
\mathscr{L}_{\rm KE}=[\Phi^\dagger \Phi]_D=\int d^4 \theta\,\Phi^\dagger\Phi\,,
\end{align}
is not gauge invariant.  But this deficiency is easily repaired.  A
gauge-invariant kinetic energy term 
with respect to the gauge transformations given in
\eqs{eq:Vgaugetransform}{eq:Phigauge} is given by,
\begin{align}
\mathscr{L}_{\rm KE}=[\Phi^\dagger e^{2gV}\Phi]_D=\int d^4 \theta\,\Phi^\dagger e^{2gV}\Phi\,.
\label{eq:LKE}
\end{align}
The proof is left as an exercise (see Problem \ref{pr:invKE}).

Normally, the exponential, $\exp(2gV)$, would yield an infinite series
of terms.  But, the series terminates  in the Wess-Zumino gauge, as
indicated in \eq{e2gV}, and we get 
\begin{align}
\begin{split}
\mathscr{L}_{\rm KE}=&(\mathcal{D}_\mu A)(\mathcal{D}^\mu A)^\dagger+
i \psi^\dagger \sigmabar^\mu \mathcal{D}_\mu \psi+F^\dagger  F\\
&
+ig\sqrt{2}(A^\dagger\lambda\psi-A\lambda^\dagger\psi^\dagger)
+gAA^\dagger  D+\text{total derivative}\,,\label{KEint}
\end{split}
\end{align}
where $\mathcal{D}_\mu\equiv\partial_\mu+igV_{\mu}$ is the usual gauge-covariant derivative.
The presence of the Yukawa interaction of the scalar-fermion-gaugino
is especially noteworthy, with a coupling proportional to the gauge
coupling $g$.  This is a consequence of supersymmetry, which relates
the gauge and Yukawa couplings that otherwise would be independent.

Another manifestation of SUSY is revealed when we consider the terms
of the Lagrangian involving the auxiliary fields $F$ and $D$.  
Consider the Lagrangian of the interacting gauge theory that consists
of contributions from \eqs{WF}{KEint}.  We can isolate those terms
that involve $F$ and $D$ explicitly, 
\begin{align}
\mathscr{L} = \biggl\{\tfrac14[\W^\alpha \W_\alpha]_F+{\rm h.c.}\biggr\}+[\Phi^\dagger e^{2gV}\Phi]_D 
 =\ldots+
F^\dagger F+\half D^2+gAA^\dagger D\,.\label{FD}
\end{align}
Solving the Lagrange field equations for $F$ and $D$,
\begin{align}
& \frac{\partial\mathscr{L}}{\partial F}=0\qquad\Longleftrightarrow\qquad F=0\,, \\
& \frac{\partial\mathscr{L}}{\partial
  D}=0\qquad\Longleftrightarrow\qquad D=-gA^\dagger A\,.\label{DAA}
\end{align}
Inserting these results back into \eq{FD} [where the terms not
explicitly given can be found in \eqs{WF}{KEint}] yields the Lagrangian in terms of its physical fields,
\begin{align}
\begin{split}
\mathscr{L} =&-\tfrac14 F_{\mu\nu}F^{\mu\nu}
+i\lambda^\dagger\sigmabar^\mu\partial_\mu\lambda+(\mathcal{D}_\mu A)(\mathcal{D}^\mu A)^\dagger 
+ i \psi^\dagger \sigmabar^\mu \mathcal{D}_\mu \psi \\
& +i\sqrt{2}\,g(A^\dagger\lambda\psi-A\lambda^\dagger\psi^\dagger)
-\half g^2 (A^\dagger A)^2\,.
\end{split}
\end{align}
Thus, a potential for the scalar field $A$ has been generated,
\begin{align}
V_{\rm scalar}=\half g^2 (A^\dagger A)^2\,.
\end{align}

There is one more possible term, called the Fayet-Iliopoulos term\cite{Fayet:1974jb},
that can appear in a renormalizable SUSY U(1) gauge theory Lagrangian,
\begin{align}
\mathscr{L}_{\rm FI}=2\xi[V]_D= \xi D+\text{total divergence}\,.
\end{align}
This modifies the form of $D$ obtained in \eq{DAA},
\begin{align}
D=-gA^\dagger A-\xi\,,
\end{align}
which in turn modifies the scalar potential,
\begin{align}
V_{\rm scalar}=\half\bigl[gA^\dagger A+\xi\bigr]^2\,.
\end{align}
The existence of a quartic scalar coupling proportional to the square of the gauge coupling (in the presence or absence of a Fayet-Iliopoulos term) is another manifestation of SUSY.

\subsection{Generalizing to more than one chiral superfield}

With only one chiral superfield, it was not possible to include a superpotential $W(\Phi)$ in our gauge theory, since 
$W$ is a holomorphic function of a charged field and hence not gauge-invariant.
But, a theory with more than one charged chiral superfield can admit a gauge invariant superpotential.

For example, consider a set of charged chiral superfields $\Phi_i$ with U(1) charges $q_i$, which transform under U(1) as
\begin{align}
\Phi_i\to e^{-2igq_i\Lambda}\Phi_i\,.
\end{align}
Suppose that a gauge-invariant superpotential can be constructed, $W(\Phi_i)$.  When we solve for the auxiliary field $F_i$, we will obtain
\begin{align}
F_i=-\left(\frac{dW}{dA_i}\right)^\dagger\,, \label{fsubi}
\end{align}
as before [cf.~\eq{f}], which provides the $F$-term contributions to the scalar potential,
\begin{align}
V_{\rm scalar}\ni \sum_i \left|\frac{dW}{dA_i}\right|^2\,.\label{niF}
\end{align}

When we solve for the auxiliary field $D$, we obtain a contribution from each scalar $A_i$,
\begin{align}
D=-\xi-\sum_i q_i g A_i^\dagger A_i\,. \label{DFI}
\end{align}
The corresponding $D$-term contributions to the scalar potential are
\begin{align}
V_{\rm scalar}\ni\half\left[\xi+\sum_i gq_iA^\dagger A\right]^2\,.\label{niD}
\end{align}
Including both the $F$-term and $D$-term contributions yields the
following scalar potential,
\begin{align}
V_{\rm scalar}=\sum_i \left|\frac{dW}{dA_i}\right|^2+\half\left[\xi+\sum_i gq_iA^\dagger A\right]^2\,,\label{vscalar1}
\end{align}
which can also be conveniently written as
\begin{align}
V_{\rm scalar}=\sum_i F_i^\dagger F_i+\half D^2\,, \label{vscalar2}
\end{align}
where $F$ and $D$ are given by \eqs{fsubi}{niD}, respectively.
Note that the form of the scalar potential [either \eq{vscalar1} or
(\ref{vscalar2})] makes clear that $V_{\rm scalar}\geq 0$.  This
observation will play an important role in the theory of supersymmetry
breaking, which is treated in Section~\ref{SSB}.

The above results can now be used to construct the supersymmetric
extension of QED.  The superfield content of SUSY-QED consists of a real vector superfield $V$, a chiral superfield $\Phi_+$ with charge $q=1$, and a chiral superfield $\Phi_-$ with charge $q=-1$.  The unique renormalizable, gauge-invariant superpotential is
\beq \label{Wsqed}
W(\Phi_+,\Phi_-)=m\Phi_+\Phi_-\,.
\eeq
The $R$-charges of both $\Phi_+$ and $\Phi_-$ can be chosen to be $+1$, in which case the theory is also $R$-invariant.
The construction of the SUSY-QED Lagrangian is left as an exercise
(see Problem \ref{pr:SUSYQED}).

\subsection{\mbox{SUSY \!Yang-Mills theory coupled to supermatter}}

The construction of the supersymmetric generalization of Yang-Mills
theory, \textit{i.e.},~a non-abelian gauge theory coupled to matter, is more
complicated than the case of an abelian gauge theory treated in
previous sections.  In this subsection, we will summarize the main
modifications.  The reader can fill in the details with the help of Refs.\cite{Bailin,Sohnius:1985qm}.

Consider a non-abelian compact simple Lie group G, with generators $T^a$ that satisfy commutation relations,
\begin{align}
\bigl[T^a\,,\,T^b\bigr]=if_{abc}T^c\,.
\end{align}
It is convenient to normalize the generators of the defining (fundamental)
representation of G such that,
\begin{align}
\Tr(T^a T^b)=\half\delta_{ab}.
\end{align}

The vector superfield, $V^a$,  possesses an adjoint index $a$, which
runs over the generators of G.
Thus, we can define the matrix gauge superfield,
\begin{align}
V\equiv V^a T^a\,.
\end{align}
The gauge transformation law for $V$ given in \eq{eq:Vgaugetransform}
is significantly more complicated in the case of a non-abelian gauge theory, 
\begin{align}
e^{2gV} \longrightarrow e^{-2ig\Lambda^\dagger}e^{2gV} e^{2ig\Lambda},\label{eq:Vnonabelian}
\end{align}
where $\Lambda\equiv (\Lambda^a T^a)_{ij}$ is the matrix chiral superfield gauge
transformation parameter.

The chiral superfields are now multiplets corresponding to representation $R$
of the gauge group G, transforming as\footnote{When acting on the $\Phi_i$, one employs the generators $T^a$ in the representation $R$.}
\begin{align}
\Phi_i\to \left(e^{-2ig\Lambda}\right)_{ij}\Phi_j\,,
\end{align}
which provides the
generalization of \eq{eq:Phigauge} to a nonabelian gauge group.
Note that $\Phi^\dagger_i\left(e^{2gV}\right)_{ij}\Phi_j$ is
gauge-invariant, if the gauge transformation law for $V$ is given by \eq{eq:Vnonabelian}.

Likewise, we define a matrix version of 
the nonabelian field-strength superfield, $\W_\alpha\equiv \W_\alpha^a T^a$, where
\begin{align}
\W_\alpha=-\frac{1}{8g}\Dbar\lsup{2}e^{-2gV}D_\alpha e^{2gV}\,.\label{eq:Wnonabelian}
\end{align}
Unlike the abelian case, $\W_\alpha$ is not gauge-invariant.  However it transforms as an adjoint field,
\begin{align}
\W_\alpha\to e^{-2ig\Lambda}\W_\alpha e^{2ig\Lambda}\,,
\end{align}
so that $\Tr(\W^\alpha \W_\alpha)$ is gauge-invariant.
In the WZ gauge,\footnote{In contrast to the abelian case, the
  expansion of  $\W_\alpha^a$ in terms of its component fields in the
  nonabelian case will necessarily contain gauge artifacts.  After imposing
  the WZ gauge condition, the expansion of  $\W_\alpha^a$ in terms of
  its component fields resembles the corresponding expression of 
  SUSY abelian gauge theory [cf.~\eq{Wdef}].}
when expanded in component fields,  $\W_\alpha^a$ depends only on the
physical fields, $\lambda^a$, $F^{\mu\nu a}$ and the auxiliary field $D^a$,
\beq
\W_\alpha^a=-i\lambda_\alpha^a+\theta_\alpha D^a-\half i(\sigma^\mu\sigmabar^\nu\theta)_\alpha F_{\mu\nu}^a-\sigma^\mu(\mathscr{D}_{\mu ab}\lambda^{\dagger b})_\alpha \,\theta\theta\,,\label{Wna}
\eeq
where 
\beq
\mathscr{D}_{\mu ab}\equiv\delta_{ab}\partial_\mu+g f_{abc} V_\mu^c\,,
\eeq
 is the gauge-covariant derivative in the adjoint representation, and 
\beq
F_{\mu\nu}^a=\partial_\mu V_\nu^a-\partial_\nu V_\mu^a-gf_{abc}V_\mu^b V_\nu^c
\eeq
is the nonabelian field strength tensor.

\subsection{The SUSY Lagrangian}
The Lagrangian for SUSY Yang-Mills theory coupled to supermatter is given by
\begin{align}
\begin{split}
\mathscr{L}=  \left[\half\int d^2\theta\,\Tr(\W^\alpha \W_\alpha)+{\rm h.c.}\right]+\int d^4\theta \,
\Phi^\dagger e^{2gV}\Phi 
+\left[\int d^2\theta\,W(\Phi_k)+{\rm h.c.}\right]. \label{SUSYYMLag}
\end{split}
\end{align}
In contrast to the abelian gauge theory,  no Fayet-Iliopoloulos term
is allowed since $[D^a]_D$ carries an adjoint index and thus is not
gauge invariant.  The superpotential $W(\Phi_k)$ is assumed to be a
gauge-invariant holomorphic function of the chiral superfields.  The
chiral superfields $\Phi_k$ taken together transform under a reducible
$d$-dimensional representation $R=\oplus_k R_k$ of the gauge group G,
where $d=\sum_k {\rm dim}~R_k$.
In terms of component fields, \eq{SUSYYMLag} yields
\begin{align}
\begin{split}
\mathscr{L}=&-\tfrac14 F_{\mu\nu}^a F^{\mu\nu a}+i\lambda^{\dagger a}\sigmabar^\mu(\mathscr{D}_\mu\lambda)^a+\half D^a D^a+F_i^\dagger F_i    
+(\mathcal{D}_\mu A)_i(\mathcal{D}^\mu A)_i^\dagger \\
&  +i \psi_i^\dagger
\sigmabar^\mu (\mathcal{D}_\mu \psi)_i +gA_i^\dagger T_{ij}^a A_j D^a   
+ig\sqrt{2}(A_i^\dagger T^a_{ij}\psi_j\lambda^a-\lambda^{\dagger a}\psi^\dagger_i T^a_{ij} A_j)
 \\
&
+F_i\frac{dW}{dA_i}+F^\dagger_i\left(\frac{dW}{dA_i}\right)^\dagger
-\half \frac{d^2 W}{dA_i dA_j}\psi_i\psi_j-
\half \left(\frac{d^2 W}{dA_i dA_j}\right)^\dagger\psi^\dagger_i\psi^\dagger_j  \,,
\end{split}
\label{eq:LSUSYcomponents}
\end{align}
where there is an implicit sum over repeated indices, and 
the labels $i$ and $j$ run over $1,2,\ldots, d$.
The corresponding covariant derivative, when acting
on the component fields $A_i$ and $\psi_i$, is $\mathcal{D}_\mu=\mathds{1}\partial_\mu+igT^a V_\mu^a$,
where $\mathds{1}$ is the $d\times d$ identity matrix and the generators $T^a$ are in the reducible 
representation $R$ of the group G.

Note that the interactions of the matter fermions and the gauginos
with the gauge fields are dictated by gauge invariance (via the
gauge covariant derivative) and do not depend on supersymmetry.
In contrast, the Yukawa interaction of the gaugino with the matter
fermion and its scalar partner (with a coupling proportional to the gauge
coupling $g$) is a consequence of supersymmetry, and relates
the gauge and Yukawa couplings that otherwise would be independent.

We can now eliminate the auxiliary fields $F_i$ and $D^a$ by employing the
Lagrange field equations.  We end up with
\begin{align}
F_i=-\left(\frac{dW}{dA_i}\right)^\dagger\,,\qquad\quad D^a=-gA_i^\dagger T^a_{ij} A_j\,.\label{FandD}
\end{align}
Substituting back into \eq{eq:LSUSYcomponents} yields the following scalar potential,
\begin{align}
V_{\rm scalar}=\sum_i \left|\frac{dW}{dA_i}\right|^2+\half g^2(A_i^\dagger T^a_{ij} A_j)^2\,.\label{vscalar3}
\end{align}
Equivalently, we can write:
\begin{align}
V_{\rm scalar}=\half D^a D^a+\sum_i F_i^\dagger F_i\,.\label{vscalar4}
\end{align}
\Eqs{vscalar3}{vscalar4} provide the nonabelian generalization of
\eqs{vscalar1}{vscalar2}.  As in the abelian case, $V_{\rm scalar}\geq 0$.

If we drop the requirement of renormalizability, then we can generalize the action of a SUSY-Yang Mills theory coupled to supermatter,
\beqa
\mathscr{L} &=&\half\int d^4\theta \bigl[K(e^{2gV}\Phi\,,\, \Phi^\dagger) +
K(\Phi\,,\, \Phi^\dagger e^{2gV})\bigr] 
+\left[\int d^2\theta\,W(\Phi_i)+{\rm h.c.}\right]\nn  \\
&& +\left[\tfrac14\int d^2\theta\,f_{ab}(\Phi) \W^{\alpha
    a}\W^b_\alpha+{\rm h.c.}\right] \label{susykahler}
\,,
\eeqa
where $K$ is the K\"ahler potential and $f_{ab}(\Phi)$ is a holomorphic function of the chiral superfields called the \textit{gauge kinetic function}.
In renormalizable global supersymmetry, the minimal versions of the K\"ahler potential and gauge kinetic function are used:
\begin{align}
K(e^{2gV}\Phi\,,\, \Phi^\dagger) &=
K(\Phi\,,\, \Phi^\dagger e^{2gV})=\Phi^\dagger e^{2gV}\Phi\,,\\
f_{ab}(\Phi) &=\delta_{ab}\,.
\end{align}

The generalization of the SUSY Lagrangian to a theory based on a
gauge group that is a direct product of compact simple Lie groups and
U(1) factors is straightforward.   There is a gauge field strength
tensor and a separate gauge coupling constant
corresponding to each group in the direct product.  Details are left
for the reader.

%%%%%%%%%%%%%%%%%%%%%%%%%%%%%%%%%%%%%%%%%%%%%%%%%%%
\subsection{Problems}
\begin{problem}
Show that $\W_\alpha$ is invariant under the gauge transformation of \eq{eq:Vgaugetransform}.
\end{problem}
\begin{problem}
\label{pr:invKE}
Show that the kinetic energy term given by \eq{eq:LKE} is invariant under  the gauge transformations for $\Phi$ and $\Phi^\dagger$ given in \eq{eq:Phigauge} and $V\to V+i(\Lambda-\Lambda^\dagger)$.
\end{problem}

\begin{problem}
%[Construct SUSY-QED]
\label{pr:SUSYQED}
Construct the full SUSY QED Lagrangian in the Wess-Zumino gauge.  Show that the physical states of the theory consist of a Dirac fermion (the ``electron''), two complex scalar ``selectrons,'' usually denoted by $\widetilde e_L$ and $\widetilde e_R$, a massless photon, and a massless photino.  
Check that the number of bosonic and fermionic degrees of freedom are equal, both off-shell and on-shell.
\end{problem}

\begin{problem}
Consider the SUSY QED theory examined in Problem~\ref{pr:SUSYQED}.
However, this time
do \textit{not} impose the Wess-Zumino gauge condition.  Instead, explore the consequences of adding the following supersymmetric gauge fixing term\cite{Ovrut:1981wa,Miller:1983pg,Dine:2016rxc}, 
\begin{align}
\mathscr{L}_{\rm GF}=-\frac{1}{8\alpha}\bigl[(D^2 V)(\overline{D}\lsup2 V)\bigr]_D\,,
\end{align}
where $\alpha$ is the gauge fixing parameter.
\end{problem}

\begin{problem}
Starting from the case where the gauge group G is nonabelian, show
that the gauge transformation law for the gauge superfield~$V$, as
deduced from \eq{eq:Vnonabelian},  reduces to
$V\rightarrow V+i(\Lambda-\Lambda^\dagger)$ in the abelian limit.
Likewise, show that $\W_\alpha$ as given in \eq{eq:Wnonabelian} reduces to
$\W_\alpha=-\tfrac{1}{4} {\Dbar}^{2} D_\alpha V $ in the abelian limit.
\end{problem}

\begin{problem}
Evaluate the contribution of the K\"ahler potential terms to the
Lagrangian given in \eq{susykahler} in terms of the component fields.
Show that your result reduces to \eq{kahler} in the limit of $g\to 0$. 
\end{problem}

\begin{problem}
Evaluate the contribution of the gauge kinetic function terms
to the
Lagrangian given in \eq{susykahler} in terms of the component fields.
How does your result simplify in the abelian limit?
\end{problem}

\begin{problem}
Starting from \eq{susykahler}, solve for the auxiliary fields $F_i$
and~$D^a$ using the Lagrange field equations.  Using these results, determine the form of the scalar
potential that generalizes the results of \eqs{vscalar3}{vscalar4}.
\end{problem}
%%%%%%%%%%%%%%%%%%%%%%%%%%%%%%%%%%%%%%%%%%%%%%%%%%

%
\section{Supersymmetry Breaking}
\renewcommand{\theequation}{\arabic{section}.\arabic{equation}}
\setcounter{equation}{0}
\label{SSB}

If supersymmetry were an exact symmetry of nature, then particles
and their superpartners, which differ in spin by half a unit, would be
degenerate in mass.  Since superpartners have not (yet) been observed,
supersymmetry must be a broken symmetry.  In light of the
non-observation of supersymmetric particles at the LHC,
the energy scale of supersymmetry breaking must lie above 
1~TeV.  

The fundamental mechanism responsible for supersymmetry  breaking is
presently unknown. In Section~\ref{sec:originsofSUSYbreaking}, we
describe some general considerations related to SUSY breaking,
and we examine several possible frameworks for the spontaneous
breaking of SUSY.  In Section~\ref{sumrule}, we examine constraints
on mass splittings within supermultiplets in the presence of
SUSY-breaking.
The possible origins of SUSY-breaking dynamics is surveyed in Section~\ref{SUSYdynamics}. 
Finally, in Section~\ref{sec:softSUSYbreaking}, we examine a
more agnostic approach, in which the supersymmetry of the effective low energy theory at
the TeV scale is softly broken.  In such an
approach, we identify the possible soft-supersymmetry breaking terms
that can appear in the Lagrangian, without making assumptions about
their fundamental origin.

\subsection{Spontaneous SUSY breaking}
\label{sec:originsofSUSYbreaking}
In Section~\ref{SUSYalg}, we derived \eq{pzero}, which states that the energy operator $P^0$ for a supersymmetric theory is given by
\begin{align}
P^0= \frac{1}{2 t } \of{   Q_1 Q_1^\dagger+Q_1^\dagger Q_1+Q_2 Q_2^\dagger+Q_2^\dagger Q_2 }\,,\label{peezero}
\end{align}
where $t$ is real and positive (conventionally, $t=2$).  Since the right-hand side of \eq{peezero} is
positive semi-definite, it
follows that the vacuum energy is zero if and only if the vacuum is supersymmetric:
\begin{align}
\vev{0\,|P^0\,|\,0}=0\quad\Longleftrightarrow\quad Q_\alpha\ket{0}=0\,.
\end{align}
Moreover, assuming the absence of fermion condensation,\footnote{That is, we assume the absence of a
fermion bilinear covariant, with the properties of a Lorentz scalar,
that acquires a nonzero vacuum expectation value.}
the vacuum energy can be identified as the vacuum expectation value of
the scalar potential.  That is, in the case of a supersymmetric vacuum,
\begin{align}
\vev{0\,|P^0\,|\,0}=0\quad\Longleftrightarrow\quad \vev{0\,|V_{\rm scalar}\,|\,0} = 0\,.
\end{align}
To appreciate the significance of $ \vev{0\,|V_{\rm scalar}\,|\,0} = 0$,
recall \eq{vscalar4}, which we repeat below for the
convenience of the reader,
\begin{align}
V_{\rm scalar}=\half D^a D^a+\sum_i F_i^* F_i\,.\label{DDFF}
\end{align}
It follows that if the vacuum is supersymmetric, then the vacuum expectation 
values of the auxiliary fields must vanish,
\begin{align}
\vev{0\,|F_i\,|\,0}=\vev{0\,|D^a\,|\,0}=0.
\end{align}

One can reach the same conclusion by considering the transformation
laws of the field components of a superfield.
For a chiral superfield, the component fermion field transforms according to,
\begin{align}
\delta_\xi\psi_{\alpha i}=i\bigl[\xi Q+\xi^\dagger Q^\dagger\,,\,\psi_{\alpha i}\bigr]=
- i\sqrt{2}\, (\sigma^\mu \xi^\dagger)_\alpha\> \partial_\mu A_i+\sqrt{2}\,\xi_\alpha F_i\,.
\end{align}
By Lorentz invariance, $\vev{0\,|\partial_\mu A_i\,|\,0}=0$.  Hence,
\begin{align}
\vev{0\,|\bigl[\xi Q+\xi^\dagger Q^\dagger\,,\,\psi_{\alpha i}\bigr]\,|\,0}=\sqrt{2}\,\xi_\alpha\vev{0\,|F_i\,|\,0}\,.
\end{align}
Thus, if $Q_\alpha\ket{0} = 0$ and $Q^\dagger_{\dalpha}\ket{0}=0$, then $\vev{0\,|F_i\,|\,0}=0$.
Likewise, for a real vector superfield, the component gaugino field transforms according to,
\begin{align}
\deltaxi\lambda^a_\alpha=i\bigl[\xi Q+\xi^\dagger Q^\dagger\,,\,\lambda^a_\alpha\bigr]=
 i\xi_\alpha D^a+\half(\sigma^\mu\sigmabar^\nu)_\alpha{}^\beta\xi_\beta F^a_{\mu\nu}\,.
\end{align}
 Since $\vev{0\,|F_{\mu\nu}^a|\,0}=0$ (again, by Lorentz invariance), it follows that
\begin{align}
 \vev{0\,|\bigl[\xi Q+\xi^\dagger Q^\dagger\,,\,\lambda^a_\alpha\bigr]\,|\,0}=i\xi_\alpha\vev{0\,|D^a|\,0}\,.
\end{align}
Thus, if $Q_\alpha\ket{0}=0$ and $Q^\dagger_{\dalpha}\ket{0}=0$, then $\vev{0\,|D^a\,|\,0}=0$.

If at least one of the components of the auxiliary fields $F_i$ or
$D_a$ has a nonzero vacuum expectation value, then SUSY is
spontaneously broken. Mechanisms of spontaneous SUSY breaking fall into two
possible categories:
$F$-type breaking, if $\vev{0\,|F_i\,|\,0}\neq 0$ for some $i$, and
$D$-type breaking if $\vev{0\,|D^a|\,0}\neq 0$ for some $a$.

\subsubsection{The O'Raifeartaigh mechanism ($F$-type breaking)}

 One way to spontaneously break SUSY is to construct a model in which
 it is impossible to simultaneously solve the Lagrange field equations
for all the components of the auxiliary fields, $F_i$.  
This is the O'Raifeartaigh mechanism\cite{ORaifeartaigh:1975nky},\footnote{A well-known
  supersymmetric joke: a graduate student returns to the University
  for the fall semester after spending a month at TASI earlier in the summer.
  The professor says to
  the student, ``Welcome back!  I see that one of the lecture courses you attended
  at TASI was an introduction to supersymmetry.  So, did you learn
  anything useful from
  these lectures?''  The student replies, ``I learned how to spell
  O'Raifeartaigh's name.''}
where the SUSY breaking arises
entirely from a nonzero $F$-term vacuum expectation
value.\footnote{Implicitly, we are assuming here that if the $D$-term
  is present, then $\vev{D^a}=0$.}

Consider the set of equations,
\begin{align}
F_i^\dagger=-\frac{dW}{dA_i}=0\,.\label{fdag}
\end{align}
A solution to these equations corresponds to the existence of 
a choice of the scalar fields, $A_i$, such that all the
equations, $F_i^\dagger=0$, are fulfilled.   Suppose that a solution,
$A_i=v_i$, solves these equations.  In light of \eq{DDFF}, this
solution must correspond to a minimum of the scalar potential, which
we identify as the vacuum (ground) state of the theory.  Since
$F_i^\dagger=0$ implies that $F_i=0$, we can conclude
that  $\vev{0\,|F_i\,|\,0}=0$ (for all $i$).
If no solution to \eq{fdag} exists, then it must be true that
$\vev{0\,|F_i\,|\,0}\neq 0$ for some $i$.  
In this latter case, SUSY must be spontaneously broken.

The simplest O'Raifeartaigh model that exhibits $F$-term SUSY breaking
contains three chiral superfields and
is treated in Problem~\ref{pr:Oraif}.

\subsubsection{$D$-type breaking via the Fayet-Iliopoulos term}

Consider SUSY-QED with a superpotential given by \eq{Wsqed} and a Fayet-Iliopoulos term.   Using
\eqs{fsubi}{DFI}, the resulting scalar potential [\eq{vscalar2}] is given by
\beq \label{FI1}
V_{\rm
  scalar}=|F_+|^2+|F_-|^2+\half D^2\,,
\eeq
where
\beq \label{FI2}
F_\pm=-mA_\pm\,,\qquad\quad D=-g\bigl(|A_+|^2-|A_-|^2\bigr)-\xi.
\eeq
Suppose that $m^2>g\xi$.  One can check that the minimum of the scalar
potential occurs for $\vev{A_+}=\vev{A_-}=0$.   Moreover, at the scalar
potential minimum, $\vev{F_+}=\vev{F_-}=0$, whereas $\vev{D}=-\xi\neq
0$.   Thus, in this model SUSY breaking arises
entirely from a nonzero $D$-term vacuum expectation value.
Additional aspects of this model are treated in Problems~\ref{pr:FI}
and \ref{pr:FI2}.

\subsubsection{The goldstino}
\label{goldstino}

From Goldstone's theorem, we know that the spontaneous breaking of a
continuous symmetry (with bosonic generators) gives rise to a massless
boson called the Nambu-Goldstone boson. Analogously, the spontaneous breaking
of supersymmetry, whose algebra contains fermionic generators, gives rise to a
massless fermion called the Goldstone fermion, which is more
commonly known as the goldstino\cite{Salam:1974zb}.

\begin{theorem}
If SUSY is spontaneously broken, then there exists a massless spin-1/2 fermion in the spectrum called the \textit{goldstino}.
\end{theorem}
\begin{proof}[Proof]
Although this theorem can be proven rigorously, independently of
perturbation theory, it is instructive to exhibit a proof based on a
tree-level analysis of a SUSY nonabelian gauge theory coupled to supermatter.
The scalar potential is given by \eq{DDFF} where [cf.~\eq{FandD}],
\begin{align}
F_i=-\left(\frac{dW}{dA_i}\right)^\dagger\,,\qquad\quad D^a=-gA_i^\dagger T^a_{ij} A_j\,.\label{FiDa}
\end{align}
At the scalar potential minimum, where $\partial V/\partial A_j=0$, the scalar fields are equal to their vacuum expectation values, $A_j=\vev{A_j}$.
Then,
\begin{align}
0=\left(\frac{ \partial V}{\partial A_j}\right)_{\vev{A}}=-gA_i^\dagger T^a_{ij} D^a\biggl|_{\vev{A}}-\sum_i
\frac{ \partial^2 W}{\partial A_i\partial A_j}F_i\biggl|_{\vev{A}}\,.
\end{align}
Hence,
\begin{align}
\sum_i
\left\langle\frac{ \partial^2 W}{\partial A_i\partial A_j}\right\rangle\vev{F_i}=-g\vev{A_i}^\dagger T^a_{ij} \vev{D^a}\,.
\label{eq:FD}
\end{align}

The superpotential $W$ must be a gauge invariant function of the
chiral superfields.  That is,
\beq
W(\Phi)=W(e^{-2ig\Lambda}\Phi)\,.
\eeq
where $\Lambda\equiv \Lambda^a T^a$ is the matrix chiral superfield gauge
transformation parameter.
Taking $\Lambda^a$ infinitesimal and expanding to first order yields
\beq
\frac{dW}{d\Phi_i}T^a_{ij}\Phi_j=0\,.
\eeq
Evaluating the hermitian conjugate of this expression, setting $\theta=\theta^\dagger=0$, and
taking the vacuum expectation value of the resulting equation, we end up with
\beq \label{FTa}
\vev{F_i}T^a_{ji}\vev{A_j}^\dagger=0\,.
\eeq

The fermion masses can be determined from the SUSY Lagrangian given by
\eq{eq:LSUSYcomponents} after setting the scalar fields to their
vacuum expectation values,
\begin{align}
-\mathscr{L}_{\rm mass} =&\half \left\langle\frac{ \partial^2 W}{\partial A_i\partial A_j}\right\rangle \psi_i\psi_j
-i\sqrt{2}\,g\vev{A_i}^\dagger T^a_{ij} \psi_j\lambda^a+{\rm h.c.}
\\
=&\half\bigl(\psi_i\quad -i\lambda^b\bigr)
\begin{pmatrix} \left\langle\displaystyle\frac{ \partial^2 W}{\partial A_i\partial A_j}\right\rangle & \quad
\sqrt{2}\,g\vev{A_j}^\dagger T^a_{ji} \\[25pt] \sqrt{2}\,g\vev{A_i}^\dagger T^b_{ij} & \quad 0\end{pmatrix}
\begin{pmatrix} \psi_j \\[25pt] -i\lambda^b\end{pmatrix}\,.\label{fmatrix}
\end{align}

Using \eqs{eq:FD}{FTa}, one can verify that the fermion mass matrix
given in \eq{fmatrix} possesses a zero eigenvalue,
\begin{align}
\begin{pmatrix} \left\langle\displaystyle\frac{ \partial^2 W}{\partial A_i\partial A_j}\right\rangle & \quad
\sqrt{2}\,g\vev{A_j}^\dagger T^a_{ji} \\[25pt] \sqrt{2}\,g\vev{A_i}^\dagger T^b_{ij} & \quad 0\end{pmatrix}
\begin{pmatrix} \vev{F_j} \\[25pt] \frac{1}{\sqrt{2}}\vev{D^a}\end{pmatrix}=0\,,\label{eigen}
\end{align}
under the assumption that at least one of the auxiliary field vacuum
expectation values is nonzero.
The corresponding eigenvector, $\of{  \vev{F_j},
  \tfrac{1}{\sqrt{2}}\vev{D^a} }$, can be identified with the massless
goldstino, $\widetilde{G}$.  That is,
\begin{align}
\widetilde G=\vev{F_j}\psi_j-\frac{i}{\sqrt{2}}\vev{D^a}\lambda^a\,.
\end{align}

The existence of the goldstino in the fermion mass spectrum is a
consequence of the assumption that the vacuum is not invariant under
SUSY transformations, in which case at least one of the auxiliary field vacuum
expectation values is nonzero, as assumed below \eq{eigen}.
In contrast, if the vacuum is supersymmetric, then $\vev{F_j}=\vev{D^a}=0$, in
which case \eqs{eq:FD}{FTa} are trivially
satisfied.  Hence in this case, one cannot conclude that a zero
eigenvalue of the fermion mass matrix exists.    
\end{proof}

\subsection{Mass Sum rules}
\label{sumrule}
 
If SUSY is broken, then there is no expectation that particles in a
would-be supermultiplet are degenerate in mass.  If the SUSY breaking
is spontaneous, then there is still some memory of supersymmetry
in the properties of the SUSY-broken theory.  In particular, the mass spectrum
of the spontaneously broken SUSY theory satisfies certain sum rules that
reflect the fact the spontaneous breaking of the supersymmetry is inherently soft\cite{Ferrara:1979wa}.

To exhibit such sum rules, we return to the Lagrangian of the SUSY
nonabelian gauge theory coupled to supermatter
given in \eq{eq:LSUSYcomponents}.  We  set the scalar fields and the
auxiliary fields to their vacuum expectation values and compute the
resulting tree-level mass spectrum.

The spin-1 masses arise from
\begin{align}
\mathscr{L}_{\rm mass}= (\mathcal{D}_\mu A)(\mathcal{D}^\mu A)^\dagger,
\end{align}
 where 
$\mathcal{D}_\mu=\partial_\mu+igT^a V^a_\mu$.   It is convenient to write the gauge boson squared-mass matrix as follows,
\begin{align}
(M^2_1)_{ab}=2g^2\vev{A^\dagger_i}T^a_{ij}T^b_{jk}\vev{A_k}=
2\left\langle \frac{\partial D^a}{\partial A_k^\dagger}\frac{\partial D^b}{\partial A_k}\right\rangle\,,
\end{align}
where we have made use of $D^a=-gA_i^\dagger T^a_{ij}A_j$ [cf.~\eq{FandD}].
Likewise, we can rewrite the spin-1/2 mass matrix [previously obtained
in \eq{fmatrix}] as,
\begin{align}
 M_{\scalebox{.8}{$\tfrac12$}}=\begin{pmatrix} \left\langle -\displaystyle\frac{ \partial F_i^\dagger}{\partial A_j}\right\rangle & \quad
-\sqrt{2}\, \displaystyle\left\langle\frac{ \partial D^a}{\partial A_i}\right\rangle\\[25pt]
-\sqrt{2}\, \displaystyle\left\langle\frac{ \partial D^b}{\partial A_j}\right\rangle& \quad 0\end{pmatrix}\,.
\end{align}

The spin-0 masses arise from the scalar potential, $V\equiv V_{\rm scalar}$. Identifying the
terms quadratic in the scalar field,
\beq
-\mathscr{L}_{\rm mass}=\frac12\bigl(A_i\quad A_j^\dagger\bigr)\begin{pmatrix}
  \displaystyle\left\langle{\frac{\partial^2 V}{\partial A_i\partial A_k^\dagger}}\right\rangle\qquad 
 \displaystyle\left\langle{\frac{\partial^2 V}{\partial A_i\partial A_\ell}}\right\rangle \\[15pt]
 \displaystyle\left\langle{\frac{\partial^2 V}{\partial A_j^\dagger\partial A_k^\dagger}}\right\rangle\qquad 
 \displaystyle\left\langle{\frac{\partial^2 V}{\partial A_j^\dagger\partial
     A_\ell}}\right\rangle\end{pmatrix}\begin{pmatrix} A_k^\dagger \\[25pt]
   A_\ell \end{pmatrix}\,.
\eeq
The scalar squared-mass matrix given above will be denoted by $M_0^2$.
 
The elements of the scalar squared-mass matrix can be
rewritten in terms of derivatives of the auxiliary fields $F_i$ and
$D^a$.   For example, noting that 
\eq{FiDa} implies that $F$
is a function of $A^\dagger$ (and likewise, $F^\dagger$ is a
function of $A$), then it follows from \eq{DDFF} that 
\beqa
\frac{\partial^2 V}{\partial A_i\partial
      A_k^\dagger}
%&=& 
%\frac{\partial  F^\dagger_m}{\partial A_i}\frac{\partial
%  F_m}{\partial A^\dagger_k}+\frac{\partial}{\partial
%  A_k^\dagger}\left(D^a\,\frac{\partial D^a}{\partial A_i}\right)
%\nn \\
&=& \frac{\partial  F^\dagger_m}{\partial A_i}\frac{\partial
  F_m}{\partial A^\dagger_k} +\frac{\partial D^a}{\partial
  A_k^\dagger}\frac{\partial D^a}{\partial
  A_i}+D^a\,\frac{\partial^2 D^a}{\partial A_k^\dagger\partial A_i}\,.
\eeqa

One can now evaluate the trace of the various squared-mass
matrices,
\beqa
\Tr M_1^2&=&2\left\langle \frac{\partial D^a}{\partial
    A_k^\dagger}\frac{\partial D^a}{\partial A_k}\right\rangle\,,  \\
\Tr M_{\scalebox{.8}{$\tfrac12$}}^\dagger M_{\scalebox{.8}{$\tfrac12$}}^{\phantom{\dagger}}&=&\left\langle \frac{\partial  F_i}{\partial A_k^\dagger}\frac{\partial
  F_i^\dagger}{\partial A^\dagger_k}\right\rangle +4\left\langle \frac{\partial D^a}{\partial
  A_k^\dagger}\frac{\partial D^a}{\partial A_k}\right\rangle\,,  \\
\Tr M_0^2&=& 2 \left\langle \frac{\partial  F^\dagger_i}{\partial A_k}\frac{\partial
  F_i}{\partial A^\dagger_k}\right\rangle +2\left\langle \frac{\partial D^a}{\partial
  A_k^\dagger}\frac{\partial D^a}{\partial
    A_k}\right\rangle+2\left\langle D^a\frac{\partial^2 D^a}
  {\partial A^\dagger_k \partial A_k}\right\rangle, \nn \\
\phantom{line} \label{trmzero}
\eeqa
where there are implicit sums over each pair of repeated indices.
We can simplify the last term of \eq{trmzero} using $D^a=-gA_i^\dagger
T^a_{ij}A_j$ to obtain.
\begin{align}
\Tr M_0^2=2\left\langle \frac{\partial F_i}{\partial A_k^\dagger}\frac{\partial F_i^\dagger}{\partial A_k}\right\rangle\
+2\left\langle \frac{\partial D^a}{\partial A_k^\dagger}\frac{\partial D^a}{\partial A_k}\right\rangle-2g\vev{D^a}\Tr T^a\,.
\end{align}
It then follows that
\beq \label{masssumrule}
\Tr(M_0^2-2M_{\scalebox{.8}{$\tfrac12$}}+3M_1^2)=-2g\vev{D^a}\Tr
T^a\,.
\eeq

We recognize the left-hand side of \eq{masssumrule} as a
supertrace, which is defined as the following weighted sum of traces,
\beqa
&&\phantom{line}\nn \\[-10pt]
&&\Str M^2\equiv \sum_J (-1)^{2J} (2J+1) \Tr M_J^2\,,\label{stracedef}
\eeqa
where 
$M_J^2$ is the squared-mass matrix of \textit{real} spin-$J$
fields.\footnote{Note that complex fields are equivalent to two
  mass-degenerate real fields.}
Note the $(-1)^{2J}$
factor, so that bosons contribute positively and fermions negatively
to the sum over $J$. 
As applied to a SUSY nonabelian gauge theory coupled to supermatter, the sum is taken over $J=0$, $\half$ and 1. 
Hence, \eq{masssumrule} assumes the following simple form,
\beq \label{eq:supersumrule}
{\rm Str}~M^2=-2g\vev{D^a}\Tr T^a\,.
\eeq

The mass sum rule can 
provide a useful check on the phenomenological viability of theories
with tree-level spontaneous supersymmetry breaking. 
Let us now see how this applies in several cases.

\subsection{The origin of SUSY-breaking dynamics}
\label{SUSYdynamics}
\subsubsection{Models of tree-level spontaneous SUSY breaking}
In the case of $F$-type breaking (\textit{i.e.}, the O'Raifeartaigh model), in which  $\vev{F_i}\neq 0$ and $\vev{D^a}=0$,
\eq{eq:supersumrule} yields 
\beq \label{strzero}
{\rm Str}~M^2=0\,.
\eeq
For example, consider the matter sector of SUSY-QED, which contains two chiral
supermultiplets [cf.~\eq{Wsqed}].  The corresponding spectrum contains a
four-component Dirac electron and its two complex scalar superpartners, the
selectrons (denoted by $\widetilde e_1$ and $\widetilde e_2$).  If SUSY
is spontaneously broken by an $F$-term vacuum expectation value, then \eq{strzero}
yields
\begin{align}
m_{\tilde{e}_1}^2 + m_{\tilde{e}_2}^2 = 2 m_e^2 ,
\end{align}
so that one selectron would be heavier than the electron and the other
selectron would be lighter than the electron.  Clearly, this is very
bad for phenomenology, since experiment demands that all superpartner
masses must be significantly heavier than their SM counterparts.

Consider next $D$-type breaking with $\vev{F_i}=0$ and $\vev{D^a}\neq 0$ in a nonabelian gauge theory.
In this case, $\Tr T^a=0$ and we again conclude that $\Str M^2=0$.
However, it turns out that when the scalar potential is minimized, it
is always possible to find a vacuum in which $\vev{D^a}=0$.  Hence,
$D$-term SUSY-breaking is not possible in this case (see Problem~\ref{pr:holo}).

Finally, consider $D$-type breaking in a gauge theory with a U(1)
factor.  The Standard Model provides an example of this case.  But in the
Standard Model, the hypercharge generator satisfies $\Tr Y=0$ when
summed over one generation of matter.  Hence we again find that $\Str
M^2=0$.  It is possible to construct models of $D$-type SUSY breaking
via the Fayet-Iliopoulos term~$\xi$.  In such models, $\vev{D}$ is
proportional to $\xi$, as shown below \eq{FI2}.  However, no realistic models of this type are known.

Based on the above considerations, we conclude that the mass sum rule
severely constrains tree-level SUSY-breaking models.  Indeed, no
phenomenologically realistic tree-level spontaneously broken SUSY model has ever been
successfully constructed.  

\subsubsection{Gauge-mediated SUSY breaking}
One way to avoid the tyranny of the mass
sum rule is to consider models in which the radiative corrections to
the tree-level masses are significant.  In general, there is no reason
why the radiative corrections should respect the tree-level relations
derived in Section~\ref{sumrule}.   For example, one can construct models with two distinct
sectors of supermatter, which are coupled by the exchange of gauge
bosons.  The particles of the Standard Model (SM) reside in one of the
supermatter sectors, whereas the source of SUSY-breaking (SSB) is located in
the second supermatter sector, whose characteristic mass scale, $M_{\rm
SSB},$ is
assumed to be significantly above 1~TeV.   Indeed, in this second supermatter
sector, the masses of particles and their superpartners are split due
to SUSY-breaking, while respecting the tree-level mass sum rule
obtained in \eq{eq:supersumrule}.   In this case, tree-level SUSY-breaking
is phenomenologically viable in light of the large
characteristic mass scale $M_{\rm SSB}$ that governs the SSB sector.

In such a setup, SUSY is unbroken in
the SM sector at tree level, in which case  $\Str M^2=0$ is trivially
satisfied (see Problem~\ref{pr:exact}).   However, there exist
radiative corrections to the sum rule induced by loops involving the
supermatter of the SSB sector.  These corrections are
responsible for SUSY-breaking in the SM sector
and the corresponding
mass splitting between the SM particles and their superpartners.
Moreover, these mass splittings are totally radiative in nature and not
subject to the tree-level sum rule of \eq{eq:supersumrule}.   Models
can easily be constructed in which the masses of the SM superpartners
are all raised above 1 TeV,  thereby avoiding conflict with the current LHC searches.  
 The end result is SUSY-breaking in the SM that is
phenomenologically viable.

In the scenario outlined above, SUSY-breaking is communicated to the
SM-sector via a messenger mechanism, in which the messengers consists
of gauge bosons that couple both to the SM sector and the
SSB sector.  Models of this type provide examples of
gauge-mediated SUSY breaking (GMSB).
Details of GMSB model building lie beyond the scope of these lectures.
For further details, you may consult Refs.~\cite{Giudice:1998bp,Luty:2005sn,Shirman:2009mt}.

\subsubsection{Local supersymmetry and the super-Higgs mechanism}

Another way of evading the tyranny of the mass sum rule is to
consider models with \textit{local} supersymmetry.

In these lectures, we  have focused on theories with global
supersymmetry, where the anticommuting SUSY translation parameter
$\xi$ is independent of the position $x$.  Suppose we attempt to
generalize this to local supersymmetry, where $\xi=\xi(x)$.  Since
the spinorial SUSY generators satisfy
$\{Q_\alpha\,,\,\overline{Q}_{\dbeta}\}=2\sigma^\mu_{\alpha\dbeta}P_\mu$, 
a theory of local supersymmetry must also be invariant under local
spacetime translations, in which the translation depends on the
position.  A theory that possesses a local spacetime translation
symmetry is a theory of gravity!  Hence, a locally supersymmetric
theory is a theory of gravity plus supersymmetry, \textit{i.e.}, supergravity\cite{sugra1,sugra2}.
 
We have already encountered the massless supermultiplet that contains
the spin-3/2 gravitino and the spin-2 graviton. 
Suppose we couple this supermultiplet to ordinary supermatter.  In
addition, suppose that the local supersymmetry is broken, which will
generate a mass splitting within the graviton supermultiplet.  
We require that the graviton remain massless, while the gravitino acquires
mass.   This can be accomplished via the super Higgs mechanism\cite{Deser:1977uq,Cremmer:1978iv}.

We have seen in Section~\ref{goldstino} that in models of
spontaneously-broken global supersymmetry, the spectrum includes a
massless goldstino.
In models of spontaneously-broken supergravity, the goldstino is ``absorbed''
by the gravitino via the super-Higgs mechanism.  
Initially, a massless gravitino possesses only two helicity states,
$\lambda=\pm\tfrac32$.  In the super-Higgs mechanism, the goldstino
provides $\lambda=\pm\half$ helicity states for a massive gravitino.  
That is, the goldstino is removed from the
physical spectrum and the gravitino acquires a mass
(denoted by $m_{3/2}$).  The gravitino now possesses the four
helicity states, $\lambda=\pm\tfrac32$, $\pm\half$, as expected
for a massive spin-$\tfrac32$ particle.

 In spontaneously broken supergravity, the tree-level mass sum rule
 obtained in \eq{eq:supersumrule} is modified.  For example, if $N$ chiral supermultiplets are minimally coupled to supergravity, then\cite{Cremmer:1982en},
 \begin{align}
 \Str M^2= (N-1)(2m_{3/2}^2-\kappa \vev{D^a D^a})-2g\vev{D^a} T^a\,,
\end{align}
 where
 $\kappa=(8\pi G_N)^{1/2}=(8\pi)^{1/2}M_{\rm PL}^{-1}$.  Typical models of interest have $\vev{D^a}=0$, in which case\cite{Cremmer:1982wb} ,
 \begin{align}
  \Str M^2= 2(N-1)m_{3/2}^2\,.
  \end{align}
  If $m_{3/2}\gsim \mathcal{O}(1~{\rm TeV})$, then one expects the superpartner masses of SM particles to lie in the TeV regime.
  
\subsubsection{Gravity-mediated SUSY-breaking}

Consider again the framework of two distinct sectors of supermatter that
are initially uncoupled.   We identify one of the sectors as the
SM sector where the SM particles and their superpartners reside.
In the second so-called ``hidden''  sector, SUSY is spontaneously
broken.  
  
Supergravity models provide a natural mechanism for
transmitting the SUSY breaking of the hidden sector to the
particle spectrum of the SM sector.  In models of gravity-mediated
SUSY breaking, gravity is the messenger of
supersymmetry breaking\cite{Nilles:1983ge ,Hall:1983iz}.
More precisely, SUSY breaking in the SM sector is mediated by effects of
gravitational strength (suppressed by inverse powers of the Planck mass).
The induced mass splittings between the SM particles and their superpartners
are of $\mathcal{O}(m_{3/2})$, whereas the gravitino couplings are
roughly gravitational in strength.

Under certain theoretical assumptions
on the structure of the K\"ahler potential (the so-called sequestered form
introduced in Ref.\cite{Randall:1998uk}), SUSY breaking is due
entirely to the super-conformal (super-Weyl) anomaly,
which is common to all supergravity models.
This approach is called anomaly-mediated supersymmetry breaking (AMSB).
Indeed, anomaly mediation is more generic than originally conceived,
and provides a ubiquitous source of supersymmetry breaking\cite{DEramo:2012vvz,Harigaya:2014sfa}.

  \subsection{A phenomenological approach: soft SUSY-breaking}
  \label{sec:softSUSYbreaking}
 
 If SUSY-breaking arises due to gauge-mediated SUSY-breaking or
 gravity-mediated SUSY-breaking, then we can formally integrate out
 the SSB sector physics at the mass scale $M_{\rm SSB}$ that
 characterizes the fundamental SUSY-breaking dynamics.  For example, in the case
 of gravity-mediated SUSY breaking, we identify $M_{\rm SSB}=M_{\rm PL}$.   In
 GMSB models, $M_{\rm SSB}$ can be much smaller than $M_{\rm PL}$ but still
 significantly larger than the scale of electroweak symmetry breaking.

 The end result is an effective broken supersymmetric theory whose Lagrangian consists  
 of supersymmetric terms and explicit SUSY-breaking terms.
The explicit SUSY-breaking terms that are present in the effective low-energy
theory (which is valid at energy scales below $M_{\rm SSB}$) are ``soft.''
 The meaning of soft in this context will be explained shortly.
 
 The phenomenological approach to SUSY-breaking takes the point of
 view that the fundamental dynamics of SUSY-breaking is unknown.
 Therefore, we should simply parameterize SUSY breaking in the
 low-energy effective theory
 by including all possible soft-SUSY-breaking terms.  The coefficients
 of these terms will be taken to be arbitrary (to be determined by
 experiment).  Ultimately, these parameters will provide clues to the
 structure of the fundamental dynamics that is responsible for SUSY-breaking.
 
 \subsubsection{A catalog of soft-SUSY-breaking terms}
\label{GGrules}
The most general set of soft-SUSY-breaking terms in a super-Yang Mills theory coupled to supermatter
was first elucidated by Girardello and Grisaru in Ref.\cite{Girardello:1981wz},
\begin{align}
-\mathscr{L}_{\rm soft}=m_{ij}^2 A_i^\dagger A_j+\half\bigl[m_{ab}\lambda^a\lambda^b+{\rm h.c.}\bigr]+\bigl[w(A)+{\rm h.c.}\bigr]\,,\label{GGsoft}
\end{align}
where there is an implicit sum over repeated indices.  The scalar squared-mass matrix $m_{ij}^2$ is hermitian and the gaugino mass matrix $m_{ab}$ is complex symmetric.  The function
$w(A)$ is a holomorphic cubic multinomial of the scalar fields,
\begin{align}
w(A)=c_i A_i+b_{ij}A_i A_j+a_{ijk}A_i A_j A_k\,.
\end{align}
Note that $c_i=0$ in the absence of any gauge singlet fields.  In the
literature, the $b_{ij}$ are called the $B$-terms and the $a_{ijk}$
are called the $A$-terms.  Note the corresponding mass dimensions, $[b_{ij}]=2$ and $[a_{ijk}]=1$.

Dimension-4 terms are not included in \eq{GGsoft}, since non-supersymmetric
dimension-4 terms would constitute a hard breaking of supersymmetry\cite{Martin:1999hc}.
One interesting feature of \eq{GGsoft} is the absence of
non-supersymmetric fermion mass terms, $m_{ij}\psi_i\psi_j+{\rm
  h.c.}$, and non-holomorphic cubic terms in the scalar fields (e.g.,
$A_i A_j A_k^\dagger$, etc.).  Although such terms are technically soft 
in models with no gauge
singlets\cite{Hall:1990ac,Jack:1999ud,Un:2014afa,Chattopadhyay:2016ivr,Ross:2016pml},
theses terms rarely arise in
actual models of fundamental SUSY-breaking, or if present are highly
suppressed\cite{Martin:1999hc}.  Henceforth, we shall
neglect them.

In general, there is no relation between $w(A)$ and the
superpotential, which under the assumption of renormalizability has
the following generic form, 
\begin{align}
W(\Phi)=\kappa_i\Phi_i+\mu_{ij}\Phi_i\Phi_j+\lambda_{ijk}\Phi_i\Phi_j\Phi_k\,.
\end{align}
But, some models of fundamental SUSY breaking yield the relations,
\begin{align}
c_i=C\kappa_i\,, \qquad\quad b_{ij}=B\mu_{ij}\,,\qquad\quad a_{ijk}=A\lambda_{ijk}\,,
\end{align}
which relate the coefficients of $w(A)$ to the coefficients of $W(\Phi)$.

 \subsubsection{Soft vs.~hard SUSY breaking and the reappearance of quadratic divergences}
 
 Consider the one-loop effective potential for a gauge theory coupled to matter,
\begin{align}
 V_{\rm eff}(A)=V_{\rm scalar}(A)+V^{(1)}(A)\,.
 \end{align}
 If we regulate the divergence of the one-loop correction by a
 momentum cutoff $\Lambda$, then\cite{HaberTASI}
\begin{align}
 V^{(1)}(A)=\frac{\Lambda^2}{32\pi^2}\Str
  M_i^2(A)+\frac{1}{64\pi^2}\Str\left\{M_i^4(A)\left[\ln\left(\frac{M_i^2(A)}{\Lambda^2}\right)-\frac12\right]\right\}\,,\label{effpot} 
 \end{align}
 where $M_i^2(A)$ are the relevant squared-mass matrices for spin 0,
 $\half$ and 1, in which the scalar vacuum expectation values are
 replaced by the corresponding scalar fields, $A$.

\Eq{effpot} implies that both in supersymmetric theories and in the case
of spontaneously broken SUSY  (assuming
 in the latter that all U(1) generators are traceless), we have
 $\Str M^2=0$, in which case the quadratic divergences [i.e., the
 terms proportional to $\Lambda^2$ in \eq{effpot}] cancel exactly!
 In Ref.\cite{Girardello:1981wz}, Girardello and Grisaru showed that if
 explicit SUSY breaking terms are present, then there 
is a catalog of possible explicit SUSY-breaking terms for which
$\Str M_i^2(A)$ is a constant \textit{independent} of the scalar
fields, $A$.  Such terms shift the vacuum energy, but in the context
of quantum field theory they have no observable effect.  Terms with
such properties are deemed ``soft,'' and are given in
\eq{GGsoft}.\footnote{Non-holomorphic cubic terms and mass terms of
  fermions that reside in a chiral supermultiplet can generate
  quadratically divergent terms in $V^{(1)}$ that are linear in the
  scalar fields, $A$.  However, if no gauge singlet fields exist in the
  model, then terms that are linear in $A$ are absent due to gauge invariance.} 
In contrast, hard SUSY-breaking terms will generate quadratically
divergent terms in $V^{(1)}$ that are scalar-field-dependent.
This is a signal that some of the parameters of the low-energy effective theory
are quadratically sensitive to UV physics.

\subsubsection{Soft SUSY-breaking: an effective theory perspective}
  
Consider a set of light chiral superfields $\Phi$ and a set of heavy
chiral superfields $\Omega$ associated with a mass scale $M\equiv
M_{\rm SSB}$.   Furthermore, assume that SUSY-breaking is generated by
an $F$-term that resides in the SSB sector,
\begin{align}
 \vev{F_\Omega}=f\neq 0\,.
 \end{align}
One can integrate out the physics of the SSB sector, as shown in the
following examples\cite{Girardello:1981wz,Pomarol:1995np,Rattazzi:1995tc}.

\begin{example}
Consider a holomorphic cubic multinomial of chiral superfields $\Phi$,
which we denote by $\widetilde{w}(\Phi)$.
A possible term in the effective Lagrangian is
\begin{align}
 \frac{1}{M}\int d^2\theta\, \Omega\, \widetilde{w}(\Phi)\,,\label{Ow}
 \end{align}
 since $\Omega\,  \widetilde{w}(\Phi)$ is a term in the
 superpotential. 
The factor of $M^{-1}$ appears on the basis of dimensional analysis.
In particular, note the mass dimensions, $[\widetilde{w}]=3$, $[\Omega]$=1 and $[\int d^2\theta]=1$.  
 
 Since the vacuum expectation value of ${F_\Omega}$, denoted by
$\vev{F_\Omega}=f$, is nonzero, it follows that $\vev{\Omega}\ni \theta\theta f$.  Inserting this into \eq{Ow} yields,
 \begin{align}
  \frac{1}{M}\int d^2\theta\, \theta\theta f\, \widetilde{w}(\Phi)=\frac{f}{M}\,\widetilde{w}(A)\,,
  \end{align}
  which produces the term,  $w(A)=(f/M)\widetilde{w}(A)$, in our
  catalog of $\delta\mathscr{L}_{\rm soft}$ given in \eq{GGsoft}.
  
In order to achieve soft-SUSY-breaking masses in the low-energy
effective theory of order 1~TeV, one must require that
$f/M\sim\mathcal{O}(1~{\rm TeV})$.  For example, in gravity-mediated SUSY breaking, $M\sim M_{\rm PL}$, in which case $f\sim (10^{11}~{\rm GeV})^2$.  Note that $f^{1/2}$ identifies the energy scale of the fundamental SUSY breaking.
 \end{example}
\begin{example}
Another possible term in the effective Lagrangian is
\begin{align}
 \frac{1}{M^2}\int d^4\theta\,\Phi_i^\dagger \left(e^{2gV}\right)_{ij}\Phi_j\,\Omega^\dagger \Omega\,,
 \end{align}
which would contribute to the K\"ahler potential.  Setting $\vev{\Omega}=\theta\theta f$ and evaluating the result in the Wess-Zumino gauge,
 \begin{align}
 \frac{f^2}{M^2}\int d^4\theta\,(\theta\theta)(\thetabar\thetabar)\Phi_i^\dagger \left(e^{2gV}\right)_{ij}\Phi_j=\frac{f^2}{M^2}A^\dagger A\,.
 \end{align}
Thus, the low-energy effective theory contains a scalar squared-mass
term of order $f/M$, which we again recognize as one of the
soft-SUSY-breaking terms of \eq{GGsoft}.

 \end{example}
 \begin{example}
 Finally, one additional possible term in the effective Lagrangian is
 \begin{align}
  \frac{1}{M}\int d^2\theta\,\Omega\Tr(W^\alpha W_\alpha)\,,
  \end{align}
  which would contribute to the gauge kinetic function.   Setting $\vev{\Omega}=\theta\theta f$,
  \begin{align}
  \frac{f}{M}\int d^2\theta\,\theta\theta \Tr(W^\alpha W_\alpha)=-\frac{f}{M}\Tr(\lambda^\alpha\lambda_\alpha)\,,
  \end{align}
 which yields a gaugino mass term of order $f/M$.
 \end{example}

We have thus demonstrated how the possible soft-SUSY-breaking terms of
\eq{GGsoft} can arise in the low-energy effective theory after
integrating out the physics associated with the SSB sector.

 %%%%%%%%

 \subsection{Problems}
 
\begin{problem}
\label{pr:Oraif}
An O'Raifeartaigh model that exhibits $F$-term SUSY breaking
must involve at least three chiral superfields\cite{ORaifeartaigh:1975nky}. 
One of the simplest
models of this type has the following superpotential,
\beq
W(\Phi_1,\Phi_2,\Phi_3)=\lambda\Phi_1(\Phi_3^2-m^2)+\mu\Phi_2\Phi_3\,,
\eeq
where $\lambda$ is dimensionless and $\mu$ and $m$ are mass parameters.
Evaluate the corresponding $F$-terms, $F_1$, $F_2$ and $F_3$ and write
out the scalar potential, $V_{\rm scalar}$.  Show that no solution for
the scalar fields $A_1$ $A_2$ and $A_3$ exist such that
$F_1=F_2=F_3=0$.  Conclude that SUSY is spontaneously broken.  
\end{problem}

\begin{problem}
Find the minimum of $V_{\rm scalar}$ obtained in Problem~\ref{pr:Oraif}, and verify that $\langle
0|V_{\rm scalar}|0\rangle > 0$. Identify the goldstino of this model.
Finally, compute the mass spectrum of the fermions and bosons and
verify that the mass sum rule, \eq{strzero}, is satisfied.
\end{problem}

\begin{problem}
\label{pr:FI}
Show that in the case of SUSY-QED with a Fayet-Iliopoulos term and
$m^2>g\xi$ [cf.~\eqs{FI1}{FI2}], SUSY is broken and the goldstino can be identified as the
photino (the supersymmetric partner of the photon). 
In the case of $m^2<g\xi$, is SUSY broken?  Is the U(1) gauge symmetry broken?
\end{problem}

\begin{problem}
\label{pr:FI2}
Referring back to Problem~\ref{pr:FI}, determine the masses of the
electron and its scalar partners and the masses of the photon and
photino in the two cases of $m^2>g\xi$ and
$m^2<g\xi$, respectively.  Evaluate ${\rm Str}~M^2$ in both cases, and
compare with \eq{eq:supersumrule}.
\end{problem}

\begin{problem}
\label{pr:exact}
 Show that the sum rule of \eq{eq:supersumrule} is valid in the limit of exact SUSY, \textit{i.e.}, when the masses of bosons and fermions are equal.
 \end{problem}
\clearpage

\begin{problem}
\label{pr:holo}
Show that in a SUSY nonabelian gauge theory that is coupled to supermatter,
only $F$-type SUSY breaking is allowed.  To prove this statement,
assume that a solution to $\vev{F_i}=0$ exists and show that one can always find
a choice of scalar fields $A_i$ that provide a solution to
\eq{fdag} such that $\vev{D^a}=0$ for all $a$. 

{\sl HINT}: If the $A_i$ provide a solution to
\eq{fdag}, then so do the corresponding gauge transformed scalar
fields, $(e^{-2ig\Lambda})_{ij}A_j$.   The key
observation is that the superpotential is a holomorphic function of
the scalar fields $A_i$.  Hence, one can generate additional
solutions to \eq{fdag} by taking $g$ complex, which 
will modify $\vev{D^a}$.   Conclude that there must then be a set of
$A_i$ such that $\vev{F_i}=\vev{D^a}=0$.   See Ref.\cite{WessBagger}
for further details.
\end{problem}

\section{\hbox{Supersymmetric extension of the Standard Model (MSSM)}}
\label{sec:MSSM}
\renewcommand{\theequation}{\arabic{section}.\arabic{equation}}
\setcounter{equation}{0}

With the necessary SUSY technology now in hand, we are ready to study its realization in extensions to the SM. 
In this section, we describe the minimal supersymmetric extension of
the Standard Model (MSSM).  Much of the presentation of this section
follows Ref.\cite{susy}, where many of the relevant references to the
original literature can be found.

In Section~\ref{sec:MSSMfields}, we begin by presenting the MSSM
field content. We then  specify the
SU(3)$\times$SU(2)$\times$U(1) gauge-invariant superpotential for the
chiral superfields in Section~\ref{sec:MSSMW}.  Given the superfield formalism developed in
Sections~\ref{sec:superspace} and \ref{sec:gaugetheories}, all the
supersymmetric interactions of the theory are now determined.
At this stage, the supersymmetry is still an exact symmetry.

We introduce SUSY breaking in the MSSM in Section~\ref{sec:MSSMSSB}.
Since the fundamental origin of SUSY-breaking is 
unknown, we parametrize the SUSY-breaking by adding all possible
soft-SUSY-breaking terms consistent with the SU(3)$\times$SU(2)$\times$U(1)
gauge symmetry and a discrete $B-L$ symmetry.   In
Section~\ref{sec:count}, we count the number of
parameters that govern the MSSM.
The resulting MSSM particle spectrum and Higgs boson spectrum are
exhibited in Sections~\ref{sec:MSSMspectrum} and \ref{higgssector}, respectively.
Finally, in Section~\ref{sec:MSSMGU}, we demonstrate the unification of
gauge couplings in the MSSM. 

As in the SM, the neutrinos of the MSSM are massless.
To incorporate massive neutrinos, one can introduce 
right-handed neutrinos and employ the seesaw mechanism. It is then
a simple matter to extend the MSSM by adding a SM singlet superfield
that contains a right-handed neutrino and the corresponding sneutrino
superpartner.  We shall not present this construction in these
lectures; for further details, see e.g.~Ref.\cite{Dedes:2007ef}.

\subsection{Field content of the MSSM}
\label{sec:MSSMfields}

\subsubsection{MSSM superfields and their component fields}
The minimal supersymmetric extension of the Standard Model (MSSM)
contains the fields of the
two-Higgs-doublet extension of the SM
and their corresponding superpartners.
The gauge fields and their superpartners are contained in real vector supermultiplets.
These gauge supermultiplets consist of the 
SU(3)$\times$SU(2)$\times$U(1) gauge bosons and their
gaugino fermionic superpartners.
The matter fields and their superpartners reside in chiral supermultiplets.
The three generations of quark and lepton supermultiplets
consist of left-handed
quarks and leptons and
their scalar superpartners (squarks and sleptons),
and the corresponding antiparticles.  
The Higgs supermultiplets
consist of two complex Higgs doublets, their
higgsino fermionic superpartners, and the
corresponding antiparticles.
The MSSM fields and their gauge quantum
numbers are shown in Table~\ref{tab:MSSMcontent}. 
\vskip -0.05in
\begin{table}[h!]
\caption{\small
The fields of the MSSM and their
SU(3)$\times$SU(2)$\times$U(1) quantum numbers are listed.
The electric charge is given in terms of the third component of
the weak isospin $T_3$ and U(1) hypercharge $Y$ by
$Q=T_3+\half Y$.
For simplicity, only one generation of quarks and leptons is exhibited.
%For each lepton, quark, and Higgs supermultiplet,
%there is a corresponding antiparticle multiplet of charge-conjugated
%fermions and their associated scalar partners.   
The left-handed charge-conjugated quark and lepton fields are denoted
by a superscript $c$.  In particular, $f^c_L\equiv
P_Lf^c=P_LC\bar{f}\lsup{\,\T}=C\bar{f}_R\lsup{\,\T}$, following
the notation of Ref.\cite{Langacker:1980js}, where $f$ is a
four-component fermion field.
The $L$ and $R$ subscripts
of the squark and slepton fields indicate the chirality of the
corresponding fermionic superpartners.
 \label{tab:MSSMcontent} }
\vskip 0.1in
\begin{tabular}{|c|c|c|c|c|c|c|} \hline
\multicolumn{7}{|c|}{Field content of the MSSM} \\ \hline
Super- & Super- & Bosonic & Fermionic &  &  &  \\
multiplets & field & fields & partners &
SU(3) & SU(2) & U(1) \\ \hline
gluon/gluino & $\wh V_8$ &  $g$ &    $\widetilde g$ &       8 & 1&  $\phm 0$ \\
gauge boson/  & $\wh V$ & $W^\pm\,,\,W^0$ & $\widetilde W^\pm\,,\widetilde W^0$ & 1 & 3 & $\phm 0$ \\
gaugino & $\wh V^\prime$ & $B$ &      $\widetilde B$ &   1 &1 & $\phm 0$ \\ \hline
slepton/ & $\wh L$ &$(\widetilde\nu_L, \widetilde e^-_L)$  & $(\nu,e^-)_L$ & 1 & 2 & $-1$ \\
lepton   & $\wh E^c$ & $\tilde e^+_R$   & $e_L^c$      & 1 & 1 & $\phm
                                                                 2$ \\ \hline
squark/ & $\wh Q$ & $(\widetilde u_L,\widetilde d_L)$ & $(u,d)_L$ & 3 &2 & $\phm 1/3$ \\
quark   & $\wh U^c$ & $\widetilde u_R^*$  & $u_L^c$ & $\bar{3}$ &1 & $-4/3$ \\
        & $\wh D^c$ & $\widetilde d_R^*$ & $d_L^c$ & $\bar{3}$ & 1 &  $\phm 2/3$ \\ \hline
Higgs boson/  & $\wh H_d$ & $(H^0_d\,,\,H_d^-)$ & $(\widetilde H^0_d,\widetilde H^-_d)$ & 1 & 2 & $-1$ \\ 
higgsino & $\wh H_u$ & $(H^+_u\,,\,H^0_u)$ & $(\widetilde H^+_u,\widetilde H^0_u)$ & 1 & 2 & $\phm 1$ \\ \hline
 \end{tabular}
 \end{table}

Table~\ref{tab:MSSMcontent} shows that one Higgs doublet superfield has hypercharge $-1$, and the other has hypercharge $+1$.
The distinction between hypercharge $\pm 1$ is irrelevant in a
 non-supersymmetric quantum field theory, where complex scalar fields are
 always accompanied by their hermitian conjugates.  However, in
supersymmetric models the distinction is important, because 
the corresponding Higgs superfields are used to construct the
superpotential.  Since the superpotential 
must be holomorphic, \textit{i.e.}~depend only on chiral superfields and not
their hermitian conjugates, it is important to keep track of the
quantum numbers of the chiral superfields of the model.

\subsubsection{Anomaly cancellation and the second Higgs doublet}
\label{sec:ac}

The enlarged Higgs sector of the MSSM constitutes the minimal structure
needed to guarantee the cancellation of gauge
anomalies generated by the 
higgsino superpartners that can appear as internal lines in one-loop triangle diagrams with
three external electroweak gauge bosons.

Potentially problematic anomalies arise from 
one-loop $VVA$ and $AAA$ triangle diagrams with three external gauge bosons, and fermions running around the loop [where $V$ refers to a $\gamma_\mu$ (vector) vertex and $A$ refers to a $\gamma_\mu\gamma\ls{5}$ (axial vector) vertex].
An anomalous theory violates unitarity and fails as a consistent quantum field theory.
Thus, we need to make sure all gauge anomalies cancel when summed over
all triangle diagrams with fixed external gauge fields\cite{anomalies}. 

The anomalies will cancel if 
certain group theoretical constraints are satisfied.  
In particular, the trace of the product of the relevant generators appearing at
the external vertices must vanish,
\begin{align}
&
W^i W^j B~\text{triangle} \qquad\Longleftrightarrow \qquad \Tr(T_3^2 Y)=0\,,\nn \\
&
BBB~\text{triangle}  \,\,\quad\qquad\Longleftrightarrow \qquad\quad\!\! \Tr(Y^3)=0\,.\nn
\end{align}
In the Standard Model, the fermion contributions  to 
 $\Tr(Y^3)$ sum to zero:
\begin{align}
\Tr(Y^3)_{\rm SM}=3\left(\tfrac{1}{27}+\tfrac{1}{27}-\tfrac{64}{27}+\tfrac{8}{27}\right)-1-1+8=0\,.
\end{align}
In contrast, in the MSSM, 
if we only add the higgsinos $(\widetilde{H}_u^+ \,,\,\widetilde{H}_u^0)$,  the resulting
anomaly factor is
$
\Tr (Y^3)=\Tr(Y^3)_{\rm SM}+2,
$
leading to a gauge anomaly.  To cancel this, we must  add a second higgsino doublet with opposite hypercharge, $(\widetilde{H}_d^0 \,,\,\widetilde{H}_d^-)$.

There is an independent argument for requiring the second Higgs
doublet in the MSSM.
With only one Higgs doublet, one cannot
generate mass for both ``up''-type and ``down''-type
quarks (and charged leptons)
in a way that is consistent with a holomorphic superpotential.

\subsubsection{Suppressed baryon and lepton number violation}
\label{sec:bml}

It is an experimental fact that baryon number $B$ and lepton number
$L$ are, to a very good approximation, global symmetries of nature.
If neutrinos are Majorana fermions, then $L$-violation is present but strongly
suppressed, with neutrino masses of order $v^2/M$, where $v$ is the
scale of electroweak symmetry breaking and $M\gg v$.  No $B$-violation
has yet been experimentally observed.  Moreover, the
current bounds on the
proton lifetime suggest that the mass scale associated with baryon
number violation cannot be below about $10^{16}$~GeV, which is a
characteristic scale of grand unification.
 
One of the remarkable features of the SM is that the suppression of  $B$ and $L$-violating
processes is a natural feature of the model.
That is, the SM Lagrangian possesses an accidental
global \hbox{$B\!\!-\!\!L$} symmetry due to the fact that 
all renormalizable terms of the Lagrangian (with dimension four or less) 
that can be composed of SM fields preserve the $B$ and $L$
global symmetries.  Indeed, $B$ and
$L$-violating operators composed of SM fields must have
dimension $d=5$ or
larger\cite{Weinberg:1979sa,Wilczek:1979hc,Weldon:1980gi}.

For example, consider the dimension-five $L$-violating operator,
\beq \label{L5}
\mathscr{L}_5=-\frac{f_{mn}}{M}(\epsilon^{ij}L_i^mH_j)(\epsilon^{k\ell}L_k^n
H_\ell)+{\rm h.c.}\,,
\eeq
where $f$ is a coefficient that depends on the lepton generation (labeled by
$m$ and $n$), $H_j$ is the complex Higgs doublet field and $L_i^a\equiv
(\nu_L^a\,,\,\ell_L^a)$ is the doublet 
of two-component lepton fields. 
After electroweak symmetry breaking, the neutral component
of the doublet Higgs
field acquires a vacuum expectation value, and a Majorana mass
matrix for the neutrinos is generated. The dimension-five term given
by \eq{L5}
is generated by new physics beyond
the SM at the scale $M$.  Likewise, one can construct dimension-six 
$B$-violating operators composed of SM fields that allow, e.g.,
for proton decay, which is suppressed by  $v^2/M_{\rm G}^2$.  Such 
terms can be generated, e.g., in grand unified theories with a
characteristic mass scale $M_{\rm G}$.
In general, $B$ and
$L$-violating effects are suppressed by $(v/M)^{d-4}$, where
$M$ is the characteristic mass scale of the physics that generates the
corresponding higher dimensional operator (of dimension $d$).  
%Indeed, values of $M$ of
%order the grand unification scale or larger yield the observed
%(approximate) stability of the proton and suppression of
%neutrino masses.

Unfortunately, the suppression of $B$ and $L$-violation is not guaranteed in a generic
supersymmetric extension of the Standard Model.  For example, it is
possible to construct gauge invariant supersymmetric dimension-four
$B$ and $L$-violating operators made up of fields of SM
particles and their superpartners.  Such operators, if present in the
theory, would yield a proton decay rate many orders of magnitude
larger than the current experimental bound.  
To avoid this catastrophic prediction, one can
introduce an additional symmetry
in the supersymmetric theory that will eliminate the $B$ and
$L$-violating operators of dimension
$d\leq 4$.  Further details are provided in the next subsection.
Nevertheless, one must admit that the SM provides a more satisfying
explanation for approximate $B$ and $L$ conservation than does its
supersymmetric extension.

\subsection{The superpotential of the MSSM}
\label{sec:MSSMW}
Given the chiral and gauge superfield content of the MSSM, we must now specify the superpotential. The most general SU(3)$\times$SU(2)$\times$U(1) gauge-invariant superpotential (omitting the right-handed neutrino superfield) is
\begin{align}
\begin{split}
W = &\ (h_u)_{mn} \widehat{Q}_m\newcdot \widehat{H}_u\, \widehat{U}_n^c + (h_d)_{mn} \widehat{H}_d\newcdot\widehat{Q}_m\, \widehat{D}_n^c \\
& + (h_e)_{mn} \widehat{H}_d \newcdot\widehat{L}_m\,\widehat{E}_n^c   +\mu \widehat{H}_u\newcdot \widehat{H}_d\,+\,W_{\rm RPV},\label{MSSMsuperpot}
\end{split}
\end{align}
where $m$ and $n$ label the generations.  That is, $h_u$, $h_d$ and $h_e$ are $3\times 3$ matrix Yukawa couplings.  Note that color indices have been suppressed, and we
 employ a dot product notation for the singlet combination of two SU(2) doublets.  For example,
\beq
 \widehat{H}_u\newcdot \widehat{H}_d \equiv \epsilon^{ij}\widehat{H}_{u\,\!i} \widehat{H}_{d\,\!j}
 =\widehat{H}_u^+ \widehat{H}_d^--\widehat{H}_u^0 \widehat{H}_d^0\,.
 \eeq
The so-called $\mu$-term above is the supersymmetric analog
of the Higgs boson squared-mass term of the SM. 

In addition to the supersymmetric generalization of the SM Yukawa
couplings and the $\mu$-term,  
the gauge symmetries of the superpotential also allow for a number of new terms that violate $B-L$ conservation.
%there are additional 
%gauge-invariant terms that can appear in the superpotential
%that violate $B-L$ conservation. 
As discussed in Section~\ref{sec:bml},
this is in contrast to the SM where there are no $B$ or
$L$-violating interactions at the renormalizable level.
% the
%possibility of such interactions must be considered in the MSSM, as
%discussed in Section~\ref{sec:bml}.
The $B-L$ violating terms of the supersymmetric model arise due to the presence of $W_{\rm RPV}$ in \eq{MSSMsuperpot} and are
given by,
\begin{align}
\begin{split}
W_{\rm RPV}=&\ 
(\lambda_L)_{pmn} \widehat L_p \widehat L_m \widehat E^c_n
+ (\lambda_L^\prime)_{pmn}\widehat L_p \widehat Q_m\widehat D^c_n  \\
& +(\lambda_B)_{pmn}\widehat U^c_p \widehat D^c_m \widehat D^c_n
+(\mu_L)_p \widehat H_u\widehat L_p\,.
\end{split}
\end{align}
Note that the term 
proportional to $\lambda_B$ violates $B$, while the other three terms
violate $L$.  
The $L$-violating term proportional to $\mu_L$ is the generalization of the
$\mu \widehat H_u\widehat H_d$ term,
in which the $Y=-1$ Higgs supermultiplet $\widehat H_d$ is replaced
by the lepton supermultiplet $\widehat L_p$.  Indeed, if $L$ violation
is present, then there is no distinction between $\wh{L}$ and $\wh{H}_d$, since the gauge quantum numbers of these two superfields are identical.

If all terms in $W_{\rm RPV}$ were allowed, the resulting model would predict 
a proton decay rate many orders of magnitude larger than the current
experimental bound.
This can be avoided by imposing an appropriate discrete symmetry that
would eliminate the undesirable terms in $W$.
%To avoid a phenomenological disaster, one can ban undesirable terms
%from $W$ by imposing an appropriate discrete symmetry.  
%

The standard choice in constructing the MSSM is to set $W_{RPV}=0$.
There are a number of ways to accomplish this.  First, one
one could directly impose a $B-L$ symmetry.
Alternatively, one can set $W_{RPV}=0$ by introducing a matter parity, under which $\wh Q$, $\wh U^c$, $\wh D^c$, $\wh L$ and $\wh E^c$ are odd, and $\wh H_u$ and $\wh H_d$ are even. 
Finally, a third option is to impose an $R$-invariant superpotential.  As discussed in Section~\ref{Rinvariance},
$W$ is $R$-invariant if the $R$ charges of the chiral superfields are
chosen such that $R(W)=2$. Thus, if we choose $R$ charges of $+\half$
for $\wh Q$, $\wh U^c$, $\wh D^c$, $\wh L$, $\wh E^c$ and $R$ charges
of $+1$ for $\wh H_u$, $\wh H_d$, then the condition of $R$-invariance
sets $W_{\rm RPV}=0$.

One has to make sure that whichever symmetry one chooses to set
$W_{\rm RPV}=0$ is also consistent with the soft-SUSY-breaking terms
that are subsequently added to the model.  In particular, in the case of the $R$-invariance, recall that $R(\lambda)=1$, which forbids the gaugino mass term,
\begin{align}
m_\lambda(\lambda\lambda+\lambda^\dagger\lambda^\dagger).\label{eq:gauginomass}
\end{align}
But phenomenology requires massive gauginos. This motivates the use of $R$-parity, described in the following subsection, rather than $R$-invariance.

\subsubsection{$R$-parity}
The  gaugino mass term in \eq{eq:gauginomass}
is an allowed soft-SUSY-breaking term.
If this term is added 
to a theory with an $R$-invariant superpotential, then
the continuous U(1)$_R$ symmetry is broken down to a discrete $\mathbb{Z}_2$ symmetry,
called {$R$-parity}\cite{Fayet:1976et,Farrar:1978xj}.  One can check that the $R$-parity of a particle with baryon number $B$, lepton number $L$ and spin $S$ is given by
\begin{align}
R=(-1)^{3(B-L)+2S}\,.\label{Rparity}
\end{align}
It is sufficient to impose $R$-parity invariance in order to set $W_{\rm
  RPV}=0$,\footnote{The effects of imposing matter parity and
  $R$-parity in the MSSM are identical for all
  renormalizable interactions.} 
which is equivalent to imposing the $B-L$ discrete symmetry.
For the remainder of these lectures, we shall assume that $R$-parity
is conserved. 

One can use \eq{Rparity} to deduce the $R$-parity quantum numbers of
all SM particles and their supersymmetric partners,
\begin{align}
R=\begin{cases} +1\,, & \quad \text{for all SM particle particles}\,,\\
-1\,,& \quad \text{for all superpartners}\,.\end{cases}
\end{align}
The conservation of $R$-parity in scattering
and decay processes has a critical impact on supersymmetric
phenomenology. 
 For example, any initial state in a scattering
experiment will involve ordinary ($R$-even) particles.
Consequently, it follows that supersymmetric particles must be
produced in pairs.  In general, these particles are highly unstable
and decay into lighter states.  Moreover, $R$-parity invariance
also implies that
the lightest supersymmetric particle (LSP) is absolutely
stable, and must eventually be produced
at the end of a decay chain initiated by the decay of a heavy unstable
supersymmetric particle.

In order to be consistent with cosmological constraints, a stable LSP
is almost certainly electrically and color neutral.
Consequently, the LSP in an $R$-parity-conserving theory is weakly
interacting with ordinary matter, \textit{i.e}\!., it behaves like a stable heavy
neutrino and will escape collider detectors without being directly
observed.  Thus, the canonical signature for conventional
$R$-parity-conserving supersymmetric theories is missing (transverse)
energy, due to the escape of the LSP.  Moreover,
the stability of the LSP in $R$-parity-conserving supersymmetry
makes it a promising candidate for dark matter.

\subsubsection{MSSM parameters of the SUSY-conserving sector}
The parameters of the SUSY-conserving
sector consist of: (i)~gauge couplings, $g_s$, $g$, and $g'$,
corresponding
to the Standard Model gauge group SU(3)$\times$SU(2)$\times$U(1)
respectively; (ii)~a
SUSY-conserving higgsino mass parameter
$\mu$; and (iii)~Higgs-fermion Yukawa coupling constants,
$\lambda_u$, $\lambda_d$, and $\lambda_e$, corresponding to
the couplings of one generation of left- and right-handed
quarks and leptons and their
superpartners to the Higgs bosons and higgsinos.  Because there is no
right-handed neutrino (or its superpartner) in the MSSM as defined
here, a Yukawa coupling $\lambda_\nu$ is not included.
The complex $\mu$ parameter and Yukawa couplings
enter via the most general renormalizable $R$-parity-conserving
superpotential given by \eq{MSSMsuperpot} with $W_{\rm RPV}=0$.

One can now obtain the scalar potential from \eq{vscalar4} as applied to
the MSSM,
\begin{align}
V_{\rm scalar}=\half\bigl[D^a D^a+(D')^2\bigr]+F_i^* F_i\,,
\end{align}
where the index $a$ runs over the SU(3) and SU(2) gauge indices and
$D'$ is the U(1)$_Y$ $D$-term.
Focusing on the terms that depend on the Higgs boson fields, one
obtains,
\clearpage

\begin{align}
V_{\rm Higgs}=|\mu|^2\bigl[|H_d|^2+|H_u|^2\bigr]+\tfrac18(g^2+g^{\prime\,2})\bigl[|H_d|^2-|H_u|^2\bigr]^2
+\half g^2|H_d^* H_u|^2\,.
\end{align}
Clearly $\vev{V_{\rm Higgs}}\equiv\vev{0|V_{\rm Higgs}|0}\geq 0$, as expected.  Moreover, $H_d=H_u=0$ minimizes the
Higgs scalar potential, which yields $\vev{V_{\rm Higgs}}=0$, corresponding to a supersymmetric vacuum.  Thus, there is no SU(2)$\times$U(1) breaking at this stage.
But after introducing soft SUSY-breaking terms, some of which involve
the Higgs fields, it will then be possible to spontaneously break the
SU(2)$\times$U(1) symmetry.  Consequently, SUSY breaking and electroweak symmetry breaking are intimately related in the MSSM.

\subsection{Supersymmetry breaking in the MSSM}
\label{sec:MSSMSSB}

Following the rules of Girardello and Grisaru\cite{Girardello:1981wz}
that were presented in Section~\ref{GGrules}, we add the
soft-SUSY-breaking terms, consistent with the
SU(3)$\times$SU(2)$\times$U(1) gauge symmetry and the assumed
$R$-parity invariance (for a review, see Ref.\cite{Chung:2003fi}).  For simplicity, we consider in this section the case of one generation of quarks,
leptons, and their scalar superpartners.

The supersymmetry-breaking
sector contains the following sets of parameters:
(i)~three complex
gaugino Majorana mass parameters, $M_3$, $M_2$, and $M_1$, associated with
the SU(3), SU(2), and U(1) subgroups of the Standard Model;
(ii)~five squark and slepton squared-mass parameters, $M^2_{\wt Q}$,
$M^2_{\wt U}$, $M^2_{\wt D}$, $M^2_{\wt L}$, and $M^2_{\wt E}$,
corresponding to the superpartners of the five electroweak multiplets of
left-handed fermion fields and their charge-conjugates, $(u, d)_L$, $u^c_L$,
$d^c_L$, $(\nu$, $e^-)_L$, and $e^c_L$
%where the superscript $c$denotes a charge-conjugated field 
[cf.~Table~\ref{tab:MSSMcontent}]; and
(iii)~three Higgs-squark-squark and Higgs-slepton-slepton trilinear
interaction terms, with complex coefficients $T_U\equiv\lambda_u A_U$,
$T_D\equiv\lambda_d A_D$, and $T_E\equiv\lambda_e A_E$
(which define the $A$-parameters).  
Following Ref.\cite{Haber:1993wf}, it is conventional to separate out the
factors of the Yukawa couplings in defining the
$A$-parameters, originally motivated by a simple class of
gravity-mediated SUSY-breaking
models\cite{Hall:1983iz,Nilles:1983ge,Martin:1997ns}.
With this definition, if the $A$-parameters 
are parametrically of the same order (or smaller) relative
to other supersymmetry-breaking mass parameters, then
only the third generation $A$-parameters will be
phenomenologically relevant.  

Finally, we have
(iv)~two real squared-mass parameters ($m_1^2$ and~$m_2^2$) and one 
complex squared-mass parameter, $m_{12}^2\equiv \mu B$
(the latter defines the $B$-parameter), which appear in the 
tree-level scalar Higgs potential, 
\beqa
V&=&(m_1^2+|\mu|^2)H_d^\dagger H_d+(m_2^2+|\mu|^2)H_u^\dagger
H_u+(m_{12}^2H_u H_d+{\rm
  h.c.}) \nn \\
&&\qquad\quad +\eighth(g^2+g^{\prime\,2})(H_d^\dagger H_d-H_u^\dagger
H_u)^2+\half|H_d^\dagger H_u|^2\,.\label{Hpot}
\eeqa
Note that the quartic Higgs couplings in \eq{Hpot} are related to the gauge
couplings $g$ and $g'$ as a consequence of supersymmetry.
The breaking of the
electroweak symmetry SU(2)$\times$U(1) to U(1)$_{\rm EM}$ is
only possible after introducing the
supersymmetry-breaking Higgs squared-mass parameters $m_1^2$, $m_2^2$
(which can be negative) and $m_{12}^2$.
After minimizing the Higgs scalar potential,
these three squared-mass
parameters can be re-expressed in terms of the two
Higgs vacuum expectation values, $\langle H_d^0\rangle\equiv v_d/\sqrt{2}$ 
and $\langle H_u^0\rangle\equiv v_u/\sqrt{2}$,
and the CP-odd Higgs mass $m_A$ [cf.~\eqs{minbeta}{minconditions} below].  
One is always free to rephase the Higgs doublet fields such that $v_d$
and $v_u$ are both real and positive.

The quantity, $v_d^2+v_u^2=
4m_W^2/g^2=(2G_F^2)^{-1/2}\simeq (246~{\rm GeV})^2$, is fixed by the
Fermi constant, $G_F$, whereas the ratio
\beq \label{eqtanbeta}
\tan \beta = \frac{v_u}{v_d}
\eeq
is a free parameter such that $0\leq\beta\leq\pi/2$.
The tree-level conditions for the scalar potential minimum
relate the diagonal and off-diagonal Higgs squared-mass parameters in terms
of $m^2_Z=\quarter(g^2+ g^{\prime\,2})(v_d^2+v_u^2)$, the angle~$\beta$, and
the CP-odd Higgs mass $m_A$:
\beqa
\sin 2\beta &=& \frac{2m_{12}^2}{m_1^2+m_2^2+2|\mu|^2}=\frac{2m_{12}^2}{m_A^2}
\,, \label{minbeta} \\[6pt]
\half m_Z^2 &=& -|\mu|^2+\frac{m_1^2-m_2^2\tan^2\beta}{\tan^2\beta-1}\,.
\label{minconditions}
\eeqa

At this stage, one can already see the tension with naturalness, if
the SUSY parameters, $|m_1|$, $|m_2|$ and $|\mu|$, are significantly larger than the scale of
electroweak symmetry breaking.  In this case, $m_Z^2$ will be the
difference of two large numbers, requiring some fine-tuning of the
SUSY parameters in order to produce the correct $Z$ boson mass.  In
the literature, this tension is referred to as the little hierarchy
problem\cite{little,little2,little3}, previous noted in Section~\ref{quadratic}.
One must also guard against the existence of 
charge and/or color breaking global minima
due to non-zero vacuum expectation values for the squark and 
charged slepton fields.  This possibility can be avoided 
if the $A$-parameters are not unduly
large\cite{AlvarezGaume:1983gj,Frere:1983ag,Derendinger:1983bz,Gunion:1987qv,Chowdhury:2013dka,Hollik:2016dcm,Casas:1995pd}.
Additional constraints must also be respected to avoid directions in scalar field space in which
the full tree-level scalar potential can become unbounded from below\cite{Casas:1995pd}.

\subsection{The MSSM parameter count}
\label{sec:count}

The total number of independent physical parameters
that define the MSSM (in its most general form) is
quite large, primarily due to the
soft-supersymmetry-breaking sector.  In particular, in the case of
three generations of quarks, leptons, and their superpartners,
$M^2_{\wt Q}$,
$M^2_{\wt U}$, $M^2_{\wt D}$, $M^2_{\wt L}$, and $M^2_{\wt E}$
are hermitian $3\times 3$ matrices, and
$A_U$, $A_D$, and $A_E$ are complex $3\times 3$
matrices.  In addition, $M_1$, $M_2$, $M_3$, $B$, and $\mu$
are in general complex parameters.  Finally, as in the Standard Model, the
Higgs-fermion Yukawa couplings, $\lambda_f$ ($f\!=\!u$, $d$, and $e$),
are complex $3\times 3$ matrices that
are related to the quark and lepton mass matrices via: $M_f=\lambda_f
v_f/\sqrt{2}$, where $v_e\equiv v_d$ [with $v_u$ and $v_d$ as defined
above \eq{eqtanbeta}].

However, not all these parameters are physical.
Some of the MSSM parameters can be eliminated by
expressing interaction eigenstates in terms of the mass eigenstates,
with an appropriate redefinition of the MSSM fields to remove unphysical
degrees of freedom.  The analysis of Refs.\cite{Dimopoulos:1995ju,Haber:2000jh} shows that the MSSM
possesses 124 independent parameters.  Of these, 18
correspond to SM parameters
(including the QCD vacuum angle, $\theta_{\rm QCD}$), one corresponds to
a Higgs sector parameter (the analogue of the SM
Higgs mass), and 105 are genuinely new parameters of the model.
The latter include: five real parameters and three CP-violating phases in
the gaugino/higgsino sector, 21 squark and slepton masses,
36 real mixing angles to define the
squark and slepton mass eigenstates, and 40 CP-violating phases that
can appear in the squark and slepton interactions.

Unfortunately, without additional restrictions on the 124 parameters,
the MSSM is not a
phenomenologically viable theory.  In particular, a generic point of
the MSSM parameter space typically exhibits:
(i)~no conservation of the separate lepton numbers
$L_e$, $L_\mu$, and $L_\tau$; (ii)~unsuppressed
flavor-changing neutral currents (FCNCs)\cite{Georgi:1986ku,Hall:1985dx};
and (iii)~new sources of CP~violation\cite{Khalil:2002qp} that are
inconsistent with the experimental bounds.
For example, the strong suppression of FCNCs observed in nature implies
that the off-diagonal matrix elements of
the soft-SUSY-breaking squark and slepton squared-mass matrices
are highly constrained\cite{Chung:2003fi,RamseyMusolf:2006vr}.

In practice, various simplifying assumptions are imposed 
on the SUSY-breaking sector to reduce the
number of parameters to a more manageable form, such that
the constraints imposed by lepton and quark flavor changing and
CP-violating processes are satisfied.  For example,
specific models of gravity-mediated and gauge-mediated supersymmetry
breaking\footnote{One of the benefits of GMSB models
  is that the SUSY-breaking is transmitted to the MSSM sector via
  gauge boson exchange, which is automatically flavor-conserving.}   
introduce a small number of fundamental parameters that provide the
source for SUSY-breaking for the MSSM,
consistent with the constraints due to flavor and CP violation.
More details can be found in Ref.\cite{susy}.

An alternative approach, called the phenomenological MSSM (pMSSM) has
been introduced\cite{Djouadi:2002ze,Berger:2008cq}, which attempts to
identify the parameters most relevant for phenomenology, subject to
a number of simplifying assumptions.
The pMSSM is governed by 19 independent real supersymmetric
parameters: the three gaugino
mass parameters $M_1$, $M_2$ and $M_3$, the Higgs sector parameters $m_A$ and
$\tan\beta$, the Higgsino mass parameter $\mu$, five squark and slepton
squared-mass parameters for the degenerate first and second
generations ($M^2_{\widetilde Q}$, $M^2_{\widetilde U}$,  $M^2_{\widetilde D}$,
$M^2_{\widetilde L}$ and $M^2_{\widetilde E}$), the five
corresponding squark and slepton squared-mass parameters for
the third generation, and three third-generation $A$-parameters
($A_t$, $A_b$ and $A_\tau$).\footnote{In Ref.\cite{deVries:2015hva}, the number of pMSSM parameters
is reduced to ten by assuming one common squark mass parameter for the
first two generations, a second common squark mass parameter for the third
generation, a common slepton mass parameter, and a common third generation
$A$ parameter.}  
The first and second generation $A$-parameters can be neglected as their
phenomenological consequences are negligible.   Such an approach 
assumes that new sources of flavor violation and/or CP-violation
are either absent or negligible.\footnote{The pMSSM approach has been
  recently extended to include additional CP-violating
SUSY-breaking parameters in Ref.\cite{Berger:2015eba}.}

\subsection{The MSSM particle spectrum}
\label{sec:MSSMspectrum}
\subsubsection{ Spin-1/2 superpartners}

The superpartners of the gauge and Higgs bosons are fermions,
whose names are obtained by appending ``ino'' to the end of the
corresponding SM particle name.  The gluino is the
color-octet Majorana fermion partner of the gluon
with mass $M_{\widetilde g}=|M_3|$.
The superpartners of the electroweak gauge
and Higgs bosons (the gauginos and higgsinos)
can mix due to SU(2)$\times$U(1) breaking effects.  As a result,
the physical states of definite mass are model-dependent linear combinations
of the charged or neutral gauginos and higgsinos,
called charginos and neutralinos, respectively
(sometimes collectively called electroweakinos).
The charginos are Dirac fermions, and
the neutralinos are Majorana fermions.

The tree-level mixing of the charged gauginos ($\widetilde W^\pm$) and 
higgsinos ($\widetilde H_u^+$ and $\widetilde H_d^-$) is governed 
by a $2\times 2$ complex
mass matrix,
\begin{align}
M_C\equiv \begin{pmatrix}
    M_2\quad
      &  gv_u/\sqrt{2} \\
       gv_d/\sqrt{2}    \quad
      &\mu \end{pmatrix}\,.
\end{align}
The physical chargino states and their
masses are obtained by
performing a singular value decomposition
of the complex matrix $M_C$ [cf.~\eq{LTMR}]:
\begin{align}
U^* M_C V^{-1}={\rm diag}(\mchipa\,,\,\mchipb)\,,
\end{align}
where $U$ and $V$ are unitary matrices.
The physical chargino states are Dirac fermions and are denoted by
$\chipma$ and $\chipmb$.  These are linear combinations of the
charged gaugino and higgsino states determined
by the matrix elements of $U$ and $V$.
The chargino masses correspond to the singular values of
$M_C$, \textit{i.e.}, the positive square roots
of the eigenvalues of $M_C^\dagger M_C$,
\begin{align}
\begin{split}
\hspace{-0.1in}
M^2_{\chipa,\chipb}=&
\half \biggl\{ |\mu|^2+|M_2|^2+2m_W^2\\
&\quad\left.
\mp
\sqrt{\left(|\mu|^2+|M_2|^2+2m_W^2\right)^2 
-4 |\mu M_2 - m_W^2 \sin2\beta|^2}\,\,
\right\rbrace\,,
\end{split}
\end{align}
where the states are ordered such that $\mchipa \leq \mchipb$.
The relative phase of $\mu$ and $M_2$ is physical and potentially observable.

The tree-level mixing of the neutral gauginos ($\widetilde B$ and
$\widetilde W^0$) and 
higgsinos ($\widetilde H_d^0$ and $\widetilde H_u^0$) is
governed by a $4\times 4$ complex symmetric mass
matrix,
\begin{align}
M_N\equiv \begin{pmatrix}
    M_1\quad & 0 \quad & -\half g' v_d \quad & \phm\half g' v_u \\
 0 \quad & M_2 \quad & \phm\half g v_d \quad & -\half g v_u \\
-\half g' v_d \quad & \phm\half g v_d \quad & 0 \quad & -\mu \\
\phm\half g' v_u \quad & -\half g v_u \quad & -\mu \quad & 0 \end{pmatrix}\,.
\end{align}
To determine the physical neutralino states and their masses,
one must perform a
Takagi-diagonalization
of the complex symmetric matrix $M_N$ [cf.~\eq{takagidef}]:
\begin{align}
W^T M_N W={\rm diag}(\mchina\,,\,\mchinb\,,\,\mchinc\,,\,\mchind)\,,
\end{align}
where $W$ is a unitary matrix.
The physical neutralino states are Majorana fermions, and are denoted by
$\chini$ ($i=1,\ldots 4$), where the states are ordered such that
$\mchina\leq\mchinb\leq\mchinc\leq\mchind$.
The $\chini$ are the linear combinations of the
neutral gaugino and higgsino states determined
by the matrix elements of $W$.
The neutralino masses correspond to the singular values of
$M_N$, \textit{i.e.}, the positive square roots
of the eigenvalues of $M_N^\dagger M_N$.

\subsubsection{Spin-0 superpartners}

The superpartners of the quarks and leptons are spin-zero
bosons:  the squarks, charged sleptons,
and sneutrinos, respectively.
For a given Dirac fermion $f$, there are two superpartners, $\widetilde
f_L$ and $\widetilde f_R$, where the $L$ and $R$ subscripts simply identify
the scalar partners that are related by supersymmetry to the left-handed and
right-handed fermions, $f_{L,R}\equiv\half(1\mp\gamma_5)f$, respectively.
(There is no $\widetilde\nu_R$ in the MSSM.)
However, $\widetilde f_L$--$\widetilde f_R$ mixing is possible,
in which case $\widetilde f_L$ and $\widetilde f_R$ are not mass
eigenstates.  

We first  consider the squarks and the sleptons.
For three generations of squarks, one
must diagonalize $6\times 6$ matrices corresponding
to the basis $(\widetilde q_{iL}, \widetilde q_{iR})$,
where $i=1,2,3$ are the generation
labels.
For simplicity, only the one-generation case is illustrated
in detail below.

Using the notation of the third family, the one-generation
tree-level squark squared-mass matrix is given by
\begin{align}
\mathcal{M}^2 =& \begin{pmatrix}
    M^2_{\widetilde Q}+ m^2_q+ L_q\quad
      & m_q X_q^* \\
    m_q X_q\quad
      &M^2_{\widetilde R}+ m^2_q+ R_q  \end{pmatrix}\,,\label{sqmassmat}
      \end{align} 
where
\begin{align}
X_q\equiv A_q-\mu^* (\cot\beta)^{2T_{3q}}\,,
\end{align} \label{Xtdef}
and 
\begin{align}
T_{3q}=\begin{cases} \phm\half\,,\quad \text{for $q=t$}\,,\\ -\half\,,\quad \text{for $q=b$}.\end{cases}
\end{align}

The diagonal squared-masses are governed by soft-SUSY-breaking
squared-masses $M^2_{\widetilde Q}$ and $M^2_{\widetilde R}\equiv
M^2_{\widetilde U}$ [$M^2_{\widetilde D}$] for $q=t$~[$b$], the
corresponding quark masses $m_t$ [$m_b$], and electroweak correction terms:
\begin{align}
L_q& \equiv
(T_{3q}-e_q\sin^2\theta_W)m_Z^2\cos 2\beta\,,\\
R_q& \equiv
e_q\sin^2\theta_W \,m_Z^2\cos 2\beta\,,
\end{align}
where $e_q=\tfrac23$ [$-\tfrac13$] for $q=t$ [$b$].

The off-diagonal squark squared-masses are
proportional to the corresponding quark masses and depend on
$\tan\beta$, the
soft-SUSY-breaking $A$-parameters and the higgsino mass parameter
$\mu$.
Assuming that the $A$-parameters
are parametrically of the same order (or smaller) relative
to other SUSY-breaking mass parameters, it then follows that
$\widetilde q_L$--$\widetilde q_R$ mixing effects
are small, with the possible exception of the third generation,
where mixing can be enhanced by factors of $m_t$ and $m_b\tan\beta$.

In the case of third generation $\widetilde q_L$--$\widetilde q_R$
mixing, the mass eigenstates (denoted by $\widetilde q_1$ and
$\widetilde q_2$, with $m_{\tilde q_1}<m_{\tilde q_2}$) are determined
by diagonalizing the $2\times 2$ matrix ${\cal M}^2$.  
The corresponding squared-masses
and mixing angle are:
\begin{align}
  m^2_{\tilde q_{1,2}} =&\half\left[{\rm Tr}\,{\cal M}^2\mp
\sqrt{({\rm Tr}{\cal M}^2)^{2}
-4\,{\rm det}\,{\cal M}^2}\right]\,,  \\
\sin 2\theta_{\tilde q} =& \frac{2m_q |X_q|}{m^2_{\tilde
q_2}-m^2_{\tilde q_1}}\,.
\end{align}
The results above
also apply to the charged sleptons with the 
substitutions: $q\to \ell$ with
$T_{3\ell}=-\half$ and $e_\ell=-1$, and the
replacement of the SUSY-breaking parameters:
$M^2_{\widetilde Q}\to M^2_{\widetilde L}$,
$M^2_{\widetilde D}\to M^2_{\widetilde E}$, and $A_q\to A_\tau$.
For the neutral sleptons, $\widetilde\nu_R$ does not exist in the
MSSM, so $\widetilde\nu_L$ is a mass eigenstate.

In the case of three generations, the supersymmetry-breaking scalar-squared
masses [$M_{\wt Q}^2$, $M_{\wt U}^2$, $M_{\wt D}^2$,
$M_{\wt L}^2$, and $M_{\wt E}^2$] and
the $A$-parameters [$A_U$, $A_D$, and $A_E$]
are now $3\times 3$ matrices.
The diagonalization of the $6\times 6$ squark mass
matrices yields $\widetilde f_{iL}$--$\widetilde f_{jR}$
mixing (for $i\neq j$).
In practice, since the
$\widetilde f_L$--$\widetilde f_R$ mixing is appreciable only for the
third generation, this additional complication can often
be neglected.

\subsection{The Higgs sector of the MSSM}
\label{higgssector}
Having completed our tour of the superpartners of the SM particles, we
now focus of the Higgs sector of the MSSM\cite{Gunion:1984yn,hhg,Djouadi:2005gj}.  
We first provide details of the structure of the Higgs sector based on
a tree-level analysis.   We then discuss the importance of radiative
corrections, in light of the observed Higgs boson with a mass of 125 GeV. 

\subsubsection{The tree-level MSSM Higgs sector}

The tree-level scalar Higgs potential, previously
given in \eq{Hpot}, is
CP-conserving.  This follows from the fact that $m_{12}^2$,
the only potentially complex parameter that appears in \eq{Hpot}, 
can be chosen real and positive by an appropriate rephasing of the Higgs fields.

After minimizing the Higgs potential, as indicated above
\eq{eqtanbeta}, one can identify the physical Higgs states.
The five physical Higgs particles
consist of a charged Higgs pair
\begin{align}
H^\pm=H_d^\pm\sin\beta+ H_u^\pm\cos\beta\,,
\end{align}
one CP-odd neutral scalar
\begin{align}
\ha= \sqrt{2}\left({\rm Im\,}H_d^0\sinb+{\rm Im\,}H_u^0\cosb
\right)\,,
\end{align}
and two CP-even neutral scalar mass eigenstates that are determined by
diagonalizing the neutral CP-even Higgs scalar squared-mass matrix,
\begin{align}
\mathcal{M}_0^2 = &
\begin{pmatrix}
\mha^2 \sin^2\beta + m^2_Z \cos^2\beta \ \ & \quad
           -(\mha^2+m^2_Z)\sin\beta\cos\beta \\
  -(\mha^2+m^2_Z)\sin\beta\cos\beta \ \   & \quad
  \mha^2\cos^2\beta+ m^2_Z \sin^2\beta \end{pmatrix}\,.\label{mzero}
\end{align}
The eigenstates of $\mathcal{M}_0^2$ are identified as the neutral CP-even Higgs bosons,
\begin{align}
\hl &= -(\sqrt{2}\,{\rm Re\,}H_d^0-v_d)\sin\alpha+
(\sqrt{2}\,{\rm Re\,}H_u^0-v_u)\cos\alpha\,,\\
\hh &= (\sqrt{2}\,{\rm Re\,}H_d^0-v_d)\cos\alpha+
(\sqrt{2}\,{\rm Re\,}H_u^0-v_u)\sin\alpha\,,
\end{align}
which defines the CP-even Higgs mixing angle $\alpha$.
%where the angle $\alpha$ parameterizes the orthogonal matrix that
%diagon\-lizes~$\mathcal{M}^2_0$. 
%\clearpage

All Higgs masses and couplings can be expressed in terms of two
parameters, usually chosen to be $\mha$ and $\tan\beta$.
The charged Higgs mass is given by
\begin{align}
\mhpm^2 =\mha^2+\mw^2\,.
\end{align}
The squared-masses of the CP-even Higgs bosons $\hl$ and $\hh$ are eigenvalues
of $\mathcal{M}_0^2$.   The trace and determinant of $\mathcal{M}_0^2$
yield,
\beq
m_h^2+m_H^2=m_A^2+m_Z^2\,,\qquad\quad m_h^2 m_H^2=m_A^2 m_Z^2\cos^2
2\beta\,,\label{trdet}
\eeq
where the CP-even Higgs masses are given by
\begin{align}
  m^2_{H,h} = \half \left( \mha^2 + m^2_Z \pm
                  \sqrt{(\mha^2+m^2_Z)^2 - 4m^2_Z \mha^2 \cos^2 2\beta}
                  \; \right)\,.\label{hH}
\end{align}
In the convention where $0\leq\beta\leq\half\pi$, it is standard
practice to choose $\alpha$ to lie in the range
$|\alpha|\leq\half\pi$.  However, because the off-diagonal element of
$\mathcal{M}_0^2$ is negative semi-definite, one finds that
$-\half\pi\leq\alpha\leq 0$.   More explicitly, the mixing angle $\alpha$ can be determined 
as a function of $m_A$ and $\tan\beta$ 
from the following expression and from \eq{hH},\footnote{The corresponding expressions
  for a general CP-conserving two Higgs doublet model can be found in 
Ref.\cite{Bernon:2015qea} .} 
\beq
\cos\alpha=\sqrt{\frac{m_A^2\sin^2\beta+m_Z^2\cos^2\beta-m_h^2}{m_H^2-m_h^2}}\,,
\eeq
and $\sin\alpha=-(1-\cos^2\alpha)^{1/2}$.

In the expression for the couplings of the Higgs bosons with the gauge
bosons, only the combination $\beta-\alpha$ appears.   For example,
the coupling of $h$ to $VV$ (where $VV=W^+ W^-$ or $ZZ$) relative to
the corresponding coupling of the SM Higgs boson, $h_{\rm SM}$, is given by,
\beq
\frac{g_{hVV}}{g_{h_{\rm SM}VV}}=\sin(\beta-\alpha)\,.\label{SMh}
\eeq
Given the range
of the angles $\alpha$ and $\beta$, it follows that $0\leq\beta-\alpha\leq\pi$.
In particular, the following expressions can be obtained,
\beqa
\cos(\beta-\alpha)&=&\frac{m_Z^2\sin 2\beta\cos
  2\beta}{\sqrt{(m_H^2-m_h^2)(m_H^2-m_Z^2\cos^2
    2\beta)}}\,.\label{cbma} \\[6pt]
\sin(\beta-\alpha)&=&\sqrt{\frac{m_H^2-m_Z^2\cos^2 2\beta}{m_H^2-m_h^2}}\,.\label{sbma}
\eeqa
One can check that \eqs{cbma}{sbma} are consistent in light of \eq{trdet}.

The Higgs--fermion Yukawa couplings are obtained from the MSSM
superpotential [\eq{MSSMsuperpot} with $W_{\rm RPV}=0$] by employing the last two terms
of \eq{eq:LSUSYcomponents}.  Focusing on the Higgs
interactions with third generation quarks, one obtains the so-called
Type-II Higgs-quark interaction\cite{Hall:1981bc},
\beq \label{typetwo}
\mathscr{L}_{\rm Yuk}=-\epsilon^{ij}\bigl[h_b \overline{b}_R H_{d\,\!i} Q_{L\,\!j}+h_t\overline{t}_R Q_{L\,\!i} H_{u\,\!j}\bigr]+{\rm h.c.}\,,
\eeq
where $Q_L\equiv(t_L\,,\,b_L)$ is the quark doublet and $i$ and $j$
are SU(2) indices.  In \eq{typetwo}, we employ four-component quark
fields, where $q_{R,L}\equiv P_{R,L}q$ and
$P_{R,L}=\half(1\pm\gamma\ls{5})$.
The quark masses are identified by replacing the Higgs fields in
\eq{typetwo} with their corresponding vacuum expectation values,
\beq \label{tbmasses}
m_b= h_b v \cos\beta/\sqrt{2}\,,\qquad\quad  m_t=h_t v \sin\beta/\sqrt{2}\,.
\eeq
The tree-level Yukawa couplings of the lightest CP-even Higgs boson to
third generation quark pairs are given by
\beqa
g_{h b\bar b} &= &-\frac{m_b}{v}\,\frac{\sin\alpha}{\cos\beta}= \frac{m_b}{v}\,\bigl[\sinbma-\cosbma \tan\beta\bigr]\,,
\label{hlbbtree}  \\[5pt]
g_{h t\bar t} & = & \phm\frac{m_t}{v}\,\frac{\cos\alpha}{\sin\beta}=\frac{m_t}{v}\,\bigl[\sinbma+\cosbma\cot\beta\bigr]\,.
\label{hltttree}
\eeqa
It is straightforward to work out the couplings of the other Higgs
bosons of the model to the quarks (and leptons).  A comprehensive set
of Feynman rules for Higgs bosons in the MSSM can be found in Refs.\cite{Gunion:1984yn,hhg}.

In the limit of $\mha\gg\mz$, the expressions for the
Higgs masses and mixing angle are given by,
\begin{align}
\mhl^2 &\simeq  \ \mz^2\cos^2 2\beta-\frac{m_Z^4 \sin^2 {4\beta}}{4m_A^2}\,, \\
\mhh^2 &\simeq  \ \mha^2+\mz^2\sin^2 2\beta\,,\\
\mhpm^2& =  \ \mha^2+\mw^2\,,\\
\cos(\beta-\alpha)&\simeq\ \frac {\mz^2\sin 4\beta}{2\mha^2}\,.
\end{align}
Two consequences are immediately apparent.
First,
\begin{align}
\mha\simeq\mhh
\simeq\mhpm,
\end{align}
up to corrections of ${\cal O}(\mz^2/\mha)$.  Second,
$\cos(\beta-\alpha) \simeq 0$, up to corrections of ${\cal O}(\mz^2/\mha^2)$.
This is the decoupling limit of the MSSM Higgs sector, since at energy scales below
 the approximately common mass of the heavy
Higgs bosons $H^\pm$, $\hh$, and $A^0$, the effective Higgs theory is
equivalent to the one-doublet Higgs sector of the SM\cite{Haber:1989xc,Gunion:2002zf}.
Indeed, one can check that in the limit of $\cos(\beta-\alpha)\to 0$,
all the $\hl$ couplings to SM particles approach their SM limits, as
in the case of the $hVV$ coupling exhibited in \eq{SMh} and in the
case of the $hq\bar{q}$ couplings exhibited in \eqs{hlbbtree}{hltttree}.

\subsubsection{Impact of radiative corrections on the MSSM Higgs sector}

The tree-level result for $m_h$ given in \eq{hH} yields a startling
prediction,
\beq
\mhl\leq\mz |\cos 2\beta|\leq\mz\,.
\eeq
This is clearly in conflict with the observed Higgs mass of 125 GeV.
However, the above
inequality receives quantum corrections.  The Higgs mass can be shifted
due to loops of particles and their superpartners exhibited below (an incomplete
cancellation, which would have been exact if supersymmetry were
unbroken).

\begin{center}
\begin{picture}(200,75)(-50,-40)
\SetScale{0.85}
\thicklines
\DashLine(-100,0)(-70,0){3}
\ArrowArcn(-40,0)(30,180,0)
\ArrowArcn(-40,0)(30,0,180)
\DashLine(20,0)(-10,0){3}
\DashLine(100,0)(130,0){3}
\DashArrowArcn(160,0)(30,180,0){3}
\DashArrowArcn(160,0)(30,0,180){3}
\DashLine(190,0)(220,0){3}
\Text(-110,0)[]{$\hl$}
\Text(30,0)[]{$\hl$}
\Text(90,0)[]{$\hl$}
\Text(230,0)[]{$\hl$}
\Text(-40,15)[]{$t$}
\Text(160,15)[]{$\widetilde t_{1,2}$}

\end{picture}
\end{center}
\vskip -0.1in
The impact of these corrections
can be significant\cite{Haber:1990aw,Okada:1990vk,Ellis:1990nz}.
In particular, the qualitative behavior of the one-loop radiative corrections
can be most easily
seen 
in the limit of large top-squark masses.
In this limit,
both the off-diagonal entries and the splitting between the two diagonal entries
of the top-squark squared-mass matrix
[\eq{sqmassmat}]
are small in comparison to
the square of the geometric mean of the two top-squark 
masses,
$\msusyy\equiv\mstopa\mstopb$.  
In this case (assuming $\mha>\mz$), the predicted upper bound for $m_h$
is approximately given by\cite{Haber:1996fp}
\begin{align}
\mhl^2\lsim \mz^2+\frac{3g^2 m_t^4}{8\pi^2\mw^2}\left[\ln\left(\frac{M_S^2}{m_t^2}\right)+\frac{X_t^2}{M_S^2}
\left(1-\frac{X_t^2}{12M_S^2}\right)\right]\,, \label{hradcorr}
\end{align}
where $X_t\equiv A_t-\mu\cot\beta$ governs stop mixing (taking $A_t$
and $\mu$ real for simplicity).
The Higgs mass upper limit is saturated when
$\tan\beta$ is large [{\it i.e.}, $\cos^2 (2\beta) \sim 1$] and $X_t=\sqrt{6}\,
M_S$, which defines the so-called maximal mixing scenario.

A more complete treatment of the radiative corrections\cite{Draper:2016pys}
shows that
\eq{hradcorr} somewhat overestimates the true upper bound of $\mhl$.
These more refined computations, which incorporate
renormalization group improvement, and the two-loop and
leading three-loop contributions, yield an upper bound of $m_{h}\lsim 135$~GeV in the
region of
large $\tan\beta$ (with an accuracy of a few GeV)
for $m_t=175$~GeV and $M_S\lsim 2$~TeV\cite{Draper:2016pys},
which is quite close to the observed value of the Higgs mass!

In certain cases, radiative corrections also can significantly modify the tree-level
Yukawa couplings.  For a review of such effects, see e.g., Ref.\cite{Carena:2002es}.

\subsection{Unification of gauge couplings}
\label{sec:MSSMGU}

Grand unification theory (GUT) predicts the unification of gauge couplings at some very high energy scale\cite{Raby,guts,Langacker:1980js,Ross}.  
The running of the couplings is dictated by the particle content of the effective theory that resides below the GUT scale.  
However, attempts to embed the Standard Model in an SU(5) or SO(10)
unified theory do not quite succeed.
In particular, the three running gauge couplings (the strong QCD
coupling $g_s$ and the electroweak gauge couplings $g$ and $g'$) do not meet at a point, as shown by the dashed lines in Fig.~\ref{fig:GUT}.
In contrast, in the case of the MSSM with superpartner masses of order
1 TeV, the renormalization group evolution is modified above the
SUSY-breaking scale.   In this case, unification of gauge couplings
can be (approximately) achieved as illustrated  by the red and blue
lines in Fig.~\ref{fig:GUT}.

\begin{figure}[h!]
\centering
\includegraphics[width=0.7\linewidth]{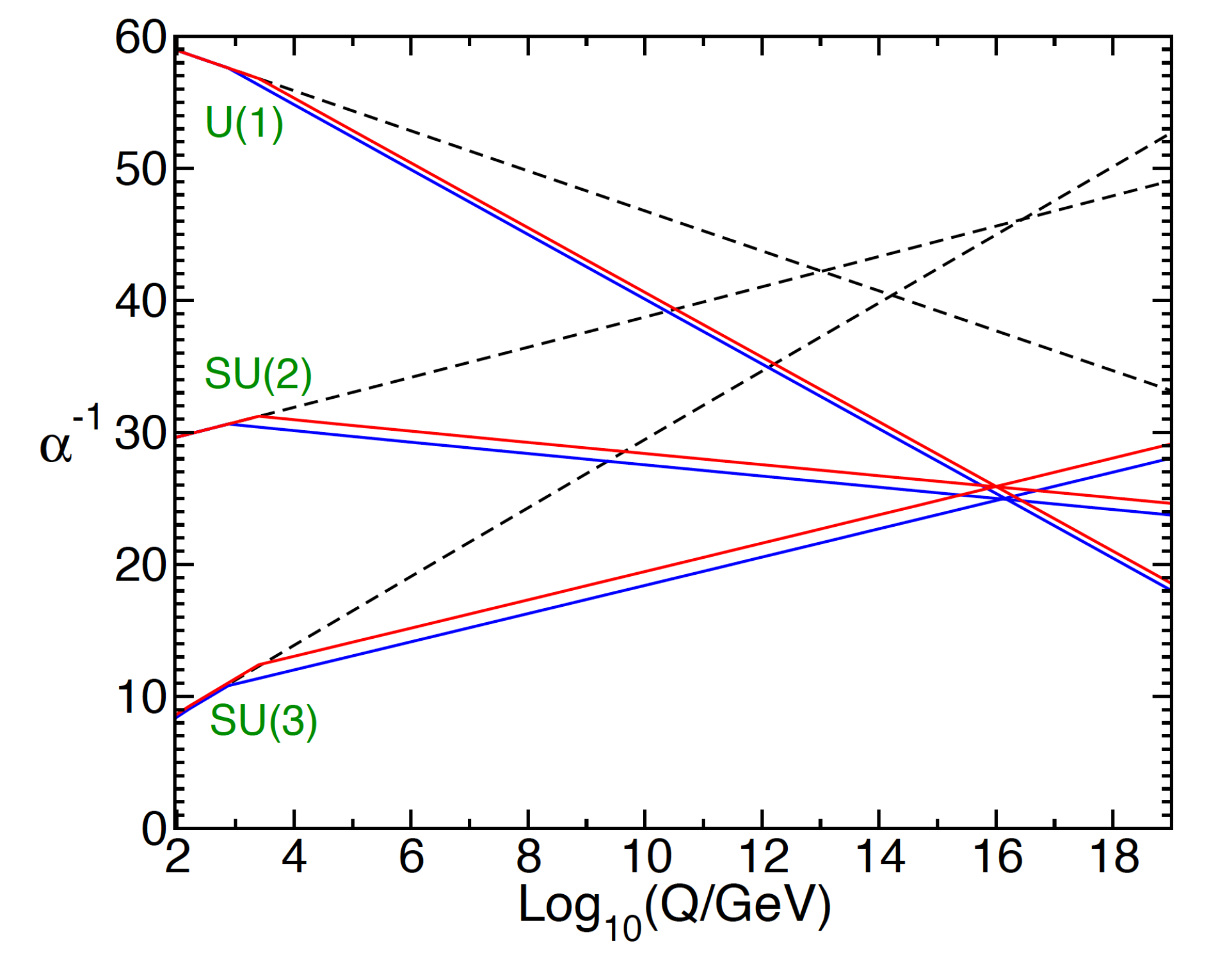}
\caption{\small
Renormalization group evolution of the inverse gauge couplings $\alpha_a^{-1}(Q)$ in the
Standard Model (dashed lines) and the MSSM (solid lines). In the MSSM
case, $\alpha_3(m_Z)$ is varied between 0.121 and 0.117, and the
supersymmetric particle mass thresholds are between 500 GeV and 1.5 TeV, for the
lower and upper solid lines, respectively. Two-loop effects are
included. Taken from Ref.\cite{Martin:1997ns}.
}
\label{fig:GUT}
\end{figure}

A quantitative assessment of the success of gauge coupling unification
can be performed as follows.  
Since the electroweak gauge couplings $g$ and $g'$ are very well
measured, first focus on these two couplings.  For a given low-energy
effective theory (below the GUT scale), we use the renormalization
group equations (RGEs) to determine the couplings $g$ and $g'$ as a
function of the energy scale.  We then define $M_{\rm GUT}$ to be the
scale at which these two couplings meet.  

 We now assume that the unification of the three
gauge couplings, $g_s$, $g$ and $g'$ occurs at $M_{\rm GUT}$.  Using
the RGEs for the gauge couplings, we can now run $g_s$ down to the
electroweak scale and compare with the experimentally measured value.

\subsubsection{Normalization of the U(1)$_{\rm Y}$ coupling}
In electroweak theory, the overall normalization of the U(1)$_{\rm Y}$ coupling is a matter of convention.  But, if the GUT group is simple and nonabelian, then the relative normalization of the U(1)$_{\rm Y}$ coupling to the SU(2) gauge coupling is fixed.   
We denote the SU(3)$\times$SU(2)$\times$U(1)$_{\rm Y}$ gauge couplings using the proper GUT normalization by $g_3$, $g_2$ and $g_1$ respectively.  Our task is to relate $g_1$ with $g'$.
To do so, let us begin by
considering the covariant derivative,
\begin{align}
D_\mu=\partial_\mu+i\sum_a g_a T^a A_\mu^a\,.
\end{align}
If the gauge group is a direct product group, then different sets of generators $T^a$ are associated with with the different group factors, and we must use the appropriate $g_a$ depending on which generator it multiplies.  
In particular, for SU(2)$\times$U(1)$_{\rm Y}$ (below the GUT scale), 
\begin{align}
g_a T^a A_\mu^q\ni gT^3 W_\mu^3+g'\frac{Y}{2}B_\mu\,.
\end{align}
Above the GUT scale, the corresponding terms of the covariant derivative are
\begin{align}
g_a T^a A_\mu^q\ni g_U( T^3 W_\mu^3+T^0B_\mu)\,,
\end{align}
where $g_U$ is the gauge coupling of the unifying GUT group and $T^0$ is the properly normalized hypercharge generator.  
In particular, the generators of the GUT group satisfy
\begin{align}
\Tr(T^a T^b) =T(R)  \delta^{ab}\,,\label{tab}
\end{align}
where $T(R)$ is a constant that depends on the representation
$R$.\footnote{Once $T(R)$ is fixed for one representation, it is then
  determined for all other representations.  It is standard practice
  to fix $T(R)=\half$ for the defining (fundamental) representation, although the
argument presented below is independent of this choice.}  
We now set the two covariant derivatives above equal at the GUT scale,
\begin{align}
 g_U( T^3 W_\mu^3+T^0B_\mu)=gT^3 W_\mu^3+g'\frac{Y}{2}B_\mu\,.
 \end{align}
Noting that $g_U=g_3=g_2=g_1$ at the GUT scale, it
follows that
 $g_2=g$ and  $g_1 T^0=g'(Y/2)$.  Since $T(R)$ only depends on the
 representation $R$, \eq{tab} yields $\Tr (T^3)^2=\Tr(T^0)^2$. Thus,
\begin{align}
 g_1^2=g^{\prime\,2}\,\frac{\Tr Y^2}{4\Tr(T^3)^2}\,.\label{gone}
\end{align}
The relevant quantum numbers are provided in
Table~\ref{tab:two-component_fields}.  
 
The traces in \eq{gone} are evaluated by summing over one generation of fermions, under the assumption that it is made up of
complete irreducible representations of the GUT group.\footnote{In an
  SU(5) GUT, one
  generation of fermions make up a 10-dimensional and the complex
  conjugate of a 5-dimensional 
  representation of SU(5).  In an SO(10) GUT, one generation of fermions (including
  the right-handed neutrino) comprise a 16 dimensional spinor
  representation of SO(10).}
Using the results of Table~\ref{tab:two-component_fields}, we simply
add up the last two columns.  Including the appropriate color factor
of 3 when tracing over the suppressed color index, 
we obtain
$\Tr (T^3)^2=2$ and $\Tr Y^2=\tfrac{40}{3}$.  Thus, 
\eq{gone} yields
 \begin{align}
g_1^2=\tfrac53 g^{\prime\,2}\,.
\end{align}

\begin{table}
\centering
\caption{\small The $T_3$ and $Y$ quantum numbers of the two-component
  fermion fields that make up one generation of SM fermions.  In
  computing the corresponding traces, one must not forget the color
  factor of 3 that arises when tracing over the (suppressed) color
  index.  \label{tab:two-component_fields}}
\vskip 0.06in
\begin{tabular}{|ccccc|} \hline
Two-component fields & $T_3$ & $Y$ &  $\Tr (T^3)^2$ & $\Tr Y^2$ \\ \hline
$\psi_{Q_1}$ & $\phm\half$ & $\phm\tfrac13$ & $3(\tfrac14)$ & $3(\tfrac19)$ \\[5pt]
$\psi_{Q_2}$ & $-\half$ & $\phm\tfrac13$ & $3(\tfrac14)$ & $3(\tfrac19)$ \\[5pt]
$\psi_{U}$ & $\phm 0$ & $-\tfrac{4}{3}$ & $3(0)$ & $3(\tfrac{16}{9})$ \\[5pt]
$\psi_{D}$ & $\phm 0$ & $\phm \tfrac23$ & $3(0)$ & $3(\tfrac{4}{9})$ \\[5pt]
$\psi_{L_1}$ & $\phm\half$ & $-1$ & $\tfrac14$ & $1$ \\[5pt]
$\psi_{L_2}$ & $-\half$ & $-1$ & $\tfrac14$ & $1$ \\[5pt]
$\psi_{E}$ & $\phm 0$ & $\phm 2$ & $0$ & $4$ \\ \hline
 \end{tabular}
 \end{table}

\subsubsection{Gauge coupling running}
We now examine the running of the gauge couplings
in the one-loop approximation, where the gauge couplings $g_i$ obey the differential equation,
 \begin{align}
 \frac{dg_i^2}{dt}=\frac{b_i g_i^4}{16\pi^2}\,,\qquad \text{for $i=1,2,3$},\label{gRGE}
\end{align}
 where $t=\ln Q^2$ ($Q$ is the energy scale) and the $b_i$ are given by 
 \begin{align}
b_i = \tfrac23 \sum_j T(R^{(i)}_j) m(R^{(i)}_j) + \tfrac16 \sum_J c_J T(R^{(i)}_J) m(R^{(i)}_J)-\tfrac{11}{3}C_A(G^{(i)})\,,\label{bi}
\end{align}
where the indices $j$ and $J$ are employed for two-component fermions and scalars, respectively.  We have assumed that the low-energy gauge group is a direct product group, $G\equiv \prod_i G^{(i)}$ [where $G^{(i)}$ is either a simple compact Lie group or U(1)],
and the $j$th two-component fermion multiplet and the $J$th scalar multiplet transform irreducibly under $G$ as $(R_j^{(1)},R_j^{(2)},\ldots)$
and $(R_J^{(1)},R_J^{(2)},\ldots)$, respectively. The multiplicity factors in \eq{bi} are given by $m(R^{(i)})=\prod_{k\neq i} d(R^{(k)})$, where $d(R^{(k)})$ is the dimension of the irreducible representation $R^{(k)}$, and $c_J=1$ [$c_J=2$] for real [complex] scalars.   Finally, $T(R^{(i)})$ is defined in \eq{tab} in a convention where $T(R)=\half$ for the defining representation of a simple compact Lie group, and
$C_A(G^{(i)})$ is the eigenvalue of the Casimir operator in the adjoint representation of $G^{(i)}$, which is defined in terms of the structure constants of the Lie group,
\beq
f_{abc}f_{abd}=C_A(G)\delta_{cd}\,.
\eeq
For example, $C_A({\rm G})=N$ for ${\rm G}={\rm SU}(N)$.  Note that for U(1)$_Y$, we have $C_A({\rm G})=0$ and
\begin{align}
T(R)=\left[\sqrt{\tfrac{3}{5}}\,\half Y\right]\lsup{2}=\tfrac{3}{20}Y^2\,,
\end{align}
where we have employed the properly normalized hypercharge generator, $\sqrt{3/5}\,(Y/2)$.

 The solution to
 \eq{gRGE} is
 \begin{align}
 \frac{1}{g_i^2(m_Z)}=\frac{1}{g_U^2}-\frac{b_i}{16\pi^2}\ln\left(\frac{m_Z^2}{M_{\rm GUT}^2}\right)\,,\label{RGEsol}
\end{align}
 where $M_{\rm GUT}$ is the GUT scale at which the three gauge
 couplings unify.  Using \eq{RGEsol}, the following two equations are obtained:
\begin{align}
 \sin^2\theta_W(m_Z)=&\frac{g^{\prime\,2}(m_Z)}{g^2(m_Z)+g^{\prime\,2}(m_Z)}=\frac{\tfrac35 g_1^2(m_Z)}{g^2(m_Z)+\frac35 g_1^2(m_Z)}  \nn \\
=&\frac{3}{8}-\frac{5}{32\pi}\,\alpha(m_Z)(b_1-b_2)\ln\left(\frac{M_{\rm GUT}^2}{m_Z^2}\right)\,,\label{sinw} \\[5pt]
 \ln\left(\frac{M_{\rm GUT}^2}{m_Z^2}\right)=&\frac{32\pi}{5b_1+3b_2-8b_3}\left(\frac{3}{8\alpha(m_Z)}-\frac{1}{\alpha_s(m_Z)}\right)\,,\label{log}
 \end{align}
 where $e=g\sin\theta_W$, $\alpha\equiv e^2/4\pi$ and $\alpha_s\equiv g_s^2/4\pi$.

It is convenient to introduce the parameter,
\begin{align} x\equiv \frac{1}{5}\left(\frac{b_2-b_3}{b_1-b_2}\right)\,.
 \end{align}
 Then, \eqs{sinw}{log} yield,
\begin{align}
 \sin^2\theta_W(m_Z)=\frac{1}{1+8x}\left[3x+\frac{\alpha(m_Z)}{\alpha_s(m_Z)}
\right]\,.
\end{align}
Once we know the value of $x$,
we can use the above equation to determine $\alpha_s(m_Z)$ given the
values of $\sin^2\theta_W$ and $\alpha$, evaluated at $m_Z$,
\beq
\alpha_s(m_Z)=\frac{\alpha(m_Z)}{(1+8x)\sin^2\theta_W(m_Z)-3x}\,.\label{alphastrong}
\eeq
The value of $x$ is determined from the values of the $b_i$, which are given by \eq{bi}.

One can now assess the success or failure of gauge coupling unification
in the SM and in the MSSM.  For details, see Problems~\ref{pr:GUT1}
and \ref{pr:GUT2}.  As advertised in Fig.~\ref{fig:GUT}, the gauge
couplings do not unify when the SM is extrapolated to the GUT scale.
In contrast, in the MSSM, the modified running of the gauge couplings
due to the supersymmetric partners of the SM particles results in
approximate unification.\footnote{For a more precise analysis, we
  should extend the calculations of this subsection to include
  two-loop running of the gauge couplings\cite{Castano:1993ri}.  One must also properly
  treat threshold corrections at the TeV scale\cite{Martens:2011uha,Allanach:2014nba} (due to mass splittings
  among superpartners) and at the GUT scale\cite{Lucas:1995ic}.  The latter are quite
  model-dependent and allows some wiggle room in achieving precise
  gauge coupling unification.}
This success has often been touted as one of
the motivations for TeV-scale supersymmetry.

\subsection{Problems}

\begin{problem}
\label{pr:spectra}
Starting with the SUSY Lagrangian for SUSY Yang Mills theory coupled
to matter given in \eq{eq:LSUSYcomponents}, 
eliminate the auxiliary fields and obtain the Lagrangian of the MSSM prior to
SUSY-breaking.  For simplicity, you may consider only one generation
of quarks and leptons and their superpartners.
Then add in the soft-SUSY-breaking terms to obtain the
complete MSSM Lagrangian.  Using this result, verify the mass spectrum
of the supersymmetric particles obtained in
Section~\ref{sec:MSSMspectrum}. 
\end{problem}

\begin{problem}
Using the results of Problem~\ref{pr:spectra}, verify the results
obtained in Section~\ref{higgssector} for the MSSM Higgs sector.
Write out the Feynman rules for the interaction of the Higgs bosons 
with the gauge bosons and with the quarks and leptons.
\end{problem} 
   
\begin{problem}
Using the results of Problem~\ref{pr:spectra}, one can obtain the complete set of Feynman rules for the MSSM with one
generation of quarks and leptons and their superpartners.
Work out as many of the rules as you can and check your results against
Ref.\cite{Rosiek:1989rs}. 
\end{problem}

\begin{problem}
\label{pr:GUT1}
Assuming $N_g$ generations of the quarks and leptons and $N_h$ copies
of the SM Higgs boson, use \eq{bi} to obtain
\begin{align}
b_3=&\tfrac{4}{3}N_g-11\,,\nn \\
b_2=&\tfrac{1}{6}N_h+\tfrac{4}{3}N_g-\tfrac{22}{3}\,,\nn \\
b_1=&\tfrac{1}{10}N_h+\tfrac{4}{3}N_g\,.\nn
\end{align}
For the SM, we have $N_g=3$ and $N_h=1$.  Check that $b_3=-7$, $b_2=-\tfrac{19}{6}$ and $b_1=\tfrac{41}{10}$. Consequently,
independently of the value of $N_g$,
\begin{align}
x=\frac{23}{218}=0.1055\,.
\end{align}
\vskip -0.05in
\label{pr:bs}
\end{problem}

\begin{problem}
\label{pr:GUT2}
 Show that the SM results of Problem~\ref{pr:bs} are modified in the
 MSSM as follows:
 \begin{align}
b_3=&2N_g-9\,,\nn \\
b_2=&\tfrac{1}{2}N_h+2N_g-6\,,\nn \\
b_1=&\tfrac{3}{10}N_h+2N_g\,.\nn
\end{align}
For the MSSM, we have $N_g=3$ and $N_h=2$.  Verify that $b_3=-3$, $b_2=1$ and $b_1=\tfrac{33}{5}$, and consequently,
$x=\tfrac17$.  Using the values for $\alpha(m_Z)$ and $\sin^2\theta_W(m_Z)$ given in Ref.\cite{pdg},
evaluate $\alpha_s$ using \eq{alphastrong}.
Show that for $x=\tfrac17$ (as predicted by the
MSSM), one obtains a value for $\alpha_s(m_Z)$ that is quite close to
the current world average\cite{pdg}.  Using  $x=0.1055$, check that the
corresponding SM prediction for $\alpha_s(m_Z)$ is significantly lower than the observed value.
\end{problem}

\section{Supersymmetry Quo Vadis?}
\label{sec:future}
\renewcommand{\theequation}{\arabic{section}.\arabic{equation}}
\setcounter{equation}{0}

In these lectures, time constraints have limited the number of topics
that we have been able to cover.   The reader can consult the many
fine
books\cite{WessBagger,Gates,Srivastava,Piguet,Freund,MullerKirsten,West,Lopuszanski,Bailin,Buchbinder,Soni,Galperin,Polonsky,Mohapatra,Drees,Baer,Aitchison,Binetruy,Terning,MullerKirsten2,Labelle,Shifman,sugra1,weinberg3,MDine,Manoukian,sugra2,Raby}
and the 
reviews and lecture notes\cite{Taylor:1983su,Nilles:1983ge,Haber:1984rc,Sohnius:1985qm,Lahanas:1986uc,Haber:1993wf,Derendinger,Lykken,Martin:1997ns,Giudice:1998bp,bilalsusy,Petrov:2001hz,FigueroaO'Farrill:2001tr,Chung:2003fi,Luty:2005sn,RamseyMusolf:2006vr,Shirman:2009mt,GKane,susy,Bertolini:2013via}
already cited in Section~\ref{Intro} to pursue various topics in
supersymmetry in greater depth.  

In Section~\ref{sec:MSSM}, we introduced the basics of the MSSM.  But
this is not the only possible supersymmetric extension of the SM.  For
example, in the MSSM as defined in Section~\ref{sec:MSSM}, the
neutrino is massless.   There are a number of ways to extend the MSSM
to allow for massive neutrinos.  For example, by relaxing the
assumption of $R$-parity conservation, one can introduce lepton number
violating terms in the MSSM Lagrangian that can be used to incorporate
massive neutrinos that are consistent with the neutrino oscillation
data.\footnote{There is a huge literature on this subject.  See, e.g.,
  Refs.\cite{Grossman:2003gq,Dedes:2006ni,Peinado:2012tp} and the
  references contained therein.}  
Alternatively, one can start with the seesaw-extended SM and
consider its supersymmetric extension\cite{Dedes:2007ef,Hisano:1995nq,Hisano:1995cp,Grossman:1997is,Casas:2001sr,Ellis:2002fe,Masiero:2004js,Arganda:2004bz,Joaquim:2006uz,Ellis:2007wz}.

Extensions of the MSSM have also been proposed to solve a variety of
theoretical problems.  One such problem involves the $\mu$ parameter of
the MSSM.  Although $\mu$ is a SUSY-{\it preserving} parameter,
it must be of order the effective SUSY-breaking scale
of the MSSM to yield a
consistent supersymmetric phenomenology\cite{Kim:1983dt}.
Any natural solution to the so-called $\mu$-problem must incorporate
a symmetry that enforces $\mu=0$ and a small symmetry-breaking
parameter that generates a value of $\mu$ that is not parametrically
larger than the effective SUSY-breaking
scale\cite{Kim:1994eu}.  

A number of proposed mechanisms in the
literature provide
concrete examples of a natural solution to the $\mu$-problem of the
MSSM  (see, {\it e.g.},
Refs.\cite{Kim:1983dt,Kim:1994eu,Giudice:1988yz,Casas:1992mk,Dvali:1996cu}).
For example, one can replace $\mu$ by the
vacuum expectation value of a new SU(3)$\times$SU(2)$\times$U(1)
singlet scalar field.  This can be achieved by adding a singlet
chiral superfield to the MSSM.  The end result is the next-to-minimal
supersymmetric extension of the~SM, otherwise known as the NMSSM, which is
reviewed in
Refs.\cite{Maniatis:2009re,Ellwanger:2009dp}.

Ultimately, in order to determine how nature chooses to incorporate
supersymmetry, one must discover evidence for supersymmetric particles
in experiments.   The phenomenology of the MSSM and its extensions is a
huge subject that requires a separate lecture course.  Since we have no
time to present a detailed treatment of supersymmetric phenomenology
here, we can only refer the reader to some of the excellent books and review
articles on this subject (see e.g.,
Refs.\cite{Haber:1984rc,Drees,Baer,GKane,nosusy}).

As discussed in Section~\ref{sec:motivation}, supersymmetry was
proposed to avoid quadratic UV-sensitivity in a theory with elementary scalars.
To avoid a significant fine-tuning of the fundamental parameters, which
is required to explain the observed Higgs and $Z$ boson masses, the SUSY-breaking
scale should be not much larger than 1 TeV.  Consequently, experiments
currently being carried out at the Large Hadron Collider (LHC) should
be on the verge of discovering supersymmetric particles.  However, so
far no evidence for SUSY has emerged from the LHC data.

\begin{figure}[t!]
\begin{center}
\vskip -0.2in
\includegraphics[width=0.74\linewidth,angle=-90]{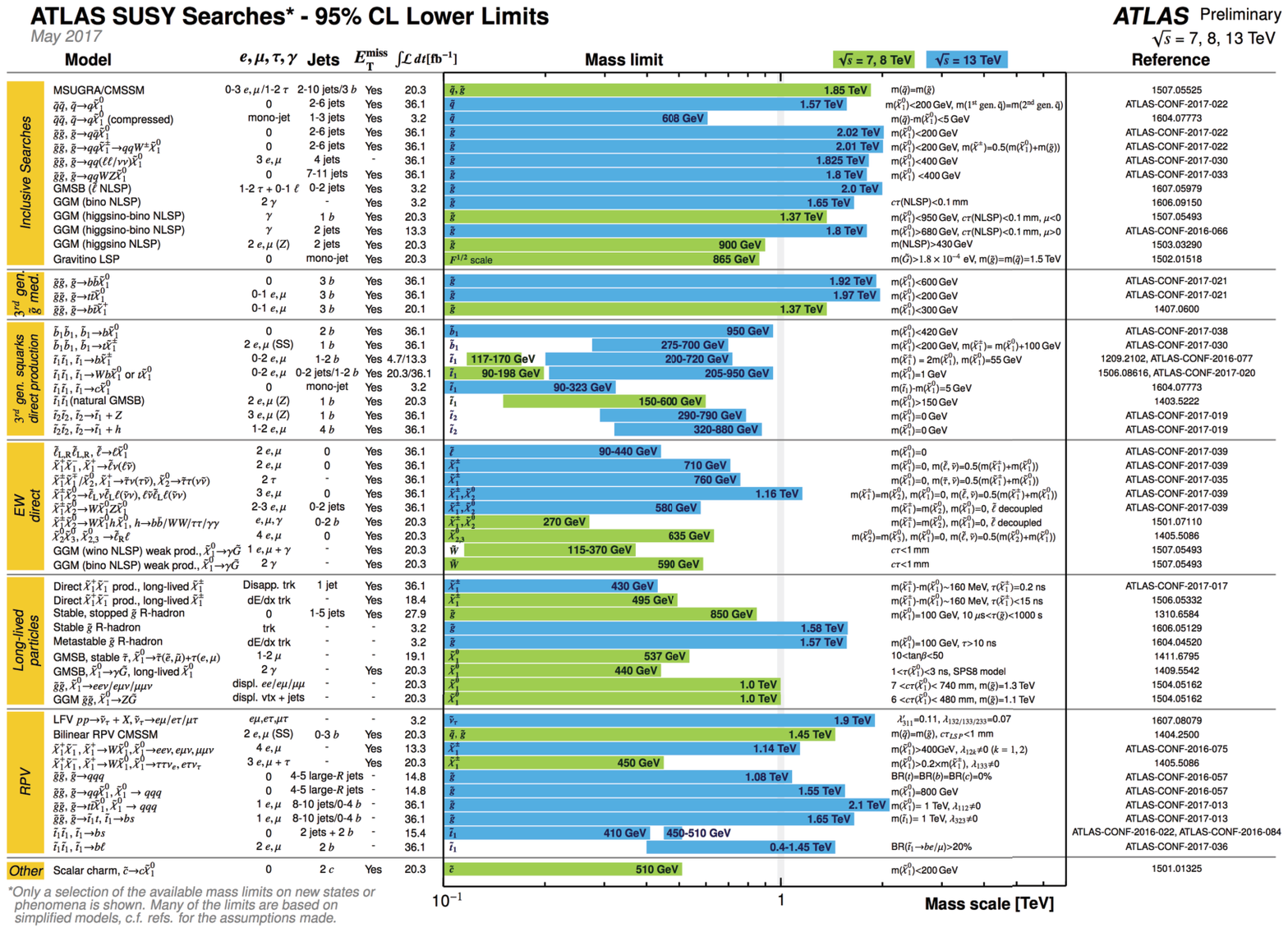}
\end{center}
\vskip -0.33in
\caption{\small Mass reach of a representative selection of ATLAS
  searches for SUSY 
as of May,
  2017.  Taken from Ref.\cite{ATLAS}.}
\label{fig:ATLAS_SUSY_Summary}
\end{figure}
\begin{figure}[h!]
\begin{center}
\vskip -0.2in
\includegraphics[width=1.02\linewidth]{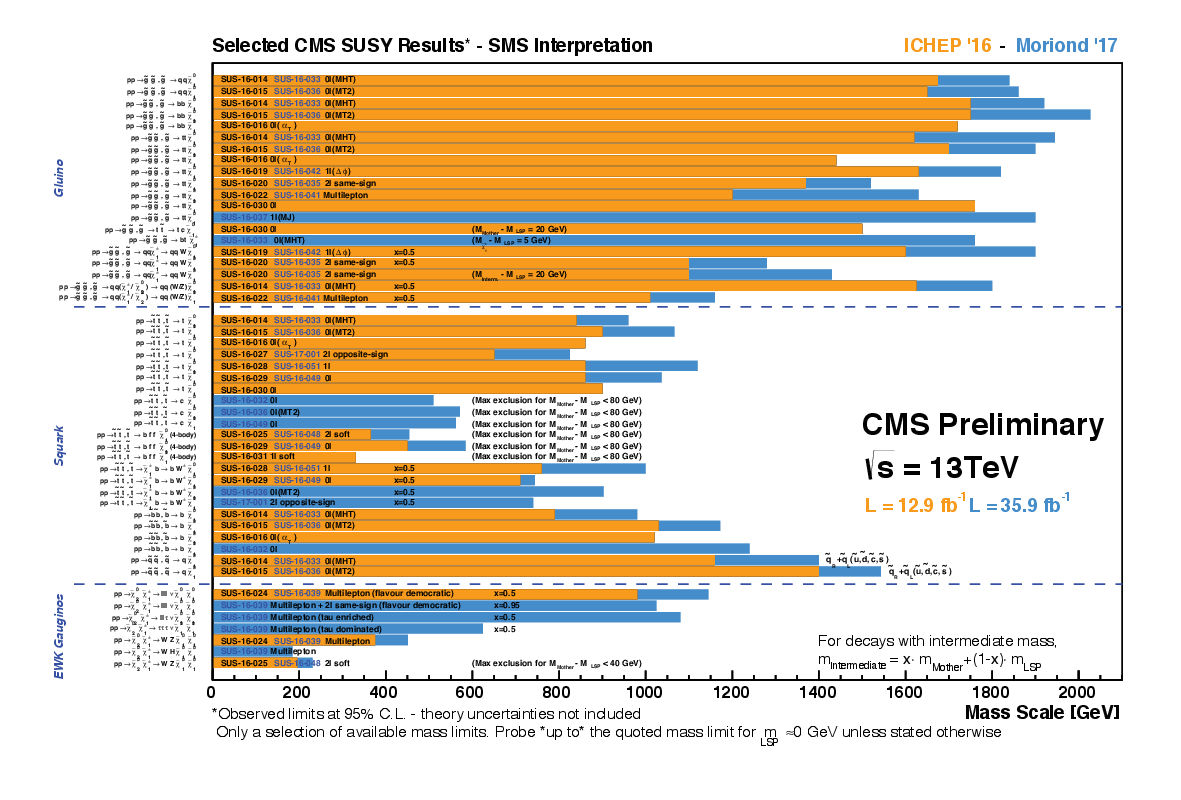}
\end{center}
\vskip -0.4in
\caption{\small Summary of exclusion limits in Simplified Model
  Spectra (SMS) from CMS searches for SUSY as of March, 2017.  Taken from Ref.\cite{CMS}.}
\label{fig:CMS_SUSY_Summary}
\end{figure}

Figs.~\ref{fig:ATLAS_SUSY_Summary}  and \ref{fig:CMS_SUSY_Summary}
summarize the limits on supersymmetric particle masses as of the
spring of 2017.  Because the LHC is a proton-proton collider, the
strongest SUSY mass bounds of about 2 TeV  are obtained for the colored superpartners (squarks
and gluinos).  
Bounds on
the top squark mass (which play an important role in assessing the degree
of fine-tuning required to accommodate the observed Higgs and $Z$
boson masses) are closer to 1 TeV.
Clearly some tension exists between the theoretical
expectations for the magnitude of the SUSY-breaking parameters
and the non-observation of supersymmetric phenomena.   
Hence, the title of this section, which is also the title of
Ref.\cite{Ross:2014mua}, where the theoretical implications of the
present LHC data for TeV-scale supersymmetry is reconsidered.

The absence of evidence for supersymmetry in the LHC data can also be
interpreted in the context of the pMSSM, which was briefly introduced
at the end of Section~\ref{sec:count}.  In a scan of the 19 parameter
pMSSM performed by the ATLAS Collaboration, the mass
of each supersymmetric particle was constrained with an upper limit
of 4 TeV, motivated to ensure a high density of models in reach of the
LHC.  Lower limits on the
supersymmetric particle masses were also applied to avoid constraints from the LEP experiments.
A summary of the sensitivity
of the ATLAS Collaboration experiment
to different types of supersymmetric particles in the 
pMSSM is shown in Fig.~\ref{fig:PMSSM}. 

\begin{figure}[t!]
\begin{center}
\includegraphics[width=\linewidth]{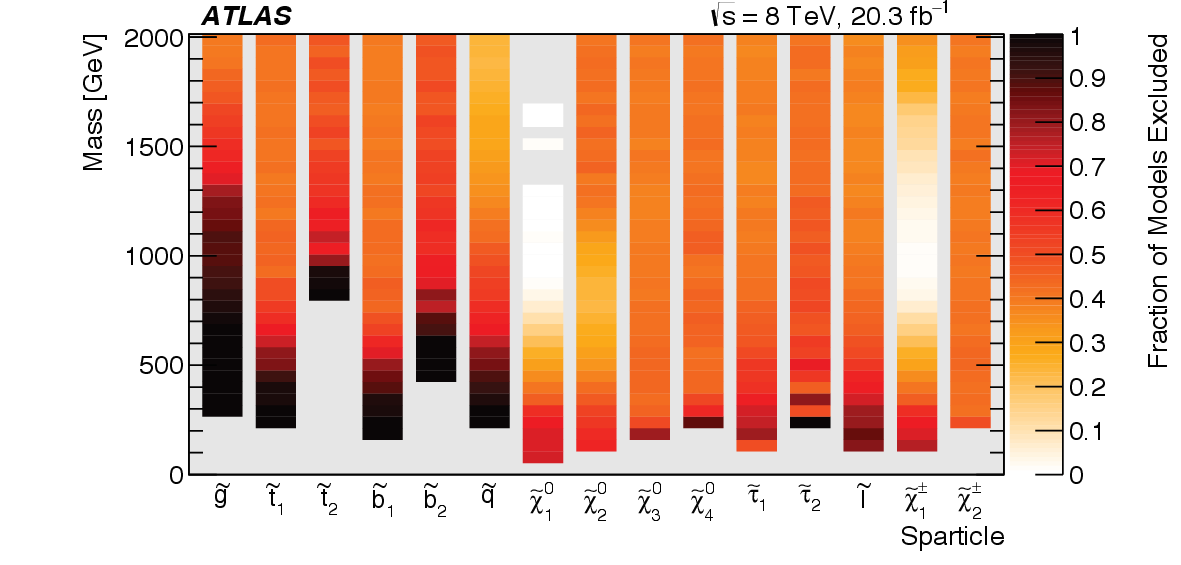}
\end{center}
\vskip -0.2in
\caption{\small A summary of the sensitivity of ATLAS to different
  types of supersymmetric particles in the 19 parameter
pMSSM. Each vertical bar is a 1D projection of the supersymmetric particle mass, with the color coding representing the
fraction of models excluded by the ATLAS searches in each bin.  This
figure taken from Ref.\cite{Fawcett:2016xoh}. 
\label{fig:PMSSM}}
\end{figure}

Of course, the LHC program is still in its infancy.  Two more years of
Run-2 data from 2017--2018 must be analyzed.  After a two year
shutdown, Run-3 follows from
2021--2023 according to the current planning schedule.   The high luminosity (HL) phase of the LHC\cite{Apollinari:2017cqg} will
commence in 2026, with an anticipation of reaching $3000~{\rm fb}^{-1}$ of
data by the year 2038.   This is a nearly 100-fold increase of the
present LHC data sample.   There is still ample room for the discovery
of SUSY at the LHC during its lifetime.

Thus, the experimental future of supersymmetry is still very much
alive.  Beyond the HL-LHC, there are possibilities of energy upgrades
at the LHC by roughly a factor of two, and considerations of the
next generation of hadron colliders with a center of mass energy of
100 TeV.   If the SUSY-breaking scale
is somewhat higher than 1~TeV (but less than say, 10~TeV), then
opportunities for discovery will be available at these future hadron
collider facilities\cite{vlhc}.

The theoretical future of supersymmetry is also quite bright.  Even if the
SUSY-breaking scale lies significantly above the TeV-scale, there are
still many opportunities for incorporating supersymmetry into the
fundamental theory of particle physics.  In this latter scenario, SUSY would
not be relevant for explaining the origin of the scale of electroweak
symmetry breaking.
Another explanation would be required, perhaps
one of the other suggested theoretical approaches mentioned in
Section~\ref{sec:motivation}.

For example, it may still be possible that some remnant
of the supersymmetric particle spectrum
survives down to the TeV-scale or below.
This is the idea of split-SUSY\cite{ArkaniHamed:2004fb,Giudice:2004tc,ArkaniHamed:2004yi,Wells:2004di,Arvanitaki:2012ps,ArkaniHamed:2012gw},
in which the squarks and sleptons
are significantly heavier
(perhaps by many orders of magnitude) than 1~TeV, whereas the fermionic 
superpartners of the gauge and Higgs bosons
may be kinematically accessible at the LHC.
Of course, the SUSY-breaking dynamics
responsible for such a split-SUSY spectrum would
not 
be related
to the origin of the scale of electroweak symmetry breaking.
Nevertheless, models of split-SUSY
can account for the dark matter (which is assumed to be the LSP
gaugino or higgsino) and gauge coupling unification, thereby
preserving two of the desirable features of TeV-scale supersymmetry.

There are many theoretical aspects of supersymmetry theory that lie beyond
the scope of these lectures but deserve further exploration.
Among these, non-perturbative approaches to supersymmetric theories,
such as holomorphy and Seiberg duality, 
have been particularly fruitful.
The power of holomorphy was briefly exhibited in
Section~\ref{sec:non-renorm}, when we reviewed Seiberg's proof of the
non-renormalization of the superpotential\cite{SeibergNR}.  There are
many other applications of holomorphy, such as  
the computation of exact $\beta$ functions in supersymmetric gauge theories\cite{Shifman:1986zi,Shifman:1991dz,ArkaniHamed:1997mj}. 
As an effective tool in non-perturbative regimes, Seiberg duality
elucidates strongly coupled gauge theories by relating them to dual
weakly coupled gauge theories.  
In Refs.\cite{Terning,Intriligator:2007cp,Terning:2003th,Shifman:1995ua,Shifman}
one can find numerous applications to the study of non-perturbative
dynamics in strongly-coupled 
supersymmetric theories and in fundamental models of SUSY-breaking. 

Another flourishing area of research 
is that of scattering amplitudes
 \cite{Henn,Elvang:2015rqa},
where novel methods are being developed to 
facilitate
 computations that were previously intractable
using the traditional Feynman-diagrammatic approach.
Here supersymmetric theories can serve as testing grounds for 
 techniques
that may eventually be extended to non-supersymmetric
quantum field theories.
For example, in $N=4$ supersymmetric Yang-Mills theory (one of the few
known examples of a finite quantum field theory in four spacetime dimensions), amplitudes are well understood, making it a relatively simple arena in which to study new computational methods\cite{Spradlin:2015kaz}.
Moreover, tree-level gluon scattering amplitudes in $N=4$ super Yang-Mills are identical to those in any other gauge theory,
so it is reasonable to expect that methods developed for SUSY gauge
theories could be adapted to the computation of QCD amplitudes.

Supersymmetry is also a powerful tool for analyzing a variety of
problems in mathematical physics,
and plays a critical role in the formulation of string theory
\cite{MDine,Green,Polchinski,becker,kiritis,barton,ibanez,BLT,schomerus}. 
Evidently, 
even in the absence of evidence for SUSY at the TeV
scale, it is very likely that supersymmetry will lead to important new insights,
both in experimental and theoretical directions.  With this in mind,
it is our hope that these lectures have provided a modest introduction
into the fascinating world of supersymmetry.

\section*{Acknowledgments}
%%%
We would like to thank Zackaria Chacko, Andrew Cohen, Michael Dine, Herbi Dreiner, Stephen Martin, Raman Sundrum, and John Terning for many enlightening discussions.
H.E.H. is grateful to Rouven Essig and Ian Low for their invitation to present these lectures at TASI 2016, and their patience in waiting for these lecture notes to be completed.   This work is supported in part by the U.S. Department of Energy grant number DE-SC0010107.
L.S.H. is also supported by the Israel
Science Foundation under grant no.~1112/17.

%\blankpage
%\printindex[aindx]                 % to print author index
%\printindex                         % to print subject index

\end{document}